\documentclass[12pt,a4paper,twoside,headinclude,DIV=12,BCOR=0.9cm,ngerman,fleqn,parskip=half,openright]{scrbook} 

\title{\"Uber die Pr\"azision interprozeduraler Analysen}
\author{Dorothea Jansen}
\date{Dezember 2010}

\usepackage{etex}

\usepackage[ngerman]{babel}

  \usepackage[colorlinks,pdftex, plainpages=false,hypertexnames=false]{hyperref}
  \hypersetup{	
  pdfauthor 		= Dorothea Jansen, 
  pdftitle		= Über die Präzision interprozeduraler Analysen,
  pdftoolbar 		= false, 
  plainpages 		= false, 
  pdfdisplaydoctitle 	= true,
  colorlinks,		
  citecolor		= black,
  filecolor		= black,
  linkcolor		= black,
  urlcolor		= black
  }

  \usepackage{fancyhdr}
  \renewcommand{\chaptermark}[1]{\markboth{\thechapter \ #1}{}}
  \renewcommand{\sectionmark}[1]{\markright{\thesection  \ #1}}
  \fancypagestyle{plain}{
	\fancyhead{}
	
  }

\usepackage[tracking=true]{microtype}
  \DeclareMicrotypeSet*[tracking]{my}%
  { font = */*/*/sc/* }%
  \SetTracking{ encoding = *, shape = sc }{ 45 }

  \usepackage[utf8]{inputenc}	

  \usepackage{amscd} 			
  \usepackage{amsmath} 			
  \usepackage{amssymb}			%
  \usepackage{amsthm}			%

  \usepackage{color, colortbl}%

  \usepackage{pstricks,pst-node,pst-plot}
  \usepackage{pgf}
  \usepackage{tikz}
	\usetikzlibrary{automata}
	\usetikzlibrary{arrows,decorations.pathmorphing,backgrounds,positioning,fit}
 \psset{arrows=->,
		 unit=1.5cm, 
		 arcangle=20,
		 labelsep=2.5pt,
		 subgriddiv=10,
		 shortput=nab}

  \usepackage{listings} 		
  \usepackage{multirow}
  \usepackage{stmaryrd}  		
  \usepackage{textcomp} 		
  \usepackage{todo}
  \usepackage{xspace} 			
  \usepackage[arrow, matrix, curve]{xy} 

   \renewcommand{\labelenumi}{\alph{enumi})}

  \renewcommand{\thesection}{\arabic{chapter}.\arabic{section}}

\newcommand{\quotes}[1]{\glqq#1\grqq\xspace}
\def\obda{ohne Beschränkung der Allgemeinheit\xspace}
\def\widening{\triangledown}
\def\node{\nu}

	\newcounter{mysection} 
	\newtheorem{satz}{Satz}[chapter]
	\newtheorem{lemma}[satz]{Lemma}
	\newtheorem{kor}[satz]{Korollar}
	
	\theoremstyle{definition}
	\newtheorem{dfn}[satz]{Definition}
	\newtheorem{bsphelp}[satz]{Beispiel}	\newenvironment{bsp}{\begin{bsphelp}}{\hfill$\boxbox$\end{bsphelp}}
	\newtheorem{bemhelp}[satz]{Bemerkung}	\newenvironment{bem}{\begin{bemhelp}}{\hfill$\boxbox$\end{bemhelp}}

\renewcommand{\labelenumi}{\alph{enumi})}

\DeclareMathOperator{\nn}{\mathbb{N}}
\DeclareMathOperator{\zz}{\mathbb{Z}}
\DeclareMathOperator{\rr}{\mathbb{R}}

\DeclareMathOperator{\eps}{\varepsilon}
\renewcommand{\phi}{\varphi}

\DeclareMathOperator{\Mat}{Mat}	
\DeclareMathOperator{\pr}{pr}

\newcommand{\mymatrix}[1]{\left(\begin{smallmatrix}#1\end{smallmatrix}\right)}
\newcommand{\wminus}{{\color{white}-}} 

\DeclareMathOperator{\Proc}{\mathtt{Proc}}
\DeclareMathOperator{\lab}{\mathtt{Label}}
\newcommand{\var}[1]{\texttt{x}_{\texttt{#1}}}

\newcommand{\sqleq}{\sqsubseteq}
\newcommand{\sqgeq}{\sqsupseteq}
\DeclareMathOperator{\sqleqmap}{\overline{\sqleq}}
\DeclareMathOperator{\sqgeqmap}{\overline{\sqgeq}}
\DeclareMathOperator{\sqleqsharpmap}{\overline{\sqleq^\sharp}}
\DeclareMathOperator{\sqgeqsharpmap}{\overline{\sqgeq^\sharp}}
\DeclareMathOperator{\alphamap}{\overline{\alpha}}
\DeclareMathOperator{\gammamap}{\overline{\gamma}}
\DeclareMathOperator{\alphacs}{\alpha}
\DeclareMathOperator{\gammacs}{\gamma}
\DeclareMathOperator{\lfp}{lfp}	
\DeclareMathOperator{\gfp}{gfp}	

\newcommand{\Var}{\mathtt{Var}}	
\DeclareMathOperator{\init}{\ensuremath{\mathtt{init}}}	
\DeclareMathOperator{\call}{\mathtt{call}}	
\newcommand{\p}{\ensuremath{\mathtt{p}}}
\newcommand{\q}{\ensuremath{\mathtt{q}}}

\DeclareMathOperator{\ret}{\mathtt{ret}}
\DeclareMathOperator{\enter}{\mathtt{enter}}
\DeclareMathOperator{\main}{\mathtt{main}}
\DeclareMathOperator{\CS}{\ensuremath{CS}}
\DeclareMathOperator{\id}{id}
\DeclareMathOperator{\SLP}{SLP}
\DeclareMathOperator{\Path}{P}
\DeclareMathOperator{\GP}{GP}
\DeclareMathOperator{\mats}{\Mat(\Sigma)}		\DeclareMathOperator{\matspot}{2^{\mats}}
\DeclareMathOperator{\rels}{\rr^n\times\rr^n}	\DeclareMathOperator{\relspot}{2^{\rels}}
\renewcommand{\subset}{\subseteq}
\DeclareMathOperator{\VerbInt}{(L_\text{Int},\sqleq_\text{Int})}
\DeclareMathOperator{\VerbIA}{(L_\text{IA},\sqleq_\text{IA})}
\DeclareMathOperator{\MOP}{MOP}
\DeclareMathOperator{\MFP}{MFP}
\DeclareMathOperator{\ugs}{\mathcal{U}}
\DeclareMathOperator{\ugssharp}{\mathcal{U}^\sharp}
\DeclareMathOperator{\abbugs}{\mathcal{F}_{\ugs}}

  \pgfkeys{/nodestyle/.style={circle,draw=#1,thick,minimum size = 2.5em}}
  
  \definecolor{javaGreen}{rgb}{0.25,0.5,0.375}
  \definecolor{javaBlue}{rgb}{0.1,0.1,0.9}

  \definecolor{darkRed}{rgb}{0.6,0,0} 
  \definecolor{darkGreen}{rgb}{0,0.6,0} 
  \definecolor{darkBlue}{rgb}{0,0,0.6} 

  \definecolor{middleRed}{rgb}{0.75,0,0} 
  \definecolor{middleGreen}{rgb}{0,0.75,0} 
  \definecolor{middleBlue}{rgb}{0,0,0.75} 

  \colorlet{stdcolor}{black} 		
  \colorlet{fadecolor}{stdcolor!50}	
  \colorlet{annocolor}{middleBlue} 		
  \colorlet{callstringcolor}{middleGreen}	
  \colorlet{edgenamecolor}{fadecolor} 	

\newenvironment{flowgraph}
 {\begin{center}
  \begin{tikzpicture}
   [auto,
	stdnode/.append style={/nodestyle=stdcolor}, 					
	startnode/.append style={stdnode,initial above, initial text=}, 
	endnode/.append style={stdnode,accepting}, 						
	fadenode/.append style={/nodestyle=fadecolor,text=fadecolor}, 	
	fadearc/.style={fadecolor}, 									
	snakearc/.style={decorate,decoration={snake,amplitude=.4mm,segment length=2mm}}, 
	callstring/.style={text=callstringcolor},						
	annnode/.style={text=annocolor}, 								
	annarc/.style={annocolor, dashed}, 								
   ]
 }{
  \end{tikzpicture}
  \end{center}
 }

\definecolor{ourGreen}{rgb}{0.25,0.5,0.375} 	
\definecolor{ourBlue}{rgb}{0.1,0.1,0.9} 	
\definecolor{ourPurple}{rgb}{0.5,0.0,0.33} 	
\definecolor{ourGrey}{rgb}{0.97,0.97,0.97} 	
\definecolor{l1}{rgb}{1.0,0,0}			
\definecolor{l2}{rgb}{1.0,0.5,0.0}		
\definecolor{l3}{rgb}{1.0,0.9,0}		
\definecolor{l4}{rgb}{0.25,0.5,0.375}		
\definecolor{l5}{rgb}{0.1,0.1,0.9}		
\definecolor{l6}{rgb}{0.54,0.17,0.88}		

\newlength{\tabwidth}
\newlength{\lastcol}

\lstnewenvironment{pseudocode} 
  {\colorlet{commentcolor}{gray}
   \lstset{	
	basicstyle=\ttfamily, 			
	showspaces=false,			
	morekeywords={forall,while,if,with},	
	keywordstyle={\sffamily\bfseries},	
	identifierstyle=, 			
	morecomment=[l]{//},
	commentstyle=\color{commentcolor},	
	showstringspaces=false,			
	breaklines=true, 			
	mathescape=true,			
  }
}{}

\lstnewenvironment{expl}[1] 
  {\stepcounter{example}\lstset{
	language=java, 				
	basicstyle = \small\ttfamily,		
	showspaces=false			
	keywordstyle=\bfseries\color{ourPurple},
	identifierstyle=, 			
	commentstyle=\color{ourGreen}, 		
	stringstyle=\ttfamily, 			
	showstringspaces=false,			
	breaklines=true, 			
	numbers = left,				
	numberstyle = \tiny, 			
	numbersep = 6.5 pt,			
	xleftmargin = 4 mm,			
	xrightmargin = 1.5 mm,			
	frame = single, 			
	title=\small{Example \thechapter.\theexample: #1}	
	}
}{}

\lstnewenvironment{code} 
  {\lstset{
	language=java, 				
	basicstyle = \ttfamily	,		
	showspaces=false			
	keywordstyle=\bfseries\color{ourPurple},
	identifierstyle=, 			
	commentstyle=\color{ourGreen}, 		
	stringstyle=\ttfamily, 			
	showstringspaces=false,			
	breaklines=true, 			
	numbers=none, 				
	frame=none, 				
	caption=, 				
	}
}{}

\lstnewenvironment{editor}[1] 
  {\lstset{
	language=ruby,
	basicstyle = \small,			
	numbers = left,				
	numberstyle = \tiny, 			
	numbersep = 6.5 pt,			
	xleftmargin = 4 mm,			
	xrightmargin = 1.5 mm,			
	frame = single, 			
	title=\small{Syntax #1}
	}
}{}

\begin{document}


\begin{titlepage}
\centering
{\includegraphics[width=10.5cm,height=2.25cm]{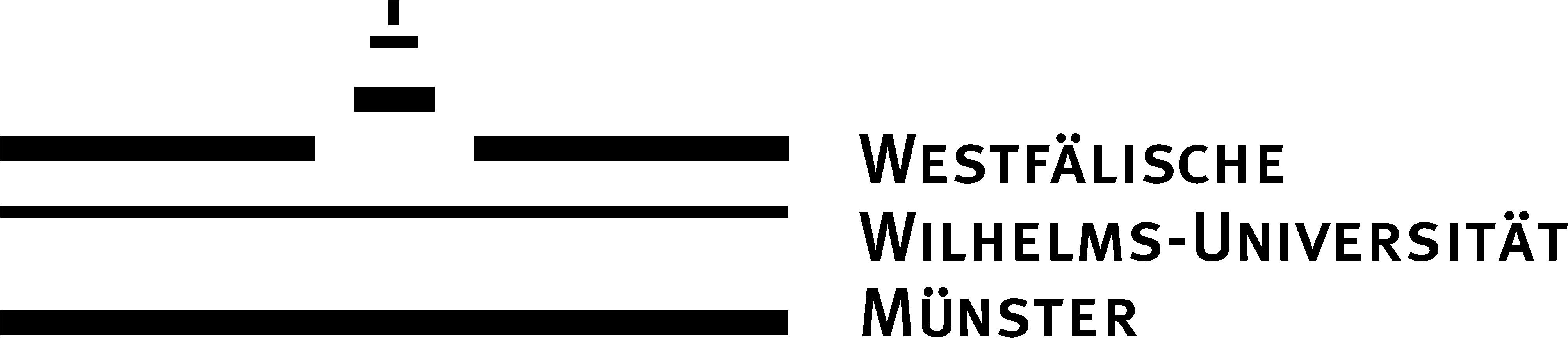}}
\vspace{9em}

{
  \LARGE 
  Diplomarbeit

  im Fach Informatik
}

\vspace{5em}

\textbf{\Huge
  Über die Präzision \\ \vspace{0.25em}			
  interprozeduraler Analysen
}

\vspace{12em}
{
  \LARGE 
  Dorothea Jansen
}

\vspace{3em}

\end{titlepage}		

\pagestyle{empty}
\cleardoublepage

\begin{titlepage}
\enlargethispage{1cm}
\centering
{\includegraphics[width=10.5cm,height=2.25cm]{./schloss.jpg}}

\vspace{7em}

{
  \LARGE 
  Diplomarbeit

  im Fach Informatik
}

\vspace{5em}

\textbf{\Huge
  Über die Präzision \\ \vspace{0.25em}			
  interprozeduraler Analysen
}

\vspace{5em}
{
  \LARGE 
  Dorothea Jansen
}

\vspace{3em}

{
  \large 
  Dezember 2010
}

\vfill
\begin{flushleft}
{\large
  Betreuer: \\
  Prof. Dr. Markus Müller-Olm \\
  Institut für Informatik \\
  Westfälische Wilhelms-Universität Münster \\
  \
}
\end{flushleft}
\end{titlepage}

\pagestyle{empty}\cleardoublepage

\setcounter{page}{0} 
\pagenumbering{roman}


\frontmatter  
\fancyhf{} \fancyhead[LE,RO]{\bfseries \thepage}
\pagestyle{fancy}
\renewcommand{\chaptermark}[1]{\markboth{\thechapter \ #1}{}}
\renewcommand{\sectionmark}[1]{\markright{\thesection  \ #1}}


\setcounter{tocdepth}{3}
\fancyhead[RE]{\bfseries Inhaltsverzeichnis} 
\fancyhead[LO]{\bfseries Inhaltsverzeichnis}
\enlargethispage{2cm} 
\small \tableofcontents \normalsize


\pagestyle{fancy}
\renewcommand{\chaptermark}[1]{\markboth{\thechapter \ #1}{}}
\renewcommand{\sectionmark}[1]{\markright{\thesection  \ #1}}

\fancyhf{} \fancyhead[LE,RO]{ \bfseries \thepage}
\fancyhead[RE]{\bfseries \leftmark} 
\fancyhead[LO]{\bfseries \rightmark}

\mainmatter


\chapter{Einleitung}\label{chapt-einf}
In der Programmanalyse wird ein Programm im Hinblick auf bestimmte Eigenschaften untersucht. 
Untersuchungsgegenstand kann dabei beispielsweise die Frage sein, 
ob es Ausdrücke im Programm gibt, die an einer Stelle stets den gleichen Wert haben, 
oder ob das Programm überflüssige Zuweisungen enthält. 
Es können aber aber auch die Werte berechnet werden, welche die Programmvariablen annehmen können, 
oder Gleichungen oder Ungleichungen über diese Werte aufgestellt werden.
Derartige Programmanalysen ermöglichen zum einen \emph{Programmoptimierung}, 
indem überflüssige Zuweisungen eliminiert 
oder konstante Ausdrücke durch diese Konstanten ersetzt werden. 
Andererseits kann Programmanalyse auch für \emph{Validierung} benutzt werden, 
indem beispielsweise überprüft wird, ob Variablen nur \quotes{erlaubte} Werte annehmen. 
Programmanalyse geschieht zur Compile-Zeit, es ist also keine Ausführung des Programmes notwendig.

Ein Bereich der \emph{Programmanalyse} ist die sogenannte 
\emph{Constraint-basierte Datenflussanalyse}, welche auch Gegenstand dieser Arbeit ist. 
Bei der Datenflussanalyse fasst man ein Programm als ein System von Graphen auf, 
deren Kanten mit den Anweisungen des Programms beschriftet sind, 
wobei jede Prozedur des Programms einem Graphen entspricht. 
Besitzt das Programm nur eine Prozedur, so wählt man zu jeder Kante eine \emph{Transferfunktion}, 
die auf dem zugrundeliegenden Datenraum operiert. 
Durch die Kanten und diese Abbildungen wird also beschrieben, 
wie sich Daten entlang verschiedener Ausführungspfade verändern. 
Diese Veränderungen werden dann durch \emph{Ungleichungen}, also \emph{Constraints}, ausgedrückt. 
Anhand dieser wiederum werden Daten für die einzelnen Programmpunkte bestimmt. 

Eine solche \emph{intraprozedurale} Analyse, 
also die Analyse eines Programmes mit nur einer Prozedur, 
benötigt für jede Kante eine Transferfunktion. 
Besitzt ein Programm dagegen mehrere Prozeduren, 
so gibt es neben solchen \emph{Basisanweisungen} auch noch \emph{Prozeduraufrufe}. 
Solche Programme sind viel mächtiger als diejenigen mit nur einer Prozedur, 
da sie Rekursion beinhalten. Um diese Programme zu analysieren, 
müssen also auch Prozeduraufrufe behandelt werden. 
Zwei Ansätze, wie dies geschehen kann, 
wurden von Sharir und Pnueli in \cite{sharir-pnueli} vorgestellt: 
Der \emph{funktionale Ansatz} berechnet dabei zunächst für jede Prozedur eine Transferfunktion, 
die den Effekt der Prozedur beschreibt, 
also die Veränderung der Daten vom Start- bis zum Endknoten der Prozedur. 
Der \emph{Call-String-Ansatz} dagegen behandelt die verschiedenen Prozeduren wie eine einzige, 
indem konzeptionell jede Aufrufkante durch den Graphen der aufgerufenen Prozedur ersetzt wird. 

Gegenstand dieser Arbeit ist die Untersuchung, wie genau die Informationen sind, 
die diese beiden Ansätze in verschiedenen Situationen liefern. 
In der Praxis können die theoretisch besten Lösungen einer gegebenen Analyse oft 
nicht endlich dargestellt oder berechnet werden 
oder es ist keine geeignete Analyse zur Betrachtung einer gegebenen Fragestellung bekannt. 
Deshalb betrachten wir nicht nur die jeweils theoretisch besten Lösungen, 
sondern untersuchen auch die Lösungen, wenn zusätzlich \emph{abstrakte Interpretationen} 
oder \emph{Widening-Operatoren} verwendet werden. 
Durch \emph{abstrakte Interpretation} wird anstelle der konkreten Information 
eines gegebenen Ungleichungssystems eine abstrakte Information berechnet. 
Damit können neue Analysen konstruiert werden, 
die die gestellte Frage beantworten oder deren Lösung endlich darstellbar und berechenbar ist. 
Mithilfe von \emph{Widening-Operatoren} 
kann Terminierung des \emph{Workset-Algorithmus} erzwungen werden. 
Der \emph{Workset-Algorithmus} ist einer Mittel zur Berechnung 
der Lösung des Ungleichungssystems einer constraint-basierten Datenflussanalyse.

Nach der Einführung einiger grundlegender Begriffe und Methoden in Kapitel \ref{chap:grund} 
werden wir in Kapitel \ref{chap:PA} zunächst den funktionalen und den Call-String-Ansatz vorstellen 
und ihre jeweils \quotes{besten} Lösungen miteinander verglichen. 
Weiter wird die Analyse interprozeduraler Programme am Beispiel der 
\emph{Polyederanalyse} von Cousot und Halbwachs \cite{CH78-POPL} veranschaulicht, 
die den funktionalen Ansatz variiert. 
Wir zeigen, dass die Berechnung der Lösungen des funktionalen und des Call-String-Ansatzes 
mithilfe des {Workset-Algorithmus} nicht immer terminiert. 
Darum stellen wir in Kapitel \ref{chap:abstr} und \ref{chap:widening} 
zwei von Cousot und Cousot \cite{CousotCousot76-1,CousotCousot77-1} bekannte Ansätze vor, 
mit denen eine solche Terminierung erreicht werden kann: 
In Kapitel \ref{chap:abstr} betrachten wir abstrakte Interpretationen und 
vergleichen die Lösungen des funktionalen und des Call-String-Ansatzes 
unter Verwendung von abstrakter Interpretation. 
Außerdem werden wir das Beispiel der Polyederanalyse fortsetzen. 
In Kapitel \ref{chap:widening} untersuchen wir anschließend {Widening-Operatoren} 
und stellen wiederum die Frage nach dem Vergleich der Lösungen 
des funktionalen und des Call-String-Ansatzes. 
Zuletzt folgen in Kapitel \ref{chap:fazit} eine Zusammenfassung der Ergebnisse 
sowie ein Ausblick auf weitere Untersuchungen, die auf diese Arbeit aufbauen könnten.


\chapter{Grundlagen}\label{chap:grund}
In diesem Kapitel werden Grundlagen zur Datenflussanalyse vorgestellt. 
Bei einer Datenflussanalyse wird zu jedem Punkt eines Programms eine gesuchte Information berechnet.
Eine solche Information kann die Menge der möglichen Werte der Programmvariablen sein 
oder ob eine Variable in einem Punkt stets denselben Wert hat. 
Wie sich eine Information durch eine Programmanweisung ändert, 
wird mithilfe von Ungleichungen beschrieben. 
Zur Berechnung der gesuchten Informationen wird ein \emph{Ungleichungssystem} aufgestellt, 
das die Informationen an den verschiedenen Punkten miteinander in Verbindung setzt. 
Für einfache Anweisungen von Programmen stellen wir bereits in diesem Kapitel Ungleichungen auf, 
während wir komplexere Anweisungen in \autoref{chap:PA} behandelt werden. 
Weiter werden wir in diesem Kapitel eine Methode vorstellen, 
mit denen ein Ungleichungssystem gelöst werden kann, den sogenannten \emph{Workset-Algorithmus}. 
Wir werden im Folgenden meist \quotes{Analyse} anstelle von 
\quotes{(constraint-basierter) Datenflussanalyse} sagen.

Wir beginnen damit, einige Begriffe und theoretische Grundlagen einzuführen. 
Diese basieren, sofern nicht anders gekennzeichnet, auf Nielson et al.~\cite{Niel}.

\section{Vollständige Verbände und Fixpunkte}\label{sec:fixpkte}
In der Datenflussanalyse wird jedem Programmpunkt eine \emph{Information} zugeordnet. 
Um mit dem Begriff der \quotes{Information} formal arbeiten zu können, 
fassen wir Informationen als Elemente einer geeigneten Menge $L$ auf. 
Diese Informationen wollen wir außerden im Hinblick darauf untersuchen, 
welche korrekte Information präziser ist als eine andere. 
Ist zum Beispiel $x=2$, so ist die Information \quotes{Der Wert von $x$ ist eine ganze Zahl} 
weniger präzise als die Information \quotes{Der Wert von $x$ ist eine gerade positive Zahl}. 
Beide Informationen sind korrekt. 
Die Aussage \quotes{Der Wert von $x$ ist $4$.} ist wiederum präziser, aber nicht korrekt. 
Die präziseste korrekte Information ist \quotes{Der Wert von $x$ ist $2$.}. 
Von Interesse sind also nur korrekte Informationen, die dabei möglichst präzise sind. 
Formal wird dieser Präzisionsbegriff durch eine \emph{Halbordnung} beschrieben. 
Eine Information heißt dann präziser, je kleiner sie bezüglich dieser Halbordnung ist. 

\begin{dfn}[Halbordnung]
Eine \emph{Halbordnung}\index{Halbordnung} $(L,\sqleq)$ ist eine Menge $L$ 
zusammen mit einer Relation $\sqleq\ \subseteq L \times L$ mit folgenden Eigenschaften:
\begin{itemize}
 \item Die Relation $\sqleq$ ist reflexiv, 
 d.h.~für alle $l \in L$ ist $l \sqleq l$.
 \item Die Relation $\sqleq$ ist transitiv, 
 d.h.~für alle $l,l',l'' \in L$ mit $l \sqleq l'$ und $l' \sqleq l''$ ist $l \sqleq l''$.
 \item Die Relation $\sqleq$ ist antisymmetrisch, 
 d.h.~für alle $l,l' \in L$ mit $l \sqleq l'$ und $l' \sqleq l$ ist $l = l'$.
\end{itemize}
Dabei schreiben wir $l \sqleq l'$ für $(l,l')\in~\sqleq$. 
Ist $l \sqleq l'$ mit $l \ne l'$, so schreiben wir auch $l \sqsubset l'$. 
Entsprechend schreiben wir $l' \sqgeq l$ für $l \sqleq l'$ und $l' \sqsupset l$ für $l \sqsubset l'$.
\end{dfn}

\begin{bsp}
Sei $X$ eine Menge und $L:=2^X$ die Potenzmenge von $X$. 
Dann ist $(L,\subseteq)$ eine Halbordnung. 
Dies folgt direkt aus der Definition der Teilmenge.
\end{bsp}

Ein auch im Folgenden wichtiges Beispiel ist das der \emph{Intervallanalyse} aus \cite{CousotCousot76-1}, für das wir nun zunächst den Verband der Intervalle einführen.
\begin{bsp}
Wir definieren 
\[\VerbInt := 
(\{[l,u] \mid l \in \zz \cup \{-\infty\}, u \in \zz \cup \{+\infty\}, l \le u\}
\cup
\{\emptyset\},\subseteq)\] 
als die Menge der beschränkten und unbeschänkten Intervalle in den ganzen Zahlen. 
Dies ist offenbar eine Halbordnung.
Wir nennen $\VerbInt$ \emph{Intervallverband}. 

Weiter betrachten wir $\VerbIA$ mit $\VerbIA := (\Var \to \VerbInt, \sqleq_\text{IA})$ 
für eine endliche Menge $\Var$ von Variablen. 
Dabei ist $\rho \sqleq_\text{IA} \rho'$ für $\rho,\rho' \in \VerbIA$, 
wenn $\rho(x) \sqleq_\text{Int} \rho'(x)$ für jede Variable $x \in \Var$ gilt. 
Man sieht leicht, dass dies eine Halbordnung definiert. 
Diese nennen wir \emph{Intervallanalyseverband}.
\end{bsp}
Den Begriff \quotes{Verband} werden wir in Kürze einführen und auch zeigen, 
dass der Intervall- und der Intervallanalyseverband tatsächlich Verbände sind.

Im Rahmen der Datenflussanalyse werden für einen Programmpunkt oft 
verschiedene Informationen berechnet. 
Da einem jeden Punkt aber eine einzige Information 
und nicht eine Menge von Informationen zugeordnet werden soll, 
muss eine neue Information bestimmt werden. 
Diese muss alle Informationen der berechneten Menge enthalten, dabei aber möglichst präzise sein. 
Der Vorgang, diese neue Information zu berechnen, 
entspricht der Bildung der \emph{kleinsten oberen Schranke} der Informationsmenge:

\begin{dfn}
Seien $(L,\sqleq)$ eine Halb\-ordnung und $A \subseteq L$. Dann heißt $l\in L$ 
\begin{itemize}
\item \emph{obere Schranke}\index{obere Schranke} von $A$, 
 falls $a \sqleq l$ für alle $a \in A$ gilt. 
\item \emph{untere Schranke}\index{untere Schranke} von $A$, 
 falls $l \sqleq a$ für alle $a \in A$ gilt. 
\item \emph{kleinste obere Schranke}\index{obere Schranke!kleinste} von $A$, 
 falls $l$ eine obere Schranke von $A$ ist 
 und für jede weitere obere Schranke $l'$ von $A$ bereits $l \sqleq l'$ gilt.
\item \emph{größte untere Schranke}\index{untere Schranke!größte} von $A$, 
 falls $l$ eine untere Schranke von $A$ ist 
 und für jede weitere untere Schranke $l'$ von $A$ bereits $l' \sqleq l$ gilt.
\end{itemize}
\end{dfn} 
Sofern sie existiert, ist die kleinste obere Schranke von $A$ eindeutig 
und wird mit $\bigsqcup A$ bezeichnet. 
Für $\bigsqcup\{a,b\}$ schreiben wir auch $a \sqcup b$.

\begin{bsp}
\label{bsp-schranken-potenzverband}
Sei $X$ eine Menge und $L:=2^X$. Seien $A,B \subseteq X$. 
Dann ist $A \cup B$ die kleinste obere Schranke von $\{A,B\}$: 
Offenbar gilt $A,B \subseteq A \cup B$. 
Somit ist $A \cup B$ eine obere Schranke. 
Ist $S$ eine weitere obere Schranke, so ist $A,B \subseteq S$ 
und damit auch $A \cup B \subseteq S$ und damit ist $A \cup B$ schon die kleinste obere Schranke.

Ebenso bestimmt man die kleinste obere Schranke einer beliebigen Teilmenge von $L$: 
Ist $\mathcal{A} \subseteq 2^X$, so ist $\bigcup\big\{A \mid A \in \mathcal{A}\big\}$ 
die kleinste obere Schranke von $\mathcal{A}$.

Betrachte nun den Intervallverband $\VerbInt$. 
Seien $I_1 = [l_1,u_1]$ und $I_2 = [l_2,u_2]$ zwei Intervalle in $\VerbInt$. 
Setze $l:=\min\{l_1,l_2\}$ und $u:=\max\{u_1,u_2\}$. 
Dann ist $I:=[l,u]$ die kleinste obere Schranke von $\{I_1,I_2\}$: 
Nach Definition ist $l \le l_i$ und $u \ge u_i$ für $i=1,2$. 
Damit ist auch $[l_i,u_i] \subseteq [l,u]$ und $I$ ist eine obere Schranke. 
Ist $I'=[l',u']$ eine weitere obere Schranke, so gilt $l' \le l_i$ und $u' \ge u_i$ für $i=1,2$. 
Damit ist auch $l' \le \min\{l_1,l_2\}=l$ und $u' \ge \max\{u_1,u_2\}=u$ und so $I' \supseteq I$. 
Also ist $I$ sogar die kleinste obere Schranke von $\{I_1,I_2\}$.

Für eine beliebige Teilmenge $\mathcal{A} \subseteq \VerbInt$ setze
\begin{align*}
l &:= \inf\{l' \mid \exists u' \in \zz\cup\{+\infty\}: [l',u'] \in \VerbInt\} \\
u &:= \sup\{u' \mid \exists l' \in \zz\cup\{-\infty\}: [l',u'] \in \VerbInt\}
\end{align*}
als das Infimum der unteren und das Supremum der oberen Intervallgrenzen.
Man zeigt nun ebenso wie für zweielementige Mengen, 
dass $[l,u]$ die kleinste obere Schranke von $\mathcal{A}$ ist.

Damit existieren auch obere Schranken in $\VerbIA$: Sei $P \subseteq \VerbIA$. 
Definiere $\rho \in \VerbIA$ durch $\rho(x) := \bigsqcup\{\rho'(x) \mid \rho' \in P\}$ 
für jedes $x \in \Var$. Dann ist $\rho$ die kleinste obere Schranke von $P$.
\end{bsp}

Im Allgemeinen existiert eine solche kleinste obere Schranke 
nicht für jede beliebige Teilmenge einer Halbordnung. 
Um also die kleinste obere Schranke gewisser oder aller Teilmengen bilden zu können, müssen wir weitere Forderungen an die Halbordnung $(L,\sqleq)$ stellen.
\begin{dfn}
Sei $(L,\sqleq)$ eine Halbordnung. Eine Teilmenge $K \subseteq L$ heißt \emph{Kette}\index{Kette}, falls $K$ total geordnet ist, d.h.~wenn $x\sqleq y$ oder $y \sqleq x$ für je zwei Elemente $x,y\in L$ gilt.

Eine Halbordnung $(L,\sqleq)$ heißt \emph{ket\-ten\-voll\-stän\-dig}\index{kettenvollständig}, falls jede Kette $K\subseteq L$ eine kleinste obere Schranke besitzt.
\end{dfn}

\begin{bsp}
Die Menge $K := \{[0,n]\mid n \in \nn\}$ ist eine Kette in $\VerbInt$: Seien $n,m \in \nn$ und \obda $n \le m$. Dann ist $[0,n] \subseteq [0,m]$. Damit sind je zwei solche Intervalle vergleichbar.

Die Menge $N:= \{[n,n+1] \mid n \in \nn\}$ ist dagegen keine Kette in $\VerbInt$: Seien wieder $n<m$ natürliche Zahlen. 
Dann sind $[n,n+1]$ und $[m,m+1]$ unvergleichbar: Angenommen, es wäre $[n,n+1] \subseteq [m,m+1]$. Dann gilt $n \ge m$ und $n+1 \le m+1$, was $n=m$ impliziert. 
Dies ist ein Widerspruch zu $n<m$. Ebenso widerlegt man $[n,n+1] \supseteq [m,m+1]$. Demnach ist $N$ keine Kette.
\end{bsp}

\begin{bem} 
In einer kettenvollständigen Halbordnung gibt es ein eindeutig bestimmtes kleinstes Element $\bot = \bigsqcup \emptyset$. 
\end{bem}

Im Folgenden werden wir meist die stärkere Forderung stellen, dass die betrachtete Halbordnung bereits ein \emph{vollständiger Verband} ist.
\begin{dfn}
Eine Halbordnung $(L,\sqleq)$ heißt \emph{voll\-stän\-di\-ger Verband}\index{vollständiger Verband}, falls jede Teilmenge $A \subseteq L$ eine kleinste obere Schranke besitzt.
\end{dfn}
\begin{bem}
Jeder vollständige Verband ist insbesondere kettenvollständig.
\end{bem}

\begin{bsp}\label{bsp:potenzverband-ist-vollst}
Sei $X$ eine Menge und $L := 2^X$. Dann sind $(L,\subseteq)$ sowie $\VerbInt$ und $\VerbIA$ vollständige Verbände. Dies folgt direkt aus den Beobachtungen in Beispiel \ref{bsp-schranken-potenzverband}.
\end{bsp}

Im funktionalen Ansatz zur Analyse interprozeduraler Programme werden wir neben einem zugrundeliegenden vollständigen Verband $L$ auch die Menge der Abbildungen $L \to L$ betrachten 
und benötigen dabei, dass auch diese wieder einen vollständigen Verband bilden. Das dies immer der Fall ist, zeigt das folgende Beispiel.
\begin{bsp}\label{bsp:Abbildungsverband}
Sei $(L,\sqleq)$ ein vollständiger Verband und $(L \to L)$ die Menge aller Abbildungen von $L$ nach $L$. Es sei nun $f \sqleqmap g$ genau dann, wenn $f(l) \sqleq g(l)$ für jedes $l \in L$ gilt. 
Dies ist offensichtlich eine Halbordnung.

Sei weiter $\mathcal{F}$ eine beliebige Teilmenge von $(L \to L)$. Wir definieren eine Abbildung $g: L \to L$ durch
\[g(l) := \bigsqcup\big\{f(l) \mid f \in \mathcal{F}\big\}.\]
Da $(L, \sqleq)$ ein vollständiger Verband ist, existieren diese kleinsten oberen Schranken und $g$ ist wohldefiniert. Weiter gilt nach Konstruktion offenbar $f \sqleqmap g$ für jedes $f \in \mathcal{F}$. 
Also ist $g$ eine obere Schranke von $\mathcal{F}$

Sei $h$ eine weitere obere Schranke. Dann gilt für jedes $f \in \mathcal{F}$ und jedes $l \in L$ bereits $f(l) \sqleq h(l)$. 
Also ist auch $g(l) = \bigsqcup\big\{f(l) \mid f \in \mathcal{F}\big\}$ für jedes $l \in L$. Damit folgt aber schon $g \sqleqmap h$ und $g$ ist die kleinste obere Schranke von $\mathcal{F}$. 
Also ist $((L\to L), \sqleqmap)$ ein vollständiger Verband.
\end{bsp}

Sind die Informationen nun Elemente einer Halbordnung $(L,\sqleq)$, so müssen wir die Programmanweisungen $\mathtt{b}$ mit Transferfunktionen $f_\mathtt{b}: L \to L$ assoziieren, 
die den Datenfluss gemäß dieser Anweisung beschreiben. Dabei erwarten wir, dass eine solche Abbildung $f_\mathtt{b}$ stets präzisionserhaltend ist. Darunter verstehen wir Folgendes: 
Wenn eine Information präziser ist als eine andere Information, dann ist auch ihr Bild unter $f_\mathtt{b}$ präziser als das der anderen. Dies entspricht dem Begriff der \emph{Monotonie}.
\begin{dfn}
Eine Abbildung $f: L \to L'$ zwischen zwei Halbordnungen $(L,\sqleq_L)$ und $(L',\sqleq_{L'})$ heißt \emph{monoton}\index{monoton}, 
falls für alle $l_1,l_2 \in L$ aus $l_1 \sqleq_L l_2$ schon $f(l_1) \sqleq_{L'} f(l_2)$ folgt.
\end{dfn}

\begin{bsp}
Auch die Menge der monotonen Abbildungen $(L \to_{\text{mon}} L)$ bildet einen vollständigen Verband: 
Sei wieder $\mathcal{F}$ eine beliebige Teilmenge und $g$ wie im vorigen Beispiel \ref{bsp:Abbildungsverband} definiert. Es reicht nun zu zeigen, dass $g$ wieder monoton ist. 
Seien also $l,l' \in L$ mit $l \sqleq l'$. Dann ist für jedes $f \in \mathcal{F}$
\begin{align*}
f(l) &\sqleq f(l') = \bigsqcup\big\{f(l') \mid f \in \mathcal{F} \big\} = g(l')
\intertext{und damit}
g(l) &= \bigsqcup\big\{f(l) \mid f \in \mathcal{F} \big\} \sqleq g(l').
\end{align*}
Also ist auch $((L \to_{\text{mon}} L), \sqleqmap)$ ein vollständiger Verband.
\end{bsp}

Wir werden später die Informationen, die an den Programmpunkten gültig sind, mithilfe von Mengen von Ungleichungen charakterisieren.   
Außerdem werden wir zeigen, dass sich die Lösungen dieser Ungleichungssystemen als \emph{Präfixpunkte} monotoner Abbildungen schreiben lassen.
\begin{dfn}
Sei $L$ eine Menge und $f: L \to L$ eine Abbildung. Dann heißt $l \in L$ \emph{Fixpunkt von $f$}\index{Fixpunkt}, falls $f(l) = l$ ist. 
Ist $(L,\sqleq)$ eine Halbordnung, definieren wir weiter, dass $l$ \emph{Präfixpunkt}\index{Präfixpunkt} heißt, falls $f(l)\sqleq l$ ist, 
und \emph{Postfixpunkt}\index{Postfixpunkt}, falls $f(l) \sqgeq l$ gilt.
\end{dfn}

Von besonderem Interesse wird dabei stets der kleinste aller Fixpunkte einer monotonen Abbildung sein, da dieser der präzisesten gültigen Information entspricht, wie wir im nächsten Abschnitt sehen werden. 
Die Existenz des kleinsten Fixpunktes liefert der Fixpunktsatz von Tarski, dessen Beweis in Nielson et al.~\cite[S.~400f.]{Niel} zu finden ist.
\begin{satz}[Fixpunktsatz von Tarski]
Seien $(L,\sqleq)$ ein vollständiger Verband und $f : L \to L$ monoton. Dann hat $f$ einen kleinsten Fixpunkt $\lfp(f)$ und einen größten Fixpunkt $\gfp(f)$ und es gilt
\begin{align*}
\lfp(f) = \bigsqcap\big\{x \in L \mid f(x) \sqleq x\big\} \text{ und } \gfp(f) = \bigsqcup\big\{x \in L \mid s \sqleq f(x)\big\}.
\end{align*}
\end{satz}
Der kleinste Fixpunkt entspricht also dem kleinsten Präfixpunkt und der größte Fixpunkt dem größten Postfixpunkt. Diese Tatsache werden wir später benötigen. 
Einen weiteren nützlichen Ansatz zur Bestimmung des kleinsten Fixpunktes einer monotonen Abbildung liefert der folgende Satz, der eine Verallgemeinerung eines Satzes ist, 
der als Fixpunktsatzes von Kleene bekannt ist. Vergleiche dazu auch \cite{folk-tale}.
\begin{satz}\label{genKleene}
Sei $(L,\sqleq)$ eine kettenvollständige Halbordnung und $f : L \to L$ monoton. Dann hat $f$ einen kleinsten Fixpunkt $\lfp(f)$ und es gilt
\[
 \lfp(f) = \bigsqcup\big\{l_\alpha \mid \alpha \text{ Ordinalzahl} \big\}
\]
mit
\begin{align*}
l_0 			&:= \bot \\
l_{\alpha+1} 	&:= f(l_\alpha) \text{ für eine Ordinalzahl } \alpha \\
l_{\alpha}		&:= \bigsqcup\big\{l_\gamma \mid \gamma < \alpha\big\} \text{ für eine Limesordinalzahl } \alpha.
\end{align*}
\end{satz}
\begin{proof}
Wir prüfen zunächst, dass diese $l_\alpha$ wohldefiniert sind. Dazu reicht es nach zu weisen, dass für alle Ordinalzahlen $\alpha$ die Menge $\{l_\gamma \mid \gamma < \alpha\}$ eine Kette ist. 
Dies wiederum folgt daraus, dass für je zwei Ordinalzahlen $\alpha$ und $\beta$ aus $\alpha < \beta$ auch $l_{\alpha} \sqleq l_{\beta}$ folgt: 
Wir zeigen dies per transfiniter Induktion für jede Ordinalzahl $\beta$: 
Für $\beta=0$ ist nichts zu zeigen. 

Sei nun $\beta = \beta' + 1$ für eine Ordinalzahl $\beta'$ und gelte die Behauptung für kleinere Ordinalzahlen. Sei $\alpha < \beta' + 1$ beliebig.
Da nach Induktionsvoraussetzung schon $l_\alpha \sqleq l_{\beta'}$ gilt, genügt es, $l_{\beta'} \sqleq l_{\beta'+1}$ nach zu weisen. 
Ist dabei $\beta'$ eine Limesordinalzahl, so gilt
\begin{align*}
l_{\beta' + 1} 
&= f\Big(\bigsqcup\big\{ l_\gamma \mid \gamma < \beta' \big\}\Big) 
\intertext{und wegen der Monotonie von $f$}
&\sqgeq \bigsqcup\big\{ f(l_\gamma) \mid \gamma < \beta' \big\} \\
&= \bigsqcup\big\{ l_{\gamma+1} \mid \gamma < \beta' \big\} \\
&= \bigsqcup \big(\{ l_{\gamma+1} \mid \gamma < \beta' \} \cup \{\bot\}\big).
\intertext{Da $\gamma$ eine Limesordinalzahl ist, folgt aus $\gamma <\beta'$ auch $\gamma+1 < \beta'$ und wir erhalten}
&= \bigsqcup\big\{ l_{\gamma} \mid \gamma < \beta' \big\} \\
&= l_{\beta'}.
\end{align*}
Andernfalls existiert $l_{\beta'-1}$ und nach Induktionsvoraussetzung ist $l_{\beta'-1} \sqleq l_{\beta'}$ 
und mit der Monotonie von $f$ dann $l_{\beta'} = f(l_{\beta'-1}) \sqleq f(l_{\beta'}) = l_{\beta'+1}$.

Sei zuletzt $\beta$ eine Limesordinalzahl und gelte die Behauptung schon für alle kleineren Ordinalzahlen. 
Dann ist für $\alpha < \beta$ bereits $l_\alpha \in \{l_\gamma \mid \gamma < \beta\}$ und damit \[l_\alpha \sqleq \bigsqcup \big\{l_\gamma \mid \gamma < \beta\big\} = l_\beta.\]

Hieraus folgt insbesondere, dass $K:=\{l_{\alpha} \mid \alpha \text{ Ordinalzahl}\}$ eine Kette ist.

Wir zeigen nun, dass, sofern $f$ Fixpunkte besitzt und $x$ ein beliebiger Fixpunkt von $f$ ist, $l_\alpha \sqleq x$ für jede Ordinalzahl $\alpha$ gilt. 
Für $\alpha=0$ ist dies  trivialerweise erfüllt, da $l_0=\bot \sqleq x$ schon für jedes Element $x \in L$ gilt.

Gilt die Behauptung für eine Ordinalzahl $\alpha$, so ist mit der Monotonie von $f$ 
\[l_{\alpha+1} = f(l_\alpha) \sqleq f(x) = x.\]
Ist $\alpha$ hingegen eine Limesordinalzahl und gilt die Behauptung für alle kleineren Ordinalzahlen, so erhalten wir
\[l_\alpha = \bigsqcup\{l_\gamma \mid \gamma < \alpha\} \sqleq \bigsqcup\{x \mid \gamma < \alpha\} = x.\]

Sei nun $\kappa$ die Kardinalität von $L$. Angenommen, die Folge $(l_{\alpha})_{\alpha \in \text{Ord}}$ wird nie stabil. 
Dann besitzt $(l_{\alpha})_{\alpha \in \text{Ord}}$ eine echt aufsteigende Teilfolge. 
Da diese aber nur aus Elementen von $L$ besteht, wird sie spätestens ab dem $\kappa$-ten Folgenelement stabil. 
Dies ist ein Widerspruch. Also wird auch $(l_\alpha)$ stabil. Es existiert also eine Ordinalzahl $\beta$ mit 
\[l_\beta = l_{\beta+1} = f(l_\beta)\]
und $l_\beta$ ist ein Fixpunkt von $f$. Wie wir bereits gesehen haben, ist $l_\beta$ aber auch kleiner oder gleich jedem anderen Fixpunkt von $f$ und damit der kleinste Fixpunkt von $f$. Es folgt also
\[\lfp(f) = l_{\beta} = \bigsqcup\big\{l_\alpha \mid \alpha \text{ Ordinalzahl}\big\}\]
und das war die Behauptung.
\end{proof}

Um später die kleinsten Fixpunkte zweier monotoner Funktionen zu vergleichen, werden wir folgende Verallgemeinerung des Transferlemmas, welches in \cite{folk-tale} zu finden ist, benutzen.
\begin{satz}\label{relFixpkte}
Seien $(L,\sqleq), (L',\sqleq')$ vollständige Verbände, $f:L\to L$, $g:L' \to L'$ monoton und $R \subseteq L \times L'$ eine Relation mit folgenden Eigenschaften:
\begin{enumerate}
 \item Für alle $l \in L, l'\in L'$ mit $(l,l') \in R$ ist auch $(f(l),g(l')) \in R$.
 \item Für alle $X \subset R$ ist $(\bigsqcup X_1, \bigsqcup X_2)\in R$. Dabei ist \[X_1 := \{x \in L \mid \exists y \in L': (x,y)\in X\}\] und \[X_2 := \{y \in L' \mid \exists x \in L: (x,y)\in X\}.\]
\end{enumerate}
Dann gilt $(\lfp(f),\lfp(g)) \in R$.
\end{satz}
\begin{proof}
Definiere induktiv
\[\begin{array}{lll}
l_0 := \bot_L & l_0' := \bot_{L'} &\\
l_{\alpha+1} := f(l_\alpha) & l'_{\alpha+1} := g(l'_\alpha) & \text{für eine Ordinalzahl } \alpha \\
l_{\alpha} := \bigsqcup\big\{l_\beta \mid \beta < \alpha\big\} & l'_{\alpha} := \bigsqcup\nolimits'\big\{l'_\beta \mid \beta < \alpha\big\} & \text{für eine Limesordinalzahl } \alpha
\end{array}\]
Dann gilt $(l_\alpha, l'_\alpha) \in R$ für jede Ordinalzahl $\alpha$:
Für $\alpha=0$ ist \[(\bot_L,\bot_{L'}) = \big(\bigsqcup \emptyset, \bigsqcup\nolimits' \emptyset\big) \in R\] nach Bedingung \textit{b)}. 
Ist $\alpha$ eine Ordinalzahl und $(l_\alpha, l'_\alpha) \in R$, so gilt \[(l_{\alpha+1},l'_{\alpha+1}) = (f(l_\alpha),g(l'_\alpha))\in R\] aufgrund von Bedingung \textit{a)}.
Für eine Limesordinalzahl $\alpha$ ist mit Induktionsvoraussetzung $(l_\beta, l'_\beta) \in R$ für jedes $\beta < \alpha$. 
Wiederum mit Bedingung \textit{b)} erhalten wir nun \[(l_\alpha, l'_\alpha) = \big(\bigsqcup\big\{l_\beta \mid \beta < \alpha\big\}, \bigsqcup\nolimits'\big\{l'_\beta \mid \beta < \alpha\big\}\big) \in R.\]

Nun existieren nach Satz \ref{genKleene} Ordinalzahlen $\beta_f, \beta_g$ mit $l_{\beta_f} = \lfp(f)$ und $l'_{\beta_g} = \lfp(g)$. 
Definiere $\beta := \max\{\beta_f, \beta_g\}$. Dann gilt \[(\lfp(f),\lfp(g)) = (l_\beta, l'_\beta) \in R,\] was den Satz beweist.
\end{proof}

Für Datenflussanalysen benötigen wir sowohl eine Halbordnung $(L,\sqleq)$ der Informationen als auch eine Menge monotoner Funktionen $L \to L$, die den Datentransfer beschreiben. 
Zusammen bezeichnen wir dies als \emph{monotones Framework}:
\begin{dfn}\label{dfn:monFW}
Ein \emph{monotones Framework}\index{monotones Framework} $(L,\sqleq,\mathcal{F})$ besteht aus einem voll\-stän\-di\-gen Verband $(L,\sqleq)$ 
und einer unter Komposition abgeschlossenen Menge $\mathcal{F} \subseteq (L\to_{\text{mon}} L)$ monotoner Abbildungen, welche die Identität $\id_L$ enthält. 
Die Abbildungen der Menge $\mathcal{F}$ heißen \emph{Transferfunktionen}.

Ein monotones Framework $(L,\sqleq,\mathcal{F})$ heißt \emph{(positiv-/universell-)distributiv}, falls jede Abbildung $f \in \mathcal{F}$ (positiv-/universell-)distributiv ist. Dabei heißt $f$
\begin{itemize}
 \item \emph{distributiv}\index{distributiv}, falls $f (l\sqcup l') = f(l) \sqcup f(l')$ für alle $l,l' \in L$ gilt.
 \item \emph{positiv-distributiv}\index{distributiv!-positiv}, falls $f(\bigsqcup A) = \bigsqcup f(A)$ für jede nichtleere Teilmenge $A\subseteq L$ gilt.
 \item \emph{universell-distributiv}\index{distributiv!-universell}, falls $f(\bigsqcup A) = \bigsqcup f(A)$ schon für jede Teilmenge $A \subseteq L$ gilt.
\end{itemize}
\end{dfn}
Dabei gelten folgende Zusammenhänge: Jede universell-distributive Abbildung ist nach Definition bereits positiv-distributiv und jede positiv-distributive Abbildung schon distributiv. 
Distributive Abbildungen sind wiederum monoton: Seien dazu $f: L \to L$ distributiv und $l,l' \in L$ mit $l \sqleq l'$. Dann ist $l' = l \sqcup l'$ und damit
\[f(l') = f(l \sqcup l') = f(l) \sqcup f(l') \sqgeq f(l)\]
und $f$ ist monoton.

Der Unterschied zwischen einer positiv-distributiven und einer universell-distri\-bu\-ti\-ven Abbildung liegt nun darin, dass letztere das kleinste Element $\bot = \bigsqcup \emptyset$ 
stets auf das kleinste Element $\bot = \bigsqcup f(\emptyset)$ abbildet, was für positiv-distributive Abbildungen nicht gefordert wird.

\begin{bsp}
Sei wieder $L := 2^X$ für eine Menge $X$ und $\mathcal{F}$ die Menge aller Abbildungen $f: L \to L$, die universell-distributiv sind. 
Dann bildet $(L,\subseteq,\mathcal{F})$ ein universell-distributives Framework:

In Beispiel \ref{bsp:potenzverband-ist-vollst} haben wir bereits gesehen, dass $(L,\subseteq)$ ein vollständiger Verband ist. 
Nach Definition von $\mathcal{F}$ ist jede Funktion $f \in \mathcal{F}$ universell-distributiv und damit monoton. 
Die Identität $\id_L$ ist offenbar universell-distributiv, also gilt $\id_L \in \mathcal{F}$. 
Man rechnet leicht nach, dass die Komposition universell-distributiver Funktionen wieder universell-distributiv ist. Also ist $\mathcal{F}$ abgeschlossen unter Komposition. 
Damit ist $(L,\subseteq,\mathcal{F})$ ein universell-distributives Framework.
\end{bsp}

\section{Ungleichungssysteme}\label{sec:UGS}
Ungleichungssysteme beschreiben Zusammenhänge von Informationen an verschiedenen Programmpunkten. Die präziseste Information, die diese Zusammenhänge erfüllt, ist die kleinste Lösung des Ungleichungssystems. 
In diesem Abschnitt geben wir eine allgemeine Beschreibung von Ungleichungssystemen und beschäftigen uns mit deren kleinsten Lösungen.

Seien $(L,\sqleq)$ ein vollständiger Verband und $\Var = \{x_i \mid i \in I\}$ eine abzählbare Menge von Variablen mit Indexmenge $I$, die Werte in $L$ annehmen. 
Die Menge der \quotes{Vektoren} über $L$ bildet zusammen mit komponentenweisem $\sqleq$ eine Halbordnung $(L^I,\sqleq^I)$. Es ist $x \sqleq^I y$ genau dann, wenn $x_i \sqleq y_i$ für jedes $i \in I$ gilt.

\begin{dfn}
Ein \emph{Ungleichungssystem}\index{Ungleichungssystem} ist eine Menge $\ugs$ von Ungleichungen der Gestalt
\begin{equation}\label{soi1}
x_j \sqgeq f(x)
\end{equation}
für Variablen $x_j \in \Var$ und monotone Abbildungen $f: L^I \to L$. Dabei schreiben wir $x$ abkürzend für $(x_i)_{i \in I}$. 
Für jede Variable $x_j$ enthalte $\ugs$ dabei nur eindliche viele Ungleichungen der Gestalt $x_y \sqgeq f(x)$.
\end{dfn}
Die Menge der verwendeten Abbildungen bezeichnen wir mit $\abbugs$. Solche Abbildungen nennen wir \emph{verallgemeinerte Transferfunktionen}. 
Für einelementige Indexmengen sind verallgemeinerte Transferfunktion bereits Transferfunktionen.

Um ein gegebenes Ungleichungssystem zu lösen, formen wir es zunächst, wie im Folgenden beschrieben, in ein äquivalentes Ungleichungssystem um, 
in welchem jede Variable auf der linken Seite von genau einer Ungleichung auftaucht:
Sind dabei $x_j \sqgeq f_1(x)$ und $x_j \sqgeq f_2(x)$ zwei Ungleichungen mit gleicher linker Seite, so werden diese ersetzt durch $x_j \sqgeq f_1(x) \sqcup f_2(x)$. 
Dieser Vorgang wird wiederholt, bis keine zwei Ungleichungen mit gleicher linker Seite mehr übrig sind. 
Für jede Variable $x_j \in \Var$, für die keine Ungleichung mit linker Seite $x_j$ existiert, wird weiter $x_j \sqgeq \bot$ hinzugefügt. 
Dazu definieren wir 
\[\mathcal{F}_{\ugs,j} := \{f \in \abbugs \mid \exists \text{ Ungleichung } x_j \sqgeq f(x) \in \mathcal{U}\}\]
als die Menge derjenigen Abbildungen, die in einer Ungleichung mit linker Seite $x_j$ benutzt werden.
Weiter sei 
\[F_j := \bigsqcup \mathcal{F}_j : L^I \to L \]
definiert durch
\[F_j(x) := \bigsqcup \big\{f(x) \mid f \in \mathcal{F}_{\ugs,j} \big\}.\]
Dies ist diejenige Abbildung, die für jedes Element die Informationen, die durch die einzelnen Abbildungen des Ungleichungssystems für dieses Element berechnet werden, \quotes{zusammen fasst}. 
Wir erhalten nun das Ungleichungssystem
\begin{equation}\label{soi2}
\ugs_I := \{x_j \sqgeq F_j(x)\}.
\end{equation}
Dies liefert auch für diejenigen Variablen, die im ursprünglichen Ungleichungssystem auf keiner linken Seite einer Ungleichung vorkommen, das Gewünschte: 
Existieren im ursprünglichen Ungleichungssystem keine Ungleichungen mit linker Seite $x_j$, so ist die Menge $\mathcal{F}_{\ugs,j}$ leer. 
Dann ist aber $\bigsqcup \mathcal{F}_{\ugs,j}$ die Abbildung, die jedes $x \in L^I$ auf $\bot$ abbildet.

\begin{lemma}
Die Lösungsmengen von $\ugs$ \eqref{soi1} und $\mathcal{U}_{I}$ \eqref{soi2} stimmen überein.
\end{lemma}
\begin{proof}
Dies folgt unmittelbar aus der Äquivalenz $x \sqgeq y \wedge x \sqgeq z \Leftrightarrow x \sqgeq y \sqcup z$: 
\quotes{$\Rightarrow$} ist offensichtlich, während \quotes{$\Leftarrow$}aus der Tatsache
$x \sqcup (y \sqcup z) 
= (x \sqcup y) \sqcup (x \sqcup z) 
= x \sqcup x
= x$ resultiert.
\end{proof}

Das Ungleichungssystem $\ugs_I$ ist nun offenbar äquivalent zu 
\[x \sqgeq^{I} F(x),\] 
wobei 
\begin{equation}\label{soi-fct}
F := (F_j)_{j \in I} : L^I \to L^I.
\end{equation}
Somit können wir $F$ an Stelle des gegebenen Ungleichungssystems \eqref{soi1} betrachten.

\begin{lemma}
Für jedes $j \in I$ ist $F_j$ monoton als Abbildung zwischen $(L^I, \sqleq^I)$ und $(L,\sqleq)$.
\end{lemma}
\begin{proof}
Allgemein gilt: Für zwei vollständige Verbände $(D, \sqleq)$ und $(D',\sqleq')$ und monotone Abbildungen $f, g: D \to D$ ist die Abbildung
\begin{align*}
h: (D,\sqleq) &\to (D,\sqleq), \\
d &\mapsto f(d) \sqcup g(d)
\end{align*} 
wiederum monoton:

Es ist $a \sqleq b \Leftrightarrow a \sqcup b=b$. Sei nun $a \sqleq b$, dann ist
\[h(a) \sqcup h(b) 
= f(a) \sqcup g(a) \sqcup f(b) \sqcup g(b)
= f(b) \sqcup g(b)
= h(b),\]
was $h(a) \sqleq h(b)$ beweist.

Für $j \in I$ ist nun $F_j: L^I \to L$ durch $F_j(x) = \bigsqcup_{i=1}^m f_i(x)$ für monotone Operatoren $f_i \in \mathcal{F}_{\ugs,j}$ definiert. 
Da $m$ endlich ist, folgt die Monotonie von $F_j$ bereits aus obiger Aussage.
\end{proof} 

Mit dem Fixpunktsatz von Tarski folgt nun:
\begin{kor} 
Die kleinste Lösung des Ungleichungssystems $\mathcal{U}$ (\ref{soi1}) ist der kleinste Präfixpunkt von $F$ und dieser stimmt mit dem kleinsten Fixpunkt $\lfp(F)$ von $F$ überein.
\end{kor}

Wir werden die kleinste Lösung eines Ungleichungssystems später mit $\underline{x}$ bezeichnen.

\section{Flussgraphen}\label{sec:fg}
Im vorigen Abschnitt haben wir uns mit Ungleichungssystemen und deren Lö\-sun\-gen befasst. 
In diesem Abschnitt erläutern wir einführend, wie zu einem ge\-ge\-be\-nen Programm ein solches Ungleichungssystem gefunden werden kann, 
dessen Lösung dann die gewünschten Informationen an den einzelnen Programmpunkten liefert.

Im Folgenden gehen wir davon aus, dass jedes Programm bereits als eine endliche Menge von Flussgraphen gegeben ist. Dabei entspricht jeder Flussgraph einer Prozedur des Programms. 
Mit $\Proc$ bezeichnen wir die (endliche) Menge von Prozeduren des Programms, die sich auch rekursiv aufrufen können. 
Um die Darstellung einfach zu halten, verzichten wir auf lokale Variablen und ebenso auch auf Prozedurparameter. Das Programm habe also nur globale Variablen.

Zunächst definieren wir ein solches Flussgraphsystem. Anschließend stellen wir einige Ungleichungen für ein Ungleichungssystem auf, dessen Lösung mit Informationen an Programmpunkten korrespondiert. 

\begin{dfn}
Ein \emph{Flussgraphsystem}\index{Flussgraphsystem} ist eine endliche Menge $G = \{G_{\p}\}_{\p\in \Proc}$ disjunkter endlicher Flussgraphen. Jeder \emph{Flussgraph}\index{Flussgraph} $G_{\p}$ besteht aus
\begin{itemize}
 \item einer endlichen Menge $N_{\p}$ von \emph{Knoten}, die Programmpunkten entsprechen,
 \item einer endlichen Menge $E_{\p} \subseteq (N_{\p} \times \lab \times N_{\p})$ von \emph{Kanten},
 \item einem ausgezeichneten \emph{Startknoten} $s_{\p} \in N_{\p}$ und
 \item einem ausgezeichneten \emph{Endknoten} $r_{\p} \in N_{\p}$ der Prozedur $\p$.
\end{itemize}
Hierbei ist $\lab$ eine Menge von \emph{Beschriftungen}, die den Programmanweisungen entsprechen. 
Solche Beschriftungen sind entweder sogenannte \emph{Basisanweisungen} (im Folgenden meist mit \quotes{$\texttt{b}$} bezeichnet) oder \emph{Prozeduraufrufe} (im Folgenden meist \quotes{$\texttt{p()}$}). 
Die entsprechenden Kanten nennen wir \emph{Basiskanten} und \emph{Aufrufkanten}. 
Wir setzen $N:= \bigcup\{N_{\p} \mid \p\in\Proc\}$ als die Menge aller Knoten und $E:= \bigcup\{E_{\p} \mid \p\in\Proc\}$ als die Menge aller Kanten.
\end{dfn}
Ein solches Flussgraphssystem entsteht aus einem Programm folgendermaßen: Zu\-nächst werden die Programmpunkte einer Prozedur als Knoten aufgefasst. 
Dann wird für jede Anweisung zwischen je zwei Programmpunkten eine gerichtete Kante zwischen den zugehörigen Knoten hinzugefügt, wobei die Kante mit dieser Anweisung beschriftet ist.

Um nun zu jedem Programmpunkt Informationen, wie beispielsweise die möglichen Werte der Programmvariablen, bestimmen zu können, stellen wir zunächst ein Ungleichungssystem auf. 
Dabei sei $(L,\sqleq, \mathcal{F})$ das zugrundeliegende Framework. Also ist $L$ die Menge der Informationen. 
Mit $I: N \to L$ bezeichnen wir die Abbildung, die am Ende der Rechnung jedem Programmpunkt $u$ die gewünschte Information $I[u]$ zuordnet.  
Als Variablen $(x_i)_{i\in I}$ im Ungleichungssystem betrachten wir hier also $(I[\node])_{\node \in N}$. 
Für jede Basisanweisung $\mathtt{b}\in\lab$ wählen wir eine Transferfunktion $f_{\mathtt{b}} \in \mathcal{F}$.

\begin{figure}[ht]\begin{flowgraph}
  \node[stdnode]	(u) 						{$u$};
  \node[stdnode]	(v) 	[below=of u] 		{$v$};
  \node[annnode]	(Iu)	[right] at(u.east)	{$I[u]$};
  \node[annnode]	(Iv)	[right] at(v.east)	{$I[v] \sqgeq f_{\mathtt{b}}(I[u])$};

  \path[->]
	(u) 	edge 					node [swap] {\texttt{b}} 	(v)
	(Iu) 	edge [annarc,bend left]	node {\texttt{b}} 			(Iv);
\end{flowgraph}
\caption{Fluss von Informationen.}
\label{bild:informationsfluss}
\end{figure}

Für jede Basiskante $e=(u,\mathtt{b},v)$ benötigen wir eine Ungleichung
\[I[v] \sqgeq f_e((I[\node])_{\node \in N}),\]
wobei $f_e$ durch $f_e((I[\node])_{\node\in N}) := f_{\mathtt{b}}(I[u])$ definiert ist. 
Also hängt $f_e$ nur von der Basisanweisung $\mathtt{b}$ und der eingehenden Information $I[u]$ ab und ist nach Definition monoton.
Diese Ungleichung wird in \autoref{bild:informationsfluss} illustriert.
Weiter wird für den Startknoten des Programms, den wir mit $s_{\main}$ bezeichnen, verlangt, dass er eine Startinformation enthält. 
Darum enthält $\mathcal{U}$ eine Ungleichung \[s_{\main} \sqgeq \init\] für eine gegebene Startinformation $\init$.
Zur Gewinnung von Ungleichungen aus den Aufrufkanten gibt es verschiedene Mög\-lich\-kei\-ten, welche wir in Kapitel \ref{chap:PA} betrachten werden.

Ist zu einem Programm nun ein Ungleichungssystem $\mathcal{U}$ gegeben, in dem ausschließlich monotone Funktionen verwendet werden, 
so können wir wie in \autoref{sec:UGS} die Funktion $F$ bestimmen, deren kleinster Fixpunkt mit der kleinsten Lösung von $\mathcal{U}$ übereinstimmt. 

\begin{bem}
Das hier vorgestellte Verfahren, um aus einem Programm ein zu untersuchendes Ungleichungssystem zu gewinnen, ist eine \emph{Vorwärtsanalyse}. Im Gegensatz dazu stehen \emph{Rückwärtsanalysen}: 
Hier wird der Endknoten $r_{\main}$ mit einer Startinformation $\init$ versehen und die Informationen laufen in umgekehrter Richtung durch die Kanten. 
Weisen wir also wieder jeder Basisanweisung $\mathtt{b} \in \lab$ eine Transferfunktion $f_\mathtt{b}$ zu, so definieren wir $f_{(u,\mathtt{b},v)} ((I[\node])_{\node \in N}) := f_\mathtt{b}(I[v])$. 
Das Ungleichungssystem besteht nun aus Ungleichungen 
\[I[u] \sqgeq f_e((I[\node])_{\node \in N})\]
für jede Basiskante $e=(u,\mathtt{b},v)$ sowie
\[r_{\main} \sqgeq \init.\]
Wie auch bei der Vorwärtsanalyse werden weitere Ungleichungen benötigt, die Prozeduren behandeln. Diese bestimmt man analog zur Vorwärtsanalyse. Deshalb betrachten wir in dieser Arbeit nur Vorwärtsanalysen.
\end{bem}

\section{Der Workset-Algorithmus}\label{sec:wla}
Ein Ungleichungssystem $\mathcal{U}$ kann mithilfe eines \emph{Worklist-Algorithmus} gelöst werden. Dieser überprüft sukzessive alle Ungleichungen. 
Ist eine Ungleichung nicht erfüllt, verändert er den Wert der Variablen auf ihrer \quotes{größeren} Seite, so dass sie erfüllt ist. 

Ein intuitiver Ansatz für einen solchen Algorithmus ist der hier angegebene. Dabei schreiben wir in allen folgenden Algorithmen stets \texttt{x} anstelle von $\mathtt{(x_i)_{i \in I}}$.

\begin{pseudocode}
//initialize worklist with all inequalities of the system of inequalities
W := $()$;
forall (($\mathtt{x_i}\sqgeq$f(x)) $\in$ $\mathcal{U}$) $\{$W := W:($\mathtt{x_i}\sqgeq$f(x))$\}$;
//initialize all variables with the minimal value $\gray{\bot \in L}$ of the complete lattice
forall (i$\in$I) $\{\mathtt{x_i}$:=$\bot;\}$
//satisfy all inequalities successively; add inequalities to worklist, if they use manipulated variables
while W $\ne ()$ $\{$
  ($\mathtt{x_j}\sqgeq$f(x)) := head(W); W := tail(W);
  t := f(x);
  if $\neg$(t $\sqleq$ $\mathtt{x_j}$) $\{$
    $\mathtt{x_j}$ := $\mathtt{x_j}$ $\sqcup$ t;
    forall (($\mathtt{x_k}\sqgeq$g(x)) $\in$ $\mathcal{U}$ : g uses $\mathtt{x_j}$) $\{$W := W:($\mathtt{x_k}\sqgeq$g(x));$\}$
  $\}$
$\}$
\end{pseudocode}
Dieser Algorithmus berechnet die kleinste Lösung des Ungleichungssystems $\mathcal{U}$ iterativ. 
Zunächst wird jeder zu berechnende Wert auf den kleinstmöglichen, also $\bot$, gesetzt. 
Dann wird eine Liste, die sogannte \emph{Worklist}, betrachtet, die zu jedem Zeitpunkt des Algorithmus alle potentiell unerfüllten Ungleichungen enthalten soll. 
Initialisiert wird diese Liste mit allen Ungleichungen von $\mathcal{U}$, da zu Beginn von keiner Ungleichung bekannt ist, ob die Variablen sie bereits erfüllen.

Bis alle Ungleichungen erfüllt sind, geht der Algorithmus folgendermaßen vor: 
Zu\-nächst erntfernt er die erste Ungleichung aus der Liste, welche die Gestalt $x_j \sqgeq f(x)$ hat, und überprüft sie. 
Dazu wird der Wert von $f(x)$ mit den zu diesem Zeitpunkten berechneten Werten der Variablen bestimmt. 
Falls die Ungleichung erfüllt ist, muss nichts weiter getan werden. 
Andernfalls muss der Wert $x_j$ so vergrößert werden, dass er mindestens so groß wie der Wert $t$ ist. 
Anschließend müssen alle Ungleichungen, die den Wert der Variable $x_j$ benutzen, erneut überprüft werden, und werden deshalb an die Worklist angehängt. 

Bei diesem Verfahren werden in der Tat nur diejenigen Ungleichungen erneut betrachtet, die vom Wert der Variable $x_j$ abhängen. Angehängt werden genauer
\[\mathcal{U} \setminus \big\{ (x_k \sqgeq f(x)) \in \mathcal{U} \mid \forall l_1,l_2\in L: f(x\{l_1/j\}) = f(x\{l_2/j\})\big\}.\]
Dabei definieren wir die Abbildung $x\{l/i\} : I \to L$ durch
\begin{align*}
x\{l/i\}(j) := \begin{cases} x_i &\text{ für } j\ne i \\ l &\text{ für } j = i. \end{cases}
\end{align*}
Die neue Abbildung $x\{l/i\}$ liefert für jeden Index außer $i$ dasselbe wie $x$. Für den Index $i$ hat $x\{l/i\}$ den Wert $l$.

\begin{lemma}\label{lem:wla-korr}
Sofern der Worklist-Algorithmus terminiert, liefert er die kleinste Lösung des Ungleichungssystems.
\end{lemma}
\begin{proof}
Während des Schleifendurchlaufes innerhalb des Algorithmus gelten folgende Invarianten:
\begin{enumerate}
\item Die jeweils aktuell berechnete Lösung für $\mathtt{x}$ approximiert den kleinsten Fixpunkt komponentenweise von unten, d.h.~es gilt $\mathtt{x_i} \sqleq \lfp(F)_i$ für jedes $i \in I$.
\item Jede Ungleichung, die nicht in der Worklist enthalten ist, ist aktuell erfüllt. Es gilt also $ \neg(\mathtt{x_j \sqgeq f(x)}) \Rightarrow (x_j \sqgeq f(x)) \in \mathsf{W}$.
\end{enumerate}
Der Algorithmus terminiert, wenn die Worklist leer ist. Mit Gültigkeit der zweiten Invariante ist also jede Ungleichung von $\mathcal{U}$ erfüllt. 
Also ist $\mathtt{x}$ in der Tat eine Lösung des Ungleichungssystems und damit ein Präfixpunkt der Abbildung $F$, welche $\mathcal{U}$ beschreibt. 
Somit gilt $\mathtt{x_i} \sqgeq \lfp(F)_i$ für jeden Index $i \in I$. 
Mit der ersten Invariante gilt dann aber schon Gleichheit. Also stimmt $\mathtt{x}$ mit der kleinsten Lösung $\lfp(F)$ von $\mathcal{U}$ überein.

Wir müssen nun noch zu zeigen, dass die Invarianten tatsächlich gelten. Vor dem ersten Schleifendurchlauf ist die erste Invariante gültig, da alle Variablen mit $\bot$ initialisiert wurden. 
Die zweite Invariante gilt, da alle Ungleichungen und damit insbeonsdere die nicht erfüllten in \textsf{W} enthalten sind. 

Wir prüfen nun, dass die Invarianten beim Ausführen des Schleifenrumpfes erhalten bleiben. Zunächst einmal wird \textsf{W} um eine Ungleichung $u$ verringert. 
Falls die Ungleichung bereits erfüllt ist, ist nichts zu tun. 
Da keine Variable verändert wird, sind auch die anderen Ungleichungen weiterhin erfüllt, die nicht in \textsf{W} enthalten sind. 
Somit bleibt die zweite Invariante gültig.

Falls $u$ dagegen nicht erfüllt ist, so wird die Variable der \quotes{größeren} Seite dieser Ungleichung verändert. 
Potentiell sind danach alle Ungleichungen, die diese Variable benutzen, nicht mehr erfüllt.
Alle anderen Ungleichungen wurden nicht verändert und sind damit weiterhin erfüllt. 
Da die potentiell nicht mehr gültigen Ungleichungen zu \textsf{W} wieder hinzugefügt werden, enthält \textsf{W} nun also alle Ungleichungen, die vorher nicht erfüllt waren, 
sowie alle Ungleichung, die nach Veränderung der einen Variable nicht mehr erfüllt sein könnten. 
Demnach ist jede unerfüllte Ungleichung in \textsf{W} enthalten. 
Die zweite Bedingung ist somit in der Tat invariant unter Ausführung der Schleife. 

Dies gilt auch für die erste Bedingung: Sei $x_j \sqgeq f(x)$ die im Schleifenrumpf betrachtete Ungleichung $u$ und $t = f(\mathtt{x})$. Nach Annahme gilt $\mathtt{x_j} \sqleq \lfp(F)_j$. 
Falls $\mathtt{x_j}$ nicht verändert wird, bleibt die Bedingung offenbar erhalten. Im anderen Fall müssen wir $x_j \sqcup t \sqleq \lfp(F)_j$ nachweisen. Dazu reicht, $t \sqleq \lfp(F)_j$ zu zeigen. 
Zunächst ist \[t = f(\mathtt{x}) \sqleq f(\lfp(F)),\] da $f$ monoton ist und nach Voraussetzung $\mathtt{x} \sqleq^I \lfp(F)$ gilt. 
Weiter gilt nach Konstruktion $f(x) \sqleq F_j(x)$ für alle $x$ und damit auch \[f(\lfp(F)) \sqleq F_j(\lfp(F)).\] Mit \[F_j(\lfp(F))=\lfp(F)_j\] folgt die Invarianz der ersten Bedingung.
\end{proof}

\begin{lemma}\label{lem:wla-term-falls-endl}
Der Algorithmus terminiert, wenn $L$ und $I$ endlich sind. 
\end{lemma}
\begin{proof}
Der Algorithmus berechnet in jedem Schritt neue Werte für die Worklist und die Variablen $x_i$. 
Wir bezeichnen nun die Werte, die nach dem $n$-ten Schritt berechnet wurden, mit $(W^{(n)}, x^{(n)}) \in 2^\mathcal{U} \times L^I$. 
Dabei ist also $W^{(n)} \in 2^\mathcal{U}$ die Liste derjenigen Ungleichungen, die nach dem $n$-ten Schritt potentiell nicht erfüllt sind. 
Der Vektor $x^{(n)} \in L^I$ beinhaltet die im $n$-ten Schritt berechneten Variablenwerte. Es ist nun zu zeigen, dass es eine natürliche Zahl $n \in \nn$ gibt, so dass $W^{(n)}$ leer ist.

Die $(W^{(n)}, x^{(n)})$ entstehen folgendermaßen:
Nach Initialisierung ist $W^{(0)}$ die Liste, die das Ungleichungssystem $\mathcal{U}$ enthält, und $x^{(0)}_i = \bot$ für jedes $i \in I$.
Sei für ein $n\ge 0$ bereits $(W^{(n)}, x^{(n)})$ konstruiert, $W^{(n)}$ nichtleer und $u=(x_j \sqgeq f(x))$ die nächste betrachtete Ungleichung. 
Dann entsteht $(W^{(n+1)}, x^{(n+1)})$ durch einen der beiden folgenden Fälle:
\begin{itemize}
\item[1.] Die Ungleichung $u$ ist erfüllt. Dann ist $x^{(n+1)} = x^{(n)}$ und $W^{(n+1)} = \text{tail}(W^{(n)})$. Insbesondere ist die neue Liste kürzer als die alte, d.h.~es gilt $|W^{(n+1)}| < |W^{(n)}|$.
Dabei bezeichnen wir mit $|W|$ die Länge einer Liste $W$.
\item[2.] Die Ungleichung $u$ ist nicht erfüllt. $W^{(n+1)}$ entsteht aus $\text{tail}(W^{(n)})$ durch An\-hän\-gen der Ungleichungen, die $x_j$ benutzen. 
Außerdem ist \[x^{(n+1)}_i = \begin{cases} x^{(n)}_i&\text{falls }i\ne j\\x^{(n)}\sqcup f(x^{(n)})&\text{sonst.}\end{cases}\] 
Insbesondere ist $x^{(n+1)} \sqsupset^I x^{(n)}$: Es ist $x^{(n+1)}_j \sqgeq x^{(n)}_j$ nach Konstruktion. 
Würde schon Gleichheit gelten, so wäre $x^{(n)}_j = x^{(n)}_j \sqcup f(x^{(n)})$ und damit $f(x^{(n)}) \sqleq x^{(n)}_j$. 
Dies ist ein Widerspruch dazu, dass die Ungleichung $u$ nach der $n$-ten Iteration nicht erfüllt ist.
\end{itemize}
Es gilt nun $(W^{(n)},x^{(n)}) \ne (W^{(m)},x^{(m)})$ für alle $n<m$: 
Wir unterscheiden zwei Fälle. Sei bei der Herleitung von $(W^{(m)},x^{(m)})$ aus $(W^{(n)},x^{(n)})$ nur der Fall 1 benutzt worden. Dann ist $|W^{(m)}| < |W^{(n)}|$
und damit $(W^{(n)},x^{(n)}) \ne (W^{(m)},x^{(m)})$.
Also sei bei der Herleitung mindestens einmal der Fall 2 aufgetreten. Damit ist aber $x^{(m)} \sqsupset^I x^{(n)}$. Auch dann ist $(W^{(n)},x^{(n)}) \ne (W^{(m)},x^{(m)})$.

Sei $\kappa:= |2^\mathcal{U} \times L^I|$ die Mächtigkeit von $2^\mathcal{U} \times L^I$. Sind nun $L$ und $I$ endlich, so ist auch $\lambda:=|\{f:L^I \to L \mid f \text{ monoton}\}| < \infty$ und 
$\mu:=|\mathcal{U}| \le |I| \cdot \lambda < \infty$. 
Es folgt $\kappa\le 2^{\lambda}\cdot |L|^{|I|} < \infty$.

Da die $(W^{(n)},x^{(n)})$ paarweise verschieden sind, kann es nur $\kappa$ verschiedene solche Tupel geben. Also ist nach spätestens $n=\kappa$ Schritten $W^{(n)}$ leer und der Algorithmus terminiert.
\end{proof}

In \autoref{sec:callstring-ansatz} werden wir auch Ungleichungssysteme einführen, die unendliche viele Ungleichungen enthalten. Darum erweitern wir das Konzept des Worklist-Algorithmus. 
Im Gegensatz zu einer endliche Liste von Ungleichungen kann eine \emph{Workset} beliebig viele Ungleichungen enthalten: 
\begin{pseudocode}
W := $\mathcal{U}$;
forall (i$\in$I) $\{\mathtt{x_i}$:=$\bot;\}$
while W $\ne ()$ $\{$
  choose u=($\mathtt{x_j}\sqgeq$f(x)) from W; W := W$\setminus\{$u$\}$;
  t := f(x);
  if $\neg$(t $\sqleq$ $\mathtt{x_j}$) $\{$
    $\mathtt{x_j}$ := $\mathtt{x_j}$ $\sqcup$ t;
    forall (($\mathtt{x_k}\sqgeq$g(x)) $\in$ $\mathcal{U}$ : g uses $\mathtt{x_j}$) $\{$W := W:($\mathtt{x_k}\sqgeq$g(x))$;\}$
  $\}$
$\}$
\end{pseudocode}
Wie zuvor berechnet dies eine Lösung des Ungleichungssystems $\mathcal{U}$.
Enthält das Ungleichungssystem nun unendlich viele Ungleichungen, so kann Terminierung in endlicher Zeit nicht garantiert werden. 
Da aber diese Berechnungsvorschrift im Fall endlich vieler Variablen wieder dem zuerst vorgestellten Algorithmus entspricht, nennen wir auch diese wieder Algorithmus.


\chapter{Klassische Ansätze zur interprozeduralen Analyse}\label{chap:PA}
Im letzten Kapitel haben wir Ungleichungssysteme und deren Lösungen betrachtet. 
Außerdem haben wir für einfachere Anweisungen eines Programms, sogenannte \emph{Basisanweisungen}, Ungleichungen aufgestellt, die den Datenfluss durch diese Anweisung beschreiben. 
Gegenstand dieses Kapitels ist nun die Behandlung der übrigen Anweisungen, der \emph{Prozeduraufrufe}.

Sharir und Pnueli \cite{sharir-pnueli} haben zwei Ansätze eingeführt, mit denen Prozeduraufrufe behandelt werden können. 
Der \emph{funktionale Ansatz} berechnet für jede Prozedur zunächst eine \emph{Summary-Information}, die den Datenfluss durch diese Prozedur beschreibt. 
Damit können anschließend Prozeduraufrufe sehr ähnlich zu Basisanweisungen mithilfe von Transferfunktionen behandelt werden. 
Die grundlegende Idee des \emph{Call-String-Ansatzes} dagegen ist, jeden Prozeduraufruf durch die Anweisungen der aufgerufenen Prozedur zu ersetzen 
und so das Programm wie eine einzige Prozedur zu behandeln. 
Ziel dieses Kapitels ist, diese beiden Ansätze im Hinblick darauf zu untersuchen, wie präzise die von ihnen berechneten Informationen sind. 

Dazu werden wir in diesem Kapitel zunächst den funktionalen und den Call-String-Ansatz vorstellen und deren theoretisch beste Lösungen untersuchen. 
Anschließend werden wir die Polyederanalyse betrachten, die von Cousot und Halbwachs \cite{CH78-POPL} eingeführt wurde und den funktionalen Ansatz variiert. 
Dabei werden wir zwei Mög\-lich\-kei\-ten studieren, die Summary-Informationen berechnen. 
Der eine Ansatz stammt dabei ebenfalls von Cousot und Halbwachs \cite{CH78-POPL} und benutzt Relationen, 
während Seidl et al.~\cite{seidl07} die Summary-Informationen mithilfe von Matrizenmengen bestimmen.

Der funktionale und der Call-String-Ansatz definieren Ungleichungssysteme. 
Wir werden uns zuletzt also der Frage zu, ob diese durch einen Workset-Algorithmus, wie er in \autoref{sec:wla} vorgestellt wurde, gelöst werden können.

Im Folgenden sei $(L,\sqleq, \mathcal{F})$ ein monotones Framework und $\init_L$ ein Startwert. 
Wir schreiben $\id$ und $\init$ anstelle von $\id_L$ bzw.~$\init_L$, falls das zugrundeliegende $L$ aus dem Zusammenhang eindeutig ist. 
Sei weiter $G = \{G_{\p}\}_{\p\in \Proc}$ ein Flussgraphsystem. Zu jeder Basisanweisung $\mathtt{b}$ wählen wir eine Transferfunktion $f_{\mathtt{b}} \in \mathcal{F}$.

\section{Der funktionale Ansatz}\label{sec:funktionaler-ansatz}
Ein Ansatz zur Berechnung von Informationen für die Punkte eines Programmes mit Prozeduren ist, zunächst den Effekt einer jeden Prozedur $\p \in \Proc$ durch eine Transferfunktion zu beschreiben. 
Dafür bestimmen wir für jeden Programmpunkt $u \in N_p$ diejenige Abbildung $T[u] \in \mathcal{F}$, die einer Information im Startpunkt der Prozedur $\p$ die zugehörige Information in $u$ zuordnet. 
Der Effekt der Prozedur, wie sich also Daten vom Start- bis zum Endknoten verändern, ist entsprechend die Abbildung für den Endknoten der Prozedur $T[r_\p]$. 

Da vom Startknoten zu sich selbst noch keine Transformation statt findet, ist die Identität eine mögliche Abbildung vom Startknoten zu sich selbst. 
Also ist $T[s_{\p}]$ mindestens so groß wie die Identität. 
Ist $e=(u,\mathtt{b},v)$ eine Basiskante, so werden Daten bis $v$ so transformiert, wie es der Hintereinanderschaltung der Transformation bis $u$ mit der Transformation $f_e$ der Kante $e$ entspricht. 
Ist $e$ dagegen eine Kante, an der die Prozedur $\q$ aufgerufen wird, so wird hinter die Transformation bis $u$ der Effekt der Prozedur $\q$ geschaltet, 
um eine Transformation der Daten bis $v$ zu erhalten.
\begin{figure}[ht]
  \begin{flowgraph}
	\node[startnode]	(sp) 					{$s_{\p}$};
	\node				(procp) [above=of sp]	{$\p$:};
	\node[stdnode]		(u)		[below=of sp] 	{$u$};
	\node[stdnode]		(v)		[below=of u]	{$v$};
	\node[startnode]	(sq) 	[right=of u]	{$s_{\q}$};
	\node				(procq) [above=of sq]	{$\q$:};
	\node[stdnode]		(rq)	[below=of sq] 	{$r_\q$};

	\node[annnode]		(Tsp)	[left] 	at (sp.west) 	{$\id \sqleq T[s_\p]$};
	\node[annnode]		(Tu)	[left] 	at (u.west) 	{$T[u]$};
	\node[annnode]		(Tv)	[left] 	at (v.west) 	{$\left.\begin{array}{r}f_{\mathtt{b}}\circ T[u]\\T[r_{\q}]\circ T[u]\end{array}\right\}\sqleq T[v]$};
	\node[annnode]		(Tsq)	[right] at (sq.east) 	{$T[s_\q]\sqgeq \id$};
	\node[annnode]		(Trq)	[right] at (rq.east) 	{$T[r_\q]$};
	\path[->]
	  (sp) 		edge [snakearc]											(u)
	  (u)		edge 			node {$\mathtt{b}\slash\mathtt{q()}$}	(v)
	  (sq)		edge [snakearc]				 							(rq)
	  (Tu)		edge [annarc,bend right]		(Tv)
	  (Trq)		edge [annarc,bend left]	(Tv)
	;
	
  \end{flowgraph}
  \caption{Informationsfluss im Ungleichungssystem \eqref{T-Ugs}.}
  \label{bild-T-Ugs}
\end{figure}

Zusammen ergibt sich also das Ungleichungssystem
\begin{equation}\label{T-Ugs}\begin{split}
T[s_{\p}]	&\sqgeq \id_L \text{ für } \p \in \Proc\\
T[v] 		&\sqgeq f_{\mathtt{b}} \circ T[u] 	\text{ für } e=(u,\mathtt{b},v) \in E \text{ Basiskante} \\
T[v] 		&\sqgeq T[r_{\q}] \circ T[u] 		\text{ für } e=(u,\mathtt{q()},v) \in E \text{ Aufrufkante}
\end{split}\end{equation}
Dieses wird in \autoref{bild-T-Ugs} verdeutlicht. 
Die dabei verwendeten Abbildungen 
\begin{align*}
f_{T,\id}: 						\mathcal{F}^N &\to \mathcal{F},\\g&\mapsto \id_L
\intertext{für den Startknoten einer Prozedur,}
f_{T,(u,\mathtt{b},v)}:			\mathcal{F}^N &\to \mathcal{F},\\g&\mapsto f_{\mathtt{b}} \circ g_u
\intertext{für Basiskanten und}
f_{T,(u,\mathtt{q()},v)}(g): 	\mathcal{F}^N &\to \mathcal{F},\\g&\mapsto g_{r_{\q}} \circ g_u  
\end{align*}
für Prozeduraufrufe sind offenbar monoton.
Wir schreiben später auch abkürzend $f_{T,\cdot}$ für diese Abbildungen.

Offenbar gibt es für jedes $T[u]$ nur endlich viele Ungleichungen. Das Ungleichungssystem \eqref{T-Ugs} hat also die in \autoref{sec:UGS} geforderte Gestalt.

Die gesuchte Information in einem Programmpunkt $u \in N_\p$ einer Prozedur $\p \in \Proc$ ist weiter gegeben durch die kleinste Lösung des Ungleichungssystems
\begin{equation}\label{R-Ugs}\begin{split}
R[s_{\main}]	&\sqgeq \init\\
R[v] 			&\sqgeq f_{\mathtt{b}}(R[u]) 		\text{ für } e=(u,\mathtt{b},v)\in E \text{ Basiskante} \\
R[s_{\q}]		&\sqgeq R[u]						\text{ für } e=(u,\mathtt{q()},v) \in E \text{ Aufrufkante} \\
R[v] 			&\sqgeq \underline{T}[r_{\q}](R[u])	\text{ für } e=(u,\mathtt{q()},v) \in E \text{ Aufrufkante}
\end{split}\end{equation}
wobei $R[u] \in L$ für alle $u \in N$. Die Zusammenhänge zwischen den Werten der einzelnen $R[u]$ werden in \autoref{bild-R-Ugs} dargestellt.

\begin{figure}[ht]
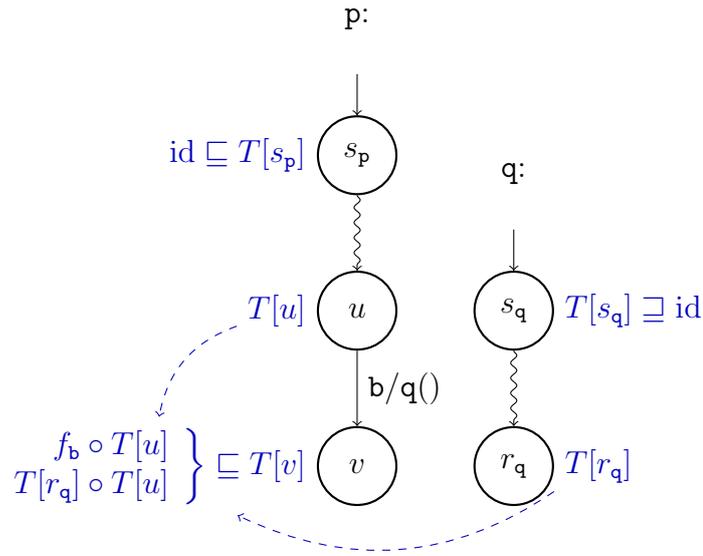

  \begin{flowgraph}
	\node[startnode]	(sp) 					{$s_{\p}$};
	\node				(procp) [above=of sp]	{$\p$:};
	\node[stdnode]		(u)		[below=of sp] 	{$u$};
	\node[stdnode]		(v)		[below=of u]	{$v$};
	\node[startnode]	(sq) 	[right=of u]	{$s_{\q}$};
	\node				(procq) [above=of sq]	{$\q$:};
	\node[stdnode]		(rq)	[below=of sq] 	{$r_\q$};

	\node[annnode]		(Ru)	[left] 	at (u.west) 	{$R[u]$};
	\node[annnode]		(Rv)	[left] 	at (v.west) 	{$\left.\begin{array}{r}f_{\mathtt{b}}(R[u])\\\underline{T}[r_{\q}](R[u])\end{array}\right\}\sqleq R[v]$};
	\node[annnode]		(Rsq)	[right] at (sq.east) 	{$R[s_\q]\sqgeq R[u]$};
	\path[->]
	  (sp) 		edge [snakearc] 										(u)
	  (u)		edge 			node {$\mathtt{b}\slash\mathtt{q()}$}	(v)
	  (sq)		edge [snakearc]				 							(rq)
	  (Ru)		edge [annarc,bend right]		(Rv)
	  (Ru)		edge [annarc,bend left]			(Rsq)
	;
  \end{flowgraph}
  \caption{Informationsfluss im Ungleichungssystem \eqref{R-Ugs}.}
  \label{bild-R-Ugs}
\end{figure}

Wie in \autoref{sec:fg} ist im Startknoten als Information mindestens die Startinformation $\init$ enthalten. 
Der Startknoten einer beliebigen Prozedur enthält die Information des Knotens, von dem aus die Prozedur aufgerufen wurde. 
Für jede Kante $e = (u,\mathtt{s},v)$ enthält $v$ die Information aus $u$, nachdem diese entsprechend des Effektes der Kante $e$ verändert wurde. 
Im Falle einer Basiskante ist dieser Effekt $f_e$, im Falle einer Aufrufkante ist es der Effekt der aufgerufenen Prozedur $\q$, d.h.~$\underline{T}[r_{\q}]$.
Die dabei verwendeten Abbildungen 
\begin{align*}
f_{R,\init} 						:L^N &\to L,\\ \mathcal{S} &\mapsto \init
\intertext{für die Startinformation,}
f_{R,(u,\mathtt{b},v)} 				:L^N &\to L,\\ \mathcal{S} &\mapsto f_{\mathtt{b}} (\mathcal{S}_u) 
\intertext{für Basiskanten und}
f_{R,(u,\mathtt{q()},v),\text{ent}} :L^N &\to L,\\ \mathcal{S} &\mapsto \mathcal{S}_u
\intertext{sowie}
f_{R,(u,\mathtt{q()},v),\text{ret}} :L^N &\to L,\\ \mathcal{S} &\mapsto \underline{T}[r_{\q}] (\mathcal{S}_u)
\end{align*}
für Prozeduraufrufe sind offenbar ebenfalls monoton.
Wie für die verallgemeinerten Transferfunktionen des Ungleichungssystems \eqref{T-Ugs} schreiben wir hier abkürzend $f_{R,\cdot}$ für diese Abbildungen.

Offenbar gibt es für jedes $R[u]$ nur endlich viele Ungleichungen. Das Ungleichungssystem \eqref{R-Ugs} hat also die in \autoref{sec:UGS} geforderte Gestalt.

Mit dem in diesem Abschnitt vorgestellten funktionalen Ansatz können Programme analysiert werden, die Prozeduren enthalten. 
Dazu müssen allerdings zunächst Transferfunktionen berechnet werden, die die Effekte der Prozeduren beschreiben. 
Den Wert in einem Rückkehrknoten nach einem Prozeduraufruf enthält man nun nicht aus den in dieser Prozedur berechneten Werten, 
sondern durch Anwenden der für die aufgerufene Prozedur berechneten Transferfunktion auf den Wert aus dem Aufrufknoten. 
Dadurch werden Reihenfolgen von Prozeduraufrufen und deren Verschachtelungen nicht explizit behandelt. 
Im nächsten Abschnitt führen wir eine Analyse ein, die die Aufrufreihenfolge von Prozeduren explizit behandelt.

\section{Der Call-String-Ansatz}\label{sec:callstring-ansatz}
Ein anderer Ansatz zur interprozeduralen Analyse ist, eine Analyse eines Programms mit nur einer Prozedur zu simulieren. 
Die Idee dabei ist, im Flussgraphen der $\main$-Prozedur des Programms jede Aufrufkante durch eine Kopie des Flussgraphen der aufgerufenen Prozedur zu ersetzen. 
Realisiert wird dies durch Identifikation der verschiedenen Kopien mit sogenannten \emph{Call-Strings}. 
Diese beschreiben, welche Prozeduraufrufe dem Aufruf der betrachteten Prozedur vorangegegangen sind, und beschreiben die Kopie dadurch eindeutig.

Anstatt vieler Flussgraphen wird also nur ein einziger betrachtet. 
Dieser erweiterte Flussgraph enthält dann nicht nur Basiskanten, die mit Basisanweisungen beschriftet sind, sondern auch Kanten, um den Eintritt und die Rückkehr aus einer Prozedur zu behandeln. 
Wir erweitern dazu die Kantenmenge $E$ um Kanten, die den Eintritt und die Rückkehr aus einer Prozedur beschreiben. Dabei unterscheiden wir zwei Arten von Prozeduraufrufen: 
Auf der einen Seite gibt es diejenigen Aufrufe, denen eine zugehörige Rückkehr folgt. Solche Aufrufe codieren wir durch $\call$, die zugehörige Rückkehr wird mit $\ret$ beschrieben. 
Auf der anderen Seite gibt es auch Aufrufe, bei denen der Pfad aus dieser aufgerufenen Prozedur nicht mehr zurückkehrt. Solche Aufrufe bezeichnen wir mit $\enter$. 
Somit können wir die Ausführungspfade des Programms beschreiben als Worte über dem Alphabet
\begin{align*}
\widetilde{E} := &\{e \in E \mid e \text{ Basiskante}\} \\
	&\cup\{(u,\call,s_{\q}) \mid \exists (u,\mathtt{q()},v) \in E\} \\
	&\cup\{(r_\q,\ret,v) \mid \exists (u,\mathtt{q()},v) \in E\} \\
	&\cup\{(u,\enter,s_\q) \mid \exists (u,\mathtt{q()},v) \in E\}
\end{align*}
Diese Pfade sollen \quotes{wohlgeschachtelt} sein, d.h.~jede Prozedur erst dann wieder verlassen werden, wenn alle in ihr aufgerufenen Prozeduren bereits beendet wurden. 
Genauer bedeutet \quotes{wohlgeschachtelt}, dass der Aufruf einer Prozedur $\p$ in einem Pfad, auf den später im Pfad eine zugehörige Rückkehr folgt, 
mit $\call$ und die entsprechende Rückkehrkante mit $\ret$ bezeichnet wird. 
Sollte zwischen diesen beiden Kanten eine weitere Prozedur $\q$ augerufen werden, so erwarten wir, 
dass es zwischen der Aufrufkante der Prozedur $\q$ und der Rückkehrkante der Prozedur $\p$ eine Kante gibt, die die Rückkehr aus der Prozedur $\q$ kennzeichnet. 
Ist dagegen $\p$ eine Prozedur, die im Pfad aufgerufen wird, ohne dass es im Pfad eine zugehörige Rückkehr gibt, so wird diese Aufrufkante mit $\enter$ gekennzeichnet. 
Diese Bedingungen implizieren insbesondere, dass auf eine $\call$-Kante niemals eine $\enter$-Kante folgen kann. Die Menge dieser Pfade werden wir in \ref{sec:koinzidenz} einführen.

Eine erste Idee, um eine intraprozedurale Analyse zu simulieren, ist es, im Flussgraphen von $\main$ eine Aufrufkante einer Prozedur $\p$ durch den zu $\p$ gehörigen Flussgraphen zu ersetzen. 
Da die Benennung der Knoten eines Flussgraphen eindeutig sein muss, können verschiedene Kopien des Flussgraphen von $\p$ nicht mit denselben Knotennamen benannt werden. 
Die verschiedenen Kopien unterscheiden sich jedoch genau durch die Call-Strings.

Intuitiv beschreibt ein Call-String eine Liste von denjenigen Aufrufkanten, deren zugehörige Prozedur noch nicht wieder beendet wurden. 
\begin{figure}[ht]
  \begin{flowgraph}
	\node[startnode]	(smain) 				{\small{$s_{\main}$}};
	\node[stdnode]		(u1) 	[below=of smain]	{$u_1$};
	\node[fadenode]		(v1) 	[below=of u1] 		{$v_1$};
	\node[stdnode]		(sq1) 	[right=of u1]		{$s_\mathtt{q_1}$};
	\node[stdnode]		(u2) 	[below=of sq1]		{$u_2$};
	\node[fadenode]		(v2) 	[below=of u2] 		{$v_2$};
	\node[stdnode]		(sq2) 	[right=of u2]		{$s_\mathtt{q_2}$};
	\node[stdnode]		(u3) 	[below=of sq2]		{$u_3$};
	\node[fadenode]		(v3) 	[below=of u3] 		{$v_3$};
	\node 				(dots)	[right=of u3]		{$\dots$};
	\node[stdnode]		(sqn) 	[right=of dots]		{$s_\mathtt{q_n}$};
	\node[stdnode]		(u) 	[below=of sqn]		{$u$};
	\node[callstring]	(csmain)[above right] at (smain.north east) {$\eps$};
	\node[callstring]	(csq1)	[right] at (sq1.north east) 	{$(u_1,\mathtt{q_1()},v_1)$};
	\node[callstring]	(csq2)	[right] at (sq2.north east) 	{$(u_1,\mathtt{q_1()},v_1)\cdot(u_2,\mathtt{q_2()},v_2)$};
	\node[callstring]	(csqn)	[right] at (sqn.north east) 	{$\begin{array}{l} (u_1,\mathtt{q_1()},v_1)\\\cdot(u_2,\mathtt{q_2()},v_2)\\\cdots(u_{n-1},\mathtt{q_{n-1}}(),v_{n-1})\end{array}$};
	\path[->]	
		(smain) edge [snakearc] 							(u1)
		(u1) 	edge [fadearc,swap]	node {$\mathtt{q_1()}$} (v1)
		(u1)	edge [dashed] 		node {$\enter$} 		(sq1)
		(sq1)	edge [snakearc] 				 			(u2)
		(u2) 	edge [fadearc,swap]	node {$\mathtt{q_2()}$}	(v2)
		(u2)	edge [dashed] 		node {$\enter$} 		(sq2)
		(sq2)	edge [snakearc] 				 			(u3)
		(u3) 	edge [fadearc,swap]	node {$\mathtt{q_3()}$}	(v3)
		(u3)	edge [dashed]		node {$\enter$}			(dots)
		(dots)	edge [dashed]		node {$\enter$}			(sqn)
		(sqn)	edge [snakearc]								(u)
		(v1) 	edge [dashed, bend right=15, callstringcolor] (csq1)
		(v2) 	edge [dashed, bend right=12, callstringcolor] (csq2)
		(csq1) 	edge [dashed, bend left=15, callstringcolor] (csq2)
	;

  \end{flowgraph} 
  \caption{Aufbau eines Call-Strings.}
  \label{bild-callstrings}
\end{figure}

Für jeden Prozeduraufruf $e=(u,\mathtt{q()},v)\in E_\p$ konkatenieren wir an die Call-Strings von $\p$ die Aufrufkante $e$ (vergleiche auch Abbildung \ref{bild-callstrings}):
\begin{equation}\label{CS-Ugs}\begin{split}
\CS[\main]	&\supseteq \{\eps\} \\
\CS[\q]		&\supseteq \CS[\p] \cdot e \text{ für } e=(u,\mathtt{q()},v)\in E_\p
\end{split}\end{equation}
Dabei ist $\CS[\p] \subset E^\ast$ für jedes $\p \in \Proc$ und $\CS[\p] \cdot e := \{\pi \cdot e \mid \pi \in \CS[\p]\}$. 
Die kleinste Lösung, die wir mit $\underline{\CS}$ bezeichnen, liefert zu jeder Prozedur dann genau die möglichen Call-Strings. 
Ein Tupel $(u,w) \in N_{\p} \times \underline{\CS}[\p]$ beschreibt nun diejenige \quotes{Kopie} des Knotens $u$, die genau durch Prozeduraufrufe auf $w$ entsteht. 
Insbesondere gehören die $(u,\eps) \in N_{\main} \times \underline{\CS}[\main]$ zu dem \quotes{ursprünglichen} Flussgraphen, in den die anderen \quotes{hineinkopiert} werden.

Die Abbildungen $f_{\CS,\eps}, f_{\CS,e}: (2^{E^\ast})^{\Proc} \to 2^{E^\ast}$, die definiert sind durch
\begin{align*}
f_{\CS,\eps}(A)	&:=\{\eps\} \\
f_{\CS,e}(A)	&:=\{w \cdot e \mid w \in A_u\} \text{ für ein } e \in E
\end{align*}
für $A \in 2^{E^\ast}$ sind offensichtlich monoton. 

Offenbar gibt es für jedes $\CS[\p]$ nur endlich viele Ungleichungen. Das Ungleichungssystem \eqref{CS-Ugs} hat also die in \autoref{sec:UGS} geforderte Gestalt.

Jede Kopie der Prozedur $\p$ korrespondiert zu einem Aufruf der Kopie $\p$. 
Wird innerhalb derselben Prozedur $\q\in\Proc$ an verschiedenen Kanten $e_1,e_2$ die Prozedur $\p$ aufgerufen, so enden die Call-Strings der Kopien auf $e_1$ bzw.~$e_2$ und sind deshalb verschieden. 
Ebenso unterscheiden sich die Aufrufkanten, wenn die aufrufenden Prozeduren verschieden sind, und mit gleichem Argument unterscheiden sich dann auch die Call-Strings.
Wird in zwei Kopien derselben Prozedur an derselben Stelle $\p$ aufgerufen, so unterscheiden sich 
die Call-Strings der aufrufenden Prozeduren und somit auch die beiden Call-Strings der Aufrufe von $\p$.

Somit ist es ein natürlicher Ansatz, die verschiedenen Kopien durch die zugehörigen Call-Strings zu unterscheiden. 
Die gesuchte Information in einem Programmpunkt $u \in N_{\p}$ bezüglich eines Call-Strings $w \in \underline{\CS}[\p]$ ist dann gegeben durch $A[u,w] \in L$.
\begin{equation}\label{A-Ugs}\begin{split}
A[s_{\main}](\eps)	&\sqgeq \init \\
A[v](w) 			&\sqgeq f_{\mathtt{b}}(A[u](w))\text{ für } e=(u,\mathtt{b},v) \in E_\p \text{ Basiskante} \\
A[s_{\q}](w\cdot e) &\sqgeq A[u](w)				\text{ für } e=(u,\mathtt{q()},v) \in E_\p \text{ Aufrufkante} \\
A[v](w) 			&\sqgeq A[r_{\q}](w\cdot e)	\text{ für } e=(u,\mathtt{q()},v) \in E_\p \text{ Aufrufkante}
\end{split}\end{equation}
für alle $\p \in \Proc$ und $w \in \underline{\CS}[\p]$. Dabei ist also $A: N_{\CS} \to L$ für \[N_{\CS} := \bigcup\{N_\p\times \underline{\CS}[\p] \mid \p \in \Proc\}.\]
Die kleinste Lösung bezeichnen wir wiederum mit $\underline{A}$. Die Zusammenhänge zwischen den einzelnen Informationen werden in \autoref{bild-A-Ugs} verdeutlicht.
\begin{figure}[ht]
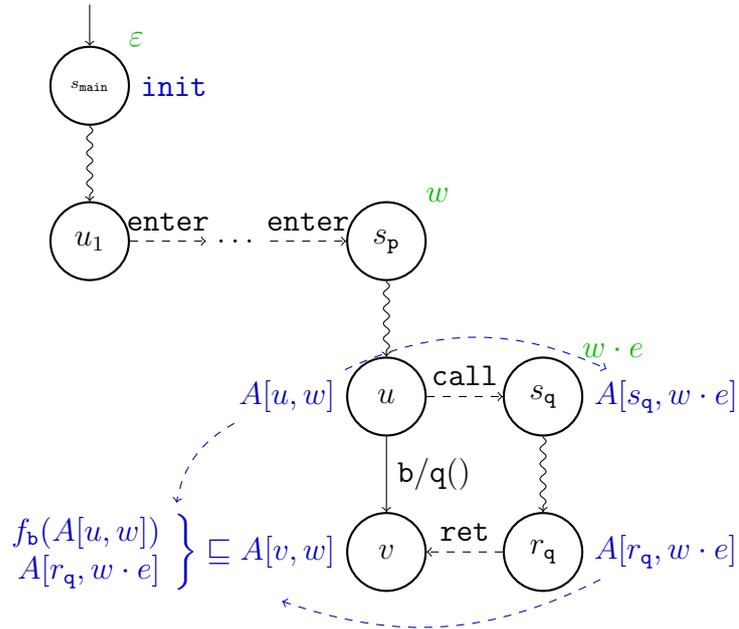

  \begin{flowgraph}
	\node[startnode]	(smain) 					{\tiny{$s_{\main}$}};
	\node[stdnode]		(u1)	[below=of smain] 	{$u_1$};
	\node 				(dots)	[right=of u1]		{$\dots$};
	\node[stdnode]		(sp)	[right=of dots]		{$s_\p$};
	\node[stdnode]		(u)		[below=of sp]		{$u$};
	\node[stdnode]		(sq)	[right=of u]		{$s_\q$};
	\node[stdnode]		(rq)	[below=of sq]		{$r_\q$};
	\node[stdnode]		(v)		[below=of u]		{$v$};
	\node[callstring] 	(csmain)[above right] at (smain.north east) {$\eps$};
	\node[callstring]	(csp)	[above right] at (sp.north east) 	{$w$};
	\node[callstring]	(csq)	[above right] at (sq.north east) 	{$w\cdot e$};
	\node[annnode]		(init)	[right] at (smain.east) {$\init$};
	\node[annnode]		(Auw)	[left] 	at (u.west) 	{$A[u,w]$};
	\node[annnode]		(Asqwv)	[right] at (sq.east) 	{$A[s_\q,w\cdot e]$};
	\node[annnode]		(Arqwv)	[right] at (rq.east) 	{$A[r_\q,w\cdot e]$};
	\node[annnode]		(Avw)	[left] 	at (v.west) 	{$\left.\begin{array}{r}f_{\mathtt{b}}(A[u,w])\\A[r_{\q},w\cdot e]\end{array}\right\}\sqleq A[v,w]$};
	\path[->]
	  (smain) 	edge [snakearc] 										(u1)
	  (u1)		edge [dashed]	node {$\enter$} 						(dots)
	  (dots)	edge [dashed]	node {$\enter$} 						(sp)
	  (sp)		edge [snakearc]				 							(u)
	  (u)		edge [dashed]	node {$\call$} 							(sq)
	  (u)		edge 			node {$\mathtt{b}\slash\mathtt{q()}$}	(v)
	  (sq)		edge [snakearc]				 							(rq)
	  (rq)		edge [dashed]	node[swap] {$\ret$} 					(v)
	  (Auw)		edge [annarc,bend left=25]								(Asqwv)
	  (Auw)		edge [annarc,bend right=25]								(Avw)
	  (Arqwv)	edge [annarc,bend left=23]								(Avw)
	;
  \end{flowgraph}
  \caption{Informationsfluss im Ungleichungssystem \eqref{A-Ugs}.}
  \label{bild-A-Ugs}
\end{figure}

Das Ungleichungssystem entsteht folgendermaßen: 
Wie im funktionalen Ansatz enthält der \quotes{ursprüngliche} Flussgraph der Prozedur $\main$ im Startknoten mindestens die Startinformation $\init$, was die erste Ungleichung widerspiegelt.
Die zweite Ungleichung beschreibt, dass jeder Knoten die durch eine eingehende Basiskante veränderte Information des Vorgängerknotens enthält, und zwar auf derselben Call-String-Ebene.
Beim Aufruf einer Prozedur wird ein Call-String $w$ um die Aufrufkante verlängert und so zu einem Call-String $w \cdot e$. 
Dabei wird in den Startknoten der aufgerufenen Prozedur mit dem entsprechend verlängerten Call-String die Information des Knotens, an dem diese Prozedur aufgerufen wird, kopiert. 
Bei der Rückkehr aus dieser Prozedur wird ebenfalls Information kopiert, und zwar aus dem Endknoten der Prozedur zum Rückkehrknoten der aufrufenden Prozedur mit entsprechend verkürztem Call-String. 
Dies wird durch die letzten beiden Ungleichungen beschrieben.

Die verallgemeinerten Transferfunktionen, die das Ungleichungssystem \eqref{A-Ugs} beschreiben, sind dann gegeben durch
\begin{align*}
f_{A,\init,\eps} 												:L^{N_{\CS}} &\to L,\\ \mathcal{S} &\mapsto \init
\intertext{für die Startinformation,}
f_{A,(u,\mathtt{b},v),w} 										:L^{N_{\CS}} &\to L,\\ \mathcal{S} &\mapsto f_{\mathtt{b}} (\mathcal{S}_{u,w}) 
\intertext{für Basiskanten und}
f_{A,(u,\mathtt{q()},v),w,\text{ent}}							:L^{N_{\CS}} &\to L,\\ \mathcal{S} &\mapsto \mathcal{S}_{u,w}
\intertext{sowie}
f_{A,(u,\mathtt{q()},v),w\cdot (u,\mathtt{q()},v),\text{ret}} 	:L^{N_{\CS}} &\to L,\\ \mathcal{S} &\mapsto \mathcal{S}_{r_\q,w\cdot(u,\mathtt{q},v)}
\end{align*}
für Prozeduraufrufe. Offenbar sind diese Abbildungen monoton.
Wir schreiben später auch abkürzend $f_{A,\cdot}$ für diese Abbildungen.

Offenbar gibt es für jedes $A[u,w]$ nur endlich viele Ungleichungen. Das Ungleichungssystem \eqref{A-Ugs} hat also die in \autoref{sec:UGS} geforderte Gestalt.

Um nun die gesuchte Information in einem Knoten $u$ zu bestimmen, haben wir also mehrere \quotes{Kopien} dieses Knotens $u$ angelegt und können für jede dieser Kopien eine Information berechnen. 
Die gewünschte Information in $u$ erhalten wir durch
\[\hat{A}[u] := \bigsqcup \big\{\underline{A}[u,w] \mid w \in \underline{\CS}[\p]\big\}.\] 
In der Praxis werden oft \emph{beschränkte} Call-Strings verwendet: Ist das Programm rekursiv, so werden die Call-Strings beliebig lang. 
Damit enthält aber das Ungleichungssystem \eqref{A-Ugs} unendliche viele Variablen und ist somit in endlicher Zeit nicht lösbar. 
Beschränkt man dagegen die Länge der Call-Strings, so sind nur endlich viele Ungleichungen zu lösen. 
Wir werden im Folgenden aber sehen, dass unbeschränkte Call-Strings notwendig sind, damit der Call-String-Ansatz präzise Lösungen liefert.

Wir haben in diesem Abschnitt einen Ansatz vorgestellt, der mithilfe von Call-Strings und zusätzlichen Kanten die verschiedenen Prozeduren des Programms als eine einzige behandelt 
und dadurch den Programmfluss explizit behandelt. 
Im folgenden Abschnitt werden wir prüfen, ob dieser und der zuvor vorgestellte funktionale Ansatz korrekte oder sogar präzise Lösungen liefern.

\section{Koinzidenz}\label{sec:koinzidenz}
In den vorigen Abschnitten haben wir zwei Ansätze zur interprozeduralen Analyse vorgestellt. 
Eine andere und natürliche Herangehensweise an die Berechnung gesuchter Informationen ist die Betrachtung aller erreichenden Pfade zu einem Programmpunkt. 
Dabei werden zunächst alle Pfade bestimmt, die zu einem Programmpunkt führen, und anschließend die Basisanweisungen entlang dieser Kanten miteinander verknpüft. 
Für jeden Pfad erhält man so eine in dem Programmpunkt gültige Information. Die kleinste obere Schranke dieser Informationen ist die gesuchte. 
Diese Lösung wird \emph{$\MOP$-Lösung} genannt. 
Dabei steht {$\MOP$} für \quotes{Meet Over All Paths}. 
Der Name kommt von der Bezeichnung \quotes{Meet} für $\sqcap$ und stammt aus den Beginnen der Analyse, wo eine Information als präziser eingestuft wird, je größer sie ist. 
In dieser Arbeit wäre also die Bezeichnung \emph{Join Over All Paths} (für den \quotes{Join} $\sqcup$) passender. 
Wir nennen eine Analyse korrekt, wenn sie diese Lösung approximiert, und präzise, 
wenn sie sogar diese Lösung berechnet. 
In diesem Abschnitt werden wir zunächst die $\MOP$-Lösungen definieren und anschließend zeigen, 
dass der funktionale und der Call-String-Ansatz korrekt und unter geeigneten Voraussetzungen sogar präzise ist. 
Entsprechende Sätze zur Übereinstimmung der {$\MOP$}- mit der kleinsten Lösung eines (Un-)Gleichungssystems finden sich in \cite{Niel}.

Bevor wir die $\MOP$-Lösung definieren können, müssen wir zunächst die Pfade definieren. Wir unterscheiden zwei Arten von Pfaden: 
\emph{Same-Level-Pfade} sind Pfade, die mit dem Startknoten einer Prozedur beginnen und mit einem beliebigen Knoten derselben Prozedur enden. 
Von Interesse sind aber auch Pfade mit Startknoten $s_{\main}$, die Konkatenationen von Same-Level-Pfaden und Prozeduraufrufen (ohne zugehörige Rückkehr) sind.
Diese beiden Sorten von Pfaden lassen sich als kleinste Lösungen zweier Ungleichungssysteme schreiben: Das Ungleichungssystem
\begin{equation}\label{SLP-Ugs}\begin{split}
\SLP[s_{\p}]& \supseteq \{\eps\} \\
\SLP[v]		& \supseteq \SLP[u] \cdot e \text{ für } e= (u,\mathtt{b},v)\in E_\p \text{ Basiskante} \\
\SLP[v]		& \supseteq \SLP[u] \cdot (u, \call, s_{\q}) \cdot \SLP[r_{\q}] \cdot (r_{\q}, \ret, v) \text{ für } e= (u,\mathtt{q()},v)\in E_\p 
\end{split}\end{equation}
liefert die Same-Level-Pfade mit $\SLP[u] \in \widetilde{E}^\ast$. 
\begin{figure}[ht]
  \begin{flowgraph}
	\node[startnode]	(sp)					{$s_\p$};
	\node[stdnode]		(u') 	[below=of sp]	{$u$};
	\node[stdnode]		(u) 	[below=of u']	{$v$};
	\node[stdnode]		(sq) 	[right=of u]		{$s_\q$};
	\node[stdnode]		(rq) 	[below=of sq]		{$r_\q$};
	\node[stdnode]		(v) 	[below=of u]		{$w$};
	\node				(empty2)[below=of v]		{};
	\path[->]	

		(sp)	edge [snakearc] 					(u')
		(u')	edge []		 	node {$\mathtt{b}$}	(u)
		(u)		edge [dashed] 	node {$\call$} 		(sq)
		(sq)	edge [snakearc] 		 			(rq)
		(rq) 	edge [swap,dashed]		node {$\ret$}		(v)
		(v) 	edge [snakearc]						(empty2)
	;

  \end{flowgraph}
  \caption{Aufbau eines Same-Leve-Pfades.}
  \label{bild-slp}
\end{figure}

Der leere Pfad $\eps$ ist ein Pfad von $s_{\p}$ nach $s_{\p}$. Dies motiviert die erste Ungleichung. 
Ist $e= (u,\mathtt{b},v)$ eine Basiskante und $\pi$ ein Same-Level-Pfad nach $u$, so ist offenbar die Konkatenation $\pi \cdot e$ ein Same-Level-Pfad nach $v$, was die zweite Ungleichung erklärt. 
Ist dagegen $e= (u,\mathtt{q()},v)$ ein Prozeduraufruf, $\pi$ ein Same-Level-Pfad nach $u$ und $\pi'$ ein Same-Level-Pfad nach $r_\q$, so erhält man einen Same-Level-Pfad nach $v$ wie folgt: 
Zunächst konkateniert man $\pi$ mit einer $\call$-Kante zum Startknoten $s_{\q}$ der aufgerufenen Prozedur $\q$. 
Daran wird der Same-Level-Pfad $\pi'$ zum Rückkehrknoten $r_{\q}$ dieser Prozedur $\q$ angehängt. 
Zuletzt wird die entsprechende $\ret$-Kante nach $v$ angefügt, die den Rückkehr aus der Prozedur beschreibt. 
Diese Vorgehensweise wird durch die dritte Ungleichung ausgedrückt. Ein solcher Pfad ist schematisch in \autoref{bild-slp} illustriert.
Die Definition garantiert offenbar, dass Same-Level-Pfade wohlgeschachtelt sind.

Ähnlich sind nun die Pfade von $s_{\main}$ zu einem beliebigen Punkt gegeben durch
\begin{equation}\label{GP-Ugs}\begin{split}
\GP[s_{\main}]	& \supseteq \{\eps\} \\
\GP[u]			& \supseteq \GP[s_{\p}] \cdot \underline{\SLP}[u] \text{ für } u \in N_{\p}\\
\GP[s_{\q}]		& \supseteq \GP[u] \cdot (u, \mathtt{\enter}, s_\q) \text{ für } e= (u,\mathtt{q()},v)\in E_\p,
\end{split}\end{equation}
wobei wiederum $\GP: N \to \widetilde{E}^\ast$. 
\begin{figure}[ht]
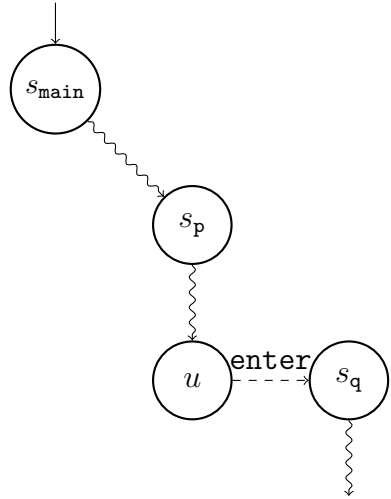

  \begin{flowgraph}
	\node[startnode]	(smain)							{\small{$s_{\main}$}};
	\node[stdnode]		(sp) 	[below right=of smain]	{$s_\p$};
	\node[stdnode]		(u) 	[below=of sp]			{$u$};
	\node[stdnode]		(sq) 	[right=of u]		{$s_\q$};
	\node				(empty)	[below=of sq]		{};
	\path[->]	
		(smain)	edge [snakearc] 					(sp)
		(sp)	edge [snakearc] 					(u)
		(u)		edge [dashed] 	node {$\enter$}		(sq)
		(sq) 	edge [snakearc]						(empty)
	;

  \end{flowgraph}
  \caption{Aufbau eines (beliebigen) Pfades.}
  \label{bild-gp}
\end{figure}

Die erste Ungleichung ergibt sich wie im vorigen Ungleichungssystem. 
Ist $e= (u,\mathtt{q()},v)$ Kante einer Prozedur $\p$ und gibt es einen Pfad nach $u$, so entsteht daraus ein Pfad nach $s_\q$ durch Konkatenation einer $\enter$-Kante nach $s_\q$ an den Pfad nach $u$.
Gibt es einen Pfad zum Startknoten $s_\p$ einer Prozedur $\p$ und einen Same-Level-Pfad zu einem Programmpunkt $u \in N_{\p}$ dieser Prozedur, 
so erhält man durch Konkatentation dieser beiden Pfade einen Pfad nach $u$. Dies erklärt die anderen beiden Ungleichungen. 
Ein solcher Pfad ist schematisch in \autoref{bild-gp} zu sehen. 
Auch diese Definition garantiert offenbar, dass die Pfade wohlgeschachtelt sind.

Für jede Basiskante $e=(u,\mathtt{b},v) \in E$ haben wir bereits ein $f_\mathtt{b} \in \mathcal{F}$ gewählt. Dies erweitern wir nun für Pfade:
Für $\pi=e_1\dots e_k \in \underline{\GP}[u]$ definieren wir 
\[f_\pi = f_k \circ \ldots \circ f_1 \in \mathcal{F},\]
wobei
\[f_i :=
\begin{cases}
f_{\mathtt{b_i}} & \text{für eine Basiskante } e_i=(u_i,\mathtt{b_i},v_i) \in E \\
\id & \text{sonst}.
\end{cases}\]
Die Abbildung $f_\pi$ beschreibt also die Verknüpfung der einzelnen Transferfunktionen entlang der Kanten in der Reihenfolge ihrer Ausführung. Weiter definieren wir $f_{\eps} := \id$.

Alternativ könnte man auf die Unterscheidung von $\enter$- und $\call$-Kanten verzichten und stattdessen die Pfade entsprechend der Ausführung der Programmes sukzessiv bilden. 
Dadurch können Aufrufkanten von bereits wieder verlassenen Prozeduren nicht von denen unterschieden werden, deren Prozeduren noch nicht abgearbeitet wurden.
In unserem Ansatz, die Pfade mithilfe eines Ungleichungssystems zu berechnen, wird diese Unterscheidung dagegen gemacht. 
Dies wird im Folgenden nützlich sein, da zum einen diese Pfade in offensichtlicher Weise wohlgeschachtelt sind und es zum anderen direkt zu der Definition eines \emph{Call-Strings} korrespondiert.

Wir wollen nun Pfade mit den Call-Strings in Verbindung bringen. Dazu definieren wir zunächst diejenigen Pfade zum Startknoten einer Prozedur, die mit einem gegebenen Call-String korrespondieren.
Sei dazu $\p \in \Proc$ eine Prozedur und $w \in \underline{\CS}[\p]$ ein Call-String der Gestalt $w = (u_1,\mathtt{p_1()}, v_1) \dots (u_n,\mathtt{p_n()}, v_n)$ mit $\mathtt{p_n} = \p$. 
Wir erhalten nun die passenden Pfade zu einem Knoten $u \in N_\p$, indem wir zunächst an einen Same-Level-Pfad zum Knoten $u_1$ eine $\enter$-Kante nach $s_{\mathtt{p_1}}$ hängen. 
Daran wiederum konkatenieren wir einen Same-Level-Pfad nach $u_2$ und setzen dies solange fort, bis daraus ein Pfad nach $s_\p$ entstanden ist. 
An diesen hängen wir zuletzt einen Same-Level-Pfad nach $u$ an. Diese Pfade sind also gegeben durch
\[\Path(u,w) := \{\pi_1 \cdot (u_1,\enter,s_{\mathtt{p_1}}) \cdot \pi_2 \cdots (u_n,\enter,s_{\mathtt{p}}) \cdot \pi_{n+1} 
\mid \pi_i \in \underline{SL}[u_i] \}.\]
Dabei sei $u_{n+1} = u$. Die Menge aller möglichen Pfade nach $u$ erhalten wir nun, indem wir die Pfade bezüglich aller Call-Strings von $\p$ vereinigen. Dies zeigt das folgende Lemma.
\begin{lemma}\label{pfadlemma}
Für jede Prozedur $\p \in \Proc$ und jeden Knoten $u \in N_p$ gilt
\[\underline{\GP}[u] = \bigcup \big\{ \Path(u,w) \mid w \in \underline{\CS}[\p] \big\} .\]
\end{lemma}
\begin{proof}
Seien $\p \in \Proc$, $u \in N_\p$ und $\Path(u) := \bigcup \{ \Path(u,w) \mid w \in \underline{\CS}[\p] \}$ die Menge der Pfade nach $u$, die zu einem Call-String korrespondieren. 
Um $\underline{\GP}[u] \subseteq \Path(u)$ zu zeigen, rechnen wir nach, dass $(\Path(\node))_{\node \in N}$ eine Lösung des Ungleichungssystem \eqref{GP-Ugs} ist.

Sei also zunächst $u = s_{\main}$. Es ist
\begin{align*}
\Path(s_{\main})
\supseteq \Path(s_{\main}, \eps)
\supseteq \{\eps\}
\end{align*}
und die erste Ungleichung des Ungleichungssystems ist erfüllt.

Sei nun $u \in N_\p$ beliebig. Zu zeigen ist $\Path(u) \supseteq \Path(s_\p) \cdot \underline{\SLP}[u] $. 
Sei $\pi \in \Path(s_\p)$. 
Dann gibt es einen Call-String $w$ mit $\pi \in \Path(s_\p,w)$ und 
$\pi$ hat die Gestalt 
\[\pi =\pi_1 \cdot (u_1,\enter,s_{\mathtt{p_1}}) \cdot \pi_2 \cdots (u_n,\enter,s_{\mathtt{p}}) \cdot \pi_{n+1}\] 
wie in \autoref{bild-pfadlemma-u} 
mit $\pi_i \in \underline{\SLP}[u_i]$ für $1 \le i \le n$ und $\pi_{n+1}\in\underline{\SLP}[s_\p]$. 
Für einen beliebigen Same-Level-Pfad $\pi' \in \underline{\SLP}[u]$ ist nach Konstruktion auch 
\[\pi_{n+1}\cdot \pi' \in \underline{\SLP}[u].\] 
Dann ist aber nach Definition 
\[\pi \cdot \pi' \in \Path(u,w)\subseteq \Path(u).\]
Somit gilt auch die zweite Ungleichung.
\begin{figure}[ht]
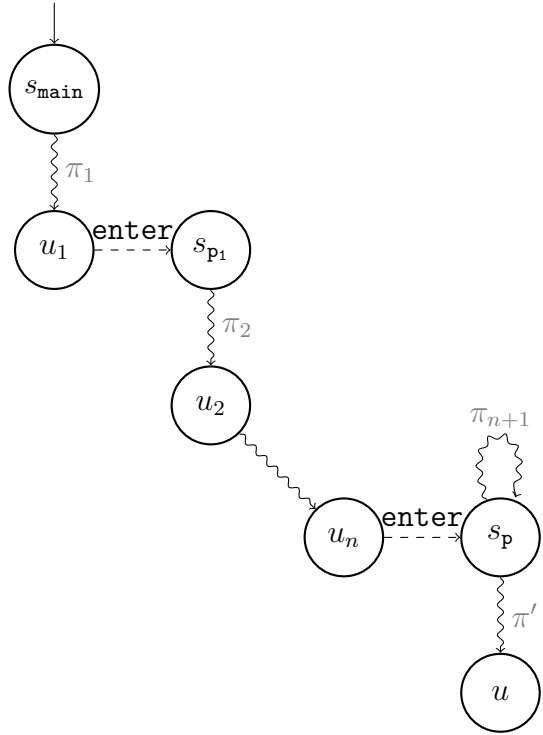

  \begin{flowgraph}
	\node[startnode]	(smain) 					{\small{$s_{\main}$}};
	\node[stdnode]		(u1) 	[below=of smain]	{$u_1$};
	\node[stdnode]		(sp1) 	[right=of u1]		{$s_\mathtt{p_1}$};
	\node[stdnode]		(u2) 	[below=of sp1]		{$u_2$};
	\node[stdnode]		(un) 	[below right=of u2]	{$u_n$};
	\node[stdnode]		(sp) 	[right=of un]		{$s_\mathtt{p}$};
	\node[stdnode]		(u) 	[below=of sp]		{$u$};
	\path[->]	
		(smain) edge [snakearc] 		node[edgenamecolor] {$\pi_1$}					(u1)
		(u1)	edge [dashed] 				node {$\enter$} 								(sp1)
		(sp1)	edge [snakearc] 		node[edgenamecolor] {$\pi_2$}					(u2)
		(u2)	edge [snakearc]															(un) 
		(un)	edge [dashed] 				node {$\enter$} 								(sp)
		(sp)	edge [out=110,in=70,loop,snakearc]	node[edgenamecolor] {$\pi_{n+1}$}	(sp)
		(sp)	edge [snakearc]			node[edgenamecolor] {$\pi'$}					(u)
	;

  \end{flowgraph}
  \caption{Beweis von Lemma \ref{pfadlemma}, $\Path(u) \supseteq \Path(s_\p) \cdot \underline{\SLP}[u]$.}
  \label{bild-pfadlemma-u}
\end{figure}

Sei nun $e=(u,\mathtt{q()}, v) \in E_\p$ eine Aufrufkante. Zu zeigen ist $\Path(s_\q) \supseteq \Path(u) \cdot (u,\enter,s_\q)$.
Sei also $\pi \in \Path(u)$. 
Dann gibt es einen Call-String $w \in \underline{\CS}[\p]$, so dass $\pi \in \Path(u,w)$. 
Der Pfad $\pi$ hat also die Gestalt
\[\pi = \pi_1 \cdot (u_1,\enter,s_{\mathtt{p_1}}) \cdot \pi_2 \cdots (u_n,\enter,s_{\mathtt{p}}) \cdot \pi_{n+1}\] 
wie in \autoref{bild-pfadlemma-sq}
mit $\pi_i \in \underline{\SLP}[u_i]$ für $1 \le i \le n$ und $\pi_{n+1}\in\underline{\SLP}[u]$.
Nun ist nach Definition \[w \cdot e \in \underline{\CS}[\p]\] und damit \[\pi \cdot (u,\enter,s_\q) \in \Path(s_\q,w\cdot e).\]
\begin{figure}[ht]
  \begin{flowgraph}
	\node[startnode]	(smain) 					{\small{$s_{\main}$}};
	\node[stdnode]		(u1) 	[below=of smain]	{$u_1$};
	\node[stdnode]		(sp1) 	[right=of u1]		{$s_\mathtt{p_1}$};
	\node[stdnode]		(u2) 	[below=of sp1]		{$u_2$};
	\node[stdnode]		(un) 	[below right=of u2]	{$u_n$};
	\node[stdnode]		(sp) 	[right=of un]		{$s_\mathtt{p}$};
	\node[stdnode]		(u) 	[below=of sp]		{$u$};
	\node[stdnode]		(sq) 	[right=of u]		{$s_\mathtt{q}$};
	\path[->]	
		(smain) edge [snakearc]						node[edgenamecolor]{$\pi_1$}		(u1)
		(u1)	edge [dashed] 						node {$\enter$} 					(sp1)
		(sp1)	edge [snakearc] 					node[edgenamecolor] {$\pi_2$}		(u2)
		(u2)	edge [snakearc]															(un) 
		(un)	edge [dashed] 						node {$\enter$} 					(sp)
		(sp)	edge [snakearc]						node[edgenamecolor] {$\pi_{n+1}$}	(u)
		(u)		edge [dashed]						node {$\enter$}						(sq)
	;
  \end{flowgraph}
  \caption{Beweis von Lemma \ref{pfadlemma}, $\Path(s_\q) \supseteq \Path(u) \cdot (u,\enter,s_\q)$.}
  \label{bild-pfadlemma-sq}
\end{figure}

Für die andere Inklusion seien $\p \in \Proc$, $u \in N_\p$ und $w \in \underline{\CS}[\p]$ beliebig. Es genügt, $\Path(u,w) \subseteq \underline{\Path}[u]$ zu zeigen. 
Dies geschieht durch vollständige Induktion nach der Länge $n$ des Call-Strings $w$.

Sei zunächst $n=0$. Dann ist der Call-String $w$ bereits das leere Wort und wir erhalten 
\[\Path(u,w) 
= \Path(u,\eps)
= \{\pi\cdot\pi'\mid\pi\in\underline{\SLP}[s_\p]\wedge\pi'\in\underline{\SLP}[u]\} = \underline{\SLP}[u] \subseteq \underline{\GP}[u].
\]
Letztere Inklusion resultiert daraus, dass $w=\eps$ bereits $\p = \main$ impliziert.

Gelte die Behauptung nun für alle Pfade der Länge kleiner oder gleich $k$. Sei $w$ ein Call-String der Länge $k+1$. 
Dann ist $w$ von der Form $w = \widetilde{w} \cdot \widetilde{e}$, wobei $\widetilde{e}=(\widetilde{u},\mathtt{p()},\widetilde{v})\in N_\q$ für eine Prozedur $\q \in \Proc$. 

Wir können einen Pfad $\pi \in \Path(u,w)$ nun schreiben als 
\[\pi = \widetilde{\pi} \cdot (\widetilde{u},\enter,s_\p) \cdot \pi'\]
für Pfade $\widetilde{\pi} \in \Path(\widetilde{u},\widetilde{w})$ und $\pi' \in \underline{\SLP}[u]$. Betrachte dazu auch \autoref{bild-pfadlemma-IS}.
\begin{figure}[ht]
  \begin{flowgraph}
	\node[startnode]	(smain) 					{\small{$s_{\main}$}};
	\node[stdnode]		(sq) 	[below right=of smain]		{$s_\mathtt{q}$};
	\node[stdnode]		(tildeu) 	[below=of sq]		{$\widetilde{u}$};
	\node[fadenode]		(tildev) 	[below=of tildeu]	{$\widetilde{v}$};
	\node[stdnode]		(sp) 	[right=of tildeu]		{$s_\mathtt{p}$};
	\node[stdnode]		(u) 	[below=of sp]		{$u$};
	\node[callstring]	(csq)[above right] at (sq.north east) {$\widetilde{w}$};
	\node[callstring]	(csp)	[right] at (sp.north east) 	{$w=\widetilde{w}\cdot\widetilde{e}$};
	\path[->]	
		(smain) edge [snakearc] 									(sq)
		(sq)	edge [snakearc] 									(tildeu)
		(tildeu)edge [fadearc]		node[edgenamecolor] {$\widetilde{e}$}		(tildev)
		(tildeu)edge [dashed] 				node {$\enter$} 			(sp)
		(sp)	edge [snakearc]		node[edgenamecolor] {$\pi'$}				(u)
	;

  \end{flowgraph}
  \caption{Beweis von Lemma \ref{pfadlemma}, Induktionsschritt.}
  \label{bild-pfadlemma-IS}
\end{figure}

Nach Induktionsvoraussetzung ist $\Path(\widetilde{u},\widetilde{w})\subseteq\underline{\GP}[\widetilde{u}]$. 
Das Ungleichungssystem \eqref{GP-Ugs} zur Definition der Pfade liefert nun, dass mit $\widetilde{\pi} \in \underline{\GP}[\widetilde{u}]$ 
\[\widetilde{\pi} \cdot (\widetilde{u},\enter,s_\p) \in \underline{\GP}[s_\p]\] 
und weiter 
\[\pi = \widetilde{\pi} \cdot (\widetilde{u},\enter,s_\p) \cdot \pi' \in \underline{\GP}[u]\] 
gilt. Es folgt die Behauptung.
\end{proof}

Wir haben nun verschiedene Pfade definiert: 
Zum einen Same-Level-Pfade, die vom Startknoten einer Prozedur zu einem Knoten derselben Prozedur führen, 
zum anderen beliebige Pfade, die im Startknoten des Programms beginnen und zu einem Knoten einer beliebigen Prozedur führen. 
Weiter haben wir für jeden Call-String die Menge derjenigen Pfade definiert, deren $\enter$-Kanten zum Call-String korrespondieren. 
Zuletzt haben wir in Lemma \ref{pfadlemma} einen Zusammenhang zwischen den beiden letzteren Mengen hergestellt. 

Wir können nun die $\MOP$-Lösung definieren durch 
\[\MOP[u] := \{f_\pi(\init) \mid \pi \in \underline{\GP}[u]\}\]
für $u \in N$. 
\begin{dfn}
Eine constraint-basierte Datenflussanalyse mit Ungleichungssystem $\ugs$ heißt
\begin{itemize}
\item \emph{korrekt}, falls $\MFP[u] \sqgeq \MOP[u]$ für alle $u \in N$ gilt.
\item \emph{präzise}, falls $\MFP = \MOP$ gilt.
\end{itemize}
Dabei bezeichnet $\MFP$ die kleinste Lösung des Ungleichungssystems $\ugs$.
\end{dfn}
Wir können nun zeigen, dass der funktionale und der Call-String-Ansatz korrekt sind. Stellen wir weiter geeignete Forderungen an die Transferfunktionen, so sind die Ansätze sogar präzise. 
Insbesondere liefern sie dann dasselbe Resultat. Wir beginnen mit dem funktionalen Ansatz.

\begin{lemma}\label{lemma-koinzidenz-T}
Für jedes $u \in N$ gilt 
\[\underline{T}[u] \sqgeqmap \bigsqcup \big\{f_\pi \mid \pi \in \underline{\SLP}[u]\big\}.\]
Ist $(L,\sqleq, \mathcal{F})$ universell-distributiv oder positiv-distributiv und in letzterem Fall weiter jeder Programmpunkt erreichbar, so gilt sogar für jedes $u \in N$
\[\underline{T}[u] = \bigsqcup \big\{f_\pi \mid \pi \in \underline{\SLP}[u]\big\}.\]
\end{lemma}
\begin{proof}
Wir schreiben abkürzend $M_u := \{f_\pi \mid \pi \in \underline{\SLP}[u]\}$. 

Zunächst zeigen wir, dass jedes Element von $M_u$ kleiner oder gleich $\underline{T}[u]$ ist. Daraus folgt $\bigsqcup M_u \sqleqmap \underline{T}[u]$.
Für die Gleichheit müssen wir dann noch die umgekehrte Ungleichung zeigen. Dazu rechnen wir nach, dass $(\bigsqcup M_\node)_{\node \in N}$ eine Lösung des Ungleichungssystems \eqref{T-Ugs} ist. 
\begin{enumerate}
\item 
  Wir zeigen folgende Aussage durch vollständige Induktion nach $k \in \nn$: 
  Für $u \in N$ und einen Pfad $\pi \in \underline{\SLP}[u]$ der Länge kleiner oder gleich $k$ ist $f_\pi \sqleqmap \underline{T}[u]$.

  Seien also $u \in N$ and $\pi \in \underline{\SLP}[u]$.
  Falls $\pi$ die Länge $0$ hat, gilt $\pi=\eps$ und $u=s_\p$. Damit folgt $f_{\pi} = \id \sqleqmap \underline{T}[s_\p]$.

  Gelte die Behauptung für Pfade, deren Länge kleiner oder gleich $k \in \nn$ ist. Sei $\pi$ ein Pfad der Länge $k+1$, d.h.~$\pi=e_1\dots e_ke_{k+1}$, wobei $e_{k+1} = (u_k,\mathtt{s},u)$. 
  Gemäß der Definition von Same-Level-Pfaden kann die letzte Kante $e_{k+1}$ weder eine Eintritts- noch eine Aufrufkante sein. Also ist $e_{k+1}$ entweder eine Basiskante oder eine Rückkehrkante. 
  Wir unterscheiden also zwei Fälle:
  \begin{enumerate}
  \item 
  Sei $e_{k+1}= (u_k,\mathtt{b},u)\in E$ eine Basiskante. Dann ist $\tilde{\pi}:=e_1\dots e_k\in \SLP[u_k]$ ein Pfad nach $u_k$ der Länge $k$. 
  Nach Induktionsvoraussetzung ist $f_{\tilde{\pi}} \sqleqmap \underline{T}[u_k]$. Weiter ist $f_{k+1} = f_\mathtt{b}$ und damit
  \[f_{\pi} \sqleqmap f_\mathtt{b} \circ \underline{T}[u_k] \sqleqmap \underline{T}[u]\]
  gemäß der Definition des Ungleichungssystems $T$.
  \item
  Sei $e_{k+1}$ eine Rückkehrkante, d.h.~$e_{k+1}=(r_\q,\ret,u)$ für ein $\q \in \Proc$ mit zugehöriger Aufrufkante $e_i=(u_{i-1},\call,s_\q)$ für ein $1 \le i < k$, 
  so dass $\pi = \pi' \cdot e_i \cdot \pi'' \cdot e_{k+1}$  wie in \autoref{bild-lemma-koinzidenz-T-IS} für Pfade $\pi' = e_1 \dots e_{i-1} \in \underline{\SLP}[u_{i-1}]$ 
  und $\pi''=e_{i+1}\dots e_k \in \underline{\SLP}[r_\q]$ und eine Aufrufkante $e=(u_{i-1},\mathtt{q()},u)\in E$. 

  \begin{figure}[ht]
	  \begin{flowgraph}
	\node[startnode]	(smain) 						{\small{$s_{\main}$}};
	\node[stdnode]		(ui-1) 	[below right=of smain]	{\small{$u_{i-1}$}};
	\node[stdnode]		(sq) 	[right=of ui-1]			{$s_\q$};
	\node[stdnode]		(rq) 	[below=of sq]			{$r_\q$};
	\node[stdnode]		(u) 	[below=of ui-1]			{$u$};
	\path[->]	
		(smain) edge [snakearc] node{$\pi'$}						(ui-1)
		(ui-1)	edge [dashed] 	node{$\call$}node[swap,edgenamecolor]{$e_i$}		(sq)
		(sq)	edge []			node {$\pi''$}						(rq)
		(rq)	edge [dashed]	node[swap]{$\ret$}node[edgenamecolor]{$e_{k+1}$}	(u)
	;

	\end{flowgraph}
	\caption{Beweis von Lemma \ref{lemma-koinzidenz-T}, Induktionsschritt.}
	\label{bild-lemma-koinzidenz-T-IS}
  \end{figure}

  Nach Induktionsvoraussetzung ist nun $f_{\pi'} \sqleqmap \underline{T}[u_{i-1}]$ und $f_{\pi''} \sqleqmap \underline{T}[r_\q]$. 
  Damit gilt \[f_\pi = f_{\pi''} \circ f_{\pi'} \sqleqmap \underline{T}[r_\q] \circ \underline{T}[u_{i-1}].\]
  Es folgt die gewünschte Ungleichung.
  \end{enumerate}
\item 
  Um zu zeigen, dass $(\bigsqcup M_\node)_{\node \in N}$ eine Lösung von \eqref{T-Ugs} ist, müssen wir folgende Ungleichungen prüfen:
  \begin{alignat}{2}
  &\bigsqcup M_{s_\p}	&&\sqgeqmap \id \label{lemma-koinzidenz-T1}\\
  &\bigsqcup M_v 		&&\sqgeqmap f_\mathtt{b} \circ \bigsqcup M_u \text{ für eine Basiskante } e = (u,\mathtt{b},v) \in E \label{lemma-koinzidenz-T2}\\
  &\bigsqcup M_v 		&&\sqgeqmap \bigsqcup M_{r_\q} \circ \bigsqcup M_u \text{ für eine Aufrufkante } e = (u,\mathtt{q()},v) \in E \label{lemma-koinzidenz-T3}
  .\end{alignat}
  Es gilt für alle $\p \in \Proc$
  \[
  \bigsqcup M_{s_\p} 
  = \big\{f_\pi \mid \pi \in \underline{\SLP}[s_\p]\big\} 
  \sqgeqmap f_{\eps} = \id
  \]
  vermöge $\eps \in \underline{\SLP}[s_\p]$. Also gilt die Ungleichung \eqref{lemma-koinzidenz-T1}.

  Sei weiter $e=(u,\mathtt{b},v)\in E$ eine Basiskante. Dann ist
  \begin{align*}
  M_v
  &= \big\{f_\pi \mid \pi \in \underline{\SLP}[v]\big\} \\
  &\supseteq \big\{f_\pi \mid \pi \in \underline{\SLP}[v] \wedge (\exists \tilde{\pi}\in \underline{\SLP}[u]:\pi = \tilde{\pi}\cdot e\big)\} \\ 
  &= \big\{f_{\tilde{\pi}\cdot e} \mid \tilde{\pi}\in \underline{\SLP}[u]\big\} \\
  &= \big\{f_\mathtt{b} \circ f_{\tilde{\pi}} \mid \tilde{\pi}\in \underline{\SLP}[u]\big\}
  .\end{align*}
  Damit folgt
  \[\bigsqcup M_v
  \sqgeqmap \bigsqcup\big\{f_\mathtt{b} \circ f_{\tilde{\pi}} \mid \tilde{\pi}\in \underline{\SLP}[u]\big\}
  = f_\mathtt{b} \circ \bigsqcup \big\{f_{\tilde{\pi}} \mid \tilde{\pi}\in \underline{\SLP}[u]\big\}
  = f_\mathtt{b} \circ \bigsqcup M_u,
  \]
  wobei die mittlere Gleichheit aus folgender Tatsache resultiert: Falls $(L,\sqleq, \mathcal{F})$ uni\-ver\-sell-distributiv ist, folgt sie gerade aus der Definition. 
  Falls $(L,\sqleq, \mathcal{F})$ positiv-distributiv ist, gilt die Aussage nur, falls $M_u$ nicht leer ist. Andererseits ist aber $M_u$ genau dann leer, wenn $\underline{\SLP}[u]$ leer ist. 
  Dies wiederum ist genau dann der Fall, wenn $u$ nicht erreichbar ist. 
  Da wir aber im Falle der positiven Distributivität von $(L,\sqleq, \mathcal{F})$ die Erreichbarkeit eines jeden Programmpunktes vorausgesetzt hatten, ist $M_u$ nicht leer. 
  Es folgt die Gültigkeit der Ungleichung \eqref{lemma-koinzidenz-T2}.

  Sei zuletzt $e = (u,\mathtt{q()},v)\in E$ eine Aufrufkante, dann gilt
  \begin{align*}
  \lefteqn{\big\{f_\pi \mid \pi \in \underline{\SLP}[v]\big\} } 
  \\ &\supseteq \big\{f_\pi \mid \pi \in \underline{\SLP}[v] \wedge (\exists \tilde{\pi}_1\in \underline{\SLP}[u], \tilde{\pi}_2\in\underline{\SLP}[r_\q]:\\
	&\phantom{~\supseteq~\big\{f_\pi \mid } \pi = \tilde{\pi}_1 \cdot (u,\call, s_\q) \cdot \tilde{\pi}_2 \cdot (r_\q, \ret, v))\big\} 
  \\ &= \big\{f_{(r_\q, \ret, v)} \circ f_{\tilde{\pi}_2} \circ f_{(u,\call, s_\q)} \circ f_{\tilde{\pi}_1} \mid \tilde{\pi}_1\in \underline{\SLP}[u], \tilde{\pi}_2\in\underline{\SLP}[r_\q]\big\}
  \\ &= \big\{f_{\tilde{\pi}_2} \circ f_{\tilde{\pi}_1} \mid \tilde{\pi}_1\in \underline{\SLP}[u], \tilde{\pi}_2\in\underline{\SLP}[r_\q]\big\}
  .\end{align*}
  Weiter ist für jedes $\hat{\pi}_2\in\underline{\SLP}[r_\q]$
  \[
  \big\{f_{\tilde{\pi}_2} \circ f_{\tilde{\pi}_1} \mid \tilde{\pi}_1\in \underline{\SLP}[u], \tilde{\pi}_2\in\underline{\SLP}[r_\q]\big\}
  \supseteq
  \big\{ f_{\hat{\pi}_2} \circ f_{\tilde{\pi}_1} \mid \tilde{\pi}_1\in\underline{\SLP}[u]\big\} 
  .\]
  Also ist
  \begin{align*}
  \bigsqcup M_v 
  \sqgeqmap \bigsqcup& \big\{f_{\tilde{\pi}_2} \circ f_{\tilde{\pi}_1} \mid \tilde{\pi}_1\in \underline{\SLP}[u], \tilde{\pi}_2\in\underline{\SLP}[r_\q]\big\}\\
  \sqgeqmap \bigsqcup& \big\{ f_{\hat{\pi}_2} \circ f_{\tilde{\pi}_1} \mid \tilde{\pi}_1\in\underline{\SLP}[u]\big\}\
  \intertext{für jedes $\hat{\pi}_2 \in \underline{\SLP}[r_\q]$. Daraus folgt in Analogie zu obiger Argumentation mit der jeweiligen Distributivität}
  \bigsqcup M_v 
  \sqgeqmap \bigsqcup& \big\{\bigsqcup\big\{ f_{\hat{\pi}_2} \circ f_{\tilde{\pi}_1} \mid \tilde{\pi}_1\in\underline{\SLP}[u]\big\} \mid \tilde{\pi}_2\in\underline{\SLP}[r_\q] \big\}\\
  = \bigsqcup&\big\{f_{\hat{\pi}_2} \circ \bigsqcup\big\{f_{\tilde{\pi}_1} \mid \tilde{\pi}_1\in\underline{\SLP}[u]\big\} \mid \tilde{\pi}_2\in\underline{\SLP}[r_\q]\big\}
  \\= \bigsqcup& \big\{f_{\hat{\pi}_2} \mid \tilde{\pi}_2\in\underline{\SLP}[r_\q]\big\} \circ \bigsqcup\big\{f_{\tilde{\pi}_1} \mid \tilde{\pi}_1\in\underline{\SLP}[u]\big\} 
  \\= \bigsqcup& M_{r_\q} \circ \bigsqcup M_u.
  \end{align*}
  Dies beweist Ungleichung \eqref{lemma-koinzidenz-T3} und damit die Behauptung.\qedhere
\end{enumerate}
\end{proof}

\begin{satz}\label{lemma-koinzidenz-R}
Für jedes $u \in N$ gilt 
\[\underline{R}[u] \sqgeq \MOP[u].\]
Ist $(L,\sqleq, \mathcal{F})$ universell-distributiv oder positiv-distributiv und in letzterem Fall weiter jeder Programmpunkt erreichbar, so gilt sogar für jedes $u \in N$
\[\underline{R}[u] = \MOP[u].\]
\end{satz}
\begin{proof}
Der Aufbau des Beweises ist analog zu dem vorangegangenen.
Dazu definieren wir $M_u := \{f_\pi(\init) \mid \pi \in \underline{P}[u]\}$. 
\begin{enumerate}
 \item 
  Wir führen einen Induktionsbeweis über $k \in \nn$: 
  Falls $u\in N$ und $\pi \in \underline{\GP}[u]$ ein Pfad der Länge kleiner oder gleich $k \in \nn$ ist, so gilt $f_{\pi}(\init) \sqleq \underline{R}[u]$.

  Seien also $u \in N$ und $\pi \in \underline{\GP}[u]$.
  Falls $\pi$ die Länge $0$ hat, so ist $\pi = \eps$. Insbesondere ist $u=s_{\main}$ und 
  \[f_{\eps}(\init) = \init \sqleq \underline{R}[s_{\main}].\]
  Gelte die Behauptung bereits für Pfade der Länge kleiner oder gleich $k \in \nn$ und sei $\pi \in \underline{\GP}[u]$ ein Pfad der Länge $k+1$, d.h.~$\pi=e_1\dots e_ke_{k+1}$, 
  wobei $e_{k+1} = (u_k,\mathtt{s},u)$. 
  Dann ist $\tilde{\pi}:=e_1\dots e_k\in P[u_k]$ ein Pfad der Länge $k$ und mit der Monotonie von $f_{k+1}$ erhalten wir nach Induktionsvoraussetzung
  \[f_{\pi}(\init) = f_{k+1}(f_{\tilde{\pi}}(\init)) \sqleq f_{k+1}(\underline{R}[u_k]).\]
  Nach Definition von $\underline{\GP}[u]$ kann die letzte Kante $e_{k+1}$ von $\pi$ eine Basis-, Eintritts- oder Rückkehrkante sein. Wir unterscheiden also drei Fälle:
  \begin{enumerate}
  \item 
  Sei $e_{k+1}=(u_k,\mathtt{b},u)\in E$ eine Basiskante. Dann gilt $f_{k+1} = f_\mathtt{b}$ und damit
  \[f_{\pi}(\init) \sqleq f_\mathtt{b}(\underline{R}[u_k]) \sqleq \underline{R}[u]\]
  nach Definition des Ungleichungssystems $R$.
  \item 
  Sei $e_{k+1}$ eine Eintrittskante, d.h.~$e_{k+1}=(u_k,\enter,s_\q)$ für ein $\q \in \Proc$ und $u = s_\q$. Damit ist $f_{k+1}=\id$, woraus
  \[f_{\pi}(\init) \sqleq \underline{R}[u_k]\] 
  folgt. Weiter existiert eine Kante $e=(u_k,\mathtt{q()},v)\in E$ und somit ist nach Definition des Ungleichungssystems $R$ auch 
  \[\underline{R}[u_k] \sqleq \underline{R}[u].\]
  Zusammen folgt 
  \[f_\pi(\init) \sqleq \underline{R}[u].\]
  \item 
  Sei $e_{k+1}$ eine Rückkehrkante, d.h.~$e_{k+1}=(r_\q,\ret,u)$ für ein $\q \in \Proc$ und $f_{k+1}=\id$. 
  Dann gibt es eine zugehörige Aufrufkante $e_i=(u_{i-1},\call,s_\q)$ für ein $1 \le i < k$. 
  Sei $\pi' := e_1 \dots e_{i-1} \in \underline{P}[u_{i-1}]$ ein Pfad der Länge $i-1 \le k$ und $\pi'':=e_{i+1}\dots e_k \in \underline{\SLP}[r_\q]$ 
  wie auch in \autoref{bild-lemma-koinzidenz-T-IS} zum Beweis von Lemma \ref{lemma-koinzidenz-T}.

  Dann gilt $\pi = \pi' \cdot e_i \cdot \pi'' \cdot e_{k+1}$ und damit $f_\pi = f_{\pi''} \circ f_{\pi'}$. 
  Nach Induktionsvoraussetzung ist dann
  \begin{align*}
  f_{\pi'} (\init) &\sqleq \underline{R}[u_{i-1}].
  \intertext{Mit der Monotonie von $f_{\pi''}$ ist}
  f_{\pi}(\init)
  &= f_{\pi''} \circ f_{\pi'} (\init) \\
  &\sqleq f_{\pi''}(\underline{R}[u_{i-1}]).
  \intertext{Da nach Lemma \ref{lemma-koinzidenz-T} $f_{\pi''} \sqleqmap \underline{T}[r_\q]$ gilt, ist weiter}
  f_{\pi''}(\underline{R}[u_{i-1}]) &\sqleq \underline{T}[r_\q](\underline{R}[u_{i-1}])\\
  &\sqleq \underline{R}[u].
  \end{align*}
  Die letzte Ungleichung folgt dabei aus der Definition des Ungleichungssystems \eqref{T-Ugs} zur Berechnung der Summary-Informationen.
  Zusammen gilt also 
  \[f_{\pi}(\init) \sqleq \underline{R}[u].\]
  Es folgt die Behauptung. 
  \end{enumerate}
\item 
  Für die andere Richtung prüfen wir
\begin{alignat}{2}
& \bigsqcup M_{s_{\main}} && \sqgeq \init \label{lemma-koinzidenz-R1}\\
& \bigsqcup M_v && \sqgeq f_\mathtt{b}\big(\bigsqcup M_u\big) \text{ für eine Basiskante } e = (u,\mathtt{b},v) \label{lemma-koinzidenz-R2}\\
& \bigsqcup M_{s_\q} && \sqgeq \bigsqcup M_u \text{ für eine Aufrufkante } e = (u,q(),v) \label{lemma-koinzidenz-R3}\\
& \bigsqcup M_v && \sqgeq \underline{T}[r_\q]\big(\bigsqcup M_u\big) \text{ für eine Aufrufkante } e = (u,q(),v). \label{lemma-koinzidenz-R4}
\end{alignat}
  Es ist
  \[
  M_{s_{\main}} 
  = \{f_\pi(\init) \mid \pi \in \underline{\GP}[s_{\main}]\}
  \ni f_{\eps}({\init}) 
  = \init
  \]
  vermöge $\eps \in \underline{\GP}[s_{\main}]$.
  Damit folgt
  \[\bigsqcup M_{s_{\main}} 
  \sqgeq \init,
  \]
  also die Gültigkeit der Ungleichung \eqref{lemma-koinzidenz-R1}.

  Sei $e=(u,\mathtt{b},v)\in E$ eine Basiskante. Dann gilt $\{\pi \cdot e \mid \pi \in \underline{\GP}[u]\} \subseteq \underline{\GP}[v]$ 
  und analog zum Beweis von Lemma \ref{lemma-koinzidenz-T} folgt Ungleichung \eqref{lemma-koinzidenz-R2} mit
  \[
  \bigsqcup M_v
  \sqgeq \bigsqcup \big\{f_\mathtt{b} (f_{\tilde{\pi}}(\init)) \mid \tilde{\pi} \in \underline{\GP}[u] \big\}
  = f_\mathtt{b}\big(\bigsqcup M_u\big).
  \]
  Sei nun $(u,\mathtt{q()},v)\in E$ eine Aufrufkante. Dann ist
  \[\{\pi \cdot (u,\enter,s_\q) \mid \pi \in \underline{\GP}[u]\} \subseteq \underline{\GP}[s_\q]\]
  und
  \[\{\pi \cdot (u,\call,s_{\q}) \cdot \tilde{\pi} \cdot (r_{\q},\ret,v) \mid \pi \in \underline{\GP}[u] \wedge \tilde{\pi} \in \underline{\SLP}[r_{\q}]\} \subseteq \underline{\GP}[v].\]
  Damit folgt einerseits Ungleichung \eqref{lemma-koinzidenz-R3} aus
  \[
  \bigsqcup M_{s_\q}
  \sqgeq \bigsqcup \big\{f_{\pi \cdot (u,\enter,s_\q)} \mid \pi \in \underline{\GP}[u]\big\}
  = \bigsqcup \big\{f_{\pi} \mid \pi \in \underline{\GP}[u]\big\}
  = \bigsqcup M_u
  \]
  und andererseits Ungleichung \eqref{lemma-koinzidenz-R4}, wobei die jeweilige Distri\-bu\-ti\-vi\-täts\-eigen\-schaft wie im Beweis von Lemma \ref{lemma-koinzidenz-T} benutzt wird, mit
  \begin{align*}
  \bigsqcup M_v
  &\sqgeq \bigsqcup \big\{f_{\pi \cdot (u,\call,s_{\q}) \cdot \tilde{\pi} \cdot (r_{\q},\ret,v)} (\init) \mid \pi \in \underline{\GP}[u] \wedge \tilde{\pi} \in \underline{\SLP}[r_{\q}]\big\} \\
  &= \bigsqcup \big\{f_{\tilde{\pi}} \circ f_\pi (\init)\mid \pi \in \underline{\GP}[u] \wedge \tilde{\pi} \in \underline{\SLP}[r_{\q}]\big\} \\
  &= \bigsqcup \big\{f_{\tilde{\pi}}\mid\tilde{\pi} \in \underline{\SLP}[r_{\q}]\big\}\big(\bigsqcup\big\{f_\pi (\init) \mid \tilde{\pi} \in \underline{\GP}[u]\big\}\big)\\
  &= \bigsqcup \big\{f_{\tilde{\pi}}\mid\tilde{\pi} \in \underline{\SLP}[r_{\q}]\big\}\big(\bigsqcup\big\{f_\pi (\init) \mid \tilde{\pi} \in \underline{\GP}[u]\big\}\big)\\
  &= \underline{T}[r_{\q}] \big( \bigsqcup(M_u) \big).
  \end{align*}
  Die letzte Gleichheit folgt dabei aus Lemma \ref{lemma-koinzidenz-T}. Es folgt die Behauptung. \qedhere
\end{enumerate}
\end{proof}

Wir betrachten nun den Call-String-Ansatz.
\begin{lemma}\label{lemma-koinzidenz-A}
Für jedes $(u,w) \in N_{\CS}$ gilt 
\[\underline{A}[u,w] \sqgeq \bigsqcup \big\{f_\pi(\init) \mid \pi \in \GP(u,w)\big\}.\]
Ist $(L,\sqleq, \mathcal{F})$ universell-distributiv oder positiv-distributiv und in letzterem Fall weiter jeder Programmpunkt erreichbar, so gilt sogar für jedes $(u,w) \in N_{\CS}$
\[\underline{A}[u,w] = \bigsqcup \big\{f_\pi(\init) \mid \pi \in \GP(u,w)\big\}.\]
\end{lemma}
\begin{proof}
Dieses Lemma beweisen wir wie die beiden vorherigen, indem wir beide Ungleichungen zeigen. Dazu sei $M_{u,w} := \{f_\pi(\init) \mid \pi \in \GP(u,w)\}$.
\begin{enumerate}
 \item 
  Wir zeigen wieder durch vollständige Induktion nach $k \in \nn$: 
  Falls $(u,w) \in N_{\CS}$ und $\pi \in \GP(u,w)$ ein Pfad der Länge kleiner oder gleich $k \in \nn$ ist, 
  so gilt $f_{\pi}(\init) \sqleq \underline{A}[u,w]$.

  Sei also $(u,w) \in N_{\CS}$. Für einen Pfad $\pi \in \GP(u,w)$ der Länge $k=0$ gilt $\pi = \eps$. Insbesondere ist $u=s_{\main}$ und $w = \eps$. Damit folgt
  \[f_{\eps}(\init) = \init \sqleq \underline{A}[s_{\main},\eps].\]
  Gelte die Behauptung nun für Pfade der Länge kleiner oder gleich $k \in \nn$ und sei $\pi$ ein Pfad der Länge $k+1$, d.h.~$\pi=e_1\dots e_ke_{k+1}$, wobei $e_{k+1} = (u_k,\mathtt{s},u)$. 
  Wiederum betrachten wir die drei Fälle, dass die letzte Kante $e_{k+1}$ von $\pi$ eine Basis-, Eintritts- oder Rückkehrkante ist.
  Wir schreiben wieder $\tilde{\pi}=e_1\cdots e_k$.

  \begin{enumerate}
  \item Sei $e_{k+1}=(u_k,\mathtt{b},u)\in E$ eine Basiskante. Dann ist auch $\tilde{\pi} \in \GP(u_k,w)$ und weiter $f_{k+1} = f_{\mathtt{b}}$. 
  Somit gilt $f_{\tilde{\pi}}(\init) \sqleq \underline{A}[u_{k},w]$ nach Induktionsvoraussetzung. Weiter ist
  \[f_{\pi}(\init) = f_{k+1} (f_{\tilde{\pi}}(\init)) \sqleq f_\mathtt{b} (\underline{A}[u_{k},w]) \sqleq \underline{A}[u,w]\] 
  vermöge der Definition des Ungleichungssystems $A$.
	  
  \item Sei $e_{k+1}$ eine Eintrittskante, d.h.~$e_{k+1}=(u_k,\enter,s_\p)$, mit zugehöriger Aufrufkante $e=(u_k,\mathtt{p()},v)\in E_\q$ für ein $\q \in \Proc$, 
  so dass $w = \tilde{w} \cdot e$ für ein $\tilde{w} \in \underline{\CS}[\q]$. Dann ist $f_{k+1} = \id$ und weiter nach Induktionsvoraussetzung
  \[f_{\pi}(\init) = f_{\tilde{\pi}}(\init) \sqleq \underline{A}[u_{k},\tilde{w}] \sqleq \underline{A}[s_\q,\tilde{w}\cdot e] = \underline{A}[u,w]\] 
  gemäß der Definition des Ungleichungssystems $A$.

  \item Sei $e_{k+1}$ eine Rückkehrkante, d.h.~$e_{k+1}=(r_\q,\ret,u)$ und es gibt eine Kante $e=(u_i,\mathtt{q()},u)\in E_\p$ für ein $\q \in \Proc$ und $1 \le i \le k$.
  Aus $w \in \underline{\CS}[\p]$ folgt $w \cdot e \in \underline{\CS}[\q]$. Hieraus erhalten wir zusammen mit $f_{k+1} = \id$ und der Induktionsvoraussetzung
  \[f_{\pi}(\init) = f_{\tilde{\pi}}(\init) \sqleq \underline{A}[r_\q,\tilde{w}\cdot e] \sqleq \underline{A}[u,w]\] 
  vermöge der Definition des Ungleichungssystems $A$. Es folgt die gewünschte Ungleichung.
  \end{enumerate} 
\item 
  Wir zeigen für jedes $(u,w) \in N_{\CS}$
  \begin{alignat}{2}
  &\bigsqcup M_{\main,\eps}	&&\sqgeq \init \label{lemma-koinzidenz-A1}\\
  &\bigsqcup M_{v,w} 			&&\sqgeq f_e\big(\bigsqcup M_{u,w}\big)	\text{ für eine Basiskante } e=(u,\mathtt{b},v) \in E_\p \label{lemma-koinzidenz-A2}\\
  &\bigsqcup M_{s_\q,w\cdot e}&&\sqgeq \bigsqcup M_{u,w}				\text{ für eine Aufrufkante } e=(u,\mathtt{q()},v) \in E_\p \label{lemma-koinzidenz-A3}\\
  &\bigsqcup M_{v,w} 			&&\sqgeq \bigsqcup M_{r_\q,w\cdot e}	\text{ für eine Aufrufkante } e=(u,\mathtt{q()},v) \in E_\p \label{lemma-koinzidenz-A4}
  \end{alignat}

  Zunächst folgt \eqref{lemma-koinzidenz-A1} direkt mit
  \begin{align*}
  M_{s_{\main},\eps}
  &= \{f_\pi(\init) \mid \pi \in \GP(s_{\main},\eps)\}\\
  &\supseteq \{f_\pi(\init) \mid \pi = \eps\}\\
  &= \{\init\}.
  \end{align*}
  Sei nun $e=(u,\mathtt{b},v) \in E_\p$ eine Basiskante und $w \in \underline{\CS}[\p]$. 
  Dann folgt aus $\pi \in \GP(u,w)$ schon $\pi \cdot e \in \GP(v,w)$. Somit ist, wobei wir die Distributivitätseigenschaften der Transferfunktionen wie im Beweis von \ref{lemma-koinzidenz-R} benutzen,
  \begin{align*}
  \bigsqcup M_{v,w} 
  &= \bigsqcup \big\{f_\pi(\init) \mid \pi \in \GP(v,w)\big\}\\
  &\sqgeq \bigsqcup\big\{f_{\tilde{\pi}\cdot e}(\init) \mid \tilde{\pi} \in \GP(u,w)\big\}\\
  &= f_\mathtt{b}\big(\bigsqcup\big\{f_{\tilde{\pi}}(\init) \mid \tilde{\pi} \in \GP(u,w)\big\}\big)\\
  &= f_\mathtt{b}(M_{u,w})
  \end{align*}
  und Ungleichung \eqref{lemma-koinzidenz-A2} ist gezeigt. 

  Sei $e = (u,\mathtt{q()},v) \in E_\p$ eine Aufrufkante und $w \in \underline{\CS}[\p]$. Dann ist
  \[
  \big\{ \tilde{\pi}\cdot (u,\call,s_\q) \mid \tilde{\pi} \in \GP(u,w)\big\}
  \subseteq \GP(s_\q,w\cdot e) 
  \]
  und
  \[
  \big\{ \tilde{\pi}\cdot (r_\q,\ret,v) \mid \tilde{\pi} \in \GP(r_\q,w\cdot e)\big\}
  \subseteq \GP(v,w)
  .\]
  Somit gilt einerseits
  \begin{align*}
  M_{s_\q,w\cdot e}
  &= \big\{f_\pi(\init) \mid \pi \in \GP(s_\q,w\cdot e)\big\}\\
  &\supseteq \big\{f_{\tilde{\pi}\cdot (u,\call,s_\q) }(\init) \mid \tilde{\pi} \in \GP(u,w)\big\}\\ 
  &= \big\{f_{\tilde{\pi}}(\init) \mid \tilde{\pi} \in \GP(u,w)\big\}\\ 
  &= M_{u,w}.
  \end{align*} 
  Daraus folgt direkt \eqref{lemma-koinzidenz-A3}.
  Andererseits ist
  \begin{align*}
  M_{v,w}
  &= \big\{f_\pi(\init) \mid \pi \in \GP(v,w)\big\}\\ 
  &\supseteq \big\{f_{\tilde{\pi}\cdot (r_\q,\ret,v) }(\init) \mid \tilde{\pi} \in \GP(r_\q,w\cdot e)\big\}\\ 
  &= \big\{f_{\tilde{\pi}}(\init) \mid \tilde{\pi} \in \GP(r_\q,w\cdot e)\big\}\\ 
  &= M_{r_\q,w\cdot e}
  \end{align*}
  und auch Ungleichung \eqref{lemma-koinzidenz-A4} ist erfüllt. Es folgt die Behauptung.\qedhere
\end{enumerate}
\end{proof}

\begin{satz}\label{lemma-koinzidenz-hatA}
Für jedes $u \in N$ gilt 
\[\hat{A}[u] \sqgeq \MOP[u].\]
Ist $(L,\sqleq, \mathcal{F})$ universell-distributiv oder positiv-distributiv und in letzterem Fall weiter jeder Programmpunkt erreichbar, so gilt sogar für jedes $u \in N$
\[\hat{A}[u] = \MOP[u].\]
\end{satz}
\begin{proof}
Seien $\p \in \Proc$ und $u \in N_\p$. 
Aus dem vorigen Lemma \ref{lemma-koinzidenz-A} folgt nun
\begin{align*}
\hat{A}[u]
&= \bigsqcup \big\{\underline{A}[u,w] \mid w \in \underline{\CS}[\p] \big\} \\
&\sqleq \bigsqcup \big\{\bigsqcup \{f_\pi(\init) \mid \pi \in \GP(u,w)\} \mid w \in \underline{\CS}[\p] \big\} \\
&= \bigsqcup \big\{f_\pi(\init) \mid \pi \in \bigcup\big\{\GP(u,w) \mid w \in \underline{\CS}[\p]\big\} \big\} \\
&= \bigsqcup \big\{f_\pi(\init) \mid \pi \in \underline{\GP}[u]\big\} 
\end{align*}
und daraus die erste Behauptung. Mit den entsprechenden Voraussetzung an die Distributivität der Transferfunktionen gilt sogar überall Gleichheit.
\end{proof}

Wir haben nun alle nötigen Aussagen beisammen, um beweisen zu können, 
dass der funktionale und der Call-String-Ansatz unter geeigneten Voraussetzungen an die verwendeten Transferfunktionen in der Tat dieselbe präziseste Information liefern:
\begin{kor}[Koinzidenztheorem]\label{satz:koinzidenz}
Für jeden Knoten $u \in N$ gilt
\[\hat{A}[u] \sqleq \underline{R}[u].\]
Ist $(L,\sqleq, \mathcal{F})$ universell-distributiv oder positiv-dis\-tri\-bu\-tiv und in letzterem Fall weiter jeder Programmpunkt erreichbar, 
so gilt \[\hat{A}[u] = \underline{R}[u]\] für jeden Knoten $u \in N$.
\end{kor}
\begin{proof}
Dies folgt direkt aus Satz \ref{lemma-koinzidenz-R} und Satz \ref{lemma-koinzidenz-hatA}.
\end{proof}
\begin{bem}
Wir haben in \autoref{sec:callstring-ansatz} bereits erwähnt, dass in der Praxis oft die Länge der Call-Strings beschränkt wird. 
Können im Programm Call-Strings beliebiger Längen auftreten, so ist die Menge der beschränkten Call-Strings $\underline{\CS}_\text{beschr}[\p]$ eine echte Teilmenge der Menge aller Call-Strings. 
Sind die Transferfunktionen universell-distributiv oder positiv-distributiv und in letzterem Fall alle Programmpunkte erreichbar, 
so folgt aus dem Beweis von Lemma \ref{lemma-koinzidenz-hatA} und Satz \ref{satz:koinzidenz}
\begin{align*}
\MOP[u] 
&= \bigsqcup \big\{f_\pi(\init) \mid \pi \in \underline{\GP}[u]\big\}\\
&= \bigsqcup \big\{\underline{A}[u,w] \mid w \in \underline{\CS}[\p] \big\}\\
&\sqsupset \bigsqcup \big\{\underline{A}[u,w] \mid w \in \underline{\CS}_\text{beschr}[\p] \big\}.
\end{align*}
Es gilt also keine Koinzidenz, wenn wir die Call-Strings beschränken. Darum werden wir im Folgenden stets nur unbeschränkte Call-Strings betrachten.
\end{bem}

In diesem Abschnitt haben wir zwei Ansätze zur Behandlung von Prozeduren vorgestellt: 
Zum einen den Call-String-Ansatz, der mithilfe von Call-Strings alle Prozeduren eines Programms wie eine einzige behandelt, 
zum anderen den funktionalen Ansatz, der die Prozeduren getrennt betrachtet und jeder Prozedur ihren Effekt zuordnet. 
Wir haben gezeigt, dass beide Ansätze korrekte constraint-basierte Datenflussanalysen sind. 
Zuletzt haben wir gesehen, dass unter geeigneten Annahmen an die Distributivität der verwendeten Transferfunktionen beide Ansätze gleichwertig sind: 
Falls das zugrundeliegende monotone Framework universell-distributiv ist oder zumindest positiv-distributiv ist und in letzterem Fall zusätzlich alle Programmpunkte erreichbar sind, 
sind sowohl der funktionale als auch der Call-String-Ansatz präzise constraint-basierte Datenflussanalysen und beide berechnen insbesondere dasselbe Ergebnis. 

Oft ist die Transferfunktion, die als Effekt einer Prozedur berechnet wurde, nicht endlich darstellbar. Es muss also eine andere, endliche Darstellung gefunden werden. 
Im nächsten Abschnitt stellen wir mit der Polyederanalyse dafür ein Beispiel vor. 
Anstelle von Transferfunktionen berechnen wir dort auf zwei verschiedene Arten Summary-Informationen, 
aus denen wir wiederum Transferfunktionen zum Beschreiben des Effektes der Prozedur herleiten können.

\section{Polyederanalyse}\label{sec:polyeder}
Ziel der Polyederanalyse ist, zu jedem Programmpunkt (affine) Ungleichungen der Gestalt $\sum_{i=1}^n a_i x_i \ge a_0$ zu bestimmen, welche von den Programmvariablen erfüllt werden. 
Damit kann beispielsweise überprüft werden, ob ein Index innerhalb gewisser Array-Grenzen liegt. Geometrisch bilden die Variablenwerte, die solche Ungleichungen erfüllen, ein konvexes Polyeder. 
Die Polyederanalyse ist ein Beispiel für eine interprozedurale Datenflussanalyse. 
Zur Berechnung gültiger affiner Ungleichungen stellen wir in diesem Abschnitt zunächst Ungleichungssysteme zur Berechnung der erreichbaren Werte auf. 
In \autoref{sec:konvexe-polyeder} werden wir diese dann zu Ungleichungssystemen erweitern, deren Lösungen tatsächlich affine Ungleichungen beschreiben.

Ein in dieser Hinsicht zu untersuchendes Programm habe eine endliche Menge $\Proc$ von Prozeduren und sei im Folgenden durch ein Flussgraphsystem $\{G_{\p}\}_{\p\in \Proc}$ wie in \autoref{sec:fg} gegeben.
Dabei habe das Programm $n$ (globale) Variablen $\var{1},\dots,\var{n}$, welche Werte $x_1,\dots,x_n$ in $\rr$ annehmen. 
Als einzigen Basisanweisungen betrachten wir \emph{affine Zuweisungen} der Form $\mathtt{x_j:=a_0 + \sum_{i=1}^na_ix_i}$ für ein $1 \le j \le n$. 
In den kommenden Abschnitten werden wir meist $\mathtt{t}$ anstelle von $\mathtt{a_0 + \sum_{i=1}^n a_i x_i}$ 
und entsprechend $\mathtt{x_j := t}$ anstatt $\mathtt{x_j:=a_0 + \sum_{i=1}^na_ix_i}$ schreiben.  

Im Folgenden geben wir zwei Ansätze an, die ähnlich dem funktionalen Ansatz zunächst Summary-Informationen für die Prozeduren berechnen, und vergleichen diese: 
Der erste Ansatz verwendet Matrizenmengen und basiert auf \cite{seidl07}, während der zweite Ansatz Relationen als Summary-Informationen berechnet und in \cite{CH78-POPL} zu finden ist.

\subsection{Framework und Zuweisungen}
Ein \emph{Zustand} ist ein Vektor $(1,x_1,\dots,x_n) \in \{1\}\times\rr^n$, der die Werte der Variablen enthält. 
Die zusätzliche erste Komponente $1$ dient dazu, die Transferfunktion einer affinen Zuweisung mithilfe einer linearen Abbildung darstellen zu können. 
Wir bezeichnen mit $\Sigma := \{1\} \times \rr^n$ den Raum der möglichen Zustände. 
Die Analyse zur Bestimmung der erreichbaren Zustände ist nun auf Mengen von solchen Zuständen definiert, also auf dem vollständigen Verband $(L,\sqleq) := (2^\Sigma,\subseteq)$. 
Zusammen mit der Menge $\mathcal{F}$ der universell-distributiven Abbildungen $L \to L$ ist $(L, \subseteq, \mathcal{F})$ also ein universell-distributives Framework.

Wir bezeichnen mit $\Mat(n)$ die Menge der reellen $n\times n$-Matrizen.
Weiter sei 
\[\Mat(\Sigma) := \{A \in \Mat(n+1) \mid a_{00}=1 \wedge \forall 1 \le j \le n.~a_{0j}=0\}\] 
die Menge derjenigen reellen $(n+1)\times(n+1)$-Matrizen, die Zustände aus $\Sigma$ wieder auf Zustände aus $\Sigma$ abbilden. 
Man sieht leicht, dass tatsächlich jede Matrix, die alle Elemente aus $\Sigma$ wieder nach $\Sigma$ abbildet, von dieser Gestalt ist.
Setzen wir 
\[
M_{\mathtt{x_j := t}}:=
\begin{pmatrix}
I_j 									& 0 			& 0 \\
\mathtt{a_0} \cdots \mathtt{a_{j-1}}	& \mathtt{a_j} 	& \mathtt{a_{j+1}} \cdots \mathtt{a_n} \\
0 										& 0 			& I_{n-j} 
\end{pmatrix}
,\]
wobei $I_k \in \Mat(k)$ die Einheitsmatrix für ein $k \ge 1$ ist,
können wir damit den Effekt der Zuweisung $\mathtt{x_j := t}$ durch
\[\begin{pmatrix}
1 \\ \tilde{x}_1 \\ \vdots \\ \tilde{x}_n
\end{pmatrix}
:= 
\begin{pmatrix}
I_j 									& 0 			& 0 \\
\mathtt{a_0} \cdots \mathtt{a_{j-1}}	& \mathtt{a_j} 	& \mathtt{a_{j+1}} \cdots \mathtt{a_n} \\
0 										& 0 			& I_{n-j} 
\end{pmatrix}
\cdot 
\begin{pmatrix}
1\\x_1\\\vdots\\x_n
\end{pmatrix}
\]
ausdrücken.
Es ist $\tilde{x}_j = a_0 + \sum_{i=1}^n a_i x_i$. Weiter gilt für $1 \le i \le n$ und $i \ne j$ offenbar $\tilde{x}_i=x_i$. 
Das beschreibt die Tatsache, dass $\var{j}$ die einzige veränderte Variable ist, und diese genau gemäß der Vorschrift der Zuweisung verändert wurde. 

Mit anderen Worten weisen wir einer Kante $e=(u,\mathtt{x_j := t}, v) \in E$ die Transferfunktion 
\begin{align*}
f_{(u,\mathtt{x_j := t}, v)} : L^N &\to L,\\ \mathcal{S} &\mapsto f_{\mathtt{x_j := t}}(\mathcal{S}_u)
\intertext{zu. Dabei ist $\mathcal{S} = (\mathcal{S}_\node)_{\node \in N}$ und}
f_{\mathtt{x_j := t}}: L &\to L,\\ S &\mapsto \{M_{\mathtt{x_j := t}} \cdot s \mid s \in S\}.\end{align*}

Wir wollen nun zunächst die in einem jeden Programmpunkt erreichbaren Zu\-stän\-de berechnen. 
Dazu bestimmen wir zunächst die Effekte der Prozeduren, um damit das Ungleichungssystem \eqref{R-Ugs} aufstellen und lösen zu können. Die Effekte wollen wir auf zwei verschiedene Arten berechnen. 
In beiden Fällen werden wir die Abbildungen nicht direkt aus $\mathcal{F}$ wählen, sondern zunächst über einem anderen Raum definieren. 
Die dann berechneten Lösungen lassen sich jedoch mit Abbildungen aus $\mathcal{F}$ identifizieren. 

\subsection{Matrizenmengen als Summary-Information}\label{subsec:polyeder-mat}
Wie im vorigen Abschnitt bereits beschrieben, lässt sich der Effekt einer affinen Zuweisung mittels einer Matrix ausdrücken. 
Der Effekt zweier aufeinander folgender Kanten ist dann genau durch das Produkt der beiden Matrizen gegeben, die die Effekte der einzelnen Kanten beschreiben. 
So lässt sich auch der Effekt entlang eines Pfades durch eine Matrix ausdrücken. 
Die Idee ist nun, einen Knoten mit einer solchen Matrix zu annotieren, um damit zu beschreiben, dass vom Startknoten bis zu diesem Knoten ein Zustand genau durch diese Matrix verändert wird. 

Nun können aber auch mehrere Pfade zu einem Knoten führen. Die Effekte dieser Pfade entsprechen also mehreren Matrizen. 
Somit ist es nicht möglich, den Effekt bis zu einem Knoten allein mit einer Matrix zu beschreiben. 
Stattdessen werden wir zu jedem Knoten eine Menge solcher Matrizen bestimmen. 
Wir definieren das Ungleichungssystem zur Berechnung der Summary-Informationen ähnlich wie im Ungleichungssystem \eqref{T-Ugs}, 
allerdings über dem vollständigen Verband \[(L_M, \subseteq) := (2^{\Mat(\Sigma)},\subseteq).\] 
Anschließend erläutern wir, wie sich eine Matrizenmenge als universell-distributive Abbildung $L \to L$ auffassen lässt.

Dazu betrachten wir die verallgemeinerten Transferfunktionen $f_{T_M,\id},f_{T_M,\mathtt{x_j:=t}}$ und $f_{T_M,\mathtt{q()}}: {L_M}^N \to L_M$ mit 
\begin{align*} 
f_{T_M,\id}						(\mathcal{A}) &:= \{I_{n+1}\}, \\
f_{T_M,(u,\mathtt{x_j:=t},v)}	(\mathcal{A}) &:= \{M_{\mathtt{x_j := t}}\} \circ \mathcal{A}_u\\
f_{T_M,(u,\mathtt{q()},v)}		(\mathcal{A}) &:= \mathcal{A}_{r_{\q}} \circ \mathcal{A}_u  
\end{align*}
für $\mathcal{A}\in {L_M}^N$. Das Ungleichungssystem zur Bestimmtung der Prozedureffekte als Transferfunktion hat also folgende Gestalt:
\begin{equation}\label{TM-Ugs}\begin{split}
T_M[s_{\p}] &\supseteq \{I_{n+1}\}\\
T_M[v] 		&\supseteq \{M_{\mathtt{x_j := t}}\} \circ T_M[u] \text { für } (u,\mathtt{x_j := t},v) \in E_p \text{ Basiskante} \\
T_M[v] 		&\supseteq T_M[r_{\q}] \circ T_M[u] \text{ für } e=(u,\mathtt{q()},v) \in E \text{ Aufrufkante}.
\end{split}\end{equation}
Dabei ist $T_M[u] \in L_M = 2^{\Mat(\Sigma)}$ für jedes $u \in N_{\p}$ und $\p \in \Proc$. 

Mit $\circ$ bezeichnen wir dabei die elementweise Matrizenmultiplikation: 
Für Matrizenmengen $\mathcal{A}, \mathcal{B} \subseteq \Mat(\Sigma)$ ist $\mathcal{A} \circ \mathcal{B} := \{A\cdot B \mid A \in \mathcal{A}, B \in \mathcal{B}\}$ 
die kanonische Fortsetzung der Matrizenmultiplikation auf Matrizenmengen.

Wir wollen nun Matrizenmengen als Transferfunktionen auffassen. Dazu definieren wir eine Abbildung $\alpha_{\text{Mat}} : L_M \to \mathcal{F}$. 
Diese bildet eine Matrix auf diejenige Funktion ab, die einen Zustand auf das Produkt der Matrix mit diesem Zustand abbildet, und setzen dies fort auf Mengen von Matrizen. 
Konkret definieren wir also für eine Matrizenmenge $\mathcal{A} \subseteq \Mat(\Sigma)$ 
\begin{align*}
\alpha_{\text{Mat}}(\mathcal{A}) : L &\to L, \\
S &\mapsto \{A \cdot s \mid A \in \mathcal{A}, s \in S\}.
\end{align*}
Die Abbildung $\alpha_{\text{Mat}}$ ist wohldefiniert: Seien dazu $\mathcal{A} \in L_M$ und $\mathcal{S} \subseteq L$. Dann ist
\begin{align*}
\alpha_{\text{Mat}}\big(\mathcal{A}\big)\big(\bigsqcup \mathcal{S}\big)
&= \big\{ A \cdot s \mid A \in \mathcal{A} \wedge s \in \bigsqcup \mathcal{S}\big\} \\
&= \big\{ A \cdot s \mid A \in \mathcal{A} \wedge (\exists S \in \mathcal{S}:s \in S) \big\} \\
&= \bigcup\big\{ \{ A \cdot s \mid A \in \mathcal{A} \wedge s \in S\} \mid S \in \mathcal{S}\big\} \\
&= \bigcup\big\{ \alpha_{\text{Mat}}(\mathcal{A})(S) \mid S \in \mathcal{S}\big\}.
\end{align*}
Also ist $\alpha_{\text{Mat}}(\mathcal{A})$ universell-distributiv. Dies zeigt die Wohldefiniertheit von $\alpha_{\text{Mat}}$.

Insbesondere gilt
$f_{(u,\mathtt{x_j := t},v)}((\mathcal{S}_\node)_{\node \in N})= \alpha_{\text{Mat}}( \{ M_{\mathtt{x_j:=t}}\}) (\mathcal{S}_u)$.

\subsection{Relationen als Summary-Information}\label{subsec:polyeder-rel}
Ein alternativer Ansatz ist, den Effekt einer Kante als Relation aufzufassen. Dabei stehen zwei Zustände $x$ und $y$ in Relation zueinander, wenn $y$ aus $x$ durch den Effekt dieser Kante entsteht. 
Genauer ordnen wir nun einem jeden Knoten eine Relation $\mathcal{R} \subseteq \rr^n \times \rr^n$ zu, wobei ein Tupel $(x,y)$ genau dann in $\mathcal{R}$ enthalten ist, 
wenn es eine Kante zu diesem Knoten gibt, die $x$ zu $y$ transfomiert.
Speziell für eine affine Zuweisung $\mathtt{x_j := t}$ ist die Relation $\mathcal{R}_{\mathtt{x_j:=t}} $, die den Effekt dieser Zuweisung beschreibt, gegeben durch
\[\big\{(x,y) \in \rr^n\times\rr^n \mid \forall 1 \le i \le n:((i\ne j \Rightarrow y_i = x_i) \wedge y_j = a_0 + \sum_{k=1}^n a_k x_k)\big\}.\]
Als zugrundeliegenden vollständigen Verband betrachen wir ähnlich wie bei den Matrizenmengen \[(L_R, \subseteq) := (2^{\rr^n \times \rr^n}, \subseteq).\]

Zum Berechnen des Effektes einer Prozedur stellen wir wieder ein Ungleichungssystem auf:
Konkret benutzen wir dazu die verallgemeinerten Transferfunktionen $f_{T_R,\id}$, $f_{T_R,\mathtt{x_j:=t}}$ und $f_{T_R,\mathtt{q()}}: {L_R}^N \to L_R$ mit
\begin{align*} 
f_{T_R,\id}						(\mathcal{R}) &:= \{(x,x) \in \rr^n \times \rr^n\}, \\
f_{T_R,(u,\mathtt{x_j:=t},v)}	(\mathcal{R}) &:= \mathcal{R}_{\mathtt{x_j := t}} \circ \mathcal{R}_u\\
f_{T_R,(u,\mathtt{q()},v)}		(\mathcal{R}) &:= \mathcal{R}_{r_{\q}} \circ \mathcal{R}_u  
\end{align*}
für $\mathcal{R} \in {L_R}^N$. Das zu lösende Ungleichungssystem hat also folgende Gestalt:
\begin{equation}\label{TR-Ugs}\begin{split}
T_R[s_{\p}] 	&\sqgeq \{(x,x) \in \rr^n\times\rr^n\} \\
T_R[v] 			&\sqgeq R_{\mathtt{x_j := t}} \circ T_R[u] \text { für } (u,\mathtt{s},v) \in E_p \text{ Basiskante} \\
T_R[v] 			&\sqgeq T_R[r_{\q}] \circ T_R[u] \text{ für } e=(u,\mathtt{q()},v) \in E \text{ Aufrufkante}.
\end{split}\end{equation}
Dabei ist $T_R[u] \in 2^{\rr^n \times \rr^n}$ für jedes $u \in N_{\p}$ und $\p \in \Proc$. 
 
Mit $\circ$ bezeichnen wir die Verknüpfung zweier Relationen: 
Fasst man eine Relation als eine Zuordnung auf und sind $(x,y)\in \mathcal{R}_1,(y,z) \in \mathcal{R}_2$, 
so wird also $x$ von $\mathcal{R}_1$ dem Wert $y$ zugeordnet und der wiederum von $\mathcal{R}_2$ dem Wert $z$. 
Die Verknüpfung $\mathcal{R}_2 \circ \mathcal{R}_1$ ordnet also dem Wert $x$ den Wert $z$ zu. 
Genauer ist
\[\mathcal{R}_2 \circ \mathcal{R}_1 := \{(x,z) \mid \exists y\in\rr^n:(x,y) \in \mathcal{R}_1 \wedge (y,z) \in \mathcal{R}_2\}.\] 
Offenbar bildet $\{(x,x) \in \rr^n\times\rr^n\}$ das neutrale Element dieser Verknüpfung.

Wiederum identifizieren wir mithilfe einer Abbildung $\alpha_{\text{Rel}} : L_R \to \mathcal{F}$ eine Relation mit einer Abbildung auf Zustandsmengen, 
indem wir eine Relation, wie oben beschrieben, als eine Zuordnung auffassen. Konkret ist für $R \subseteq \rr^n \times \rr^n$ dann 
\begin{align*}
\alpha_{\text{Rel}}(\mathcal{R}) : L &\to L, \\ S &\mapsto \{(1,y) \mid \exists (1,x)\in S:(x,y) \in \mathcal{R}\}.
\end{align*}
Dabei schreiben wir abkürzend $(1,x)$ für $(1,x_1,\dots,x_n)$ mit $x=(x_1,\dots,x_n) \in \rr$. 

Auch die Abbildung $\alpha_{\text{Rel}}$ ist wohldefiniert: Seien dazu $\mathcal{R} \in L_R$ und $\mathcal{S} \subseteq L$. Dann ist
\begin{align*}
\alpha_{\text{Rel}}\big(\mathcal{R}\big)\big(\bigsqcup \mathcal{S}\big)
&= \big\{ (1,y) \mid \exists x \in \rr^n:(x,y) \in \mathcal{R} \wedge x \in \bigsqcup \mathcal{S}\big\} \\
&= \big\{ (1,y) \mid \exists x \in \rr^n:(x,y) \in \mathcal{R} \wedge \exists S \in \mathcal{S}:x \in S\big\} \\
&= \bigcup\big\{ \{(1,y) \mid \exists x \in S:(x,y) \in \mathcal{R}\} \mid S \in \mathcal{S}\big\} \\
&= \bigcup\big\{ \alpha_{\text{Rel}}(\mathcal{R})(S) \mid S \in \mathcal{S}\big\}
\end{align*}
Also ist $\alpha_{\text{Rel}}(\mathcal{R})$ universell-distributiv. Dies zeigt die Wohldefiniertheit von $\alpha_{\text{Rel}}$.

Um später die kleinsten Lösungen $\underline{T}_M$ und $\underline{T}_R$ vergleichen zu können, benötigen wir einen Zusammenhang zwischen Matrizenmengen und Relationen, den wir später einführen werden.
Speziell für $\{M_{\mathtt{x_j := t}}\}$ und $\mathcal{R}_{\mathtt{x_j := t}}$, die ja beide den Effekt derselben affinen Zuweisung beschreiben, können wir einen solchen aber sofort herstellen: 
Es ist $(1,y) = M_{\mathtt{x_j := t}} \cdot (1,x)$ äquivalent zu 
$y_i = x_i$ für alle $i \ne j$ und $y_j = a_0 + \sum_{i=1}^n a_i x_i$, was nach Definition genau $(x,y) \in \mathcal{R}_{\mathtt{x_j := t}}$ bedeutet. 
Mit anderen Worten ist also
\[\mathcal{R}_{\mathtt{x_j := t}} = \{(x,y) \in \rr^n\times\rr^n \mid (1,y) = M_{\mathtt{x_j := t}} \cdot (1,x)\}.\]
Dies liefert also einen Zusammenhang zwischen $M_{\mathtt{x_j := t}}$ und $\mathcal{R}_{\mathtt{x_j := t}}$. 
Desweiteren erhalten wir sofort $f_{(u,\mathtt{x_j := t},v)}((\mathcal{S}_\node)_{\node \in N}) = \alpha_{\text{Rel}}(\mathcal{R}_{\mathtt{x_j:=t}}) (\mathcal{S}_u)$.

\subsection{Erreichbare Zustände}
In den vorigen Abschnitten \ref{subsec:polyeder-mat} und \ref{subsec:polyeder-rel} haben wir Summary-Informationen nicht direkt als Abbildungen $L \to L$ berechnet, 
sondern zunächst als Matrizenmengen bzw.~Relationen ausgedrückt. 
Außerdem haben wir Abbildungen $\alpha_\text{Mat}$ und $\alpha_\text{Rel}$ definiert, mit denen diese Matrizenmengen und Relationen als Abbildungen $L \to L$ aufgefasst werden können. 
Derartige Abbildungen benötigen wir, um die erreichbaren Zustände in den Programmknoten mit einem Ungleichungssystem der Gestalt \eqref{R-Ugs} aus dem funktionalen Ansatz zu berechnen. 
Für den Aufruf einer Prozedur $\q$ wird in diesem Ungleichungssystem die Abbildung $\underline{T}[r_\q]$ benötigt.
Diese haben wir hier nicht direkt zur Verfügung, da wir das Ungleichungssystem \eqref{T-Ugs}, mit dem diese Abbildungen ausgerechnet werden, 
durch die Ungleichungssysteme \eqref{TM-Ugs} und \eqref{TR-Ugs} ersetzt haben. 
Wir können jedoch die Abbildung $\underline{T}[r_\q]$ ausdrücken durch $\alpha_\text{Mat}(\underline{T_M}[r_\q])$ bzw.~$\alpha_\text{Rel}(\underline{T_R}[r_\q])$. 
Die entsprechenden Ungleichungssysteme bezeichnen wir mit $R_M$ und $R_R$ und die zugehörenden verallgemeinerten Transferfunktionen nennen wir $f_{R_M,\cdot}: L^N \to L$ und $f_{R_R,\cdot}: L^N \to L$. 
Sie sind definieren wir durch
\begin{align*}
f_{R_M,\init}(\mathcal{S}) 							&:= \Sigma \\
f_{R_M,(u,\mathtt{x_j:=t},v)}(\mathcal{S}) 			&:= f_{x_j:=t} (\mathcal{S}_u) \\
f_{R_M,(u,\mathtt{q()},v),\text{ent}}(\mathcal{S}) 	&:= \mathcal{S}_u\\
f_{R_M,(u,\mathtt{q()},v),\text{ret}}(\mathcal{S}) 	&:= \alpha_{\text{Mat}}(\underline{T_M}[r_{\q}]) (\mathcal{S}_u)
\end{align*}
und
\begin{align*}
f_{R_R,\init}(\mathcal{S}) 							&:= \Sigma \\
f_{R_R,(u,\mathtt{x_j:=t},v)}(\mathcal{S}) 			&:= f_{x_j:=t} (\mathcal{S}_u) \\
f_{R_R,(u,\mathtt{q()},v),\text{ent}}(\mathcal{S}) 	&:= \mathcal{S}_u\\
f_{R_R,(u,\mathtt{q()},v),\text{ret}}(\mathcal{S}) 	&:= \alpha_{\text{Rel}}(\underline{T_R}[r_{\q}]) (\mathcal{S}_u).
\end{align*}
Bis auf die Abbildungen $f_{R_M,(u,\mathtt{q()},v),\text{ret}}$ und $f_{R_R,(u,\mathtt{q()},v),\text{ret}}$ stimmen die verwendeten Abbildungen also überein.

\subsection{Vergleich der beiden Summary-basierten Ansätze}
Im Folgenden untersuchen wir, ob die Ansätze mit Matrizenmengen bzw.~Relationen zu jedem Programmpunkt stets dieselben erreichbaren Zustände berechnen 
oder ob ein Ansatz eine präzisere Lösung liefert als der andere.

Zunächst definieren wir eine Abbildung zwischen Matrizenmengen und Relationen. 
Dabei ordnen wir einer Matrizenmenge diejenigen Paare zu, bei denen die zweite Komponente aus der ersten durch Anwenden einer Matrix dieser Menge entsteht. 
Da die Matrizenmengen auf $\{1\}\times\rr^n$ operieren, die Relationen dagegen aus Tupeln von Elementen aus $\rr^n$ bestehen, brauchen wir zusätzlich noch eine Inklusions- und Projektionsabbildung, 
um diese beide Räumen durch eine Bijektion miteinander in Verbindung zu setzen:
\begin{align*}
 \iota: \rr^n &\to \{1\}\times\rr^n,			& \pi: \{1\}\times\rr^n &\to \rr^n,
\\(x_1,\dots,x_n) &\mapsto (1,x_1,\dots,x_n) 	& (1,x_1,\dots,x_n) &\mapsto (x_1,\dots,x_n).
\end{align*}

Wir ordnen nun einer Matrix eine Relation zu mithilfe der Abbildung
\begin{align*}\tilde{\Phi}: \mats &\to \relspot,\\A &\mapsto \tilde{\Phi}_A:= \{(x, \pi(A \cdot \iota(x))) \mid x \in \rr^n\}\end{align*}
und setzen diese kanonisch auf Matrizenmengen fort durch
\begin{align*}
\Phi: \matspot& \to \relspot,\\\mathcal{A} 
&\mapsto \bigcup\big\{\tilde{\Phi}_A \mid A \in \mathcal{A}\big\} 
= \big\{(x, \pi(A \cdot \iota(x))) \mid x \in \rr^n \wedge A \in \mathcal{A}\big\}.
\end{align*}

\begin{lemma}\label{T-lemma}
Seien $X_1, X_2 \in \mats$ and $\mathcal{X} \subseteq \matspot$. Dann gilt:
\begin{enumerate}
 \item $\Phi_{I_{n+1}} = \{(x,x) \mid x \in \rr^n\}$,
 \item $\Phi_{M_{\mathtt{x_j := t}}} = \mathcal{R}_{\mathtt{x_j := t}}$,
 \item $\Phi_{X_2 \circ X_1} = \Phi_{X_2} \circ \Phi_{X_1}$,
 \item $\Phi_{\bigcup\mathcal{X}} = \bigcup\{\Phi_X \mid X \in \mathcal{X}\}$.
\end{enumerate}
\end{lemma}
Hierbei bezeichnen wir mit $\circ$, wie in den vorigen Abschnitten definiert, die Komposition von Matrizenmengen bzw.~von Relationen.
\begin{proof}
Die Aussagen a) und b) resultieren direkt aus der Tatsache, dass $\Phi_{\{A\}} = \tilde{\Phi}_A$ für jede Matrix $A\in\Mat(\Sigma)$ gilt, sowie den Identitäten 
\[\pi (I_{n+1} \cdot \iota(x)) = x\] und \[\pi (M_{\mathtt{x_j := t}} \cdot \iota (x)) = (x_1,\dots, x_{j-1}, a_0 + \sum_{i=1}^n a_i x_i, x_{j+1},\dots,x_n).\]

Seien nun $A \in X_1$ und $B \in X_2$. Dann ist
\begin{align*}
\tilde{\Phi}_B \circ \tilde{\Phi}_A 
&= \{ (x,z) \mid \exists y:(x,y)\in \tilde{\Phi}_A \wedge (y,z) \in \tilde{\Phi}_B \} \\
&= \{ (x,z) \mid \exists y:y = \pi(A \cdot \iota (x)) \wedge z = \pi(B \cdot \iota (y)) \} \\
&= \{ (x,z) \mid z = \pi( B \cdot \iota (\pi(A \cdot \iota (x))))\} \\
&= \{ (x,z) \mid z = \pi ( B \cdot A \cdot \iota (x)) \} \\
&= \tilde{\Phi}_{BA}.
\end{align*}
Mit dieser Identität folgt nun c) aus
\begin{align*}
\Phi_{X_2 \circ X_1}
&= \Phi_{\{BA \mid A \in X_1 \wedge B \in X_2\}} \\
&= \bigcup \big\{\tilde{\Phi}_{BA} \mid A \in X_1 \wedge B \in X_2\big\} \\
&= \bigcup \big\{\tilde{\Phi}_{B} \circ \tilde{\Phi}_A \mid A \in X_1 \wedge B \in X_2\big\} \\
&= \bigcup \big\{\tilde{\Phi}_{B} \mid B \in X_2\big\} \circ \bigcup \big\{\tilde{\Phi}_A \mid A \in X_1 \big\} \\
&= \Phi_{X_2} \circ \Phi_{X_1}.
\end{align*}
Ähnlich erhalten wir für $\mathcal{X} \subset \mathcal{M}$
\begin{align*}
\Phi_{\bigcup \mathcal{X}}
&= \Phi_{\bigcup \{X \mid X \in \mathcal{X}\}} \\
&= \bigcup \big\{\tilde{\Phi}_A \mid A \in \bigcup \big\{X \mid X \in \mathcal{X}\big\} \big\} \\
&= \bigcup \big\{ \bigcup \big\{\tilde{\Phi}_A \mid A \in X \big\} \mid X \in \mathcal{X}\big\} \\
&= \bigcup \big\{\tilde{\Phi}_X \mid X \in \mathcal{X}\big\},
\end{align*}
was d) beweist.
\end{proof}

Wir wollen nun die kleinsten Lösungen der Ungleichungssysteme $T_M$ und $T_R$ zueinander in Verbindung setzen. Dabei ist $\underline{T_M}[u] \in \matspot$
für jeden Knoten $u\in N$, also $\underline{T_M} \in (\matspot)^{N}$. Hierbei identifizieren wir einen Vektor ${A}$ mit der Abbildung 
\begin{align*} 
N &\to 2^{\Mat(\Sigma)},\\ 
u &\mapsto {{A}}_u.
\end{align*} 
Entsprechend ist $\underline{T_R} \in (\relspot)^{N}$. Wir definieren also eine Relation \[R \subseteq (\matspot)^{N} \times (\relspot)^{N},\] um die beiden Lösungen miteinander zu vergleichen. Genauer sei 
\[({\mathcal{A}}, \mathcal{B}) \in R :\Longleftrightarrow \forall u \in N:\mathcal{B}(u) = \Phi_{{\mathcal{A}}(u)}.\] 

\begin{satz}
Es gilt $(\underline{T_M}, \underline{T_R}) \in R$.
\end{satz} 
\begin{proof}
Seien $F_M$ bzw.~$F_R$ wiederum diejenigen Abbildungen, die die Ungleichungssysteme $T_M$ bzw.~$T_R$ beschreiben. 
Damit ist $\underline{T_M}=\lfp(F_M)$ und entsprechend $\underline{T_R}=\lfp(F_R)$. 
Also reicht es, die Bedingungen von Satz \ref{relFixpkte} für eben diese Abbildungen $F_M$ und $F_R$ nachzurechnen. 
Offenbar sind $((\matspot)^{N}, \subseteq^N)$ und $((\relspot)^{N},\subseteq^N)$ vollständige Verbände und die $F^M$ und $F^R$ monoton, da sie aus monotonen Funktionen entstehen.

Seien ${\mathcal{A}}\in(\matspot)^{N}, \mathcal{B}\in(\relspot)^{N}$ mit $({\mathcal{A}},\mathcal{B}) \in R$. Um Bedingung a) zu zeigen, muss \[(F_M(\mathcal{A}),F_R(\mathcal{B}))\in R\] gezeigt werden.
Konkret müssen wir also 
\[F_R(\mathcal{B})(u) = \Phi_{F_M(\mathcal{A})(u)}\]
für jedes $u \in N$ nachweisen. Nach Teil d) von Lemma \ref{T-lemma} genügt es, dies für alle Funktionen der ursprünglichen Ungleichungssysteme zu zeigen.

Für die Initialwerte folgt \[(I_{n+1},\{(x,x)\in\rels\}) \in R\] direkt aus Teil a) von Lemma \ref{T-lemma}.
Seien weiter $u \in N$ und $T_R[u] = \Phi_{T_M[u]}$. Mit den Aussagen c) und b) von Lemma \ref{T-lemma} folgt nun
\[
\Phi_{\{M_{\mathtt{x_j:=t}}\} \circ T_M[u]}
= \Phi_{\{M_{\mathtt{x_j:=t}}\}}\circ \Phi_{T_M[u]} 
= R_{\mathtt{x_j:=t}}\circ T_R[u],
\]
was $(\{M_{\mathtt{x_j:=t}}\}, \mathcal{R}_{\mathtt{x_j:=t}})\in R$ liefert.

Sei nun $u \in N$, $\p \in \Proc$, $T_R[u] = \Phi_{T_M[u]}$ and $T_R[r_\p] = \Phi_{T_M[r_\p]}$. 
Wiederum mit den Teilen c) and b) von Lemma \ref{T-lemma} erhalten wir
\begin{align*}
\Phi_{T_M[r_\p]\circ T_M[u]}
= \Phi_{T_M[r_\p]}\circ \Phi_{T_M[u]} 
= T_R[r_\p] \circ T_R[u] 
.\end{align*}
Damit ist auch $(T_M[r_\p], T_R[r_\p])\in R$ und Bedingung a) von Satz \ref{relFixpkte} ist gezeigt.

Seien also $X \subseteq R$ und $u \in N$. 
Für jedes $A \in X_1(u)$ gibt es also ${\mathcal{A}} \in X_1$ mit $A = {\mathcal{A}}(u)$ und ein $\mathcal{B} \in X_2$ mit $(\mathcal{A},\mathcal{B})\in R$ und damit $\Phi_A = B$ für $B:=\mathcal{B}(u)$. 
Dann ist \[\{\Phi_A \mid A \in X_1(u) \} \subseteq \{B \mid B \in X_2(u)\}.\] 
Die andere Inklusion zeigt man analog. Daraus erhalten wir
\[
\Phi_{\bigsqcup X_1(u)}
= \bigsqcup\big\{\Phi_A \mid A \in X_1(u)\big\}
= \bigsqcup\big\{B \mid B \in X_2(u)\big\}
= \bigsqcup X_2(u).
\]
Auch Bedingung b) von Satz \ref{relFixpkte} ist gezeigt. Es folgt die Behauptung.
\end{proof}

\begin{satz}
Für jedes $u \in N$ gilt $\underline{R_M[u]} = \underline{R_R[u]}$.
\end{satz}
\begin{proof}
Die Ungleichungssysteme unterscheiden sich nur durch die Definitionen der Abbildung $f_{R_M,(u,\mathtt{q()},v),\text{ret}}$ und $f_{R_R,(u,\mathtt{q()},v),\text{ret}}$ für diejenigen Ungleichungen, 
die die Rück\-kehr aus einer Prozedur $\q \in \Proc$ beschreiben. es ist
\[f_{R_M,(u,\mathtt{q()},v),\text{ret}} = \alpha_{\text{Mat}}({\underline{T_M}[r_{\q}]}) \circ \pr_u\]
und 
\[f_{R_R,(u,\mathtt{q()},v),\text{ret}} = \alpha_{\text{Rel}}(\underline{T_R}[r_{\q}]) \circ \pr_u\]
für die Projektion $\pr_u: L^N \to L$ auf die Komponente $u$. Darum genügt es, \[\alpha_{\text{Mat}}(\underline{T_M}[r_{\q}]) = \alpha_{\text{Rel}}(\underline{T_R}[r_{\q}])\] zu zeigen.
Sei dazu ${S} \in L$. Zunächst beobachten wir
\begin{align*}
\alpha_{\text{Mat}}(\underline{T_R}[r_{\q}]) ({S})
&= \{ (1,\tilde{y}) \mid \exists (1,\tilde{x}) \in {S} :(\tilde{x},\tilde{y})\in \underline{T_R}[r_{\q}]\} \\
&= \{ y \mid \exists x \in S :(\pi(x),\pi(y))\in \underline{T_R}[r_{\q}]\}.
\end{align*}
Nach dem vorigen Lemma wissen wir bereits $\underline{T_R}[r_{\q}] = \Phi_{\underline{T_M}[r_{\q}]}$. Somit gilt weiter
\begin{align*}
\lefteqn{\{ y \mid \exists x \in {S} :(\pi(x),\pi(y))\in \underline{T_R}[r_{\q}]\}}\\
&= \{ y \mid \exists x \in {S}~\exists A \in \underline{T_M}[r_{\q}]:\pi(y) = \pi(A \cdot \iota(\pi(x)))\} \\
&= \{ y \mid \exists x \in {S}~\exists A \in \underline{T_M}[r_{\q}]:y = A \cdot x\} \\
&= \{ A \cdot x \mid x \in {S} \wedge A \in \underline{T_M}[r_{\q}]\} \\
&= \alpha_{\text{Mat}}(\underline{T_M}[r_\q])(S)
\end{align*}
Es folgt die Behauptung.
\end{proof}

Die beiden in den letzten Abschnitten vorgestellen Analysen mithilfe von Matrizenmengen bzw.~Relationen als Summary-Informationen liefern also das gleiche Ergebnis. 
Später werden wir überprüfen, ob dies für den nächsten Schritt der Polyederanalyse, in dem zusätzlich konvexe Hüllen gebildet werden, immer noch gilt. 
Die Bildung konvexer Hüllen ist eine \emph{abstrakte} Interpretation. Mit solchen werden wir uns in \autoref{chap:abstr} beschäftigen.

\section{Terminierung des Workset-Algorithmus}\label{sec:eff}
In \autoref{sec:funktionaler-ansatz} und \autoref{sec:callstring-ansatz} haben wir zwei Ansätze zur interprozeduralen Datenflussanalyse vorgestellt. 
In beiden Ansätzen werden Ungleichungssysteme aufgestellt, deren kleinste Lösungen genau die gewünschten Informationen über das Programm liefsern. 
In Abschnitt \ref{sec:wla} haben als Hilfsmittel zum Lösen von Ungleichungssystemen den {Workset-Algorithmus} vorgestellt. Falls er terminiert, so liefert er eine Lösung dieses Ungleichungssystems.
Außerdem hatten wir gesehen, dass ein solcher Algorithmus tatsächlich terminiert, falls der zugrundeliegende Verband und die Menge der Ungleichungen endlich sind. 
Ist der Verband dagegen nicht endlich, so kann Terminierung nicht garantiert werden, wie das folgende Beispiel zeigt.
\begin{bsp}
Wir betrachten das Ungleichungssystem $\mathcal{U}$, das gegeben ist durch
\begin{align*}
&A[1] \sqgeq \init\\
&A[2] \sqgeq f_{x:=0}(A[1])\\
&A[2] \sqgeq f_{x:=x+1}(A[2])\\
&A[3] \sqgeq f_{x:=0}(A[2])
\end{align*}
mit $\init := [-\infty,+\infty]$ sowie
\begin{align*}
f_{x:=0}(I) &= \{[0,0]\} 
\intertext{und}
f_{x:=x+1}(I) &= \begin{cases} \emptyset & \text{ falls } I = \emptyset \\ [l+1,u+1] & \text{ falls } I =[l,u] \end{cases}
\end{align*}
für $I \in L_{\text{Int}}$.

Im Algorithmus bekommt nun $A[2]$ aufgrund der zweiten Ungleichung den Wert $[0,0]$. Durch die dritte Ungleichung wird dieser zunächst in einem \quotes{ersten} Schritt ersetzt durch
\[[0,0] \cup [1,1] = [0,1].\]
Anschließend ist die dritte Ungleichung wiederum nicht erfüllt. Im \quotes{zweiten} Schritt wird der Wert von $A[2]$ also ersetzt durch
\[[0,1] \cup [1,2] = [0,2].\]
Die dritte Ungleichung ist weiterhin nicht erfüllt. Dies setzt sich fort, so dass $A[2]$ nach $n$ Schritten den Wert $[0,n]$ hat. 
Da die dritte Ungleichung auch für diesen Wert nicht erfüllt ist, kann der Algorithmus in endlicher Zeit nicht terminieren.
\end{bsp}

Wir haben in diesem Kapitel zwei Ansätze zur Analyse interprozeduraler Programme untersucht. 
Der Call-String-Ansatz simulitert eine interprozedurale Analyse mithilfe von Call-Strings, 
während der funktionale Ansatz zunächst die Effekte der Prozeduren bestimmt und diese dann wie die Transferfunktionen für Basiskanten benutzt. 
Weiter haben wir gezeigt, dass die kleinsten Lösungen beider Ansätze übereinstimmen, sofern die verwendeten Abbildungen gewissen Distributivitätseigenschaften genügen. 

Zuletzt haben wir gesehen, dass eine Berechnung der Lösung von Ungleichungssystemen durch den Workset-Algorithmus nicht unbedingt terminieren muss. 
Deshalb müssen in der Praxis Strategien benutzt werden, mit deren Hilfe neue Analysen gefunden, zu denen Lösungen in endlicher Zeit berechnet werden können. 
Zwei solche Mittel sind die \emph{abstrakte Interpretation} und die Verwendung eines \emph{Widening-Operators}. 
Diese Konzepte werden wir in den beiden kommenden Kapiteln vorstellen und dort wiederum die Ergebnisse des funktionalen und des Call-String-Ansatzes vergleichen.


\chapter{Abstrakte Interpretation}\label{chap:abstr}
Abstrakte Interpretation ist ein von Cousot und Cousot \cite{CousotCousot77-1} entwickeltes Mittel, um den Berechnungsaufwand in praktischen Analysen zu reduzieren. 
Die Idee dabei ist, anstelle der gesuchten eine \quotes{abstrakte} Information zu berechnen. 
Dazu abstrahiert man die verallgemeinerten Transferfunktionen des Ungleichungssystems zu Abbildungen auf dem abstrakten Verband. 
Dadurch erhält man ein Ungleichungssystem, dessen kleinste Lösung unter geeigneten Voraussetzungen genau die abstrakte Version der gesuchten Information ist. 
In dem Fall nennen wir das abstrahierte Ungleichungssystem \emph{präzise}. 
Die Lösung des abstrahierten Ungleichungssystems kann nun im Gegensatz zu der des konkreten Ungleichungssystems endlich darstellbar und in endlicher Zeit berechenbar sein. 
Darum ist abstrakte Interpretation ein starkes Werkzeug bei der Konstruktion effizienter Analysen.

Abstrakte Interpretation kann aber auch ein Mittel zur Konstruktion neuer korrekter oder präziser Analysen sein. 
Dazu wird eine gesuchte Information aus einer anderen berechnet, zu deren Berechnung wiederum eine korrekte oder präzise Anlayse bekannt ist. 
Ist eine solche Analyse gegeben und abstrahiert man diese korrekt und präzise, so gewinnt man daraus wieder eine korrekte bzw.~präzise Analyse. 
Die Bedingungen, die an die abstrakte Interpretation gestellt werden, damit diese korrekt oder präzise ist, können also dazu verwendet werden, aus einer gegebenen Analyse eine neue zu konstruieren.

In \autoref{chap:PA} haben wir mit dem funktionalen und dem Call-String-Ansatz zwei Ansätze zur Analyse interprozeduraler Programme betrachtet, 
die unter geeigneten Voraussetzungen dieselbe kleinste Lösung liefern. 
In diesem Kapitel wollen wir diese beiden Ansätze zusätzlich abstrakt interpretieren, Eigenschaften der so abstrahierten Ungleichungssysteme untersuchen und überprüfen, 
ob beide Ansätze wieder die gleichen Lösungen liefern, wenn zusätzlich abstrakte Interpretation benutzt wird.
Dazu führen wir zunächst \emph{Abstraktionen} und \emph{abstrahierte Ungleichungssysteme} ein 
und untersuchen letztere im Hinblick auf die sogenannte \emph{Korrektheit} und \emph{Präzision} ihrer kleinsten Lösungen. 
Die entsprechenden Begrifflichkeiten finden sich in \cite{CousotCousot77-1} oder auch in \cite{Niel}. 
Anschließend übertragen wir diese Vorgehensweise auf die in Kapitel \ref{chap:PA} vorgestellten Ansätze zur \emph{interprozeduralen} Analyse. 
Zuletzt betrachten wir wieder die Polyederanalyse, bei der wir mithilfe der Bildung \textit{konvexer Hüllen} die Ungleichungssysteme abstrahieren wollen.

\section{Galois-Verbindungen}\label{sec:Galois}
Um einen \quotes{konkreten} und einen \quotes{abstrakten} vollständigen Verband miteinander in Verbindung zu setzen, 
benötigen wir nicht nur eine \emph{Abstraktion}, die einem konkreten Element ein abstraktes zuordnet, 
sondern auch eine \emph{Konkretisierung}, die zu jedem abstrakten Element ein konkretes liefert. 
Solche Abbildungen erhält man durch eine Galois-Verbindung.
\begin{dfn}
Seien $(L,\sqleq)$ und $(L^\sharp, \sqleq^\sharp)$ Halbordnungen und $\alpha: L \to L^\sharp$ sowie $\gamma:L^\sharp \to L$ Abbildungen. 
Dann heißt $(\alpha,\gamma)$ \emph{Galois-Verbindung}\index{Galois-Verbindung}, falls für alle $x \in L$ und $y \in L^\sharp$ gilt:
\[\alpha(x) \sqleq^\sharp y \Longleftrightarrow x \sqleq \gamma(y).\]
Dabei nennt man $\alpha$ \emph{Abstraktion}\index{Abstraktion} und $\gamma$ \emph{Konkretisierung}\index{Konkretisierung}.
\end{dfn}

Wir erinnern an folgendes Beispiel in \cite{Wilhelm-Maurer}.
\begin{bsp}
\newcommand{\mybar}[1]{\bar{#1}}
Zu zeigen ist die Aussage, dass für jede ganze Zahl $z \in \zz$ der Ausdruck $z^2(z+1)^2$ durch $4$ teilbar ist. 
Dazu abstrahieren wir die ganzen Zahlen, indem wir eine Zahl mit dem Rest identifizieren, der bei Division dieser Zahl durch $4$ übrig bleibt. 
Wir haben also
\begin{align*}
L = \zz \text{ und } L^\sharp = \{\mybar{0}, \mybar{1}, \mybar{2}, \mybar{3}\}
\end{align*}
und definieren $\alpha= L \to L^\sharp$ durch $\alpha(z) := \mybar{r}$ für dasjenige $r \in \{0,1,2,3\}$ mit \[z \equiv x \text{ mod } 4.\]

Weiter definieren wir Abstraktionen $\mybar{+}: L^\sharp \times L^\sharp \to L^\sharp$ und $\mybar{\cdot}: L^\sharp \times L^\sharp \to L^\sharp$ 
der Addition $+ : L \times L \to L$ und der Multiplikation $\cdot : L \times L \to L$ durch
\begin{align*}
\mybar{z_1} \mybar{+} \mybar{z_2}		:= \overline{\mybar{z_1} + \mybar{z_2}} 
\text{ und }
\mybar{z_1} \mybar{\cdot} \mybar{z_2}	:= \overline{\mybar{z_1} \cdot \mybar{z_2}}
\end{align*}

Man sieht nun leicht die Gültigkeit von
\begin{align*}
\overline{z_1 + z_2}		= \mybar{z_1} \mybar{+} \mybar{z_2}
\text{ und }
\overline{z_1 \cdot z_2}	= \mybar{z_1} \mybar{\cdot} \mybar{z_2}.
\end{align*}
Es kann nun der Ausdruck 
\[{z^2(z+1)^2}\]
abstrakt interpretiert werden.
Berechnet man diesen abstrakte Ausdruck für jeden möglichen Wert $z \in \{0,1,2,3\}$, so erhält man in jedem der vier Fälle den Wert $\mybar{0}$. Die Aussage
\[\overline{z^2(z+1)^2} = \mybar{0}\]
ist nun äquivalent zu der gewünschten Aussage, dass $z^2(z+1)^2$ durch $4$ teilbar ist.
\end{bsp}
In diesem Beispiel wurde also, anstatt eine Aussage für unendliche viele Elemente $z \in \zz = L$ zu zeigen, eine äquivalente Aussage in einem endlichen Raum $L^\sharp$ bewiesen. 
Auf ähnliche Weise kann es also möglich sein, ein Ungleichungssystem über einem unendlichen Verband $(L, \sqleq)$ so zu abstrahieren, 
dass ein entsprechendes Ungleichungssystem über einem Verband $(L^\sharp,\sqleq^\sharp)$ definiert ist, der endlich oder unendlich mit endlicher Kettenhöhe ist.

Das nächste Lemma zeigt, dass sich Abstraktionen und Konkretisierungen bereits eindeutig gegenseitig festlegen.
\begin{lemma}\label{lemma-eind-abstr}
Sei $(\alpha,\gamma)$ eine Galois-Verbindung zwischen den vollständigen Verbänden $(L\sqleq)$ und $(L^\sharp,\sqleq^\sharp)$. 
Dann bestimmen $\alpha$ und $\gamma$ sich eindeutig durch
\begin{align*}
\alpha(x) = &\bigsqcap\nolimits^\sharp \big\{ y \in L^\sharp \mid x \sqleq \gamma(y) \big\} \\
\gamma(y) = &\bigsqcup \big\{ x \in L \mid \alpha(x) \sqleq y \big\}.
\end{align*}
\end{lemma}
\begin{proof}
Für beliebiges $x \in L$ gilt
\begin{align*}
\alpha(x) 
= \bigsqcap\nolimits^\sharp \big\{ y \in L^\sharp \mid \alpha(x) \sqleq^\sharp y \big\}
= \bigsqcap\nolimits^\sharp \big\{ y \in L^\sharp \mid x \sqleq \gamma(y) \big\}.
\end{align*}
Dabei resultiert die letzte Gleichheit aus der Äquivalenz von $\alpha(x) \sqleq^\sharp y$ und $\gamma(y) \sqleq x$.

Diese Darstellung definiert $\alpha$ bereits eindeutig.
Die eindeutige Darstellung von $\gamma$ zeigt man entsprechend.
\end{proof}
Die Konkretisierung zu einer Abstraktion ist also eindeutig. Mit dieser eindeutigen Darstellung kann man nun leicht zeigen, dass Abstraktionen und Konkretisierungen monoton sind.
\begin{kor}\label{kor:galois-ist-mon}
Sei $(\alpha,\gamma)$ eine Galois-Verbindung zwischen den voll\-stän\-di\-gen Ver\-bän\-den $(L\sqleq)$ und $(L^\sharp,\sqleq^\sharp)$. Dann sind $\alpha$ und $\gamma$ monoton.
\end{kor}
\begin{proof}
Seien zunächst $x_1,x_2 \in L$ mit $x_1 \sqleq x_2$. Dann ist offenbar
\[
 \big\{ y \in L^\sharp \mid x_2 \sqleq \gamma(y) \big\}
\subset
 \big\{ y \in L^\sharp \mid x_1 \sqleq \gamma(y) \big\}.
\]
Damit gilt nach Lemma \ref{lemma-eind-abstr}
\[
 \alpha(x_2)
= \bigsqcap\nolimits^\sharp \big\{ y \in L^\sharp \mid x_2 \sqleq \gamma(y) \big\}
\sqgeq^\sharp \bigsqcap\nolimits^\sharp \big\{ y \in L^\sharp \mid x_1 \sqleq \gamma(y) \big\}
= \alpha(x_1).
\]
Also ist $\alpha$ monoton. Die Monotonie von $\gamma$ zeigt man genauso.
\end{proof}

Wir können nun zeigen, dass eine Abbildung zwischen zwei vollständigen Ver\-bän\-den genau dann eine Abstraktion in einer Galois-Verbindung ist, wenn sie universell-distributiv ist. 
\begin{lemma}\label{lem:abstr-sind-univ-distr}
Sind $(L,\sqleq)$ und $(L^\sharp, \sqleq^\sharp)$ zwei vollständige Verbände und $\alpha: L \to L^\sharp$ eine Abbildung, 
so ist $\alpha$ genau dann universell-distributiv, wenn $\alpha$ eine Abstraktion ist.
\end{lemma}
\begin{proof}
Sei $\alpha$ eine Abstraktion und $\gamma$ die zugehörige Konkretisierung. Sei $X \subseteq L$ beliebig. Offenbar gilt für jedes $x \in X$ bereits 
\[\alpha(x) \sqleq^\sharp \bigsqcup\nolimits^\sharp \big\{\alpha(x) \mid x \in X\big\} = \bigsqcup\nolimits^\sharp \alpha(X)\] 
und nach Definition einer Galois-Verbindung damit auch 
\[x \sqleq \gamma\big(\bigsqcup\nolimits^\sharp \alpha(X)\big).\] 
Daraus folgt direkt 
\[\bigsqcup X \sqleq \gamma\big(\bigsqcup\nolimits^\sharp \alpha(X)\big),\] 
was wiederum äquivalent ist zu 
\[\alpha\big(\bigsqcup X\big) \sqleq^\sharp \bigsqcup\nolimits^\sharp \alpha(X).\] 
Nach Korollar \ref{kor:galois-ist-mon} ist $\alpha$ monoton und dies impliziert die umgekehrte Ungleichung und somit Gleichheit. Also ist $\alpha$ universell-distributiv.

Sei nun $\alpha$ universell-distributiv. Definiere wie in Lemma \ref{lemma-eind-abstr} 
\begin{align*}
\gamma: L^\sharp &\to L, \\
y &\mapsto \bigsqcup L_y
\end{align*} 

für \[L_y := \{x \in L \mid \alpha(x) \sqleq^\sharp y\}.\] Da $(L,\sqleq)$ ein vollständiger Verband ist, existiert $\bigsqcup L_y$ und $\gamma$ ist wohldefiniert. 
Ist $y \sqleq^\sharp y'$, so ist offenbar $L_y \subseteq L_{y'}$ und damit $\gamma(y) \sqleq \gamma(y')$. Also ist $\gamma$ monoton. 
Es bleibt zu zeigen, dass $(\alpha,\gamma)$ tatsächlich eine Galois-Verbindung ist.

Seien dazu $x \in L$ und $y \in L^\sharp$. Gelte zunächst $\alpha(x) \sqleq^\sharp y$. Dann ist $x \in L_y$ und damit \[x \sqleq \bigsqcup L_y = \gamma(y).\]

Sei nun umgekehrt $x \sqleq \gamma(y) = \bigsqcup L_y$. Mit der Monotonie und universellen Distributivität von $\alpha$ folgt 
\[\alpha(x) \sqleq^\sharp \alpha\big(\bigsqcup L_y\big) = \bigsqcup\nolimits^\sharp \big\{\alpha(x) \mid x \in L_y\big\} \sqleq^\sharp y,\] 
wobei letzteres daraus resultiert, dass $\alpha(x) \sqleq^\sharp y$ für jedes $x \in L_y$. Damit folgt die Behauptung.
\end{proof}

Das folgende Lemma liefert eine äquivalente Beschreibung des Begriffs \emph{Galois-Ver\-bin\-dung}, die wir im Folgenden benötigen.
\begin{lemma}\label{lemma-eig-galois}
Seien $(L\sqleq)$ und $(L^\sharp,\sqleq^\sharp)$ vollständige Verbände und $\alpha: L \to L^\sharp$ und $\gamma: L^\sharp \to L$. Dann sind die folgenden Aussagen äquivalent:
\begin{enumerate}
 \item Es ist $(\alpha,\gamma)$ ist eine Galois-Verbindung.
 \item Die Abbildungen $\alpha$ und $\gamma$ sind monoton und erfüllen $\alpha \circ \gamma \sqleqmap \id_{L^\sharp}$ und $\gamma \circ \alpha \sqgeqmap \id_L$.
\end{enumerate}
\end{lemma}
\begin{proof}
Sei zunächst $(\alpha,\gamma)$ eine Galois-Verbindung. Die Monotonie von $\alpha$ und $\gamma$ haben wir bereits in Korollar \ref{kor:galois-ist-mon} gezeigt. 
Für beliebiges $x\in L$ ist trivialerweise $\alpha(x) \sqleq^\sharp \alpha(x)$. Nach Definition der Galois-Verbindung folgt daraus $x \sqleq \gamma \circ \alpha (x)$. 
Ebenso zeigt man $\alpha \circ \gamma \sqleqmap \id_{L^\sharp}$. 

Seien nun $\alpha$ und $\gamma$ monoton und gelte $\alpha \circ \gamma \sqleqmap \id_{L^\sharp}$ und $\gamma \circ \alpha \sqgeqmap \id_L$. 
Sei $x \in L$ und $y \in L^\sharp$ mit $\alpha(x) \sqleq^\sharp y$. Nach Voraussetzung und mit der Monotonie von $\gamma$ ist dann $x \sqleq \gamma \circ \alpha (x) \sqleq \gamma(y)$. 
Analog zeigt man, dass $x \sqleq \gamma(y)$ schon $\alpha(x) \sqleq^\sharp y$ impliziert. Also ist $(\alpha,\gamma)$ eine Galois-Verbindung.
\end{proof}

Die Verknüpfungen von Abstraktion und Konkretisierung liefern also noch nicht die Identität, approximieren diese aber. 
Verknüpft man die Verknüpfungen wiederum mit der Abstraktion oder der Konkretisierung, so erhält man dabei aber die Abstraktion bzw.~die Konkretisierung. Dies zeigt das folgende Lemma.
\begin{lemma}\label{lemma:eig-galois-verknuepfungen}
Sei $(\alpha,\gamma)$ eine Galois-Verbindung zwischen den vollständigen Verbänden $(L,\sqleq)$ und $(L^\sharp,\sqleq^\sharp)$. Dann gilt
\begin{align*}
 \alpha \circ \gamma \circ \alpha = \alpha \text{ und } \gamma \circ \alpha \circ \gamma = \gamma.
\end{align*}
\end{lemma}
\begin{proof}
Nach Lemma \ref{lemma-eig-galois} ist $\alpha \circ \gamma \sqleqmap \id_{L^\sharp}$ und $\gamma \circ \alpha \sqgeqmap \id_L$. Damit und mit der Monotonie von $\alpha$ folgt für $x \in L$
\begin{align*}
				\alpha(x) 
\sqleq^\sharp 	\alpha \circ (\gamma \circ \alpha) (x)
= 				(\alpha \circ \gamma) \circ \alpha (x)
\sqleq^\sharp 	\alpha(x).
\end{align*}
Somit gilt überall bereits Gleichheit und es folgt die erste Aussage. Die andere Aussage zeigt man völlig analog.
\end{proof}

\begin{bsp}\label{bsp:summary-inform-von-polyederana-sind-abstr}
In \autoref{chap:PA} haben wir den funktionalen Ansatz zur interprozeduralen Datenflussanalyse untersucht, welcher Transferfunktionen als Summary-Informationen berechnet. 
In der Polyederanalyse in \autoref{sec:polyeder} haben wir diese Summary-Informationen jedoch nicht direkt als Transferfunktionen berechnet, sondern zunächst Matrizenmengen bzw.~Relationen berechnet. 
Diese haben wir mithilfe von Abbildungen $\alpha_{\text{Mat}}$ und $\alpha_{\text{Rel}}$ als Transferfunktionen aufgefasst. Diese Abbildungen sind Abstraktionen. 
Um dies zu sehen, reicht es nach Lemma \ref{lem:abstr-sind-univ-distr} nachzuweisen, dass die Abbildungen universell-distributiv sind.

Seien dazu zunächst $\mathcal{M} \subset L_M$ eine Menge von Matrizenmengen und $S \in L$ eine Menge von Zuständen. Dann ist
\begin{align*}
\lefteqn{\alpha_{\text{Mat}} \big(\bigcup \mathcal{M} \big) \big(S\big)}\\
&= \big\{ A \cdot s \mid s \in S \wedge A \in \bigcup\mathcal{M} \big\} \\
&= \big\{ A \cdot s \mid s \in S \wedge (\exists \mathcal{A} \in \mathcal{M}:A \in \mathcal{A} \big) \} \\
&= \bigcup\big\{ \{A \cdot s \mid s \in S \wedge A \in \mathcal{A}\} \mid \mathcal{A} \in \mathcal{M} \big \} \\
&= \bigcup\big\{ \alpha_{\text{Mat}}(\mathcal{A} \mid \mathcal{A} \in \mathcal{M} \} \big( S \big).
\end{align*}
Entsprechend ist für $\mathcal{R} \subset L_R$ und $S \in L$
\begin{align*}
\lefteqn{\alpha_{\text{Rel}} \big(\bigcup \mathcal{R} \big) \big(S\big)} \\
&= \big\{ (1,y) \mid \exists x \in \rr^n: ((1,x) \in S \wedge (x,y) \in \bigcup\mathcal{R}) \big\} \\
&= \big\{ (1,y) \mid \exists x \in \rr^n \exists R \in \bigcup\mathcal{R}: ((1,x) \in S \wedge (x,y) \in R) \big\} \\
&= \big\{ \{(1,y) \mid \exists x \in \rr^n: ((1,x) \in S \wedge (x,y) \in R) \} \mid  R \in \bigcup\mathcal{R} \big\} \\
&= \bigcup\big\{ \alpha_{\text{Rel}}(\mathcal{R} \mid R \in \mathcal{R} \} \big( S \big).
\end{align*}
Also sind $\alpha_{\text{Mat}}$ und $\alpha_{\text{Rel}}$ Abstraktionen.
\end{bsp}

In \autoref{sec:konvexe-polyeder} werden wir sogenannte \quotes{konvexe Hüllen} bilden. 
Dazu definieren wir hier Hüllenoperatoren und zeigen, dass Hüllenoperatoren und Galois-Verbindungen sich gegenseitig induzieren.
\begin{dfn}[Hüllenoperator] 
Sei $(L,\sqleq)$ eine Halbordnung. Eine Abbildung $h: L \to L$ heißt \emph{Hüllenoperator}\index{Hüllenoperator}, falls gilt:
\begin{enumerate}
 \item Die Abbildung $h$ ist \emph{extensiv}, d.h.~für alle $x \in L$ ist $x \sqleq h(x)$.
 \item Die Abbildung $h$ ist \emph{monoton}, d.h.~für alle $x,y \in L$ ist $(x \sqleq y \Rightarrow h(x) \sqleq h(y))$.
 \item Die Abbildung $h$ ist \emph{idempotent}, d.h.~für alle $x \sqleq L$ ist $h(h(x)) = h(x)$.
\end{enumerate}
\end{dfn}

\begin{lemma}\label{lemma-korresp-huell-galois}
Seien $(L,\sqleq)$ und $(L^\sharp,\sqleq^\sharp)$ zwei Halbordnungen. Dann gilt:
\begin{enumerate} 
 \item Ist $(\alpha,\gamma)$ eine Galois-Verbindung zwischen $(L,\sqleq)$ und $(L^\sharp,\sqleq^\sharp)$, so ist die Abbildung $h := \gamma \circ \alpha : L \to L$ ein Hüllenoperator.
 \item Ist $h: L \to L$ ein Hüllenoperator, so ist $(h,\id)$ eine Galois-Verbindung zwischen $(L,\sqleq)$ und $(h(L),\sqleq)$.
\end{enumerate}
\end{lemma}
\begin{proof} Beide Aussagen lassen sich schnell nachrechnen:
\begin{enumerate}
 \item Setze $h := \gamma \circ \alpha$. Da $\alpha$ und $\gamma$ monoton sind, ist auch $h$ monoton. Weiter gilt $h = \gamma \circ \alpha \sqgeqmap \id_L$. Also ist $h$ auch extensiv. Damit folgt weiter 
 \[ h \circ h \sqgeqmap h \circ \id = h. \]
 Andererseits ist aber auch $\alpha \circ \gamma \sqleqsharpmap \id_{L^\sharp}$ und damit
 \[ h \circ h = \gamma \circ (\alpha \circ \gamma) \circ \alpha \sqleqsharpmap \gamma \circ \id_{L^\sharp} \circ \alpha = h.\]
 Zusammen folgt die Idempotenz von $h$. Somit ist $h$ ein Hüllenoperator.
 \item Seien $x \in L$ sowie $y \in h(L)$ beliebig. Wir zeigen nun die Äquivalenz 
 \[h(x) \sqleq y \Longleftrightarrow x \sqleq \id(y).\]
 Gelte zunächst $h(x) \sqleq y$. Da $h$ extensiv ist, folgt sofort \[x \sqleq h(x) \sqleq y.\]
 Sei nun $x \sqleq y$. Da $y \in h(L)$, existiert ein $x' \in L$ mit $y = h(x')$. Mit der Monotonie und Idempotenz von $h$ erhalten wir nun \[h(x) \sqleq h(y) = h(h(x')) = h(x') = y.\] 
 Also ist $(\alpha,\id)$ eine Galois-Verbindung. \qedhere
\end{enumerate}
\end{proof}

\section{Abstrakte Interpretation von Ungleichungssystemen}
In diesem Abschnitt betrachten wir abstrakte Interpretationen von Ungleichungssystemen. 
Wie zu Beginn des Kapitels erwähnt, soll dabei zu einem Ungleichungssystem $\ugs$ über $L$ und einer Abstraktion $\alpha : L \to L^\sharp$ 
ein Ungleichungssystem $\ugssharp$ über $L^\sharp$ aufgestellt werden. 
Dessen kleinste Lösung soll dem abstrahierten Wert der kleinsten Lösung von $\ugs$ entsprechen oder diesen zumindest approximieren. 

Formal erhält man ein solches abstrahiertes Ungleichungssystem durch eine \emph{abstrakte Interpretation}, 
die den verallgemeinerten Transferfunktionen Abbildungen auf dem {abstrakten} Raum $L^\sharp$ zuordnet. 
Dazu benötigen wir für eine Indexmenge $I$ zunächst einen Lift einer Abstraktion $\alpha: L \to L^\sharp$ auf den Raum $L^I \to L$, der die verallgemeinerten Transferfunktionen enthält.
\begin{lemma}
Sei $(\alpha,\gamma)$ eine Galois-Verbindung zwischen den voll\-stän\-di\-gen Ver\-bän\-den $(L,\sqleq)$ und $(L^\sharp,\sqleq^\sharp)$. 
Definiere $\hat{\alpha}: L^I \to (L^\sharp)^I$ und $\hat{\gamma}: (L^\sharp)^I \to L^I$ durch
\begin{align*}
 \hat{\alpha}((x_i)_{i \in I}) &:= (\alpha(x_i))_{i \in I} \\
 \hat{\gamma}((y_i)_{i \in I}) &:= (\gamma(y_i))_{i \in I}.
\end{align*}
Dann ist $(\hat{\alpha},\hat{\gamma})$ eine Galois-Verbindung zwischen $(L^I,\sqleq^I)$ und $((L^\sharp)^I, (\sqleq^\sharp)^I)$.
\end{lemma}
\begin{proof}
Seien $x \in L^I$ und $y \in (L^\sharp)^I$. Dann gilt
\begin{align*}
\hat{\alpha}(x) (\sqleq^\sharp)^I y
&\Longleftrightarrow \forall i \in I:\alpha(x_i) \sqleq^\sharp y_i \\
&\Longleftrightarrow \forall i \in I:x_i \sqleq \gamma(y_i) \\
&\Longleftrightarrow x \sqleq^I \hat{\gamma}(y).
\end{align*}
Es folgt die Behauptung.
\end{proof}

Wir können nun definieren, was wir unter einer abstrakten Interpretation eines Ungleichungssystems verstehen:
\begin{dfn}
Seien $I$ eine Indexmenge und $(L,\sqleq)$ und $(L^\sharp,\sqleq^\sharp)$ vollständige Verbände. 
Eine \emph{abstrakte Interpretation} bezüglich einer Abstraktion $\alpha : L \to L^\sharp$ ist eine Zuordnung $(L^I \to L) \to ((L^\sharp)^I \to L^\sharp)$, 
die jeder konkreten Abbildung $f:L^I \to L$ eine abstrakte Abbildung $f^\sharp:(L^\sharp)^I \to L^\sharp$ zuordnet. 
Falls 
\[\alpha \circ f(x) \sqleq f^\sharp \circ \hat{\alpha}(x)\] 
für alle Abbildungen $f$ und alle $x \in L^N$ gilt,
nennen wir die abstrakte Interpretation \emph{$\alpha$-korrekt}.
Falls sogar
\[\alpha \circ f = f^\sharp \circ \hat{\alpha}\] 
für alle Abbildungen $f$ gilt,
nennen wir die abstrakte Interpretation \emph{$\alpha$-präzise}.

Das \emph{abstrahierte Ungleichungssystem} zu einem Ungleichungssystem $\ugs$ bezüglich einer solchen abstrakten Interpretation ist definiert durch
\[\ugssharp := \{(x_j \sqgeq^\sharp f^\sharp(x)) \mid (x_j \sqgeq f(x)) \in \ugs \}.\]
\end{dfn}

\begin{lemma}\label{lemma-equiv-korr-abstr}
Die folgenden Aussagen sind äquivalent:
\begin{enumerate}
 \item Für alle $x \in L^I$ gilt 			$\alpha \circ f(x) 						\sqleq^\sharp 	f^\sharp \circ \hat{\alpha}(x)$.
 \item Für alle $x \in L^I$ gilt 			$f(x) 									\sqleq 			\gamma \circ f^\sharp \circ \hat{\alpha}(x)$.
 \item Für alle $y \in (L^\sharp)^I$ gilt 	$f \circ \hat{\gamma}(y) 				\sqleq 			\gamma \circ f^\sharp (y)$.
 \item Für alle $y \in (L^\sharp)^I$ gilt 	$\alpha \circ f \circ \hat{\gamma}(y)	\sqleq^\sharp 	f^\sharp (x)$.
\end{enumerate}
\end{lemma}
\begin{proof}
Die Äquivalenz der Aussagen a) und b) bzw.~c) und d) folgt jeweils sofort aus der Definition einer Galois-Verbindung. Es genügt also zu zeigen, dass die Aussagen b) und c) äquivalent sind.

Gelte zunächst b) und sei $y \in (L^\sharp)^I$ beliebig. Dann ist $\hat{\gamma}(y) \in L^I$ und nach b) gilt
\[f \circ \hat{\gamma} (y) \sqleq \gamma \circ f^\sharp \circ \hat{\alpha} \circ \hat{\gamma} (y).\]
Da $(\hat{\alpha}, \hat{\gamma})$ eine Galois-Verbindung ist, gilt nach Lemma \ref{lemma-eig-galois} $\hat{\alpha} \circ \hat{\gamma} (y) \sqleq^I y$. Damit folgt das Gewünschte.

Die andere Richtung zeigt man analog. Es folgt die Behauptung.
\end{proof}

Seien im Folgenden eine Galois-Verbindung $(\alpha,\gamma)$ zwischen vollständigen Ver\-bän\-den $(L,\sqleq)$ und $(L^\sharp, \sqleq^\sharp)$, 
eine Indexmenge $I$, ein Ungleichungssystem $\ugs$, eine abstrakte Interpretation und das zugehörige abstrahierte Ungleichungssystem $\ugssharp$ fixiert. 
Wir nennen dann auch $\ugs$ das \emph{konkrete} und $\ugssharp$ das \emph{abstrahierte} Ungleichungssystem. 
Mit $F$ und $F^\sharp$ bezeichnen wir diejenigen Abbildungen, deren Präfixpunkte mit den Lösungen der Ungleichungssysteme übereinstimmen.

Wie eingangs beschrieben, soll mit der abstrakten Interpretation eines Ungleichungssystems nach Möglichkeit der abstrahierter Wert der kleinsten Lösung des Ungleichungssystems berechnet 
oder zumindest approximiert werden. 
Der berechnete Wert $\lfp(F^\sharp)$ soll also mindestens soviel Information enthalten wie $\hat{\alpha}(\lfp(F))$. 
Dass dieser abstrahierte Wert tatsächlich immer approximiert wird, wenn die benutzte abstrakte Interpretation $\alpha$-korrekt ist, ist die Aussage des folgenden Satzes. 
Um den abstrahierten Wert tatsächlich zu erreichen, ist sogar $\alpha$-Präzision vonnöten, wie wir anschließend zeigen werden.
\begin{satz}\label{satz:abstrint-korr}
Sei die zugrundeliegende abstrakte Interpretation $\alpha$-korrekt.
Dann gilt \[\lfp(F^\sharp) \sqgeq^I \hat{\alpha}(\lfp(F)).\]
\end{satz}
\begin{proof}
Wir benutzen Satz \ref{relFixpkte}. 
Dabei betrachten wir die Relation \[R := \{ (x,y) \in L^I \times {L^\sharp}^I \mid \hat{\alpha}(x) {\sqleq^\sharp}^I y\}.\]
Wenn wir für die Abbildungen $F$ und $F^\sharp$ die Voraussetzungen von Satz \ref{relFixpkte} zeigen können, so folgt $(\lfp(F),\lfp(F^\sharp)) \in R$, was genau die gewünschte Aussage ist.

Sei nun $(x,y) \in R$. Um $(F(x),F^\sharp(y)) \in R$ zu zeigen, genügt es, $\alpha(F_i(x)) \sqleq^\sharp F^\sharp_i(y)$ für alle $i \in I$ nach zu weisen. 
Für $i \in I$ ist wegen der universellen Distributivität von $\alpha$
\begin{align*}
\alpha(F_i(x)) 
&= \alpha\big(\bigsqcup \big\{f(x) \mid f\in\mathcal{F}_i\big\} \big)\\
&= \bigsqcup\nolimits^\sharp\big\{\alpha \circ f(x) \mid f\in\mathcal{F}_i \big\} \\
&\sqleq \bigsqcup\nolimits^\sharp\big\{f^\sharp(\hat{\alpha}(x)) \mid f\in\mathcal{F}_i \big\} \\
&= F^\sharp_i(\hat{\alpha}(x)) \\
&\sqleq^\sharp F^\sharp_i(y),
\end{align*}
da $\hat{\alpha}(x) {\sqleq^\sharp}^I y$ nach Voraussetzung und $F^\sharp$ monoton ist. Damit ist Voraussetzung a) aus Satz \ref{relFixpkte} gezeigt.
 
Sei nun $X \subset R$, d.h.~für alle $x=(x^1,x^2) \in X$ und für alle $i \in I$ ist $\hat{\alpha}(x^1_i) \sqleq^\sharp x^2_i$. 
Um $(\bigsqcup X^1, \bigsqcup X^2)\in R$ zu beweisen,
zeigen wir $\hat{\alpha}(\bigsqcup X^1))_i \sqleq^\sharp\ (\bigsqcup^\sharp X^2)_i$ für alle $i \in I$:
\begin{align*}
\big(\hat{\alpha}\big(\bigsqcup X^1\big)\big)_i
&= \alpha\big(\big(\bigsqcup{X}^1\big)_i\big) \\
&= \alpha\big(\bigsqcup{X}^1_i\big) \\
&= \bigsqcup\nolimits^{\sharp}\alpha({X}^1_i) \\
&= \bigsqcup\nolimits^{\sharp}\big\{\alpha(x^1_i)\mid x \in X\big\}
\end{align*}
Für alle $x \in {X}$ haben wir nun $\hat{\alpha}(x^1_i) {\sqleq^\sharp}^N x^2_i$. Damit ist $\bigsqcup^\sharp\{x^2_i \mid x \in {X}\}$ obere Schranke von $\{\hat{\alpha}(x^1_i)\mid x \in {X}\}$. 
Wir erhalten
\[\big(\hat{\alpha}\big(\bigsqcup{X}^1\big)\big)_i
\sqleq^\sharp \bigsqcup\nolimits^{\sharp}\big\{x^2_i\mid X \in {X}\big\}
= \big(\bigsqcup\nolimits^{\sharp} {X}_2\big)_i\]
und auch Voraussetzung b) aus Satz \ref{relFixpkte} ist erfüllt. Es folgt die Behauptung.
\end{proof}

\begin{satz}\label{satz:abstrint-praez}
Sei die zugrundeliegende abstrakte Interpretation $\alpha$-präzise. 
Dann gilt \[\lfp(F^\sharp) = \hat{\alpha}(\lfp(F)).\]
\end{satz}
\begin{proof}
Nach Satz \ref{satz:abstrint-korr} wissen wir bereits, dass $\lfp(F^\sharp) \sqgeq^I \hat{\alpha}(\lfp(F))$ gilt. 
Um die umgekehrte Ungleichung zu zeigen, genügt es nach zu weisen, dass $\hat{\alpha}(\lfp(F))$ eine Lösung des Ungleichungssystems $\ugssharp$ ist. 
Sei dazu $x_j \sqgeq f(x)$ eine beliebige Ungleichung von $\ugs$. 
Da $\lfp(F)$ eine Lösung dieses Ungleichungssystems ist, gilt $\lfp(F)_j \sqgeq f(\lfp(F))$. Mit den Voraussetzungen an $f$ und $f^\sharp$ und der Monotonie von $\alpha$ folgt nun
\begin{align*}
f^\sharp(\hat{\alpha}(\lfp(F)))
= \alpha(f(\lfp(F)))
\sqleq \alpha((\lfp(F))_j)
= (\hat{\alpha}(\lfp(F)))_j.
\end{align*}
Also erfüllt $\hat{\alpha}(\lfp(F))$ die Ungleichung $x_j \sqgeq^\sharp f^\sharp(x)$. Es folgt die Behauptung.
\end{proof}
Wir sagen auch, dass $\ugssharp$ die kleinste Lösung von $\ugs$ \emph{präzise abstrakt} berechnet.

Später werden wir das folgende Lemma benötigen, das eine weitere Eigenschaft von solchen präzisen abstrakt interpretierten Abbildungen liefert.
\begin{lemma}\label{lemma:hilfe-abstr-T-praez}
Seien $f : L^I \to L$ und $f^\sharp: (L^\sharp)^I \to L^\sharp$ monoton mit $f^\sharp \circ \hat{\alpha} = \alpha \circ f$. Dann gilt
\[\alpha \circ f \circ \hat{\gamma} \circ \hat{\alpha} = \alpha \circ f.\]
\end{lemma}
\begin{proof}
Aus der Voraussetzung für $f$ und $f^\sharp$ und Lemma \ref{lemma:eig-galois-verknuepfungen} folgt unmittelbar
\begin{align*}
\alpha \circ f 
= f^\sharp \circ \hat{\alpha} 
= f^\sharp \circ \hat{\alpha} \circ \hat{\gamma} \circ \hat{\alpha} 
= \alpha \circ f  \circ \hat{\gamma} \circ \hat{\alpha} 
\end{align*}
und dies ist die Behauptung.
\end{proof}

\begin{bsp}
In Beispiel \ref{bsp:summary-inform-von-polyederana-sind-abstr} haben wir gesehen, dass die Abbildungen $\alpha_\text{Mat}$ und $\alpha_\text{Rel}$, 
Matrizenmengen oder Relationen als Transferfunktionen auffassen, Abstraktionen sind. 
Definiert man nun ein Ungleichungssystem $T$ der Transferfunktionen durch das Ungleichungssystem \eqref{T-Ugs} aus Abschnitt \ref{sec:funktionaler-ansatz}, 
so ist dies eine präzise abstrakte Interpretation:

Seien dazu zunächst $\mathcal{A}_1, \mathcal{A}_2 \in L_M$ und $S \in L$. Dann gilt 
\begin{align*}
\lefteqn{(\alpha_{\text{Mat}}(\mathcal{A}_2) \circ \alpha_{\text{Mat}}(\mathcal{A}_1)) (S)} \\
&= \alpha_{\text{Mat}}(\mathcal{A}_2) ( \{A \cdot s \mid s \in S \wedge A \in \mathcal{A}_1 \} ) \\
&= \alpha_{\text{Mat}}( \{ B \cdot A \cdot s \mid s \in S \wedge A \in \mathcal{A}_1 \wedge B \in \mathcal{A}_2 \} ) \\
&= \alpha_{\text{Mat}}( \{ C \cdot s \mid s \in S \wedge C \in \mathcal{A}_2 \circ \mathcal{A}_1 \} ) \\
&= \alpha_{\text{Mat}}(\mathcal{A}_2 \circ \mathcal{A}_1)(S).
\end{align*}
Entsprechend ist für $R_1, R_2 \in L_R$ und $S \in L$
\begin{align*}
\lefteqn{(\alpha_{\text{Rel}}(R_2) \circ \alpha_{\text{Rel}}(R_1)) (S)} \\
&= \alpha_{\text{Rel}}(R_2) ( \{(1,y) \mid \exists x \in \rr^n: ((1,x) \in S \wedge (x,y) \in R_1) \} ) \\
&= ( \{(1,z) \mid \exists x,y \in \rr^n: ((1,x) \in S \wedge (x,y) \in R_1 \wedge (y,z) \in R_2)\} ) \\
&= ( \{(1,z) \mid \exists x \in \rr^n:((1,x) \in S \wedge (x,z) \in R_2\circ R_1)\} ) \\
&= \alpha_{\text{Rel}}(R_2 \circ R_1)(S).
\end{align*}
Mit diesen Identitäten können wir nun die gewünschte Aussage leicht aus Satz \ref{satz:abstrint-praez} folgern.
Seien dafür zunächst $\mathcal{A} \in (L_M)^N$ und $S \in L$. Dann gilt
\begin{align*}
\alpha_{\text{Mat}} \circ f_{T_M,\id} (\mathcal{A}) (S)
&= \alpha_{\text{Mat}} ( \{E_{n+1}\} ) (S)\\
&= \{E_{n+1}\cdot s \mid s \in S\}\\
&= S
\end{align*}
und andererseits auch
\begin{align*}
f_{T,\id} \circ \hat{\alpha}_{\text{Mat}} (\mathcal{A}) (S) 
&= \id (S) \\
& = S.
\end{align*}
Also ist $\alpha_M \circ f_{T_M,\id} = f_{T,\id} \circ \hat{\alpha}_M$. 
Weiter ist mit der Rechnung zu Beginn dieses Beispiels
\begin{align*}
\alpha_{\text{Mat}} \circ f_{T_M,(u,\mathtt{x_j := t},v)} (\mathcal{A})
&= \alpha_{\text{Mat}} ( \{ M_{\mathtt{x_j:=t}} \} \circ \mathcal{A}_u ) \\
&= \alpha_{\text{Mat}} ( \{ M_{\mathtt{x_j:=t}} \}) \circ \alpha_{\text{Mat}}(\mathcal{A}_u ) \\
&= f_{\mathtt{x_j:=t}} \circ \alpha_{\text{Mat}}(\mathcal{A}_u ) \\
&= f_{T,(u,\mathtt{x_j:=t},v)} \circ \hat{\alpha}_{\text{Mat}}(\mathcal{A})
\end{align*}
und damit $\alpha_{\text{Mat}} \circ f_{T_M,(u,\mathtt{x_j := t},v)} = f_{T,(u,\mathtt{x_j:=t},v)} \circ \hat{\alpha}_{\text{Mat}}$. 
Ebenso wie für Zuweisungen folgt für Prozeduraufrufe 
$\alpha_{\text{Mat}} \circ f_{T_M,(u,\mathtt{q()},v)} = f_{T,(u,\mathtt{q()},v)} \circ \hat{\alpha}_{\text{Mat}}$ aus
\begin{align*}
\alpha_{\text{Mat}} \circ f_{T_M,(u,\mathtt{q()},v)} (\mathcal{A})
&= \alpha_{\text{Mat}} ( \mathcal{A}_{r_\q} \circ \mathcal{A}_u ) \\
&= \alpha_{\text{Mat}} ( \mathcal{A}_{r_\q}) \circ \alpha_M(\mathcal{A}_u ) \\
&= f_{T,(u,\mathtt{q()},v)} \circ \hat{\alpha}_{\text{Mat}}(\mathcal{A}).
\end{align*}
Nach Satz \ref{satz:abstrint-praez} ist also $\alpha_{\text{Mat}}(\underline{T_M}[u]) = \underline{T}[u]$ für jedes $u \in N$. 
Anstatt das Ungleichungssystem $T$ in \eqref{T-Ugs} direkt zu lösen, 
kann also stattdessen das Ungleichungssystem $T_M$ in \eqref{TM-Ugs} gelöst und dessen kleinste Lösung durch $\alpha_M$ abstrahiert werden, wie es in Abschnitt \ref{sec:polyeder} auch geschehen ist.

Genauso zeigt man $\alpha_{\text{Rel}}(\underline{T_R}[u]) = \underline{T}[u]$ für jedes $u \in N$.
\end{bsp}

Für jede Galois-Verbindung existiert eine $\alpha$-korrekte abstrakte Interpretation: 
Die folgende Definition zeigt, dass sich zu Abbildungen über $L$ in kanonischer Weise Abbildungen auf $L^\sharp$ angeben lassen. 
Der anschließende Satz zeigt, dass diese Zuordnung eine $\alpha$-korrekte abstrakte Interpretation liefert.
\begin{dfn}
Wir nennen die Zuordnung $(L^I \to L) \to ((L^\sharp)^I \to L^\sharp)$ , die einem $f: L^I \to L$ die Abbildung \[f^+:=\alpha \circ f \circ \gamma\] zuordnet, 
\emph{$\alpha$-kanonische abstrakte Interpretation}.
\end{dfn}
Das dadurch abstrahierte Ungleichungssystem $\ugs^+$ nennen wir entsprechend das \emph{$\alpha$-kanonisch abstrahierte Ungleichungssystem}.

\begin{satz}\label{satz:kanonisch-ist-am-besten}
Sei eine beliebige abstrakte Interpretation gegeben, die jedem $f$ ein $f^\sharp$ zuordnet.
\begin{enumerate}
\item Die abstrakte Interpretation ist genau dann $\alpha$-korrekt, wenn $f^+(y) \sqleq^\sharp f^\sharp(y)$ für alle $f$ und alle $y \in (L^\sharp)^I$ gilt.
Insbesondere ist die kanonische abstrakte Interpretation $\alpha$-korrekt und liefert von allen abstrakten Interpretation das kleinste Ergebnis.
\item Falls es eine $\alpha$-präzise abstrakte Interpretation gibt, ist die kanonische abstrakte Interpretation ebenfalls $\alpha$-präzise.
\end{enumerate}
\end{satz}
\begin{proof}
Die gewünschte Äquivalenzaussage folgt direkt aus Lemma \ref{lemma-equiv-korr-abstr}. 
Sie impliziert sofort, dass die kanonische abstrakte Interpretation korrekt ist.
Nach Konstruktion ist nun für alle $y \in (L^\sharp)^I$
\begin{align*}
F^\sharp_i (y) 
& = \bigsqcup\nolimits^\sharp \big\{ f^\sharp (y) \mid f \in \mathcal{F}_i \big\} \\
& \sqgeq^\sharp \bigsqcup\nolimits^\sharp \big\{ f^+ (y) \mid f \in \mathcal{F}_i \big\} \\
&= F^+_i(y).
\end{align*}
Damit ist aber auch $F^+ \overline{(\sqleq^\sharp)^I} F^\sharp$ und so $\lfp(F^+) (\sqleq^\sharp)^I \lfp(F^\sharp)$. Dies beweist den ersten Teil.

Gebe es nun eine $\alpha$-präzise abstrakte Interpretation. Nach Lemma \ref{lemma:hilfe-abstr-T-praez} ist dann $f^+ \circ \alpha = \alpha \circ f \circ \gamma \circ \alpha = \alpha \circ f$, 
also liefert auch $f^+$ eine $\alpha$-präzise abstrakte Interpretation. Es folgt die Behauptung.
\end{proof}

Eine $\alpha$-präzise abstrakte Interpretation zu einer gegeben Galois-Verbindung $(\alpha,\gamma)$ und einem gegebenen Ungleichungssystem muss nicht existieren. 
Dies folgt aus vorigem Satz \ref{satz:kanonisch-ist-am-besten}, wenn man zeigen kann, dass eine $\alpha$-kanonische abstrakte Interpretation nicht die gewünschte präzise Lösung liefert. 
Ein Beispiel dafür ist die folgende abstrakte Interpretation.
\begin{bsp}\label{bsp:kan-nicht-praez}
Wir betrachten die vollständigen Ver\-bän\-de, die durch die Hasse-Dia\-gram\-me in Abbildung \ref{bild:kan-nicht-praez-verband} gegeben sind.
\begin{figure}[ht]\begin{center}\begin{tikzpicture}[node distance = 0.25cm]
\node 	(L)						{$L=$};
\node	(invL)	[right=of L] 	{};
\node	(top)	[right=of invL]	{$\top$};
\node	(c) 	[below =of top]	{$c$};
\node	(b) 	[below =of c]	{$b$};
\node	(a) 	[below =of b]	{$a$};
\node	(bot)	[below =of a]	{$\bot$};
\node	(inv')	[right=of top] 	{};
\node	(inv'')	[right=of inv'] {};
\node	(inv''')[right=of inv'']{};
\node 	(L')	[right=of inv''']{$L^\sharp=$};
\node	(top') 	[right=of L']	{$\top^\sharp$};

\node	(e) 	[below =of top']{$e$};
\node	(d) 	[below =of e]	{$d$};
\node	(bot') 	[below =of d]	{$\bot^\sharp$};
\path[-]
  (top)	edge node {} (c)
  (c)	edge node {} (b)
  (b)	edge node {} (a)
  (a)	edge node {} (bot)
  (top')edge node {} (e)
  (e)	edge node {} (d) 
  (d)	edge node {} (bot')
;
\end{tikzpicture}
\caption{Hasse-Diagramm des zugrundeliegenden Verbandes in Beispiel \ref{bsp:kan-nicht-praez}.}
\label{bild:kan-nicht-praez-verband} 
\end{center}\end{figure}
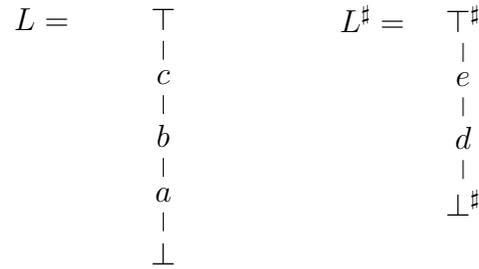

Wir definieren nun eine Galois-Verbindung $(\alpha,\gamma)$ zwischen diesen vollständigen Ver\-bän\-den durch
\begin{align*}
\alpha(x) := 
  \begin{cases}
  \bot^\sharp 	&\text{falls } x=\bot \\
  d 			&\text{falls } x= a\\
  e 			&\text{falls } x\in\{b,c\} \\
  \top^\sharp 	&\text{falls } x= \top
  \end{cases}
&&\text{und}&&
\gamma(y) := 
  \begin{cases}
  \bot	&\text{falls } y=\bot^\sharp \\
  a 	&\text{falls } y= d \\
  c 	&\text{falls } y= e\\
  \top	&\text{falls } y= \top^\sharp.
  \end{cases}
\end{align*}
Man sieht leicht
\begin{align*}
\alpha\circ\gamma&=\id_{L^\sharp}, \\
\gamma\circ\alpha_{|\{\bot,a,c,\top\}} &= \id_{|\{\bot,a,c,\top\}} \intertext{und}
\gamma\circ\alpha(b)&=c\sqsupset b.
\end{align*}
Also ist $\gamma \circ \alpha \sqgeqmap \id_L$ und nach Lemma \ref{lemma-eig-galois} ist $(\alpha,\gamma)$ eine Galois-Verbindung.

Wie betrachten nun ein Ungleichungssystem, das nur eine Variable $x$ benutzt. Die Abbildungen, die das Ungleichungssystem definieren, sind also auf $(L \to L)$ definiert.

Wir definieren eine Abbildung $f : L \to L$ durch
\begin{align*}
f(x) :=& 
  \begin{cases}
  \bot  		&\text{falls } x=\bot \\
  b 			&\text{falls } x\in\{a,b\}\\
  \top 			&\text{falls } x\in\{c,\top\}. 
  \end{cases}
\intertext{Offenbar ist $f$ universell-distributiv. Mit $\mathcal{F}$ bezeichnen wir das kleinste unter Komposition abschlossene System von Abbildungen, das $\id_L$ und $f$enthält. 
Da die Komposition universell-distributiver Abbildungen wieder universell-distributiv ist, ist $(L,\sqleq,\mathcal{F})$ ein universell-distributives monotones Framework. 
Die zugehörige optimale Abstraktion $f^+ = \alpha \circ f \circ \gamma$ erfüllt}
f^+(y) =& 
  \begin{cases}
  \bot^\sharp	&\text{falls } y=\bot^\sharp \\
  e 			&\text{falls } y=d\\
  \top^\sharp	&\text{falls } y\in\{e,\top^\sharp\}. 
  \end{cases}
\end{align*}
Wir betrachten nun ein Ungleichungssystem $\ugs$, das aus den folgenden beiden Ungleichungen besteht:
\begin{alignat*}{3}
 x &\sqgeq a 		&\text{ und } 	x &\sqgeq f(x).
\intertext{Dieses Ungleichungssystem wird durch die Abbildung $F:L \to L$ beschrieben, die durch $F(x) := a \sqcup f(x)$ definiert ist.
Das zugehörige abstrahierte Ungleichungssystem ist definiert durch}
 y &\sqgeq^\sharp d	&\text{ und } 	x &\sqgeq^\sharp f^\sharp(x).
\end{alignat*}
und wird beschrieben durch die Abbildung $F^\sharp: L^\sharp \to L^\sharp$ mit $F^\sharp(y) = d \sqcup^\sharp f^\sharp(y)$.

Man rechnet nun leicht für die kleinsten Fixpunkte
\begin{align*}
 \lfp(F) = b	&\text{ und }	\lfp(F^\sharp) = \top^\sharp
\end{align*}
aus. Es ist aber $\alpha(b) = e \ne \top^\sharp$. Also ist $\lfp(F^\sharp) \ne \alpha(\lfp(F))$. 
Die abstrakte Interpretation ist somit zwar optimal, aber nicht präzise, womit gezeigt ist, dass eine $\alpha$-präzise abstrakte Interpretation nicht existieren muss.
\end{bsp}
Wenn die $\alpha$-kanonische abstrakte Interpretation also nicht das gewünschte Ergebnis liefert, 
so liefert sie doch nach Satz \ref{satz:kanonisch-ist-am-besten} zumindest von allen möglichen abstrakten Interpretationen die präziseste Lösung.

\section{Abstrakte Interpretation für interprozedurale Datenflussanalysen}\label{sec:abstrint-fuer-PA}
Im letzten Abschnitt haben wir abstrakte Interpretationen von Ungleichungssystemen betrachtet, indem wir jede der verwendeten Abbildungen durch eine {abstrahierte} Abbildung ersetzt haben. 
Wir werden in diesem Abschnitt nun die Ungleichungssysteme abstrahieren, die wir in \autoref{chap:PA} zur Analyse interprozeduraler Programme eingeführt haben, 
indem wir in den verallgemeinerten Transferfunktionen jede benutzte Transferfunktion durch eine abstrahierte Abbildung ersetzen. 

Um die so abstrahierten Ungleichungssysteme untersuchen zu können, müssen wir eine zugrundeliegende Galois-Verbindung auf die Räume fortsetzen, über denen die Ungleichungssystem definiert sind:
Zunächst benötigen wir eine Fortsetzung der Abstraktion auf $(L \to L) \to (L^\sharp \to L^\sharp)$, um auch diejenigen Abbildungen, mit denen Transferfunktionen berechnet werden, abstrahieren zu können.
\begin{dfn}
Sei $(\alpha, \gamma)$ eine Galois-Verbindung zwischen den voll\-stän\-di\-gen Ver\-bän\-den $(L,\sqleq)$ und $(L^\sharp, \sqleq^\sharp)$. Wir definieren nun
\begin{align*}
   \alphamap : (L \to_\text{mon} L) &\to (L^\sharp \to_\text{mon} L^\sharp)
  & \text{und}& 
  &\gammamap : (L^\sharp \to_\text{mon} L^\sharp) &\to (L \to_\text{mon} L)
\\ f &\mapsto \alpha \circ f \circ \gamma 
  & &
  &g &\mapsto \gamma \circ g \circ \alpha
\end{align*}
als Fortsetzungen von $\alpha$ und $\gamma$ auf Abbildungen. Diese werden in Abbildung \ref{bild-Forts-Abstr-Transfer} illustriert.
\begin{figure}[ht]
  \begin{center} 
  $ \xymatrix{
	  L \ar[r]^{f} & L \ar[d]^{\alpha} 
  &&	L \ar[r]^{\overline{\alpha}(g)} \ar[d]_{\alpha} & L
  \\	L^\sharp \ar[r]^{\overline{\alpha}(f)} \ar[u]^{\gamma} & L^\sharp
  &&	L^\sharp \ar[r]^{g} & L^\sharp \ar[u]_{\gamma} 
  }$
  \end{center}
\caption{Fortsetzung von Abstraktionen für Transferfunktionen.}
\label{bild-Forts-Abstr-Transfer}
\end{figure}
\end{dfn}
Das folgende Lemma zeigt, dass wir dadurch tatsächlich wieder eine Abstraktion erhalten.
\begin{lemma}
Die Fortsetzung $(\alphamap,\gammamap)$ ist eine Galois-Verbindung zwischen den voll\-stän\-di\-gen Verbänden $((L\to_\text{mon} L),\sqleqmap)$ und $((L^\sharp\to_\text{mon} L^\sharp),\sqleqsharpmap)$. 
\end{lemma}
\begin{proof}
Seien dazu $f: L \to L$ und $g: L^\sharp \to L^\sharp$ beliebige monotone Abbildungen.

Gelte zunächst $\alphamap(f) \sqleqsharpmap g$. Zu zeigen ist $f \sqleqmap \gammamap(g)$. Sei dazu $x \in L$ beliebig. Offenbar gilt
\begin{align*}
\alpha(x) &\sqleq^\sharp \alpha(x).
\intertext{Da $(\alpha,\gamma)$ eine Galois-Verbindung ist, ist dies äquivalent zu}
x &\sqleq \gamma (\alpha (x)).
\intertext{Mit der Monotonie von $f$ gilt damit weiter}
f(x) &\sqleq f(\gamma (\alpha (x))).
\intertext{Nach Voraussetzung gilt für jedes $y \in L^\sharp$}
\alpha \circ f \circ \gamma (y) &= \alphamap(f)(y) \sqleq^\sharp g(y).
\intertext{Da $(\alpha,\gamma)$ eine Galois-Verbindung ist, ist dies äquivalent zu}
f (\gamma (y)) &\sqleq \gamma (g(y)).
\intertext{Für $y := \alpha(x)$ gilt also}
f(x) &\sqleq f(\gamma (\alpha (x))) \\
& \sqleq\gamma(g(\alpha(x))) \\
&= \gammamap(y).
\end{align*}
Dies impliziert $f \sqleqmap \gammamap(g)$. Die umgekehrte Richtung zeigt man analog. Also ist $(\alphamap,\gammamap)$ eine Galois-Verbindung.
\end{proof}

Sei im Folgenden $G = (N_\p,E_\p,s_\p,r_\p)_{\p \in \Proc}$ ein Flussgraphsystem, das ein Programm mit einer endlichen Menge $\Proc$ von Prozeduren beschreibt. 
Sei $(L,\sqleq,\mathcal{F})$ das zugrundeliegende monotone Framework, $(L^\sharp, \sqleq^\sharp)$ ein vollständiger Verband 
und $(\alpha,\gamma)$ eine Galois-Verbindung zwischen diesen vollständigen Verbänden. 
Sei außerdem eine abstrakte Interpretation $\mathcal{F} \to (L^\sharp \to L^\sharp)$ fest ausgewählt, die einer Abbildung $f$ eine Abbildung $f^\sharp$ zuordnet.

Wir definieren nun eine abstrakte Interpretation eines Ungleichungssystems zur Analyse des Programms, 
indem wir in den verallgemeinerten Transferfunktionen $f_{T,\cdot}$, $f_{R,\cdot}$ und $f_{A,\cdot}$ die Abbildungen $f_b$, die für die Basiskanten benutzt werden, abstrakt interpretieren. 
Die dadurch entstehenden Abbildungen nennen wir $f_{T,\cdot}^\sharp, f_{R,\cdot}^\sharp$ bzw.~$f_{A,\cdot}^\sharp$. 

Die Abbildungen $f_{T,\cdot}^\sharp : (L^\sharp \to L^\sharp)^N \to (L^\sharp \to L^\sharp)$ sind gegeben durch
\begin{align*}
 f_{T,\id}^\sharp (h) &:= \id_{L^\sharp}, \\
 f_{T,(u,\mathtt{b},v)}^\sharp (h) &:= f_\mathtt{b}^\sharp \circ h_u, \\
 f_{T,(u,\mathtt{q()},v)}^\sharp (h) &:= h_{r_\q} \circ h_u
\end{align*}
für $h=(h_\node)_{\node \in N} \in (L^\sharp \to L^\sharp)^N$, Basiskanten $(u,\mathtt{b},v)$ und Aufrufkanten $(u,\mathtt{q()},v)$.

Die Abbildungen $f_{R,\cdot}^\sharp : (L^\sharp)^N \to L^\sharp$ sind definiert durch
\begin{align*}
 f_{R,\init}^\sharp (y) &:= \init^\sharp, \\
 f_{R,(u,\mathtt{b},v)}^\sharp (y) &:= f_\mathtt{b}^\sharp (y_u), \\
 f_{R,(u,\mathtt{q()},v),\text{ent}}^\sharp (y) &:= y_u, \\
 f_{R,(u,\mathtt{q()},v),\text{ret}}^\sharp (y) &:= \underline{T}^\sharp[r_\q] (y_{u})
\end{align*}
für $y=(y_\node)_{\node \in N} \in (L^\sharp)^N$, Basiskanten $(u,\mathtt{b},v)$ und Aufrufkanten $(u,\mathtt{q()},v)$.

Zuletzt sind die Abbildungen $f_{A,\cdot}^\sharp : (L^\sharp)^{N_{\CS}} \to L^\sharp$ 
definiert durch
\begin{align*}
f_{A,\init}^\sharp(y) &:= \init^\sharp \\
f_{A,(u,\mathtt{b},v),w}^\sharp(y) &:= f_{\mathtt{b}}^\sharp (y_{(u,w)})\\
f_{A,(u,\mathtt{q()},v),w,\text{ent}}^\sharp (y)&:= y_{(u,w)} \\
f_{A,(u,\mathtt{q()},v),w\cdot(u,\mathtt{q()},v),\text{ret}}^\sharp(y) &:= y_{(r_\q,w\cdot(u,\mathtt{q()},v))}
\end{align*}
wobei wiederum $(u,\mathtt{b},v)$ Basiskanten und $(u,\mathtt{q()},v)$ Aufrufkanten beschreiben.

Um die Lösungen der Ungleichungssysteme besser unterscheiden zu können, schreiben wir im abstrahierten Ungleichungssystem $T^\sharp[u], R^\sharp[u]$ und $A^\sharp[u,w]$ 
anstelle von $T[u], R[u]$ bzw.~$A[u,w]$. 
Damit wir auch wieder über die kleinsten Lösungen der abstrahierten Ungleichungssysteme reden können, fordern wir weiter für die abstrakte Interpretation, 
dass jeder monotonen Abbildung $f \in \mathcal{F}$ eine monotone Abbildung $f^\sharp$ zugeordnert wird.

Wir untersuchen in den folgendenden Abschnitten, ob eine $\alpha$-korrekte, $\alpha$-präzise oder $\alpha$-kanonische abstrakte Interpretation der Transferfunktionen 
auch schon eine entsprechende abstrakte Interpretation der verallgemeinertern Transferfunktionen induziert. 
Dazu müssen wir die Abstraktion $\alpha$ auf die entsprechenden Verbände fortsetzen, die bei den verschiedenen Ungleichungssystemen benutzt werden.

\subsection{Korrekte abstrakte Interpretation}\label{subsec:korrekt-abstrint}
Zu Beginn dieses Abschnittes haben wir eingeführt, wie wir bezüglich einer vorgegebenen abstrakten Interpretation der Transferfunktionen 
auch die verallgemeinerten Transferfunktionen des funktionalen und des Call-String-Ansatzes abstrakt interpretieren wollen. 
In diesem Abschnitt wollen wir untersuchen, ob eine so induzierte abstrakte Interpretation der verallgemeinerten Transferfunktionen korrekt ist, 
wenn die benutzten Transferfunktionen korrekt abstrakt interpretiert werden.

Zunächst betrachten wir die konkreten und die zugehörigen abstrahierten Ungleichungssysteme aus dem Ansatz zur Analyse interprozeduraler Programme, der Summary-Informationen benutzt. 
Die konkreten Systeme sind durch die Ungleichungssysteme \eqref{T-Ugs} und \eqref{R-Ugs} für den funktionalen Ansatz und durch das Ungleichungssystem \eqref{A-Ugs} für den Call-String-Ansatz gegeben. 
Um ein abstrahiertes System zu erhalten, ersetzen wir also, wie zuvor beschrieben, jedes $f_\mathtt{b}$ durch eine abstrahierte Abbildung $f_\mathtt{b}^\sharp$. 
Speziell ändern wir die Initialwerte $\id_L$ zu $\id_{L^\sharp}$ und $\init$ zu $\init^\sharp$. Dies induziert korrekt abstrahierte Ungleichungssysteme, wie der folgende Satz zeigt.
\begin{satz}\label{satz:abstr-T-korr}
Sei die abstrakte Interpretation der Transferfunktionen $\alpha$-korrekt. 
Dann ist die dadurch induzierte abstrakte Interpretation der verallgemeinerten Transferfunktionen $f_{T,\cdot}$ auch $\alphamap$-korrekt.
\end{satz}
\begin{proof}
Wir beweisen die Aussage, in dem wir die verschiedenen Abbildungen $f_{T,\cdot}$ der Reihe nach untersuchen. Sei dazu $h = (h_\node)_{\node \in N} \in \mathcal{F}^N$ beliebig. 

Nach Definition gilt
\[(\alphamap \circ f_{T,\cdot}) (h) 
= \alphamap (f_{T,\cdot}(h))
= \alpha \circ f_{T,\cdot}(h) \circ \gamma \]
und
\[(f_{T,\cdot}^\sharp \circ \hat{\alphamap}) (h)
= f_{T,\cdot}^\sharp (\hat{\alphamap} (h)) 
= f_{T,\cdot}^\sharp ( (\alphamap(h_\node))_{\node \in N})
= f_{T,\cdot}^\sharp ( (\alpha \circ h_\node \circ \gamma))_{\node \in N})
.\]
Es ist also
\[\alpha \circ f_{T,\cdot}(h) \circ \gamma \sqleqsharpmap f_{T,\cdot}^\sharp ( (\alpha \circ h_\node \circ \gamma))_{\node \in N}) \]
zu überprüfen.
\begin{enumerate}\renewcommand{\labelenumi}{\roman{enumi})}
\item 
Zunächst einmal ist 
\begin{align*}
(\alpha \circ f_{T,\id})(h) \circ \gamma
= \alpha \circ \id_L \circ \gamma
= \alpha \circ \gamma.
\end{align*}
Außerdem ist 
\begin{align*} 
f_{T,\id}^\sharp ((\alpha \circ h_\node \circ \gamma)_{\node \in N})
= \id_{L^\sharp}.
\end{align*}
Da $(\alpha,\gamma)$ eine Galois-Verbindung ist, ist $\alpha \circ \gamma \sqleqsharpmap \id_{L^\sharp}$. Es gilt also stets
die gewünschte Ungleichung. 
\item 
Weiter ist für eine Basiskante $(u,\mathtt{b},v)\in E$
\begin{align*}
\alpha \circ f_{T,(u,\mathtt{b},v)}(h) \circ \gamma
= \alpha \circ f_\mathtt{b} \circ h_u \circ \gamma 
\end{align*}
und 
\begin{align*} 
f_{T,(u,\mathtt{b},v)}^\sharp ((\alpha\circ h_\node \circ \gamma)_{\node \in N}))
= f_\mathtt{b}^\sharp \circ \alpha \circ h_u \circ \gamma.
\end{align*}
Dann gilt nach Voraussetzung also bereits 
$\alpha \circ f_\mathtt{b} \sqleqsharpmap f_\mathtt{b}^\sharp \circ \alpha$
und damit ist
die gewünschte Ungleichung erfüllt.
\item 
Zuletzt sei $(u,\mathtt{q()},v) \in E$ eine Aufrufkante. Für diesen ist
\begin{align*}
\alpha \circ f_{T,(u,\mathtt{q()},v)}(h) \circ \gamma
= \alpha \circ h_{r_\q} \circ h_u \circ \gamma
\end{align*}
und 
\begin{align*} 
f_{T,(u,\mathtt{q()},v)}^\sharp ((\alpha \circ h_\node \circ \gamma)_{\node \in N})
&= (\alpha \circ h_{r_\q} \circ \gamma) \circ (\alpha \circ h_u \circ \gamma)\\
&= \alpha \circ h_{r_\q} \circ (\gamma \circ \alpha) \circ h_u \circ \gamma
\end{align*}
Da stets $\id_L \sqleqmap \gamma \circ \alpha$ gilt, ist
die gewünschte Ungleichung immer erfüllt.
\end{enumerate}\renewcommand{\labelenumi}{\alph{enumi})}
Es folgt die Behauptung.
\end{proof}
Insbesondere gilt also für jeden Knoten $u \in N$
\begin{align*}
\underline{T}^\sharp[u] \sqgeqsharpmap \alphamap(\underline{T}[u]).
\end{align*}

Wir betrachten nun das Ungleichungssystem zur Berechnung der Informationen an den Programmpunkten. Als Initialwert des abstrahierten Ungleichungssystem benutzen wir $\init^\sharp := \alpha(\init)$.
\begin{satz}\label{satz:abstr-R-korr}
Sei die abstrakte Interpretation der Transferfunktionen $\alpha$-korrekt. 
Dann ist die dadurch induzierte abstrakte Interpretation der verallgemeinerten Transferfunktionen $f_{R,\cdot}$ auch $\alpha$-korrekt.
\end{satz}
\begin{proof}
Sei $x = (x_\node)_{\node \in N} \in L^I$. Wir überprüfen die Aussage wieder der Reihe nach für alle verwendeten Abbildungen.
\begin{enumerate}\renewcommand{\labelenumi}{\roman{enumi})}
\item
Zunächst ist \[(\alpha \circ f_{R,\init})(x) = \alpha(\init)\] und \[(f_{R,\init}^\sharp \circ \hat{\alpha}) (x) = \init^\sharp.\]
Wegen $\init^\sharp = \alpha(\init)$ ist also sogar Gleichheit erfüllt. 
\item
Für Basiskanten $(u,\mathtt{b},v) \in E$ ist nun
\[(\alpha \circ f_{R,(u,\mathtt{b},v)})(x) = \alpha(f_\mathtt{b}(x_u))\]
 und 
\[(f_{R,(u,\mathtt{b},v)}^\sharp \circ \hat{\alpha}) (x) = f_\mathtt{b}^\sharp(\alpha(x_u)).\]
Nach Voraussetzung gilt also 
die gewünschte Ungleichung.
\item Sei nun $(u,\mathtt{q()},v)$ eine Aufrufkante. Dann ist einerseits
\[(\alpha \circ f_{R,(u,\mathtt{q()},v),\text{ent}})(x) 
= \alpha(x_u)
\] 
und 
\[(f_{R,(u,\mathtt{q()},v),\text{ent}}^\sharp \circ \hat{\alpha}) (x) 
= \alpha(x_u)\]
und es gilt immer Gleichheit.
Andererseits ist 
\[(\alpha \circ f_{R,(u,\mathtt{q()},v),\text{ret}})(x) 
= \alpha(\underline{T}[r_\q](x_u))
\] 
und 
\[(f_{R,(u,\mathtt{q()},v),\text{ret}}^\sharp \circ \hat{\alpha})(x) 
= \underline{T^\sharp}[r_\q](\alpha(x_u)).\]
Wie wir in Satz \ref{satz:abstr-T-korr} gesehen haben, ist die abstrakte Interpretation von $T$ stets korrekt. Somit erhalten wir auch hier direkt
die gewünschte Ungleichung.
\end{enumerate}\renewcommand{\labelenumi}{\alph{enumi})}
Es folgt die Behauptung. 
\end{proof}
Für jeden Knoten $u \in N$ gilt also
\begin{align*}
\underline{R}^\sharp[u] \sqgeq \alpha(\underline{R}[u]).
\end{align*}

Wir betrachten nun das Ungleichungssystem \eqref{A-Ugs}, das eine Analyse mithilfe von Call-Strings beschreibt.
Wie im funktionalen Ansatz fordern wir als Startinformation des abstrakt interpretierten Ungleichungssystems $\init^\sharp = \alpha(\init)$.
\begin{satz}\label{satz:abstr-A-korr}
Sei die abstrakte Interpretation der Transferfunktionen $\alpha$-korrekt. 
Dann ist die dadurch induzierte abstrakte Interpretation der verallgemeinerten Transferfunktionen $f_{A,\cdot}$ auch $\alphacs$-korrekt.
\end{satz}
\begin{proof}
Wie zuvor betrachten wir wieder die verschiedenen Abbildungen. Sei dazu $x \in L^{N_{\CS}}$.
\begin{enumerate}\renewcommand{\labelenumi}{\roman{enumi})}
\item Zunächst ist wiederum
\begin{align*}
(\alpha \circ f_{A,\init})(x) 
= \alpha(\init)
\end{align*}
und
\begin{align*}
(f_{A,\cdot}^\sharp \circ \widehat{\alpha})(x)
= \init^\sharp
\end{align*}
Da $\init^\sharp = \alpha(\init)$, gilt hier wiederum sogar Gleichheit.

\item Weiter ist für Basiskanten $(u,\mathtt{b},v) \in E$
\begin{align*}
({\alpha} \circ f_{A,(u,\mathtt{b},v),w})(x) 
= \alpha(f_{\mathtt{b}}(x_{u,w}))
\end{align*}
und 
\begin{align*}
(f_{A,(u,\mathtt{b},v),w}^\sharp \circ \widehat{{\alpha}}) (x)
= f_{\mathtt{b}}^\sharp (\alpha(x_{u,w})).
\end{align*}
Nach Voraussetzung gilt also die gewünschte Ungleichung.

\item Wir betrachten nun eine Aufrufkante $e=(u,\mathtt{q()},v) \in E$. Dann ist
\begin{align*}
({\alpha} \circ f_{A,(u,\mathtt{q()},v),w,\text{ent}})(y)
= \alpha(y_{u,w})
\end{align*}
und
\begin{align*}
(f_{A,(u,\mathtt{q()},v),w,\text{ent}}^\sharp \circ \widehat{{\alpha}}) (y)
= \alpha(y_{u,w})
\end{align*}
und es gilt sogar Gleichheit. Für die Abbildung, die die Rückkehr aus einer Prozedur beschreibt, ist
\begin{align*}
({\alpha} \circ f_{A,(u,\mathtt{q()},v),w\cdot e,\text{ret}})(y) 
= \alpha(h_{r_\q,w \cdot e})
\end{align*}
und ebenso
\begin{align*}
(f_{A,(u,\mathtt{q()},v),w\cdot e,\text{ret}}^\sharp \circ \widehat{{\alpha}}) (y)
= \alpha(h_{r_\q,w \cdot e}).
\end{align*}
Auch hier gilt also stets Gleichheit.
\end{enumerate}\renewcommand{\labelenumi}{\alph{enumi})}
Es folgt die Behauptung.
\end{proof}
Für jeden Knoten $u \in N$ und Call-String $w \in \CS$ erhalten wir also 
\begin{align*}
\underline{A}^\sharp[u,w]
\sqgeq^\sharp \alpha(\underline{A}[u,w]).
\end{align*}
Damit ist
\begin{align*}
\hat{A}^\sharp[u]
&= \bigsqcup\nolimits^\sharp \big\{ \underline{A}^\sharp[u,w] \mid w \in \underline{\CS}[\p] \big\} \\
&\sqgeq^\sharp \bigsqcup\nolimits^\sharp \big\{ \alpha(\underline{A}[u,w]) \mid w \in \underline{\CS}[\p] \big\} \\
&= \alpha \big(\bigsqcup \big\{ \underline{A}[u,w] \mid w \in \underline{\CS}[\p] \big\} \big)\\
&= \alpha (\hat{A}[u]).
\end{align*}
Die Lösung des Ungleichungssystem, das eine Analyse durch Call-Strings beschreibt, 
wird also durch eine $\alpha$-korrekte abstrakte Interpretation der Transferfunktionen wieder $\alpha$-korrekt abstrahiert. 
Im nächsten Abschnitt zeigen wir, dass auch $\alpha$-Präzision übertragen wird.

\subsection{Präzise abstrakte Interpretation}\label{subsec:praezise-abstrint}
In \autoref{subsec:korrekt-abstrint} haben wir Transferfunktionen korrekt abstrakt interpretiert 
und dadurch in beiden Ansätzen zur interprozeduralen Datenflussanalyse korrekte abstrakte Interpretation der Ungleichungssysteme erhalten. 
Wir werden nun sehen, dass durch eine präzise abstrakte Interpretation der Transferfunktionen auch die Ungleichungssysteme beider Ansätze präzise abstrahiert werden, 
wenn wir für die Galois-Verbindung zusätzlich $\alpha \circ \gamma = \id_{L^\sharp}$ fordern.

Diese Forderung ist natürlich: Nach Definition gilt zunächst $\alpha \circ \gamma \sqleqsharpmap \id_{L^\sharp}$. 
Ist nun für ein $y \in L^\sharp$ tatsächlich 
\[y' := (\alpha \circ \gamma) (y) \sqsubset y,\]
so folgt mit Lemma \ref{lemma:eig-galois-verknuepfungen}
\[\gamma(y') = (\gamma \circ \alpha \circ \gamma) (y) = \gamma(y).\]
Die beiden abstrakten Werte $y$ und $y'$ werden also demselben konkreten Wert zugeordnet. 
Somit können Werte im abstrakten Raum unterschieden werden, die im konkreten Raum nicht unterschieden werden können. 
Dies ist unnatürlich, da der konkrete Raum genauere Informationen liefern sollte als der abstrakte.

\begin{satz}\label{satz:abstr-T-praez}
Sei die abstrakte Interpretation der Transferfunktionen $\alpha$-präzise und gelte $\alpha \circ \gamma = \id_{L^{\sharp}}$. 
Dann ist die dadurch induzierte abstrakte Interpretation der verallgemeinerten Transferfunktionen $f_{T,\cdot}$ auch $\alphamap$-präzise.
\end{satz}
\begin{proof}
Sei $h = (h_\node)_{\node \in N} \in \mathcal{F}^N$ beliebig. 
Nach dem Beweis von Satz \ref{satz:abstr-T-korr} bleibt nun
\begin{align*}
\alpha \circ \gamma &= \id_{L^\sharp},
\intertext{für eine Basiskante $(u,\mathtt{b},v)\in E$}
\alpha \circ f_\mathtt{b} \circ h_u \circ \gamma &= f_\mathtt{b}^\sharp \circ \alpha \circ h_u \circ \gamma,
\intertext{und für einen Prozeduraufruf $(u,\mathtt{q()},v) \in E$}
\alpha \circ h_{r_\q} \circ h_u \circ \gamma &= \alpha \circ h_{r_\q} \circ (\gamma \circ \alpha) \circ h_u \circ \gamma
\end{align*}
zu zeigen. Die ersten beiden Gleichungen folgen direkt aus der Voraussetzung. Die letzte Gleichung folgt aus Lemma \ref{lemma:hilfe-abstr-T-praez}.
\end{proof}
Es gilt also $\underline{T}^\sharp[u] = \alphamap(\underline{T}[u])$ für jeden Knoten $u \in N$.

Wir betrachten nun das Ungleichungssystem zur Berechnung der eigentlichen Informationen.
\begin{satz}\label{satz:abstr-R-praez}
Sei die abstrakte Interpretation der Transferfunktionen $\alpha$-präzise und gelte $\alpha \circ \gamma = \id_{L^{\sharp}}$. 
Dann ist die dadurch induzierte abstrakte Interpretation der verallgemeinerten Transferfunktionen $f_{R,\cdot}$ auch $\alpha$-präzise.
\end{satz}
\begin{proof}
Sei $x = (x_\node)_{\node \in N} \in L^I$. 
Nach dem Beweis von Satz \ref{satz:abstr-R-korr} 
muss nur noch für Basiskanten $(u,\mathtt{b},v) \in E$
\begin{align*}
\alpha(f_\mathtt{b}(x_u)) &= f_\mathtt{b}^\sharp(\alpha(x_u))
\intertext{und für Prozeduraufrufe $(u,\mathtt{q()},v) \in E$}
\alpha(\underline{T}[r_\q](x_u)) &= \underline{T^\sharp}[r_\q](\alpha(x_u))
\end {align*}
gezeigt werden. Ersteres folgt direkt aus der Voraussetzung und letzteres aus Satz \ref{satz:abstr-T-praez}.
\end{proof}
Es gilt also $\underline{R}^\sharp[u] = \alpha(\underline{R}[u])$ für jeden Knoten $u \in N$.

Unter der zusätzlichen Voraussetzung $\alpha \circ \gamma = \id_{L^\sharp}$ gewinnt man aus der präzisen abstrakten Interpretation der Transferfunktionen bereits präzise abstrahierte Ungleichungssysteme,
wie sie zum Berechnen von Informationen mithilfe von Summary-Informationen benutzt werden. Ein entsprechendes Resultat gilt auch für den Call-String-Ansatz.
\begin{satz}\label{satz:abstr-A-praez}
Sei die abstrakte Interpretation der Transferfunktionen $\alpha$-präzise. 
Dann ist die dadurch induzierte abstrakte Interpretation der verallgemeinerten Transferfunktionen $f_{A,\cdot}$ auch $\alphacs$-präzise.
\end{satz}
\begin{proof}
Nach dem Beweis von Satz \ref{satz:abstr-A-korr} genügt es, für Prozeduren $\p \in \Proc$, Basiskanten $(u,\mathtt{b},v) \in E_\p$, Call-Strings $w \in \underline{\CS}[p]$ und $x \in L^{N_{\CS}}$
\begin{align*}
\alpha(f_{\mathtt{b}}(y_{u,w})) &= f_{\mathtt{b}}^\sharp (\alpha(y_{u,w}))
\end{align*}
zu zeigen. Dies folgt aber direkt aus der Voraussetzung.
\end{proof}
Es gilt also $\underline{A}^\sharp[u,w] = \alphacs(\underline{A}[u,w])$ für alle $(u,w) \in N_{\CS}$. Ähnlich wie für Satz \ref{satz:abstr-A-korr} erhalten wir nun 
\begin{align*}
\hat{A}^\sharp[u] = \alpha (\hat{A}[u]).
\end{align*}
Die Lösung des Ungleichungssystem aus dem Call-String-Ansatz  wird also durch eine präzise abstrakte Interpretation der Transferfunktionen wieder präzise abstrahiert.

In diesem Fall liefern also beide Ansätze zur Analyse interprozeduraler Programme wieder dasselbe Ergebnis, wenn die verwendeten Transferfunktionen geeignete Distributivitätsbedingungen erfüllen.
\begin{kor}
Sei das monotone Framework $(L,\sqleq,\mathcal{F})$ universell-distributiv oder positiv-distributiv und in letzterem Fall alle Programmpunkte erreichbar. 
Sei die zugrundeliegende abstrakte Interpretation der Transferfunktionen $\alpha$-präzise und gelte $\alpha \circ \gamma = \id_{L^\sharp}$. 
Dann stimmen die kleinsten Lösungen der abstrahierten Ungleichungssysteme überein, d.h.~es gilt \[\underline{R}^\sharp = \hat{A}^\sharp.\]
\end{kor}
\begin{proof}
Dies folgt unmittelbar aus den Sätzen \ref{satz:abstr-R-praez}, \ref{satz:koinzidenz} und \ref{satz:abstr-A-praez}: Für beliebiges $u \in N$ gilt
\[\underline{R}^\sharp[u] 
= \alpha(\underline{R}[u])
= \alpha(\hat{A}[u])
= \hat{A}^\sharp[u].\]
Dies ist die Behauptung.
\end{proof}

\subsection{Kanonische abstrakte Interpretation}\label{subsec:kanonische-abstrint}
In \autoref{subsec:korrekt-abstrint} und \autoref{subsec:praezise-abstrint} haben wir gesehen, 
dass unter den Voraussetzungen $\alpha \circ \gamma = \id_{L^\sharp}$ und $\init^\sharp = \alpha(\init)$ 
eine präzise abstrakte Interpretation der verwendeten Transferfunktionen zu einer präzisen abstrakten Interpretation der Ungleichungssysteme aus dem funktionalen und dem Call-String-Ansatz führt. 
Wir werden nun zeigen, dass unter denselben Forderungen eine kanonische abstrakte Interpretation der Transferfunktionen 
wieder eine kanonische abstrakte Interpretation der verallgemeinerten Transferfunktionen induziert. 
Sei also im Folgenden $f^+ = \alpha \circ f \circ \gamma$ für die Transferfunktionen $f \in \mathcal{F}$. 
Die dadurch abstrahierten verallgemeinerten Transferfunktionen bezeichnen wir mit $f_{T,\cdot}^\sharp$, $f_{R,\cdot}^\sharp$ und $f_{A,\cdot}^\sharp$. 
Die kanonisch abstrahierten verallgemeinerten Transferfunktionen bezeichnen wir entsprechend mit $f_{T,\cdot}^+$, $f_{R,\cdot}^+$ und $f_{A,\cdot}^+$. 
Diese sind gegeben durch
\begin{align*}
 f^+_{T,\cdot} &:= \alphamap \circ f_{T,\cdot} \circ \hat{\gammamap}, \\
 f^+_{R,\cdot} &:= \alpha \circ f_{R,\cdot} \circ \hat{\gamma}, \\
 f^+_{A,\cdot} &:= \alphacs \circ f_{A,\cdot} \circ \hat{\gammacs}.
\end{align*}
Wir zeigen nun, dass die kanonisch oder durch die $f^+$ abstrahierten verallgemeinerten Transferfunktionen jeweils übereinstimmen, wenn wir zusätzlich $\alpha \circ \gamma = \id_{L^\sharp}$ fordern. 
Wir beginnen mit dem funktionalen Ansatz.

\begin{satz}\label{satz:abstr-T-kan}
Sei $\alpha \circ \gamma = \id_{L^\sharp}$.
Dann gilt $f_{T,\cdot}^\sharp = f_{T,\cdot}^+$.
\end{satz}
\begin{proof}
Wir betrachten wiederum der Reihe nach die verschiedenen Abbildungen $f_{T,\cdot}$ und überprüfen $f_{T,\cdot}^\sharp = f_{T,\cdot}^+$. 
Dabei ist $f_{T,\cdot}^+ := \alphamap \circ f_{T,\cdot} \circ \hat{\gammamap}$ wie in Abbildung \ref{bild:kan-lift-T}. 
\begin{figure}[ht]
 \begin{center} 
  $ \xymatrix{
		\mathcal{F}^N 			\ar[r]^{f_{T,\cdot}}
  &		\mathcal{F}										\ar[d]^{\alphamap}
  \\	(\mathcal{F}^\sharp)^N 	\ar[r]^{f_{T,\cdot}^+}	\ar[u]_{\hat{\gammamap}}
  &		\mathcal{F}^\sharp
  }$
  \end{center}
\caption{Kanonische abstrakte Interpretation der Abbildungen $f_{T,\cdot}$.}
\label{bild:kan-lift-T}
\end{figure}

Für $h =(h_\node)_{\node \in N} \in (\mathcal{F}^\sharp)^N$ ist also 
\begin{align*}
f_{T,\cdot}^+ (h) 
= \alphamap (f_{T,\cdot} (\hat{\gammamap}(h)))
= \alpha \circ (f_{T,\cdot} ((\gamma \circ h_\node \circ \alpha)_{\node \in N}) \circ \gamma.
\end{align*}

Wir benutzen bei jeder der folgenden Rechnungen $\alpha \circ \gamma = \id_{L^\sharp}$. Zunächst gilt für die Startinformation
\begin{align*} 
f_{T,\id}^+ (h)
&= \alpha \circ f_{T,\id} ((\gamma \circ h_\node \circ \alpha)_{\node \in N}) \circ \gamma \\
&= \alpha \circ \id_L \circ \gamma \\
&= \alpha \circ \gamma \\
&= \id_{L^\sharp} \\
&= f_{T,\id}^\sharp (h).
\intertext{
Sei $(u,\mathtt{b},v)\in E$ eine Basiskante. Dann ist 
}
f_{T,(u,\mathtt{b},v)}^+ (h)
&= \alpha \circ f_{T,(u,\mathtt{b},v)} ((\gamma \circ h_\node \circ \alpha)_{\node \in N}) \circ \gamma \\
&= \alpha \circ (f_\mathtt{b} \circ (\gamma \circ h_u \circ \alpha)) \circ \gamma \\
&= (\alpha \circ f_\mathtt{b} \circ \gamma) \circ h_u \circ (\alpha \circ \gamma) \\
&= f_\mathtt{b}^+ \circ h_u \\
&= f_{T,(u,\mathtt{b},v)}^\sharp (h).
\intertext{
Sei nun $(u,\mathtt{q()},v) \in E$ eine Aufrufkante. Dann gilt
}
f_{T,(u,\mathtt{q()},v)}^+ (h)
&= \alpha \circ f_{T,(u,\mathtt{()},v)} ((\gamma \circ h_\node \circ \alpha)_{\node \in N}) \circ \gamma \\
&= \alpha \circ (\gamma \circ h_{r_\q} \circ \alpha) \circ (\gamma \circ h_u \circ \alpha) \circ \gamma \\
&= (\alpha \circ \gamma) \circ h_{r_\q} \circ (\alpha \circ \gamma) \circ h_u \circ (\alpha \circ \gamma) \\
&= h_{r_\q} \circ h_u \\
&= f_{T,(u,\mathtt{b},v)}^\sharp (h). 
\end{align*}
Es folgt die Behauptung.
\end{proof}

Eine entsprechende Aussage können wir auch für das Ungleichungssystem zur Berechnung der gesuchten Informationen zeigen:
\begin{satz}\label{satz:abstr-R-kan}
Sei $\alpha \circ \gamma = \id_{L^\sharp}$.
Dann gilt $f_{R,\cdot}^\sharp = f_{R,\cdot}^+$.
\end{satz}
\begin{proof}
Wir überprüfen also $f_{R,\cdot}^\sharp = f_{R,\cdot}^+$ für $f_{R,\cdot}^+ := \alpha \circ f_{R,\cdot} \circ \hat{\gamma}$ wie in Abbildung \ref{bild:kan-lift-R}. 
\begin{figure}[ht]
 \begin{center} 
  $ \xymatrix{
		L^N 			\ar[r]^{f_{R,\cdot}}
  &		L										\ar[d]^{\alpha}
  \\	(L^\sharp)^N 	\ar[r]^{f_{R,\cdot}^+}	\ar[u]_{\hat{\gamma}}
  &		L^\sharp
  }$
  \end{center}
\caption{Kanonische abstrakte Interpretation der Abbildungen $f_{R,\cdot}$.}
\label{bild:kan-lift-R}
\end{figure}

Sei $y = (y_\node)_{\node \in N} \in (L^\sharp)^N$ beliebig. Dann gilt für die Startinformation
\begin{align*}
f_{R,\init}^+ (y)
&= (\alpha \circ f_{R,\init} \circ \hat{\gamma}) (y) \\
&= \alpha(\init) \\
&= \init^\sharp \\
&= f_{R,\init}^\sharp (y).
\intertext{
Für eine Basiskante $(u,\mathtt{b},v) \in E$ ist
}
f_{R,(u,\mathtt{b},v)}^+ (y)
&= (\alpha \circ f_{R,(u,\mathtt{b},v)} \circ \hat{\gamma}) (y) \\
&= \alpha(f_\mathtt{b}(\gamma(y_u))) \\
&= f_\mathtt{b}^+(y_u) \\
&= f_{R,(u,\mathtt{b},v)}^\sharp (y).
\intertext{
Zuletzt gilt für Aufrufe $(u,\mathtt{q()},v) \in E$ einerseits
}
f_{R,(u,\mathtt{q()},v),\text{ent}}^+ (y)
&= (\alpha \circ f_{R,(u,\mathtt{q()},v),\text{ent}} \circ \hat{\gamma}) (y) \\
&= \alpha(\gamma(y_u)) \\
&= y_u \\
&= f_{R,(u,\mathtt{q()},v),\text{ent}}^\sharp (y).
\intertext{
und andererseits mit Satz \ref{satz:abstr-T-kan}}
f_{R,(u,\mathtt{q()},v),\text{ret}}^+ (y)
&= (\alpha \circ f_{R,(u,\mathtt{q()},v),\text{ret}} \circ \hat{\gamma}) (y) \\
&= \alpha(\underline{T}[r_\q]((\gamma(y_u))) \\
&= \alphamap(\underline{T}[r_\q])(y_u) \\
&= \underline{T}^+[r_\q](y_u) \\
&= f_{R,(u,\mathtt{q()},v),\text{ret}}^\sharp (y).
\end{align*}
Es folgt die Behauptung.
\end{proof}
Im funktionalen Ansatz erhält man also eine kanonische Abstraktion der verallgemeinerten Transferfunktionen, wenn die benutzten Transferfunktionen kanonisch abstrahiert werden.
Dies gilt auch für den Call-String-Ansatz, wie der folgende Satz zeigt.
\begin{satz}\label{satz:abstr-A-kan}
Sei $\alpha \circ \gamma = \id_{L^\sharp}$.
Dann gilt $f_{A,\cdot}^\sharp = f_{A,\cdot}^+$.
\end{satz}
\begin{proof}
Wir überprüfen also $f_{A,\cdot}^\sharp = f_{A,\cdot}^+$ für $f_{A,\cdot}^+ := \alphacs \circ f_{A,\cdot} \circ \hat{\gammacs}$ wie in Abbildung \ref{bild:kan-lift-A}. 
\begin{figure}[ht]
 \begin{center} 
  $ \xymatrix{
		L^{N_{\CS}} 			\ar[r]^{f_{A,\cdot}}
  &		L										\ar[d]^{\alphacs}
  \\	(L^\sharp)^{N_{\CS}} 	\ar[r]^{f_{A,\cdot}^+}	\ar[u]_{\hat{\gammacs}}
  &		L^\sharp
  }$
  \end{center}
\caption{Kanonische abstrakte Interpretation der Abbildungen $f_{A,\cdot}$.}
\label{bild:kan-lift-A}
\end{figure}

Sei nun $y \in (L^{\sharp})^{N_{\CS}}$. Für die Abbildungen gilt allgemein
\begin{align*}
f_{A,\cdot}^+(y)
= (\alphacs \circ f_{A,\cdot} \circ \hat{\gammacs})(y)
= \alpha \big(f_{A,\cdot}((\gammacs(y_{\node,w}))_{(\node,w) \in N_{\CS}})\big).
\end{align*}
Wir betrachten nun die einzelnen Abbildungen.

Für die Startinformation gilt dann
\begin{align*}
f_{A,\init}^+ (y)
&= \alpha \big(f_{A,\init}((\gammacs(y_{\node,w}))_{(\node,w) \in N_{\CS}})\big)\\
&= \alpha(\init) \\
&= \init^\sharp \\
&= f_{A,\init}^\sharp (y).
\intertext{
Ist $(u,\mathtt{b},v) \in E_\p$ eine Basiskante, so ist
}
f_{A,(u,\mathtt{b},v),w}^+ (y)
&= \alpha \big(f_{A,(u,\mathtt{b},v),w}((\gammacs(y_{\node,w}))_{(\node,w) \in N_{\CS}})\big)\\
&= \alpha\big( f_\mathtt{b} (\gamma(y_{u,w})) \big) \\
&= f_\mathtt{b}^+ (y_{u,w}) \\
&= f_{A,(u,\mathtt{b},v),w}^\sharp (y).
\intertext{
Entsprechend gilt für Aufrufkanten $e=(u,\mathtt{q()},v) \in E_\p$ 
}
f_{A,(u,\mathtt{q()},v),\text{ent}}^+ (y)
&= \alpha \big(f_{A,(u,\mathtt{q()},v),w, \text{ent}}((\gamma(y_{\node,w}))_{(\node,w) \in N_{\CS}})\big)\\
&= \alpha\big(\gammacs(y_{u,w}) \big) \\
&= y_{u,w} \\
&= f_{A,(u,\mathtt{q()},v),w,\text{ent}}^\sharp (y).
\intertext{
und
}
f_{A,(u,\mathtt{q()},v),\text{ret}}^+ (y)
&= \alpha \big(f_{A,(u,\mathtt{q()},v),w \cdot e ,\text{ret}}((\gamma(y_{\node,w}))_{(\node,w) \in N_{\CS}})\big)\\
&= \alpha(\gamma(y_{r_\q,w \cdot e})) \\
&= y_{r_\q,w \cdot e} \\
&= f_{A,(u,\mathtt{q()},v),w \cdot e,\text{ret}}^\sharp (y).
\end{align*}
Es folgt die Behauptung.
\end{proof}
In einer Situation, in der die Transferfunktionen nicht präzise abstrakt interpretiert werden können, ist es also grundsätzlich sinnvoll, sie kanonisch zu interpretieren, 
da dies nach Satz \ref{satz:kanonisch-ist-am-besten} die \quotes{beste}, d.h.~kleinste korrekte Lösung auf dem abstrakten Verband liefert. 
Die vorigen Sätze \ref{satz:abstr-R-kan} und \ref{satz:abstr-A-kan} zeigen nun, dass dadurch auch die verallgemeinerten Transferfunktionen kanonisch abstrahiert werden. 
Die zugehörigen abstrahierten Ungleichungssysteme liefern also auch wieder die beste Lösung auf dem abstrakten Verband. 
Für den funktionalen und den Call-String-Ansatz zur interprozeduralen Datenflussanalyse ist demnach zu empfehlen, die Transferfunktionen stets kanonisch abstrakt zu interpretieren, 
da dies in jedem Fall die kleinste korrekte Lösung liefert. 
In Situationen, in denen es sogar präzise abstrakte Interpretationen gibt, erhält man eine solche durch die kanonische abstrakte Interpretation.

Im folgenden Abschnitten betrachten wir wiederum die Polyederanalyse und studieren abstrakte Interpretationen anhand der Bildung konvexer Hüllen. 
Die letzten beiden Lemmata liefern ein Indiz dafür, dass die Bildung konvexer Hüllen im Rahmen der Polyederanalyse eine korrekte oder sogar präzise abstrakte Interpretation ist.

\section{Polyederanalyse und konvexe Hüllenbildung}\label{sec:konvexe-polyeder}
Im Rahmen der Polyederanalyse werden zu jedem Programmpunkt affine Ungleichungen der Gestalt $\sum_{i=1}^n a_i x_i \ge a_0$ bestimmt, 
welche für die möglichen Werte der Variablen an diesem Programmpunkt gelten. 
Die Menge aller Werte, die solche Ungleichungen erfüllen, bildet ein konvexes Polyeder. 

In \autoref{sec:polyeder} haben wir bereits Datenflussanalysen zur Berechnung der erreichbaren Variablenwerte eingeführt. 
Gültige Ungleichungen über diese Variablenwerte erhalten wir nun, indem wir die konvexen Hüllen dieser Wertemengen bestimmen. Dabei orientieren wir uns an \cite{seidl07}. 
Wir werden zeigen, dass die Bildung konvexer Hüllen eine Abstraktion ist. 
Ziel ist nun, eine präzise abstrakte Interpretation der Transferfunktionen aus \autoref{sec:polyeder} zu finden, so dass die konvexen Mengen direkt berechnet werden können, 
ohne dafür zuerst die Ungleichungen zur Berechnung der erreichbaren Variablenwerte zu lösen. 
Wir werden also in diesem Abschnitt zunächst eine abstrakte Interpretationen für die Ungleichungssysteme aus \autoref{sec:polyeder} einführen 
und diese dann im Hinblick auf Korrektheit und Präzision untersuchen. 
Wir sagen im Folgenden \quotes{korrekt} anstelle von \quotes{$\alpha$-korrekt} und \quotes{präzise} anstelle von \quotes{$\alpha$-präzise} für einen konkreten Hüllenoperator $\alpha$. 
Wir werden sehen, dass wir für den Ansatz mit Matrizenmengen eine präzise abstrakte Interpretation angeben können, 
wohingegeben eine analog definierte abstrakte Interpretation für den Ansatz mit Relationen zwar korrekt, aber im Allgemeinen nicht präzise ist. 
Im Gegensatz zu \autoref{sec:polyeder} liefern die beiden Ansätze hier also nicht dasselbe Ergebnis.

Wir werden die Ungleichungssysteme aus \autoref{sec:polyeder} durch die Bildung konvexer Hüllen abstrahieren. 
Um anschließende Aussagen über die Korrektheit und Präzision einer solchen abstrakten Interpretation zu treffen, 
benutzen wir die folgenden beiden Lemmata, die Aussagen über beliebige Hüllenoperatoren treffen. 
Das erste Lemma zeigt dabei, dass eine abstrakte Interpretation durch Hüllenbildung bereits die kanonische abstrakte Interpretation ist. 
Das zweite Lemma benennt Spezialfälle für Abbildungen, die bezüglich dieser abstrakten Interpretation schon die Gestalt haben, 
wie sie in der Definition einer präzisen abstrakten Interpretation gefordert wird. 
Wir erinnern daran, dass eine kanonische abstrakte Interpretation nur dann präzise ist, wenn die zugrundeliegende Galois-Verbindung eine präzise abstrakte Interpretation zulässt, 
und dass letzteres nicht immer der Fall ist. Dies haben wir durch Satz \ref{satz:kanonisch-ist-am-besten} und Beispiel \ref{bsp:kan-nicht-praez} gezeigt. 
Das zweite Lemma ist also nicht trivial.

\begin{lemma}\label{lemma:huell-korr}
Sei $(L,\sqleq)$ ein vollständiger Verband, $\alpha : L \to L$ ein Hül\-len\-ope\-ra\-tor, $f:L^I \to L$ für eine Indexmenge $I$ und $f^\sharp := \alpha \circ f$ für jeden verwendeten Operator $f$. 
Dann ist
$f^\sharp = f^+$.
Insbesondere ist eine solche abstrakte Interpretation $\alpha$-korrekt und sogar $\alpha$-präzise, falls $(\alpha,\id)$ eine präzise abstrakte Interpretation erlaubt.
\end{lemma}
\begin{proof}
Dies folgt unmittelbar aus Lemma \ref{lemma-korresp-huell-galois} und Satz \ref{satz:kanonisch-ist-am-besten}.
\end{proof}

\begin{lemma}\label{lemma:huell-prec}
Sei $(L,\sqleq)$ ein vollständiger Verband, $\alpha : L \to L$ ein Hül\-len\-ope\-ra\-tor, $f: L^I \to L $ monoton und $f^\sharp := \alpha \circ f$. 
Ist $f$ konstant oder eine Projektion, so gilt $f^\sharp \circ \hat{\alpha} = \alpha \circ f$.
\end{lemma}
\begin{proof}
Sei $f$ zunächst konstant, d.h.~für ein $l\in L$ sei $f(x) = l$ für jedes $x \in L^I$. Dann ist
\begin{align*}
f^\sharp \circ \hat{\alpha}(x)
= \alpha(f(\hat{\alpha}(x)))
= \alpha(l)
= \alpha(f(x))
= \alpha \circ f (x).
\end{align*}
Ist $f$ eine Projektion, das heißt $f(x) = x_i$ für ein $i \in I$ und jedes $x \in L^I$, so gilt
\begin{align*}
f^\sharp \circ \hat{\alpha}(x)
= \alpha(f(\hat{\alpha}(x)))
= \alpha(\alpha(x_i))
= \alpha(x_i)
= \alpha(f(x))
= \alpha \circ f (x). 
\end{align*}
Es folgt die Behauptung.
\end{proof}

Wir definieren nun, was eine \emph{konvexe Hülle} ist.
\begin{dfn}
Sei $V$ ein $\rr$-Vektorraum. Eine Menge $X\subseteq V$ heißt \emph{konvex}\index{konvex}, falls $\lambda v + (1- \lambda) w \in X$ für alle $0 \le \lambda \le 1$ und $v,w \in X$.

Die \emph{konvexe Hülle}\index{konvexe Hülle} einer Teilmenge $X \subseteq V$ ist die kleinste konvexe Menge $\langle X \rangle$, die $X$ enthält.
\end{dfn}

\begin{bem}\label{bem:technisches-zu-konvexen-huellen}
Sei $V$ ein $\rr$-Vektorraum.
\begin{enumerate}
\item Sei $\mathcal{C}(V):=\{C \subset V\mid C \text{ konvex}\}$ die Menge der konvexen Teilmengen von $V$. 
Dann ist $(\mathcal{C}(V), \subseteq)$ ein vollständiger Verband, wobei $\bigsqcup_{\mathcal{C}(V)} \mathcal{X} = \langle \bigcup \mathcal{X} \rangle$ für $\mathcal{X} \subseteq \mathcal{C}(V)$.
\item Schreibe $X \subseteq V$ als $X= \{x_i \mid i \in I\}$ für eine Indexmenge $I$. 
Dann ist die konvexe Hülle von $X$ gegeben durch \[\langle X \rangle = \Big\{\sum_{i\in I} \lambda_i x_i \mid \lambda_i\ge 0, \text{fast alle } \lambda_i = 0, \sum_{i \in I} \lambda_i = 1\Big\}.\]
\item Die Abbildung $\alpha: 2^V \to \mathcal{C}(V), X \mapsto \langle X \rangle$ ist ein Hüllenoperator und damit eine Abstraktion.
\item Falls $X \subseteq V$ eine konvexe Menge ist, so ist für jedes $A \subseteq X$ auch $\langle A \rangle \subseteq X$. 
Wir können die Bildung konvexer Hüllen also einschränken zu einer Abbildung $2^X \to \mathcal{C}(X)$.
\item Die Mengen $\Sigma=\{1\}\times\rr^n \subseteq \rr^{n+1}$, $\Mat(\Sigma)$ und $\rr^n \times \rr^n$ sind konvexe Teilmengen von Vektorräumen. 
\end{enumerate}
Sowohl für Matritzenmengen, Relationen als auch Zustandsmengen gilt also, dass die Bildung konvexer Hüllen auf dem zugrundeliegenden Verband operiert.

Wir bezeichnen dabei die Bildung konvexer Hüllen mit 
\begin{align*}
\alpha : L &\to L^\sharp, \\
\alpha_M : L_M &\to L_M^\sharp, \\
\alpha_R : L_R &\to L_R^\sharp .
\end{align*}
Dabei hatten wir 
\begin{align*}
L &= 2^\Sigma , \\
L_M &= 2^{\Mat(\Sigma)}, \\
L_R &= 2^{\rr^n \times \rr^n}
\end{align*}
definiert und setzen weiter
\begin{align*}
L^\sharp &:= \mathcal{C}(\Sigma) \subseteq L, \\
L_M^\sharp &:= \mathcal{C}(\Mat(\Sigma)) \subseteq L_M, \\
L_R^\sharp &:= \mathcal{C}(\rr^n \times \rr^n) \subseteq L_R.
\end{align*}
als die entsprechenden Mengen konvexer Teilmengen. Nach Aussage c) und d) sind diese Abbildungen wohldefiniert.
\end{bem}

Wir abstrahieren nun die Ungleichungssysteme, indem wir konvexe Hüllen bilden. Die abstrahierten Ungleichungssysteme versehen wir mit dem Symbol $^\sharp$. 
Die abstrahierten Abbildungen, welche diese Ungleichungssysteme definieren, sind nun gegeben durch
\begin{align*}
f_{T_M,\cdot}^\sharp &:= \alpha_M \circ f_{T_M,\cdot} & f_{R_M,\cdot}^\sharp &:= \alpha \circ f_{R_M,\cdot} & f_{A,\cdot}^\sharp &:= \alpha \circ f_{A,\cdot}\\
f_{T_R,\cdot}^\sharp &:= \alpha_R \circ f_{T_R,\cdot} & f_{R_R,\cdot}^\sharp &:= \alpha \circ f_{R_R,\cdot} &
\end{align*}
für alle Ungleichungen bis auf diejenigen, die beim funktionalen Ansatz die Rückkehr aus einer Prozedur beschreiben. 
Diese werden zusätzlich so geändert, dass nicht der im Ungleichungssystem $T_M$ bzw.~$T_R$ berechnete Effekt einer Prozedur verwendet wird, sondern stattdessen $T_M^\sharp$ bzw.~$T_R^\sharp$:
\begin{align*}
f_{R_M,(u,\mathtt{q()},v),\text{ret}}^\sharp &:= \alpha \circ (\alpha_{\text{Mat}}({\underline{T_M^\sharp}[r_{\q}]}) \circ \pr_u),\\
f_{R_R,(u,\mathtt{q()},v),\text{ret}}^\sharp &:= \alpha \circ (\alpha_{\text{Rel}}({\underline{T_R^\sharp}[r_{\q}]}) \circ \pr_u).
\end{align*}

\begin{bem}
Es reicht nicht, $f_{R_M,(u,\mathtt{q()},v),\text{ret}}^\sharp := f_{\underline{T_M^\sharp}[r_{\q}]} \circ \pr_u$ zu betrachten, 
da $f_{\underline{T_M^\sharp}[r_{\q}]} \circ \pr_u$ nicht auf konvexen Mengen operiert:
Dazu betrachten wir
\begin{alignat*}{3}
  A & = \begin{pmatrix} 1 & -1 \\ 1 & \wminus 1 \end{pmatrix} \in \Mat(2,\rr), \
& x & = \begin{pmatrix} 1 \\ 0 \end{pmatrix} \in \rr^2, \\
  B & = \begin{pmatrix}  -1 & \wminus 1 \\ -1 & -1 \end{pmatrix} \in \Mat(2,\rr), \
& y & = \begin{pmatrix} 0 \\ 1 \end{pmatrix} \in \rr^2.
\end{alignat*}
Für die Matrizenmenge $\mathcal{A} = \{A,B\}$ und die Zustandsmenge $S = \{x,y\}$ ist dann
\begin{align*}
\alpha_{\text{Mat}} (\mathcal{A})
&= \left\{\begin{pmatrix} 
	2 \lambda - 1	& 1 - 2 \lambda \\
	2 \lambda - 1	& 2 \lambda - 1
  \end{pmatrix}
  \mid 0 \le \lambda \le 1 \right\}
\\&= \left\{a \cdot 
  \begin{pmatrix} 
	1	& -1 \\
	1	& \wminus 1
  \end{pmatrix}
  \mid -1 \le a \le 1 \right\}
\intertext{und}
\alpha(S)
&= \left\{\begin{pmatrix} 
	\mu \\
	1 - \mu
  \end{pmatrix}
  \mid 0 \le \mu \le 1 \right\}
.\end{align*}
Damit erhalten wir nun
\begin{align*}
X &:= \alpha_{\text{Mat}}(\alpha_M(\mathcal{A}))(\alpha(S))
\\&= \left\{a \cdot
  \begin{pmatrix} 
	\mu - (1-\mu)\\
	\mu + (1-\mu)
  \end{pmatrix}
  \mid -1 \le a \le 1 \wedge 0 \le \mu \le 1 \right\}
\\&= \left\{a \cdot
  \begin{pmatrix} 
	2\mu - 1\\
	1
  \end{pmatrix}
  \mid -1 \le a \le 1 \wedge 0 \le \mu \le 1 \right\}
\\&= \left\{a \cdot 
  \begin{pmatrix} 
	b\\
	1
  \end{pmatrix}
  \mid -1 \le a \le 1 \wedge -1 \le b \le 1 \right\}.
\end{align*}
Es gilt für $v\in X$: Ist $v_2 = 0$, so ist auch $v_1=0$. Ist nämlich $v = \mymatrix{ab \\ a} \in X$ und $v_2=a=0$, so ist auch $v_1=ab=0$.
Nun sind $\mymatrix{1 \\ 1}$ und $\mymatrix{\wminus 1 \\ -1} \in X$ vermöge $a=b=1$ bzw.~$a=b=-1$. Andererseits gilt aber für deren Konvexkombination
\[\begin{pmatrix}1\\0\end{pmatrix} 
= \tfrac{1}{2} \begin{pmatrix}1 \\ 1\end{pmatrix} + \tfrac{1}{2} \begin{pmatrix}\wminus 1 \\ -1\end{pmatrix} \notin X.\] 
Folglich kann X nicht konvex sein. Diese Menge ist in \autoref{bild:notw-huell} dargestellt.
\begin{figure}[ht]
 \begin{center}
  \begin{tikzpicture}
   \filldraw[fill=gray!20, draw=black]
      (-1,1) node[above] {$Ay$} -- (1,1) node[above] {$Ax$} -- 
      (-1,-1) node[below] {$Bx$} -- (1,-1) node[below] {$By$} -- (-1,1);
    \draw[help lines,step=1cm] (-1.8,-1.8) grid (1.8,1.8);
    \draw[->] (-2,0) -- (2,0) node[right] {$x_1$} coordinate(x axis);
    \draw[->] (0,-2) -- (0,2) node[above] {$x_2$} coordinate(y axis);
    \foreach \x/
      \xtext in {-1,1} 
      \draw[xshift=\x cm] (0pt,2pt) -- (0pt,-2pt) node[below] {$\xtext$};
      \draw[yshift=1 cm,thick] (2pt,0pt) -- (-2pt,0pt);
      \draw (0,-1.2) node[left] {$-1$}; 
      \draw[yshift=-1 cm,thick] (2pt,0pt) -- (-2pt,0pt);
      \draw (0,1.2) node[left] {$1$}; 

  \end{tikzpicture}
  \caption{Notwendigkeit der konvexen Hüllenbildung.}
  \label{bild:notw-huell}
 \end{center}
\end{figure}
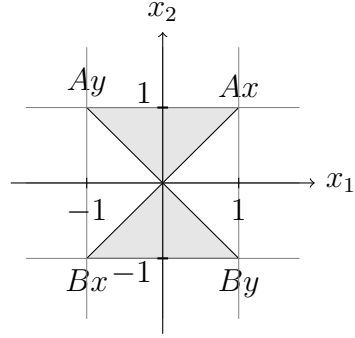

Setzen wir nun 
\begin{alignat*}{3}
\tilde{A} &:= 
  \begin{pmatrix}
	1 	& 0 	& 0 \\
	0	& \multicolumn{2}{c}{\multirow{2}{*}{A}} \\
	0	
  \end{pmatrix}, \
&\tilde{x} &:=
  \begin{pmatrix}
	1 \\
	x_1 \\ x_2  
  \end{pmatrix}, \\
\tilde{B} &:= 
  \begin{pmatrix}
	1 	& 0 	& 0 \\
	0	& \multicolumn{2}{c}{\multirow{2}{*}{B}} \\
	0	
  \end{pmatrix}, \
&\tilde{y} &:=
  \begin{pmatrix}
	1 \\
	y_1 \\ y_2  
  \end{pmatrix},
\end{alignat*} 
so erhalten wir ein Gegenbeispiel in $\{1\} \times \rr^2$ und $\Mat(\{1\} \times \rr^2)$.
\end{bem}

Wir wollen nun untersuchen, ob die Bildung konvexer Hüllen für die in \autoref{sec:polyeder} eingeführten Analysen korrekt oder sogar präzise ist. 
Da die Ungleichungssysteme $T_M$ und $T_R$ nicht direkt die Transferfunktionen berechnen, sind die Resultate aus Abschnitt \ref{sec:abstrint-fuer-PA} hier nicht anwendbar. 
Wir brauchen also neue Betrachtungen, die wir im folgenden Abschnitt präsentieren werden.
Wir orientieren uns dabei erneut an \cite{CH78-POPL} und \cite{seidl07}.

\begin{satz}
Die Ungleichungssysteme $T_M$, $R_M$, $T_R$ und $R_R$ werden durch $T_M^\sharp$, $R_M^\sharp$, $T_R^\sharp$ und $R_R^\sharp$ korrekt abstrahiert.
\end{satz}
\begin{proof}
Nach Lemma \ref{lemma:huell-korr} wissen wir bereits, dass die abstrakten Interpretationen $T_M^\sharp$ von $T_M$ und $T_R^\sharp$ von $T_R$ korrekt ist.
Da $\alpha_M$ und $\alpha_R$ Hüllenoperatoren sind, gilt also für jede Prozedur $\q \in \Proc$
\begin{align*}
\underline{T_M^\sharp}[r_q] \supseteq \alpha_M(\underline{T_M}[r_q]) \supseteq \underline{T_M}[r_q]
\text{ und } \underline{T_R^\sharp}[r_q] \supseteq \alpha_R(\underline{T_R}[r_q]) \supseteq \underline{T_R}[r_q].
\end{align*}
In Beispiel \ref{bsp:summary-inform-von-polyederana-sind-abstr} haben wir gesehen, dass die Abbildungen $\alpha_\text{Mat}$ und $\alpha_\text{Rel}$ Abstraktionen und damit insbesondere monoton sind. 
Da auch $\alpha$ monoton ist, erhalten wir
\begin{align*}
f_{R_M,(u,\mathtt{q()},v),\text{ret}}^\sharp 
= &\alpha \circ (\alpha_\text{Mat}( \underline{T_M^\sharp}[r_q]) \circ \pr_u) \\
\sqgeqmap &\alpha \circ (\alpha_\text{Mat}( \underline{T_M}[r_q]) \circ \pr_u) \\
 = &\alpha \circ f_{R_M,(u,\mathtt{q()},v),\text{ret}}
\intertext{und ebenso}
f_{R_R,(u,\mathtt{q()},v),\text{ret}}^\sharp 
\sqgeqmap&\alpha \circ f_{R_R,(u,\mathtt{q()},v),\text{ret}}.
\end{align*}
Nach Satz \ref{satz:kanonisch-ist-am-besten} abstrahieren also $R_M^\sharp$ und $R_R^\sharp$ die Ungleichungssysteme $R_M$ und $R_R$ korrekt. 
\end{proof}

Zunächst zeigen wir, dass $T_M^\sharp$ auch eine präzise abstrakte Interpretation von $T_M$ beschreibt. Dazu benötigen wir die folgenden beiden technischen Hilfsmittel.
\begin{lemma}\label{lemma-mat-correct}
Sei $X \subseteq \{1\}\times\rr^n$ und $\mathcal{A}, \mathcal{A}_1, \mathcal{A}_2 \subseteq \Mat(n+1,\rr)$. Dann gelten folgende Identitäten:
\begin{enumerate}
\item  
$\alpha(\alpha_{\text{Mat}}(\mathcal{A})(X))
= \alpha(\alpha_{\text{Mat}}(\alpha_M(\mathcal{A}))(\alpha(X)))$
\item
$\alpha_{\text{Mat}} (\mathcal{A}_1 \circ \mathcal{A}_2 )
= \alpha_{\text{Mat}}(\alpha_{\text{Mat}}(\mathcal{A}_1) \circ \alpha_{\text{Mat}}(\mathcal{A}_2))$.
\end{enumerate}
\end{lemma}
\begin{proof} Es ist nach Definition
\begin{align*}
\lefteqn{\alpha( \alpha_{\text{Mat}}(\alpha_M (\mathcal{A}))(\alpha(X)))} \\
&= {\alpha(\{Ax \mid x \in \alpha(X) \wedge A\in \alpha_M(\mathcal{A}) \})} \\
&= \alpha\big(\big\{(\lambda A_1 + (1- \lambda) A_2) x \mid 0 \le \lambda \le 1 \wedge A_1,A_2 \in \mathcal{A} \wedge x \in \alpha(X)\}\\
&= \alpha\big(\big\{\lambda A_1 x + (1- \lambda) A_2 x \mid 0 \le \lambda \le 1 \wedge A_1,A_2 \in \mathcal{A} \wedge x \in \alpha(X)\}\\
&= \alpha(\{A x \mid x \in \alpha(X) \wedge A \in \mathcal{A} \}) \\
&= \alpha\big(\big\{A (\lambda x_1 + (1- \lambda) x_2) \mid 0 \le \lambda \le 1 \wedge A \in \mathcal{A} \wedge x_1,x_2 \in X\}\\
&= \alpha\big(\big\{\lambda A x_1 + (1- \lambda) A x_2 \mid 0 \le \lambda \le 1 \wedge A \in \mathcal{A} \wedge x_1,x_2 \in X\}\\
&= \alpha(\{Ax \mid x \in X \wedge A \in \mathcal{A} \}) \\
&= \alpha(\alpha_{\text{Mat}}(\mathcal{A})(X))
\end{align*}
und 
\begin{align*}
\lefteqn{\alpha_M (\alpha_M (\mathcal{A}_1) \circ \alpha_M (\mathcal{A}_2) )} \\
&= {\alpha_M (\{A_1 A_2 \mid A_1 \in \alpha_M (\mathcal{A}_1),A_2 \in \alpha_M (\mathcal{A}_2)\})} \\ 
&= \alpha_M \big(\big\{(\lambda A_1^1 +(1- \lambda) A_1^2) A_2 \mid 0 \le \lambda \le 1, A_1^1,A_1^2 \in \mathcal{A}_1, A_2 \in \alpha_M (\mathcal{A}_2)\big\}\big)\\ 
&= \alpha_M \big(\big\{\lambda A_1^1 A_2 +(1- \lambda) A_1^2 A_2 \mid 0 \le \lambda \le 1, A_1^1,A_1^2 \in \mathcal{A}_1, A_2 \in \alpha_M (\mathcal{A}_2)\big\}\big)\\ 
&= \alpha_M (\{A_1 A_2 \mid A_1 \in \mathcal{A}_1, A_2 \in \alpha_M (\mathcal{A}_2)\big\}\big)\\ 
&= \alpha_M \big(\big\{A_1 (\lambda A_2^1 + (1-\lambda) A_2^2) \mid 0 \le \lambda \le 1, A_1 \in \mathcal{A}_1, A_2^1,A_2^2 \in \mathcal{A}_2\big\}\big)\\ 
&= \alpha_M \big(\big\{\lambda A_1 A_2^1 + (1-\lambda) A_1 A_2^2 \mid 0 \le \lambda \le 1, A_1 \in \mathcal{A}_1, A_2^1,A_2^2 \in \mathcal{A}_2\big\}\big)\\ 
&= \alpha_M (\{A_1 A_2 \mid A_i \in \mathcal{A}_i\}) \\ 
&= \alpha_M ( \mathcal{A}_1 \circ \mathcal{A}_2).
\end{align*}
Es folgt die Behauptung.
\end{proof}

\begin{kor}\label{zuweisung-prec}
Es gilt $\alpha \circ (f_{\mathtt{x_j := t}} \circ \pr_u) \circ \hat{\alpha} = \alpha \circ (f_{\mathtt{x_j := t}} \circ \pr_u)$.
\end{kor}
\begin{proof}
Sei $S \in L$, dann ist nach Lemma \ref{lemma-mat-correct}
\begin{align*}
\alpha \circ f_{\mathtt{x_j := t}} (S)
&= \alpha (\{M_{\mathtt{x_j := t}} \cdot s \mid s \in S\}) \\
&= \alpha (\{M_{\mathtt{x_j := t}} \cdot s \mid s \in \alpha (S)\}) \\
&= \alpha \circ f_{\mathtt{x_j := t}} \circ \alpha (S).
\end{align*}
Damit und mit der Identität $\alpha \circ \pr_u \circ~\hat{\alpha} = \alpha \circ \pr_u$ ist
\begin{align*}
\alpha \circ f_{\mathtt{x_j := t}} \circ \pr_u \circ~\hat{\alpha}
&= \alpha \circ f_{\mathtt{x_j := t}} \circ \alpha \circ \pr_u \circ~\hat{\alpha} \\
&= \alpha \circ f_{\mathtt{x_j := t}} \circ \alpha \circ \pr_u \\
&= \alpha \circ f_{\mathtt{x_j := t}} \circ \pr_u,
\end{align*}
wobei letzte Gleichheit vermöge $\alpha \circ f_{\mathtt{x_j := t}} \circ \alpha = \alpha \circ f_{\mathtt{x_j := t}}$ gilt.
\end{proof}

\begin{satz}\label{satz:TM+RM-praez}
Die abstrakte Interpretationen ${T_M^\sharp}$ und ${R_M^\sharp}$ von $T_M$ bzw.~$R_M$ sind präzise.
\end{satz}
\begin{proof}
Nach Definition ist
$f_{T_M, \cdot}^\sharp \circ \hat{\alpha}_M = \alpha_M \circ f_{T_M,\cdot}$ und $f_{R_M, \cdot}^\sharp \circ \hat{\alpha}_R = \alpha_R \circ f_{R_M,\cdot}$ 
zu zeigen.

Da $f_{T_M,\id}$ und $f_{R_M,\init}$ konstante Abbildungen sind und $f_{R_M,(u,\mathtt{q()},v),\text{ent}} = \pr_u$ eine Projektion ist, 
folgt hierfür die gewünschte Gleichheit schon nach Korollar \ref{lemma:huell-prec}. 
Für $f_{T_M,\mathtt{x_j:=t}} = f_{\mathtt{x_j := t}} \circ \pr_u$ und $f_{R_M,(u,\mathtt{x_j:=t},v)} = f_{\mathtt{x_j := t}} \circ \pr_u$ erhalten wir die Gleichheit mit Korollar \ref{zuweisung-prec}.
Nach Lemma \ref{lemma-mat-correct} gilt nun
\begin{align*}
f^\sharp_{T_M,\mathtt{q()}} \circ \hat{\alpha}_M (\mathcal{A})
&= \alpha_M(\mathcal{A}_{r_{\q}} \circ \alpha_M(\mathcal{A}_u)) \\
&= \alpha_M(\mathcal{A}_{r_{\q}} \circ \mathcal{A}_u) \\
&= \alpha_M(f_{T_M,\mathtt{q()}} (\mathcal{A})).
\end{align*}
Die Bildung konvexer Hüllen ist also bei der Berechnung von Matrizenmengen zur Bildung von Summary-Informationen präzise. 
Insbesondere gilt $\alpha_M(\underline{T_M}[r_\q]) = \underline{T^\sharp_M}[r_\q]$ für jede Prozedur $\q \in \Proc$. Daraus folgt nun, wiederum mit Lemmma \ref{lemma-mat-correct},
\begin{align*} 
f^\sharp_{R_M,(u,\mathtt{q()},v),\text{ret}}\circ\hat{\alpha}(\mathcal{S})
&= \alpha(\alpha_\text{Mat}( \underline{T^\sharp_M}[r_{\q}]) (\alpha(\mathcal{S}_u)) \\
&= \alpha(\alpha_\text{Mat}( \alpha_M(\underline{T_M}[r_{\q}])) (\alpha(\mathcal{S}_u)) \\
&= \alpha(\alpha_\text{Mat}( \underline{T_M}[r_{\q}]) (\mathcal{S}_u) \\
&= f_{R_M,(u,\mathtt{q()},v),\text{ret}}(\mathcal{S}).
\end{align*}
Somit ist auch die Bildung konvexer Hüllen bei den erreichbaren Zuständen eine präzise abstrakte Interpretation, wenn die Summary-Informationen mithilfe von Matrizenmengen berechnet werden.
\end{proof}

Verwendet man anstelle von Matrizenmengen dagegen Relationen, um Sum\-ma\-ry-In\-for\-ma\-tio\-nen zu berechnen, so ist die Bildung konvexer Hüllen keine präzise abstrakte Interpretation.
\begin{satz}
Die abstrakten Interpretationen ${T_R^\sharp}$ und ${R_R^\sharp}$ von $T_R$ bzw.~$R_R$ sind im Allgemeinen nicht präzise.
\end{satz}
\begin{proof}
Es genügt, ein Beispiel zu finden, bei dem konvexe Hüllenbildung nicht präzise ist.
Betrachte dazu den Flussgraphen aus Abbildung \ref{bild-unpr-rel}.
\begin{figure}[ht]
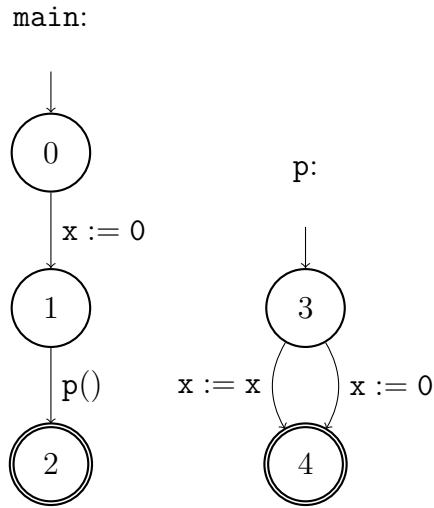

  \begin{flowgraph}
	\node[startnode]	(0) 						{$0$};
	\node				(procmain)	[above=of 0]	{$\main$:};
	\node[stdnode]		(1)			[below=of 0] 	{$1$};
	\node[endnode]		(2)			[below=of 1]	{$2$};
	\node				(inv)		[right=of 1]	{};
	\node[startnode]	(3) 		[right=of inv]	{$3$};
	\node				(procp) 	[above=of 3]	{$\p$:};
	\node[stdnode,accepting](4)		[below=of 3] 	{$4$};
	\path[->]
	  (0) 	edge 					node {$\mathtt{x:=0}$} (1)
	  (1)	edge 					node {$\mathtt{p()}$}  (2)
	  (3)	edge [bend right,swap] 	node {$\mathtt{x:=x}$} (4)
	  (3)	edge [bend left] 		node {$\mathtt{x:=0}$} (4)
	  
	;
  \end{flowgraph}
  \caption{Die Bildung konvexer Hüllen für Relationen als Summary-Information kann unpräzise sein.}	
  \label{bild-unpr-rel}
\end{figure}
Es ist
\begin{align*}
T[4] &= \{(x,x) \mid x \in \rr\} \cup \{(x,0) \mid x \in \rr\} 
\intertext{und}
R[0] &= \rr \\
R[1] &= \{0\} \\
R[2] &= \{0\}\cup\{0\}=\{0\}.
\intertext{
Andererseits haben wir
}
T^\sharp[4] &= \langle \{(x,x) \mid x \in \rr\} \cup \{(x,0) \mid x \in \rr\} \rangle \\
&= \{\lambda (x,x) + (1-\lambda) (y,0) \mid 0 \le \lambda \le 1,x,y \in \rr\} \\
&= \{(z,x)\mid x,z \in \rr\} \\
&= \rr \times \rr 
\intertext{sowie}
R^\sharp[0] &= \rr \\
R^\sharp[1] &= \{0\} \\
R^\sharp[2] &= \rr \supsetneq\{0\}=\langle R[2]\rangle.
\end{align*}
Die abstrakte Interpretation ist also nicht präzise.
\end{proof}

In dem Schritt der Berechnung, in dem konvexe Hüllen gebildet werden, liefert der Ansatz mit Matrizenmengen, im Gegensatz zu dem relationalen Ansatz, ein präzises Ergebnis. 
Zur Berechnung gültiger Ungleichungen in einem Programmpunkt ist demnach die Verwendung von Matrizenmengen als Summary-Information der bessere Ansatz.

Das folgende Beispiel zeigt, dass die Bildung konvexer Hüllen noch keine Terminierung garantiert.
\begin{bsp}
Betrachte das Programm aus Abbildung \ref{bild-nichteff-abstraktion}. 
 \begin{figure}[ht]
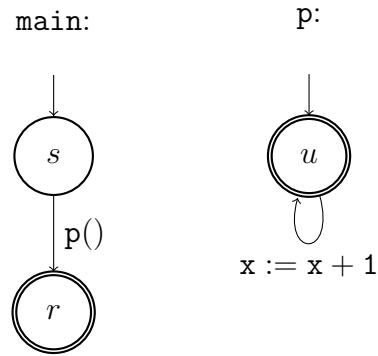

  \begin{flowgraph}
    \node[startnode]			(s) 						{$s$};
    \node						(procmain)	[above=of s]	{$\main$:};
    \node[endnode]				(r)			[below=of s]	{$r$};
    \node						(inv)		[right=of s]	{};
    \node[startnode,accepting]	(u) 		[right=of inv]	{$u$};
    \node						(procp) 	[above=of u]	{$\p$:};
    \path[->]
      (s)	edge 						node {$\mathtt{p()}$}  	(r)
      (u)	edge [out=285,in=255,loop] 	node {$\mathtt{x:=x+1}$}(u)
    ;

  \end{flowgraph}\caption{Konvexe Hüllenbildung garantiert keine Terminierung.}\label{bild-nichteff-abstraktion}
 \end{figure}

Das zugehörige Ungleichungssystem zur Berechnung konvexer Matrizenmengen besteht für die Prozedur $\p$ aus den beiden Ungleichungen
\begin{align*}
T_M^\sharp[u] \supseteq \{E_2\} \text{ und } T_M^\sharp[u] \supseteq \alpha_M( \{M\}\circ T_M^\sharp[u] )
\end{align*}
für die Matrix
\[ M = 
  \begin{pmatrix}
	1 & 0 \\
	1 & 1 
  \end{pmatrix}. 
\]
Da im Workset-Algorithmus die Berechnung einer kleinsten Lösung für den Programmpunkt $u$ unabhängig von den Werten in den Programmpunkten $s$ und $r$ ist, 
können wir die Ungleichungen für diese Punkte vernachlässigen. Für $u$ produziert der Workset-Algorithmus folgende Werte:
Zunächst erhält $T_M^\sharp[u]$ den Wert $\{E_2\}$ und die Workset enthält die zweite Ungleichung. Dann wird
\[
\alpha_M(\{M\} \circ \{E_2\} ) = \alpha_M(\{M\}) = \{M\}
\]
als neuer Wert berechnet, den $T_M^\sharp[u]$ mindestens haben muss. Da $\{E_2\} \not\supseteq\{M\}$, erhält also $T_M^\sharp[u]$ als neuen Wert 
\[\alpha_M(\{E_2\} \cup \{M\})
 = \big\{ 
M_t \mid 0 \le t \le 1
\big\}
\]
für \[M_t :=   \begin{pmatrix}
	1 & 0 \\
	t & 1 
  \end{pmatrix}
\]
und die Workset enthält wiederum die zweite Ungleichung.
Im dritten Schritt wird entsprechend zunächst
\[
\{M\} \circ \{M_t \mid 0 \le t \le 1\}
= \{M_t \mid 1 \le t \le 2\} 
\]
berechnet. Die zweiten Ungleichung ist also weiterhin unerfüllt und es wird
\[
T_M^\sharp[u] 
= \alpha_M( \{M_t \mid 0 \le t \le 1\} \cup \{M\} \circ \{M_t \mid 0 \le t \le 1\} )
= \{M_t \mid 0 \le t \le 2\}
\]
gesetzt.
Analog erhält man im $n$-ten Schritt der Berechnung
\[
T_M^\sharp[u]
= \{M_t \mid 0 \le t \le n-1\}.
\]
Die zweite Ungleichung ist wiederum nicht erfüllt. Also terminiert die Berechnung nicht.
\end{bsp}

Wir haben in diesem Kapitel eingeführt, wie wir die in \autoref{chap:PA} eingeführten Ansätze zur Analyse interprozeduraler Programme abstrakt interpretieren. 
Diese abstrakte Interpretation hängt allein von einer zuvor festgelegten abstrakten Interpretation der Transferfunktionen des monotonen Frameworks ab. 
Wir haben gezeigt, dass die abstrakte Interpretation korrekt, präzise bzw.~kanonisch ist, wenn die abstrakte Interpretation der Transferfunktionen die entsprechende Eigenschaft hat. 
Weiter haben wir gesehen, dass die kanonische Interpretation von allen korrekten die bestmögliche Lösung liefert. 
Falls die zugrundeliegende Galois-Verbindung präzise abstrakte Interpretationen zulässt, so ist die kanonische abstrakte Interpretation sogar präzise. 
Die bestmögliche abstrakte Interpretation erhält man also im funktionalen wie auch im Call-String-Ansatz durch die kanonische abstrakte Interpretation der Transferfunktionen. 
Außerdem haben wir gesehen, dass die beiden Ansätze der Polyederanalyse verschiedene 
Ergebnisse liefern: Der Ansatz mit Matritzenmengen liefert präzise Werte, während der Ansatz mit Relationen zwar korrekte, aber nicht unbedingt präzise Werte berechnet.

Zudem haben wir gesehen, dass das Abstrahieren von Ungleichungssystemen keine Terminierung garantiert. 
Ein weiteres Mittel, das Terminierung erreichen kann, ist die Verwendung eines sogenannten Widening-Operators im Workset-Algorithmus. Dieses werden wir im kommenden Kapitel untersuchen.


\chapter{Widening}\label{chap:widening}
In \autoref{chap:PA} haben wir zwei verschiedene Ansätze zur Analyse interprozeduraler Programme vorgestellt und im Hinblick auf ihre Präzision ihrer Ergebnisse verglichen. 
Eine solches Ergebnis kann mithilfe eines sogenannnten {Workset-Algorithmus}\index{Workset-Algorithmus} bestimmt werden, der aber nicht terminieren muss.

Ein von Cousot und Cousot in \cite{CousotCousot76-1} vorgestellter Ansatz, um Terminierung des Workset-Al\-go\-rith\-mus zu erreichen, 
ist die Verwendung eines \emph{Widening-Operators} im Workset-Algorithmus. 
In diesem Kapitel werden wir zunächst Widening-Operatoren allgemein vorstellen. 
Weiter werden wir feststellen, dass durch die Verwendung eines Widening-Operators im Workset-Algorithmus keine eindeutige Lösung berechnet wird 
und dass - im Gegensatz zu Analysen mit nur einer Prozedur - noch immer keine Terminierung garantiert werden kann. 
Zuletzt werden wir sehen, dass die Lösungen, die der Workset-Algorithmus mit Widening-Operatoren für den funktionalen und den Call-String-Ansatz berechnet, im Allgemeinen unvergleichbar sind.

\section{Workset-Algorithmus mit Widening}
In Abschnitt \ref{sec:eff} haben wir gesehen, dass ein Workset-Algorithmus schon für ein endliches Ungleichungssystem nicht terminieren muss, 
wenn der zugrundeliegende vollständige Verband unendliche echt aufsteigende Ketten enthält: Die Werte, die der Algorithmus für eine Variable berechnet, bilden eine aufsteigende Kette. 
Die Berechnung eines Wertes für die Variable terminiert also nur, wenn die Folge, die aus dieser Kette besteht, stabil wird. 
Enthält der Verband aber unendliche echt aufsteigende Ketten, so ist dies nicht gewährleistet. 
Um dem Abhilfe zu schaffen, wird bei der Berechnung eines neuen Wertes nicht mehr die kleinste obere Schranke der eingehenden Werte bestimmt, sondern ein \emph{Widening-Operator} verwendet.
\begin{dfn}
Eine Abbildung $\widening: L \times L \to L$ heißt \emph{Extrapolationsoperator}\index{Extrapolationsperator}, falls \[x \widening y \sqgeq x \sqcup y\] für alle $x,y \in L$ gilt.

Ein Extrapolationsoperator heißt \emph{Widening-Operator}\index{Widening-Operator}, 
falls für jede Folge $(x_n)_{n \in \nn}$ in $L$ die Folge $(l_n)_{n \in \nn}$ mit $l_0 := x_0$ und $l_{n+1} := l_n \widening x_{n+1}$ stabil wird, 
d.h.~es ein $N \in \nn$ gibt, so dass $l_{N+i} = l_N$ für jedes $i \ge 0$ gilt.
\end{dfn}
Verwendet man nun anstelle der kleinsten oberen Schranke bei der Berechnung der neuen Werte im Workset-Algorithmus einen Extrapolationsoperator, so ist eine so berechnete Lösung noch immer korrekt. 
Dies sieht man analog zum Beweis von Lemma \ref{lem:wla-korr}, in dem gezeigt wurde, dass der Worklist-Algorithmus, sofern er terminiert, die kleinste Lösung des betrachteten Ungleichungssystems liefert.
Betrachten wir dazu eine feste Variable. Mit $x_n$ bezeichnen wir den $n$-ten Wert, der für diese Variable berechnet wird. Das im $n$-ten Schritt berechnete \quotes{$\mathtt{t}$} bezeichnen wir mit $l_n$.
Dann entstehen die $l_n$ aus den $x_n$ durch $l_{n+1} := l_n \sqcup x_{n+1}$. 
Verwendet man anstelle von $\sqcup$ einen Widening-Operator $\widening$, so wird die Folge $(l_n)_{n \in \nn}$ stabil. 
Die Berechnung eines Wertes für diese Variable terminiert also.

\begin{bem}
Ist $(L,\sqleq)$ ein vollständiger Verband ohne unendliche echt aufsteigende Ketten und $\widening: L \times L \to L$ ein Extrapolationsoperator, so ist $\widening$ bereits ein Widening-Operator. 
Dies folgt direkt daraus, dass die Folge $(l_n)_{n \ge 0}$ aus der Definition von Widening-Operatoren eine Kette bildet.

Insbesondere ist also auf endlichen vollständigen Verbänden jeder Extrapolationsoperator bereits ein Widening-Operator.
\end{bem}

Wir betrachten nun ein Beispiel für einen Widening-Operator für die Intervallanalyse. Der Widening-Operator wurde ebenfalls von Cousot und Cousot in \cite{CousotCousot76-1} eingeführt.
\begin{bsp}
\label{bsp-widening-IA-std}
Definiere $\widening: L_\text{Int} \times L_\text{Int} \to L_\text{Int}$ durch $[l_1,u_1] \widening [l_2,u_2] = [l,u]$ mit
\begin{align*}
 l &:= \begin{cases} l_1 & \text{falls } l_1 \le l_2 \\ -\infty &\text{sonst} \end{cases} 
&&\text{und}&
 u &:= \begin{cases} u_1 & \text{falls } u_1 \ge u_2 \\ +\infty &\text{sonst}. \end{cases} 
\end{align*}
{sowie}
\[ \emptyset \widening I := I\]
{und}
\[ I \widening \emptyset := I\]
für jedes $I \in L_\text{Int}$.
Ist die untere Grenze des ersten Intervalles im zweiten enthalten, so bildet dies auch die unteren Grenzen des neuen Intervalles. Ansonsten wird die Grenze auf $-\infty$ gesetzt. 
Entsprechend ist die neue obere Grenze die des ersten Intervalles oder $+\infty$. Dies ist ein Widening-Operator: 

Seien $[l_1,u_1]$ und $[l_2,u_2]$ beliebige Intervalle und $[l,u]=[l_1,u_1] \widening [l_2,u_2]$. Falls $l_1 \le l_2$, so ist $l=l_1\le l_2$. Andernfalls ist $l=-\infty\le l_2 \le l_1$. 
In beiden Fällen erhalten wir also $l_1,l_2 \ge l$. Ebenso zeigt man $u_1,u_2 \le u$. Zusammen folgt $[l_i,u_i] \subseteq [l,u]$ für $i=1,2$. Damit ist $\widening$ ein Extrapolationsoperator.

Sei nun $(I_n)_{n \in \nn}$ eine Folge in $L_\text{Int}$. Definiere $I'_0 := I_0$ und $I'_{n+1} := I'_n \widening I_{n+1}$. 
Da $\widening$ ein Extrapolationsoperator ist, ist die Folge $(I'_n)_{n \in \nn}$ aufsteigend. Angenommen, diese Folge wird nicht stabil. 
Dann gibt es eine echt aufsteigende Teilfolge, die wir wieder mit $(I'_n)_{n \in \nn}$ bezeichnen. 
Schreibe $I_n = [l_n,u_n]$ und $I'_n=[l'_n,u'_n]$.

Angenommen, die Folge $(l'_n)_{n \ge 0}$ ist nicht konstant. Dann gibt es einen Index $m \ge 0$ mit $l'_{m+1} \sqsubset l'_m$. Nach Konstruktion ist dann $l'_{m+1} = - \infty$. 
Insbesondere ist $l'_n = - \infty$ für alle $n \ge m+1$. Also muss $(u'_n)_{n \ge m+1}$ echt aufsteigend sein. 
Insbesondere ist $u'_{m+2} \sqsupset u'_{m+1}$ und nach Konstruktion wiederum $u'_n= +\infty$ für alle $n \ge m+2$. Also ist die Folge $(I'_n)_{n \ge m+2}$ konstant. Dies ist ein Widerspruch. 
Also muss $(l'_n)_{n \ge 0}$ konstant sein. Es gilt demnach $l'_n = l_0$ für alle $n \ge 0$. 
Analog zeigt man $u'_n = u_0$ für alle $n \ge 0$. Damit ist $(I'_n)_{n \in \nn}$ konstant. Dies ist ein Widerspruch. 
Also ist $\widening$ ein Widening-Operator. 
\end{bsp}

Wir geben nun einen Workset-Algorithmus an, der Widening benutzt. Zusätzlich fordern wir, dass die Strategie, nach der Ungleichungen aus der Workset ausgewählt werden, in einem gewissen Sinne fair ist: 
Für jede Ungleichung und jeden Zeitpunkt $t$ soll gelten, dass, wenn sich die Ungleichung zum Zeitpunkt $t$ in der Workset befindet, sie zu einem späteren Zeitpunkt betrachtet wird. 
Dies garantiert, dass keine Ungleichung gewissermaßen \quotes{verhungert}, sondern alle Ungleichungen immer wieder betrachtet werden. 
Diesen Algorithmus nennen wir \quotes{Workset-Algorithmus mit Widening und Fairness}.
\begin{pseudocode}
W := $\mathcal{U}$;
forall (i$\in$I) $\{\mathtt{x_i}$:=$\bot;\}$
while W $\ne ()$ $\{$
  choose u=($\mathtt{x_j}\sqgeq$f(x)) fairly from W; W := W$\setminus\{$u$\}$;
  t := f(x);
  if $\neg$(t $\sqleq$ $\mathtt{x_j}$) $\{$
    //compute new value using $\widening$
    $\mathtt{x_j}$ := $\mathtt{x_j}$ $\widening$ t;
    forall (($\mathtt{x_k}\sqgeq$g(x)) $\in$ $\mathcal{U}$ : g uses $\mathtt{x_j}$) $\{$W := W:($\mathtt{x_k}\sqgeq$g(x));$\}$
  $\}$
$\}$
\end{pseudocode}

\begin{lemma}\label{lem:wsa-mit-widening-und-fairness-terminiert-fuer-jede-var}
Für jede Variable garantiert der Workset-Algorithmus mit Widening und Fairness Terminierung.
\end{lemma}
\begin{proof}
Sei $i \in I$ ein beliebiger Index. Wir zeigen nun, dass für die Variable $\mathtt{x_i}$ in endlicher Zeit ein Wert berechnet wird.

Wie bei der Betrachtung des Workset-Algorithmus ohne Widening in \autoref{sec:wla} betrachten wir die Folgen $(x^{(n)}_i)_{n \in \nn}$. Es gilt wieder $x^{(n)}_i \sqleq x^{(m)}_i$ für $n \le m$. 
Wir konstruieren daraus eine echt aufsteigende Teilfolge, indem wir alle diejenigen Folgenglieder weglassen, die mit dem vorherigen übereinstimmen:
\begin{align*}
x^{(n_0)}_i &:= x^{(0)}_i \\
x^{(n_{k+1})}_i &:= x_i^{(j)} \text{ mit } j=\min\{l \mid x^{(l)}_i \sqsupset x^{(n_k)}_i \}
\end{align*}
Nach dem $n_k$-ten Schritt ist also der $j$-te Schritt der nächste, in dem der Wert der Variablen $x_i$ verändert wird. 
Dieser Index existiert, da die Ungleichungen fair aus der Workset herausgenommen werden. 
Also wird jede Ungleichung, auf deren linker Seite die Variable $x_i$ steht und die in der Workset enthalten ist, irgendwann nach dem $n_k$-ten Schritt wieder betrachtet. 
Insbesondere gilt also $x_i^{(j-1)} = x_i^{(n_k)}$.

Nach Konstruktion entsteht $x^{(j)}_i$ aus $x^{(j-1)}$ durch Betrachten einer unerfüllten Ungleichung $x_i \sqleq f(x)$. 
Damit ist dann \[x^{(n_{k+1})}_i = x^{(j)}_i = x^{(j-1)}_i \widening f(x^{(j-1)}).\] 
Nach Definition eines Widening-Operators wird diese Teilfolge ab einem Index $m_i$ stabil. Damit wird aber auch $(x^{(n)}_i)$ ab $m_i$ stabil.
\end{proof}

\begin{kor}
Wenn das betrachtete Ungleichungssystem endlich ist, garantiert der Workset-Algorithmus mit Widening und Fairness Terminierung.
\end{kor}
\begin{proof}
Wie in Lemma \ref{lem:wsa-mit-widening-und-fairness-terminiert-fuer-jede-var} betrachten wir die Folgen $(x^{(n)}_i)_{n \in \nn}$ 
der für eine Variable $\mathtt{x_i}$ durch den Workset-Algorithmus mit Widening und Fairness berechnete Werte. 
Wir haben im Beweis von Lemma \ref{lem:wsa-mit-widening-und-fairness-terminiert-fuer-jede-var} gezeigt, dass ein endlicher Indes $m_i$ existiert, ab dem diese Folge stabil wird.

Dies gilt für jedes $i \in I$. Also wird $(x^{(n)})$ stabil ab $m := \max\{m_i \mid i\in I\}$. 
Ab dem Index $m$ ist dann jede Ungleichung erfüllt. Wie für den Workset-Algorithmus ohne Widening gilt, dass jede nichterfüllte Ungleichung auf der Workset enthalten ist. 
Wir zeigen nun $W^{(n+1)} = \text{tail}(W^{(n)})$ für alle $n \ge m$, für die die Workset nicht leer ist. 

Sei also $u \in W^{(n)}$ von der Gestalt $x_j \sqgeq f(x)$. Wäre $u$ nicht erfüllt, so würde als neuer Wert 
\[x^{(n+1)}_j = x^{(n)}_j \widening f(x^{(n)})\] 
berechnet. Da $(x^{(n)}_j)$ ab $m$ stabil ist, gilt aber auch 
\[x^{(n+1)}_j = x^{(n)}_j.\] 
Aus $x^{(n)}_j = x^{(n)}_j \widening f(x^{(n)})$ folgt aber 
\[x^{(n)}_j \sqgeq f(x^{(n)})\]
und dies ist ein Widerspruch. Also ist $u$ erfüllt und die neue Workset entsteht aus der alten durch Herausnehmen von $u$. 
Nach $m$ Schritten ist also $(x^{(n)})$ stabil und nach weitern $|W^{(m)}|$ Schritten die Workset leer. Somit terminiert der Algorithmus.
\end{proof}

Diese Lösung muss aber nicht minimal sein, wie das folgende Beispiel zeigt.
\begin{bsp}
Wir betrachten wiederum eine Intervallanalyse. Das Ungleichungssystem sei
\begin{align*}
x &\sqgeq [0,1] \\ x &\sqgeq [0,2].
\end{align*}
Zunächst wird $x$ mit $\emptyset$ und $W$ mit $\{(x \sqgeq [0,1]),(x \sqgeq [0,2])\}$ initialisiert. 
Nun wird die erste Ungleichung betrachtet. Da $\emptyset \not\sqgeq [0,1]$, wird nun $x:= \emptyset \widening [0,1] = [0,1]$ gesetzt. $W$ wird um die erste Ungleichung verringert. 

Auch die zweite Ungleichung ist nicht erfüllt und $x$ erhält den neuen Wert \[x=[0,1] \widening [0,2] = [0,\infty].\] 
Nun ist $W$ leer und der Algorithmus terminiert mit $x=[0,\infty]$. Die kleinste Lösung des Ungleichungssystems ist dagegen \[x = [0,1] \sqcup [0,2] = [0,2].\]
Also ist der für $x$ berechnete Wert nicht minimal.
\end{bsp}

Ist dagegen $I$ nicht endlich, so terminiert der Algorithmus nicht. Allerdings werden für jede Variable nur endlich viele Schritte benötigt, bis für sie ein endgültiger Wert berechnet ist. 
Die Fairness garantiert nun, dass die Ungleichungen, durch die diese Variable verändert wird, oft genug betrachtet werden. 
So kann zwar nicht in endlicher Zeit eine Lösung für alle Variablen berechnet werden, aber für jede einzelne Variable liefert der Algorithmus nach endlicher Zeit einen richtigen Wert. 
Somit kann der Workset-Algorithmus mit Widening zwar nicht unbedingt zum Lösen von unendlichen Ungleichungssystemen benutzt werden, 
wir werden ihn aber dennoch in den folgenden Abschnitten als Referenzpunkt zum Vergleichen verschiedener Analyseansätze benutzen. 
Zunächst werden wir uns jedoch einige Eigenschaften von Widening-Operatoren genauer ansehen.

\section{Nicht-Monotonie}\label{sec:widening-nicht-monoton}
In \autoref{sec:koinzidenz} haben wir zum Vergleich der kleinsten Fixpunkte der Ungleichungssysteme wesentlich 
die Monotonie der verwendeten Abbildungen und sogar geeignete Distributivitätseigenschaften benutzt. 
Dies ist hier nicht möglich: Wie wir im Folgenden zeigen werden, ist ein Widening-Operator im Allgemeinen nicht monoton und damit auch nicht distributiv. 
Wir werden nun an die Definition von Monotonie erinnern und anschließend monotone Widening-Operatoren untersuchen.
\begin{dfn}
Sei $(L,\sqleq)$ ein vollständiger Verband und $\widening: L \times L \to L$ ein Extrapolationsoperator.
Dieser heißt \emph{monoton}\index{monoton}, falls $x \widening y \sqleq x' \widening y'$ für alle $x \sqleq x'$ und $y \sqleq y'$ gilt.
Weiter heißt $\widening$ \emph{idempotent}\index{idempotent}, falls $x \widening x = x$ für jedes $x \in L$ gilt.
\end{dfn}

\begin{bsp} 
Wir betrachten erneut den Widening-Operator $\widening$ des Intervallverbandes. Dieser Widening-Operator ist nach Konstruktion idempontent. 
Aber er ist nicht monoton: Es ist $[0,1] \sqleq [0,2]$, aber $[0,1] \widening [0,2] = [0,\infty] \not\sqleq [0,2] = [0,2] \widening [0,2]$.
\end{bsp}

\begin{satz}
Sei $(L,\sqleq)$ ein vollständiger Verband und $\widening: L \times L \to L$ ein idempotenter und monotoner Extrapolationsoperator. 
Dann gilt bereits $\widening = \sqcup$.
\end{satz}
\begin{proof}
Seien $x,y \in L$ beliebig. Dann gilt mit der Monotonie und der Idempotenz von $\widening$ und zuletzt mit der Tatsache, dass $\widening$ ein Extrapolationsoperator ist
\begin{align*}
x \widening y 
\sqleq (x \sqcup y) \widening (x \sqcup y)
= x \sqcup y
\sqleq x \widening y.
\end{align*}
Somit gilt in obigen Ungleichungen bereits Gleichheit. Es folgt die Behauptung.
\end{proof}
Im Workset-Algorithmus wird der Widening-Operator innerhalb einer Fallunterscheidung benutzt. 
Diese können wir durch eine Abbildung $\widening_{idp}$ kodieren. Sie ist idempotent, wie der folgendende Satz zeigt. 
Wir sagen auch, der Widening-Operator werde \quotes{in idempotenter Weise verwendet}.
\begin{satz}
Sei $(L,\sqleq)$ ein vollständiger Verband und $\widening: L \times L \to L$ ein Extrapolationsoperator. 
Definiere $\widening_\text{idp}: L \times L \to L$ durch 
\[x \widening_\text{idp} y = \begin{cases} x & \text{falls } y \sqleq x \\ x \widening y & \text{sonst} \end{cases}\]
für $x,y \in L$. Dann ist $\widening_\text{idp}$ idempotent.
\end{satz}
\begin{proof} Dies gilt offenbar nach Definition. \end{proof}
Wir können den Workset-Algorithmus nun äquivalent schreiben als
\begin{pseudocode}
W := $\mathcal{U}$;
forall (i$\in$I) $\{\mathtt{x_i}$:=$\bot;\}$
while W $\ne ()$ $\{$
  choose u=($\mathtt{x_j}\sqgeq$f(x)) fairly from W; W := W$\setminus\{$u$\}$;
  t := f(x);
  if $\neg$(t $\sqleq$ $\mathtt{x_j}$) $\{$
    $\mathtt{x_j}$ := $\mathtt{x_j}$ $\widening_\text{idp}$ t;
    forall (($\mathtt{x_k}\sqgeq$g(x)) $\in$ $\mathcal{U}$ : g uses $\mathtt{x_j}$) $\{$W := W:($\mathtt{x_k}\sqgeq$g(x));$\}$
  $\}$
$\}$
\end{pseudocode}
Der Widening-Operator wird also in idempotenter Weise verwendet.
In einer Situation, in der das Bilden der kleinsten oberen Schranke durch $\sqcup$ keine Terminierung garantiert, kann die Verwendung eines Widening-Operators also nicht monoton sein. 
Da es in anderen Situationen unnötig ist, einen Widening-Operator zu verwenden, gehen wir im Folgenden von einer nicht-monotonen Verwendung von Widening-Operatoren aus.

Dadurch entstehen weitere Probleme. Sind die verwendeten Abbildungen alle monoton, so liefert der \quotes{normale} Workset-Algorithmus eine eindeutige Lösung. 
Wird dagegen ein Widening-Operator im Workset-Algorithmus verwendet, so ist dies nicht mehr der Fall:
\begin{bsp}
Wir betrachten das Programm aus Abbildung \ref{bild-uneind-wla-lsg} und machen wieder eine Intervallanalyse.
\begin{figure}[ht]
 \begin{flowgraph}
    \node[startnode]	(s) {$s$};
    \node[stdnode]	(u) 	[below=of s] 			{$u$};
    \node[stdnode]	(v1) 	[below left=of u] 		{$v_1$};
    \node[stdnode]	(v2) 	[below right=of u] 		{$v_2$};
    \node[endnode]	(r) 	[below right=of v1] 	{$r$};
    
    \path[->]	
	    (s) 	edge node {\texttt{x:=0}} 			(u)
	    (u) 	edge node [swap]{\texttt{x:=17}} 	(v1)
	    (u) 	edge node {\texttt{x:=42}} 			(v2)
	    (v1)	edge node [swap]{\texttt{skip}} 	(r)
	    (v2)	edge node {\texttt{skip}} 			(r)
    ;

  \end{flowgraph}
  \caption{Flussgraph zu \texttt{S} = \texttt{x:=0;(x:=17|x:=42)}.}
 \label{bild-uneind-wla-lsg}
\end{figure}

Es ist $\Var = \{x\}$ und $N = \{s,u,v_1,v_2,r\}$. 
Für eine natürliche Zahl $n$ ordnen wir der Zuweisung \texttt{x:=n} die Transferfunktion $f_{\mathtt{x:=n}} : L_\text{IA} \to L_\text{IA}$ mit $f_{\mathtt{x:=n}}(\rho)(x) = [n,n]$ zu. 
Nach Ausführen der Zuweisung hat die Variable $x$ also genau den ihr soeben zugewiesenen Wert.
Die Basisanweisung \texttt{skip} besagt, dass die Variablenwerte unverändert bleiben. Dazu korrespondiert die Identität $\id: L_\text{IA} \to L_\text{IA}$ als Transferfunktion.
Da zu Beginn der Ausführung noch nichts über den Wert von $x$ bekannt ist, ordnen wir jeder Variable im Startpunkt als Initialwert das größtmögliche Intervall zu. Dieses ist $[-\infty,+\infty]$.
Wir erhalten damit als Ungleichungssystem $\mathcal{U}$
\begin{align*}
\tag{$u_1$} \label{u1} I[s](x)	&\supseteq [-\infty,+\infty] \\
\tag{$u_2$} \label{u2} I[u](x)	&\supseteq f_{\mathtt{x:=0}}(I[s])=[0,0]\\
\tag{$u_3$} \label{u3} I[v_1](x)&\supseteq f_{\mathtt{x:=17}}(I[u])=[17,17] \\
\tag{$u_4$} \label{u4} I[v_2](x)&\supseteq f_{\mathtt{x:=42}}(I[u])=[42,42] \\ 
\tag{$u_5$} \label{u5} I[r](x)	&\supseteq f_{\mathtt{skip}}(I[v_1])=I[v_1] \\
\tag{$u_6$} \label{u6} I[r](x)	&\supseteq f_{\mathtt{skip}}(I[v_2])=I[v_2]
\end{align*}

Als Widening-Operator betrachten wir wieder den Operator aus Beispiel \ref{bsp-widening-IA-std}.

Wir berechnen nun eine Lösung des Ungleichungssystems $\mathcal{U}$, indem wir einen Workset-Algorithmus mithilfe des Widening-Operators $\widening$ ausführen.
Dazu initialisieren wir die Workset $W$ mit dem Ungleichungssystem und die Variablen in jedem Programmpunkt mit dem minimalen Element $\emptyset$ von $L_\text{Int}$. 
Es gilt also $W = \mathcal{U}$ und $I[\node](x) = \emptyset$ für alle $\node \in N$.

Wir betrachten nun zunächst die ersten vier Ungleichungen. 
Zu Ungleichung $\eqref{u1}$ setzen wir zunächst $t:= [-\infty,+\infty]$. Es ist nun $\emptyset \not\supseteq [-\infty,+\infty]$ und die Ungleichung ist nicht erfüllt. 
Wir berechnen also als neuen Wert $I[s](x) = \emptyset \widening [-\infty,+\infty] = [-\infty,+\infty]$ und erhalten
\begin{align*}
W = \{\text{\ref{u2}},\dots,\text{\ref{u6}}\} &\text{ und } 
  I[\node](x) = \begin{cases} 
	[-\infty,+\infty] &\text{für } \node=s \\ 
	\emptyset & \text{sonst} .
  \end{cases} 
\intertext{Ebenso ist Ungleichung $\eqref{u2}$ nicht erfüllt. Somit berechnen wir im Programmpunkt $u$ als neuen Wert $I[u](x) = \emptyset \widening [0,0] = [0,0]$ und erhalten}
W = \{\text{\ref{u3}},\dots,\text{\ref{u6}}\} &\text{ und } 
  I[\node](x) = \begin{cases} 
	[-\infty,+\infty] &\text{für } \node=s \\ 
	[0,0] &\text{für } \node=u \\ 
	\emptyset & \text{sonst} .
  \end{cases} 
\intertext{Die Ungleichungen $\eqref{u3}$ und $\eqref{u4}$ sind ebenfalls nicht erfüllt. Mit analogen Rechnungen wie in den ersten beiden Schritten erhalten wir zunächst}
W = \{\text{\ref{u4}},\text{\ref{u5}},\text{\ref{u6}}\} &\text{ und } 
  I[\node](x) = \begin{cases} 
	[-\infty,+\infty] &\text{für } \node=s \\ 
	[0,0] &\text{für } \node=u \\ 
	[17,17] &\text{für } \node=v_1 \\ 
	\emptyset & \text{sonst} 
  \end{cases}
\intertext{und dann}
W = \{\text{\ref{u5}},\text{\ref{u6}}\} &\text{ und } 
  I[\node](x) = \begin{cases} 
	[-\infty,+\infty] &\text{für } \node=s \\ 
	[0,0] &\text{für } \node=u \\ 
	[17,17] &\text{für } \node=v_1 \\ 
	[42,42] &\text{für } \node=v_2 \\ 
	\emptyset & \text{für } \node=r  .
  \end{cases}
\intertext{Die nächsten beiden Ungleichungen betrachten wir nun in verschiedenen Reihenfolgen. Zunächst beginnen wir mit Ungleichung $\eqref{u5}$. 
Auch diese ist nicht erfüllt und wir erhalten analog zu oben}
W = \{\text{\ref{u6}}\} &\text{ und } 
  I[\node](x) = \begin{cases} 
	[-\infty,+\infty] &\text{für } \node=s \\ 
	[0,0] &\text{für } \node=u \\ 
	[17,17] &\text{für } \node=v_1 \\ 
	[42,42] &\text{für } \node=v_2 \\ 
	[17,17] & \text{für } \node=r  .
  \end{cases}
\intertext{Auch die letzte Ungleichung ist nicht erfüllt und wir berechnen \[I[r](x) = [17,17] \widening [42,42]= [17,\infty].\] Damit erhalten wir}
W = \emptyset &\text{ und } 
  I[\node](x) = \begin{cases} 
	[-\infty,+\infty] &\text{für } \node=s \\ 
	[0,0] &\text{für } \node=u \\ 
	[17,17] &\text{für } \node=v_1 \\ 
	[42,42] &\text{für } \node=v_2 \\ 
	[17,\infty] & \text{für } \node=r  .
  \end{cases}
\intertext{Falls wir umgekehrt zunächst $\eqref{u6}$ betrachten, erhalten wir}
W = \{\text{\ref{u5}}\} &\text{ und } 
  I[\node](x) = \begin{cases} 
	[-\infty,+\infty] &\text{für } \node=s \\ 
	[0,0] &\text{für } \node=u \\ 
	[17,17] &\text{für } \node=v_1 \\ 
	[42,42] &\text{für } \node=v_2 \\ 
	[42,42] & \text{für } \node=r  .
  \end{cases}
\intertext{und nach Betrachtung von Ungleichung $\eqref{u5}$ schließlich}
W = \emptyset &\text{ und } 
  I[\node](x) = \begin{cases} 
	[-\infty,+\infty] &\text{für } \node=s \\ 
	[0,0] &\text{für } \node=u \\ 
	[17,17] &\text{für } \node=v_1 \\ 
	[42,42] &\text{für } \node=v_2 \\ 
	[-\infty,42] & \text{für } \node=r  .
  \end{cases}
\end{align*}
Die berechneten Werte im Programmpunkt $r$ unterscheiden sich also und sind sogar unvergleichbar. 
Somit hängt das Ergebnis des Workset-Algorithmus von der Strategie ab, nach der die Ungleichungen der Workset betrachtet werden.
\end{bsp}
Dieses Problem kann dadurch gelöst werden, dass man anstatt einer Lösung die gesamte Lösungsmenge berechnet 
und anschließend die Lösungsmengen der Ungleichungssysteme des funktionalen und des Call-String-Ansatzes vergleicht. 
Mögliche Ergebnisse dabei könnten sein, dass die eine Menge die andere enthält oder dass zu jedem Element der einen Menge ein kleineres in der anderen existiert.

\section{Nicht-Terminierung}
In \autoref{sec:widening-nicht-monoton} haben wir festgestellt, dass Widening-Operatoren im Workset-Al\-go\-rith\-mus im Allgemeinen in einer nicht-monotonen Weise verwendet werden. 
Dies führte dazu, dass die berechnete Lösung nicht eindeutig sein muss. Deshalb wollen wir die Mengen solcher Lösungen betrachten. 
Beim Berechnen dieser Lösungsmengen gibt es aber weitere Probleme, die wir im Folgenden untersuchen werden. 

Wie in \autoref{sec:wla} bereits gesehen, terminiert der Workset-Algorithmus nicht, wenn das betrachtete Ungleichungssystem unendlich viele Ungleichungen enthält. 
Im Allgemeinen werden für ein Programm unendlich viele Call-Strings berechnet. 
Das Ungleichungssystem \eqref{A-Ugs} aus dem funktionalen Ansatz enthält dann unendlich viele Ungleichungen und kann mit einem Workset-Algorithmus nicht gelöst werden. 

In Lemma \ref{lem:wsa-mit-widening-und-fairness-terminiert-fuer-jede-var} haben wir aber gesehen, dass für jede Variable in endlicher Zeit ein Wert berechnet wird, 
wenn wir im Workset-Algorithmus einen Widening-Operator und Fairness benutzen. 
Somit kann zwar in endlicher Zeit keine Lösung mehr berechnet werden, aber aus theoretischer Sicht ist dennoch vertretbar, 
auch für Call-String einen Workset-Algorithmus mit Widening und Fairness zu benutzten, 
wenn dieser als Referenzpunkt zum Vergleich der Lösungen des funktionalen und des Call-String-Ansatzes benutzt werden soll.

Auch die Ungleichungssysteme, die Transferfunktionen als Summary-Informationen berechnen bzw.~benutzen, sind nicht direkt lösbar. 
Um deren Lösung sinnvoll mit der des Ungleichungssystems aus dem Call-String-Ansatz vergleichen zu können, muss in beiden Fällen derselbe Widening-Operator im Workset-Algorithmus benutzt werden. 
Dazu wählt man zunächst einen Widening-Operator $\widening: L \times L \to L$, um mit diesem die beiden Ungleichungssysteme \eqref{R-Ugs} und \eqref{A-Ugs} zu lösen, 
welche die eigentlichen Informationen berechnen. Zunächst muss jedoch das Ungleichungssystem \eqref{T-Ugs} zur Berechnung der Summary-Informationen gelöst werden. 
Dieses benötigt nun einen Widening-Operator $\widening^\ast: (L \to L) \times (L \to L) \to (L \to L)$. 
Der kanonische Lift von $\widening$ zu einem Operator auf $(L \to L)$ ist aber im Allgemeinen kein Widening-Operator, wie das folgende Lemma zeigt.
\begin{lemma}\label{lem-widening-lift}
Sei $\widening: L \times L \to L$ ein Widening-Operator und 
\[\widening^\ast: (L \to L) \times (L \to L) \to (L \to L) \]
mit
\[(f \widening^\ast g)(l) := f(l) \widening g(l)\]
der natürliche Lift von $\widening$ auf $(L \to L)$.
Dann ist $\widening^\ast$ ein Extrapolationsoperator. Außerdem ist $\widening^\ast$ genau dann ein Widening-Operator, wenn für jede Folge $(f_n)_{n\in\nn}$ gilt:
\[m_0 := \sup\{m_l \mid l \in L\} < \infty.\]
Dabei sei 
\[h_n := \begin{cases}
f_1 & n=1 \\
h_{n-1} \widening^\ast f_n & \text{sonst}
\end{cases}\]
und $m_l := \min\{n \in \nn \mid \forall i\ge 0:h_{n+i}(l) = h_n(l)\}$.
\end{lemma}
\begin{proof}
Offenbar ist für jedes $x \in L$
\[ f(x) \sqleq f(x) \widening g(x) = (f \widening^\ast g)(x). \]
Ebenso ist $g(x) \sqleq (f \widening^\ast g)(x)$ für jedes $x \in L$. Also ist $\widening^\ast$ ein Extrapolationsoperator.

Sei nun $(f_n)_{n \in \nn}$ beliebige Folge in $(L \to L)$ und $(h_n)_{n \in \nn}$ wie oben definiert. 
Der geliftete Operator $\widening^\ast$ ist also ein Widening-Operator, wenn $(h_n)_{n \in \nn}$ irgendwann stabil wird. Dies wiederum ist genau dann der Fall, wenn
\[m_h := \inf\{m \in \nn \mid \forall i \ge 0:h_{m+i}=h_m\}<\infty.\] 
Ist $m_h$ unendlich, so wird die Folge nicht stabil. Andernfalls ist $m_h$ der Index, ab dem $(h_n)_{n \in \nn}$ stabil ist.

Nach Definition ist \[h_{n+1}(l) = h_{n}(l) \widening f_{n+1}(l).\] Da $\widening$ ein Widening-Operator ist, wird die so definierte Folge stabil. 
Also existiert \[m_l = \min\{n \in \nn \mid \forall i\ge 0:h_{n+i}(l) = h_n(l)\}\] und $m_l$ ist der Index, ab dem die Folge $(h_n(l))_{l \in L}$ stabil wird. 
Die Idee ist nun, dass die Folge $(h_n)_{n \in \nn}$ stabil wird, wenn $(h_n(l))_{n \in \nn}$ für jedes $l \in L$ stabil ist. 
Letzteres ist der Fall, wenn $m_0 = \sup\{m_l \mid l \in L\}$ endlich ist, ersteres, wenn $m_h$ endlich ist. Dies motiviert die Äquivalenz. Diese zeigen wir nun formal, in dem wir sogar $m_h=m_0$ zeigen.

Nach Definition von $m_0$ gilt $m_l \le m_0$ für jedes $l \in L$. Es gilt also 
\[h_{m_0+i}(l) = h_{m_0}(l)\]
für alle $l \in L$ und $i \ge 0$.
Es folgt $m_0 \ge m_h$.

Andererseits gilt auch $m_l \le m_h$ für alle $l \in L$. Also ist $m_h$ eine obere Schranke der Menge $\{m_l \mid l \in L\}$ und es folgt $m_0 \le m_h$.
\end{proof}
Ist der geliftete Operator $\widening^\ast$ kein Widening-Operator, so muss ein von $\widening$ unabhängiger Operator $\widening^\ast$ gewählt werden, um das Ungleichungssystem zu lösen. 
Dann ist aber im Allgemeinen nicht zu erwarten, die Lösungsmengen vom funktionalen und vom Call-String-Ansatz vergleichen zu können. 
Alternativ könnte zum Lösen dieses Ungleichungssystems kein Widening-Operator verwendet werden. In dem Fall muss allerdings im Allgemeinen auch auf Terminierung verzichtet werden.

\section{Nicht-Koinzidenz}
In Kapitel \ref{chap:PA} haben wir zwei Ansätze zur Analyse interprozeduraler Programme vorgestellt und gezeigt, 
dass diese die gleiche Lösung liefern, wenn die verwendeten Abbildungen gewissen Distributivitätseigenschaften genügen. 
Nun wäre wünschenswert, ein entsprechendes Resultat zu haben, wenn die Ungleichungssysteme mit einem Workset-Algorithmus mit Widening gelöst werden. 
Zur Herleitung des Koinzidenzresultates in Kapitel \ref{chap:PA} haben wir wesentlich die Monotonie der verwendeten Abbildungen benutzt. 
Dies geht nun nicht mehr, da ein Widening-Operator nicht monoton sein muss. 
So könnte beispielsweise im Rahmen der Intervallanalyse für einen dieser beiden Ansätze $[0,1] \widening [0,2] = [0,\infty]$ berechnet werden. 
Der andere Ansatz berechnet aber eventuell an einer entsprechenden Stelle den Wert $[0,2] \widening [0,2] = [0,2]$. 
Es ist also nicht möglich, so zu argumentieren, dass der eine Ansatz für eine Variable eine Kette berechnet, in der die Werte immer oberhalb der entsprechenden Kette im 
anderen Ansatz liegen. Diese \quotes{Sprünge} des Widening-Operators lassen also vermuten, dass es vielleicht unmöglich ist, wieder ein Koinzidenztheorem zu gewinnen. 
Wir werden im Folgenden zeigen, dass diese Vermutung stimmt. Weiterhin werden wir sehen, dass es auch unmöglich ist zu zeigen, 
dass die Lösungsmenge des einen Ansatzes immer in der des anderen enthalten ist. 
Dazu geben wir zwei einfache Beispiele an, die gleichzeitig alle der im Folgenden aufgelisteten Aussagen widerlegen.

Seien $\mathcal{L}_R$ bzw.~$\mathcal{L}_A$ die Lö\-sungs\-men\-gen des funktionalen bzw.~ des Call-String-Ansatzes, 
die vom Workset-Algorithmus mit Widening für ein zugrundeliegendes Framework $(L,\sqleq,\mathcal{F})$, ein Programm $G$ und einen Widening-Operator $\widening$ berechnet werden.

\textit{Koinzidenz:}
\begin{enumerate}
\item[(A)] Für alle monotonen Frameworks $(L,\sqleq,\mathcal{F})$, Widening-Operatoren $\widening$ und Flussgraphsysteme $G$ gilt
	  \[\mathcal{L}_A = \mathcal{L}_R.\]
\end{enumerate}
\textit{Inklusion: Die Lösungen des einen Ansatzes enthalten die des anderen.}
\begin{enumerate}
\item[(B)] Für alle monotonen Frameworks $(L,\sqleq,\mathcal{F})$, Widening-Operatoren $\widening$ und Flussgraphsysteme $G$ gilt
	  \[\mathcal{L}_A \supseteq \mathcal{L}_R.\]
\item[(C)] Für alle monotonen Frameworks $(L,\sqleq,\mathcal{F})$, Widening-Operatoren $\widening$ und Flussgraphsysteme $G$ gilt
	  \[\mathcal{L}_R \supseteq \mathcal{L}_A.\]
\end{enumerate}
\textit{Der eine Ansatz liefert nur bessere oder genauso gute Lösungen wie der andere.}
\begin{enumerate}
\item[(D)] Für alle monotonen Frameworks $(L,\sqleq,\mathcal{F})$, Widening-Operatoren $\widening$ und Flussgraphsysteme $G$ gilt
	  \[\forall x \in \mathcal{L}_A~\forall y \in \mathcal{L}_R:~x \sqleq y.\]
\item[(E)] Für alle monotonen Frameworks $(L,\sqleq,\mathcal{F})$, Widening-Operatoren $\widening$ und Flussgraphsysteme $G$ gilt
	  \[\forall x \in \mathcal{L}_R~\forall y \in \mathcal{L}_A:~x \sqleq y.\]
\end{enumerate}
\textit{Der eine Ansatz liefert eine Lösung, die besser ist oder genauso gut wie alle Lösungen des anderen.}
\begin{enumerate}
\item[(F)] Für alle monotonen Frameworks $(L,\sqleq,\mathcal{F})$, Widening-Operatoren $\widening$ und Flussgraphsysteme $G$ gilt
	  \[\exists x \in \mathcal{L}_A~\forall y \in \mathcal{L}_R:~x \sqleq y.\]
\item[(G)] Für alle monotonen Frameworks $(L,\sqleq,\mathcal{F})$, Widening-Operatoren $\widening$ und Flussgraphsysteme $G$ gilt
	  \[\exists x \in \mathcal{L}_R~\forall y \in \mathcal{L}_A:~x \sqleq y.\]
\end{enumerate}
\textit{Für jede Lösung des anderen findet der eine Ansatz eine Lösung, die besser oder genauso gut ist.}
\begin{enumerate}
\item[(H)] Für alle monotonen Frameworks $(L,\sqleq,\mathcal{F})$, Widening-Operatoren $\widening$ und Flussgraphsysteme $G$ gilt
	  \[\forall y \in \mathcal{L}_A~\exists x \in \mathcal{L}_R:~x \sqleq y.\]	  
\item[(I)] Für alle monotonen Frameworks $(L,\sqleq,\mathcal{F})$, Widening-Operatoren $\widening$ und Flussgraphsysteme $G$ gilt
	  \[\forall y \in \mathcal{L}_R~\exists x \in \mathcal{L}_A:~x \sqleq y.\]	  
\end{enumerate}
\textit{Für jede Lösung des einen Ansatzes gibt es im anderen eine Lösung, so dass die des ersten besser oder genauso gut ist.}
\begin{enumerate}
\item[(J)] Für alle monotonen Frameworks $(L,\sqleq,\mathcal{F})$, Widening-Operatoren $\widening$ und Flussgraphsysteme $G$ gilt
	  \[\forall x \in \mathcal{L}_R~\exists y \in \mathcal{L}_A:~x \sqleq y.\]	  
\item[(K)] Für alle monotonen Frameworks $(L,\sqleq,\mathcal{F})$, Widening-Operatoren $\widening$ und Flussgraphsysteme $G$ gilt
	  \[\forall x \in \mathcal{L}_A~\exists y \in \mathcal{L}_R:~x \sqleq y.\]	  
\end{enumerate}
In jedem dieser Fälle wäre der \quotes{eine} Ansatz dem \quotes{anderen} vorzuziehen.

Das erste Beispiel zeigt, dass es ein Programm, ein universell-distributives Framework und einen Widening-Operator gibt, 
so dass der Call-String Ansatz eine Lösung liefert, die zu allen Lösungen des funktionalen Ansatzes unvergleichbar ist. 
Dabei nutzen wir aus, dass zum Berechnen des Effektes einer Prozedur einmalig entschieden wird, in welcher Reihenfolge die entsprechenden Ungleichungen gelöst werden. 
Diese Ungleichungen korrespondieren, bis auf Prozeduraufrufe, mit denen vom Call-String-Ansatz. 
Gibt es nun mehrere Prozeduraufrufe, so berechnet der Call-String-Ansatz Informationen für die verschiedenen Kopien. Dabei können die Ausführungsreihenfolgen unabhängig voneinander gewählt werden. 
Werden sie in den einzelnen Kopien nun verschieden gewählt, so kann dies durch die Summary-Information der Prozedur nicht nachgebildet werden, 
da diese nur zu einer der möglichen Reihenfolgen korrespondiert.
\begin{bsp}\label{bsp1}
Wir betrachten den vollständigen Verband $(L, \sqleq)$, der durch das Hasse-Diagramm in Abbildung \ref{ggbsp1-verband} gegeben ist.
\begin{figure}[ht]\begin{center}\begin{tikzpicture}[node distance = 0.5cm]
\node	(top) 					{$\top$};
\node	(8) [below right=of top]{$l_8$};
\node	(7) [below left=of top]	{$l_7$};
\node	(6) [below right=of 7]	{$l_6$};
\node	(5) [below right=of 6]	{$l_5$};
\node	(4) [below left=of 6]	{$l_4$};
\node	(3) [below right=of 4]	{$l_3$};
\node	(2) [below =of 3]		{$l_2$};
\node	(1) [below =of 2]		{$l_1$};
\node	(bot) [below =of 1]		{$\bot$};
\node	(inv) [left=of top] 	{};
\node 	(L)	[left=of inv]		{$L=$};
\path[-]
  (top)	edge node {} (8)
  (top)	edge node {} (7)
  (8)	edge node {} (6)
  (7)	edge node {} (6)
  (6)	edge node {} (4)
  (6)	edge node {} (5)
  (5)	edge node {} (3) 
  (4)	edge node {} (3) 
  (3)	edge node {} (2)
  (2)	edge node {} (1)
  (1)	edge node {} (bot)
;
\end{tikzpicture}
\caption{Hasse-Diagramm des zugrundeliegenden Verbandes in Beispiel \ref{bsp1}.}
\label{ggbsp1-verband} 
\end{center}\end{figure}
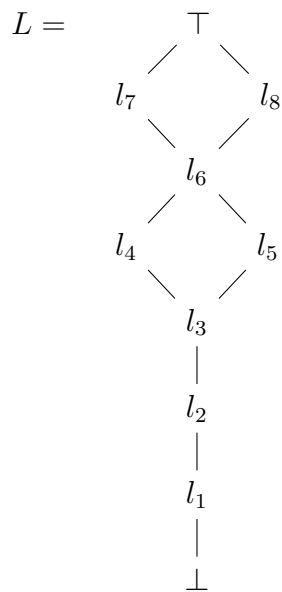

Definiere nun $f,g:L \to L$ durch
\begin{align*}
f(l) &= \begin{cases} 					
  \bot 	& \text{ falls } l=\bot \\
  l_1 	& \text{ falls } l=l_1 \\
  l_4 	& \text{ falls } l=l_2,l_3 \\
  \top 	& \text{ sonst}
\end{cases} &\text{bzw.}&&
g(l) &= \begin{cases} 					
  \bot 	& \text{ falls } l=\bot \\
  l_2 	& \text{ falls } l=l_1 \\
  l_3 	& \text{ falls } l=l_2 \\
  l_5 	& \text{ falls } l=l_3 \\
  \top 	& \text{ sonst}.
\end{cases} 
\end{align*}
Diese Abbildungen sind offenbar universell-distributiv. Mit $\mathcal{F}$ bezeichnen wir das kleinste unter Komposition abschlossene System von Abbildungen, das $\id_L$, $f$ und $g$ enthält. 
Da die Komposition universell-distributiver Abbildungen wieder universell-distributiv ist, ist $(L,\sqleq,\mathcal{F})$ ein universell-distributives monotones Framework.

Zuletzt definiere $\widening: L \times L \to L$ durch
\[
x \widening y = \begin{cases}
  l_3 & \text{ falls } (x,y) = (l_1,l_2) \\
  l_7 & \text{ falls } (x,y) = (l_3,l_4) \\
  l_7 & \text{ falls } (x,y) = (l_4,l_5) \\
  l_8 & \text{ falls } (x,y) = (l_5,l_4) \\
  x \sqcup y & \text{ sonst}.
\end{cases}
\]
Dies ist offenbar ein Extrapolationsoperator. Da $L$ endlich ist, ist $\widening$ sogar ein Widening-Operator. Außerdem ist $x \widening y = x$ für $y \sqleq x$ erfüllt. 
Also stimmt $\widening$ mit seiner idempotenten Variante $\widening_{\text{idp}}$ überein. 
Außerdem ist nach Lemma \ref{lem-widening-lift} der Lift $\widening^\ast$ von $\widening$ auf $(L\to L)$ wieder ein Widening-Operator.

Sei nun $\Var = \{x\}$ und das Programm gegeben als das Flussgraphsystem aus Abbildung \ref{ggbsp1-fgs}.
\begin{figure}[ht]\begin{flowgraph}
\node				(procmain)						{$\main$:};
\node[startnode]	(s) 		[below=of procmain]	{$s$};
\node[stdnode]		(u)			[below=of s] 	{$u$};
\node[endnode]		(r)			[below=of u]	{$r$};
\node				(inv1)		[right=of procmain]	{};
\node				(inv2)		[right=of inv1]	{};
\node				(inv3)		[right=of inv2]	{};
\node				(procp) 	[right=of inv3]		{$\p$:};
\node[startnode]	(sp) 		[below=of procp]	{$s_\p$};
\node[endnode]		(rp)		[below=of sp] 		{$r_\p$};
\path[->]
  (s) 	edge 					node {$\mathtt{p()}$} (u)
  (u)	edge 					node {$\mathtt{p()}$} (r)
  (sp)	edge [bend right,swap] 	node {$\mathtt{x:=f(x)}$} (rp)
  (sp)	edge [bend left] 		node {$\mathtt{x:=g(x)}$} (rp)  
;
\end{flowgraph}
\caption{Flussgraphsystem in Beispiel \ref{bsp1}.}
\label{ggbsp1-fgs}
\end{figure}

Wir stellen die zugehörigen Ungleichungssysteme auf, die den Wert von $x$ in den Programmpunkten bestimmen. Dabei sei $\init = l_1$. 
Weiter setzen wir $e_1 := (s,\mathtt{p()},u)$ und $e_2 := (u,\mathtt{p()},r)$. Das Ungleichungssystem zur Berechnung der Call-Strings hat folgende Gestalt:
\begin{align*}
\CS[\main] 	&\supseteq \{\eps\} \\
\CS[\p] 	&\supseteq \{e_1\} \\
\CS[\p] 	&\supseteq \{e_2\} 
\end{align*}
Dies hat als Lösung
\begin{align*} \underline{\CS}[\main] \supseteq \{\eps\} \text{ und } \underline{\CS}[\p] \supseteq \{e_1,e_2\}. \end{align*}
Das Ungleichungssystem zur Bestimmung des Wertes von $x$ in den Programmpunkten ist nun 
\begin{align*}
A[s,\eps]	&\sqgeq \init			\tag{a1} \label{bsp1:a1}	\\
A[s_\p,e_1]	&\sqgeq A[s,\eps]		\tag{a2} \label{bsp1:a2}	\\
A[r_\p,e_1]	&\sqgeq f(A[s_p,e_1])	\tag{a3} \label{bsp1:a3}	\\
A[r_\p,e_1]	&\sqgeq g(A[s_p,e_1])	\tag{a4} \label{bsp1:a4}	\\
A[u,\eps]	&\sqgeq A[r_\p,e_1]		\tag{a5} \label{bsp1:a5}	\\
A[s_\p,e_2]	&\sqgeq A[u,\eps] 		\tag{a6} \label{bsp1:a6}	\\
A[r_\p,e_2]	&\sqgeq f(A[s_p,e_2])	\tag{a7} \label{bsp1:a7}	\\
A[r_\p,e_2]	&\sqgeq g(A[s_p,e_2])	\tag{a8} \label{bsp1:a8}	\\
A[r,\eps]	&\sqgeq A[r_\p,e_2]		\tag{a9} \label{bsp1:a9}
\end{align*}
Eine mögliche Ausführung des Workset-Algorithmus mit Widening ist die folgende:
\begin{itemize}
\item Zunächst wird jeder Programmpunkt mit $\bot$ initialisiert. Die Workset ist nun $W = \{\eqref{bsp1:a1},\dots,\eqref{bsp1:a9}\}$. 
\item Dann wird Ungleichung \eqref{bsp1:a1} betrachtet. Diese ist nicht erfüllt, also wird $A[s,\eps]$ auf den Wert $\bot \sqcup \init = l_1$ gesetzt. 
Dann ist $W = \{\eqref{bsp1:a2},\dots,\eqref{bsp1:a9}\}$.
\item Nun ist Ungleichung \eqref{bsp1:a2} unerfüllt. Also wird $A[s_\p,e_1]$ auf den Wert $\bot \sqcup l_1 = l_1$ gesetzt. Es ist $W = \{\eqref{bsp1:a3},\dots,\eqref{bsp1:a9}\}$.
\item Als nächstes wird Ungleichung \eqref{bsp1:a3} betrachtet. Also wird $A[r_\p,e_1]$ auf den Wert $\bot \sqcup f(l_1) = l_1$ gesetzt. Es ist $W = \{\eqref{bsp1:a4},\dots,\eqref{bsp1:a9}\}$.
\item Ebenso ist Ungleichung \eqref{bsp1:a4} unerfüllt. Also wird für $A[r_\p,e_1]$ als neuer Wert $l_1 \widening g(l_1) = l_1 \widening l_2 = l_3$ berechnet. 
Damit ist $W = \{\eqref{bsp1:a5},\dots,\eqref{bsp1:a9}\}$.
\item Anschließend wird Ungleichung \eqref{bsp1:a5} betrachtet, die nicht erfüllt ist. Dann wird $A[u,\eps]$ auf $l_3$ gesetzt und es ist $W=\{\eqref{bsp1:a6},\dots,\eqref{bsp1:a9}\}$.
\item Damit ist Ungleichung \eqref{bsp1:a6} unerfüllt und es wird auch $A[s_\p,e_2]$ auf $l_3$ gesetzt. Es ist $W = \{\eqref{bsp1:a7},\eqref{bsp1:a8},\eqref{bsp1:a9}\}$.
\item Nun sind die Ungleichungen \eqref{bsp1:a7} und \eqref{bsp1:a8} unerfüllt. Wir betrachten diese in umgekehrter Reihenfolge wie für die erste Kopie und beginnen also mit der letzteren. 
Damit erhält $A[r_\p,e_2]$ den Wert $g(l_3) = l_5$ und es ist $W = \{\eqref{bsp1:a7},\eqref{bsp1:a9}\}$.
\item Danach wird \eqref{bsp1:a7} betrachtet. Auch diese Ungleichung ist nicht erfüllt. 
Also wird $A[r_\p,e_2]$ zu $l_5 \widening f(l_3) = l_5 \widening l_4 = l_8$ berechnet und es ist $W = \{\eqref{bsp1:a9}\}$.
\item Zuletzt wird also \eqref{bsp1:a9} betrachtet. Da diese Ungleichung wiederum nicht erfüllt ist, wird $A[r,\eps]$ auf den Wert $l_8$ gesetzt und die Workset ist leer.
\end{itemize}
Nun wird die Lösung $\hat{A}$ berechnet. Dazu wird für jeden Programmpunkt die kleinste obere Schranke derjenigen Werte bestimmt, die in den einzelnen Kopien berechnet wurden. 
Dann ist $\hat{A}$ gegeben durch
\begin{align*}
\hat{A}[s]	&= l_1 	& \hat{A}[s_\p]	&= l_1 \sqcup l_3 = l_3\\
\hat{A}[u]	&= l_3 	& \hat{A}[r_\p]	&= l_1 \sqcup l_8 = l_8\\
\hat{A}[r]	&= l_8. 	&&\\
\end{align*}
Als nächstes betrachten wir den funktionalen Ansatz. Die beiden Ungleichungssysteme hierzu sind gegeben durch 
\begin{align*}
T[s]	&\sqgeqmap \id 				\tag{t1} \label{bsp1:t1}	\\
T[u]	&\sqgeqmap T[r_\p] \circ T[s] \tag{t2} \label{bsp1:t2}	\\
T[r]	&\sqgeqmap T[r_\p] \circ T[u] \tag{t3} \label{bsp1:t3}	\\
T[s_p]	&\sqgeqmap \id				\tag{t4} \label{bsp1:t4}	\\
T[r_p]	&\sqgeqmap f \circ T[s_p] 	\tag{t5} \label{bsp1:t5}	\\
T[r_p]	&\sqgeqmap g \circ T[s_p] 	\tag{t6} \label{bsp1:t6}	
\end{align*}
sowie
\begin{align*}
R[s]	&\sqgeq \init 				\tag{r1} \label{bsp1:r1}	\\
R[u]	&\sqgeq T[r_\p](R[s])		\tag{r2} \label{bsp1:r2}	\\
R[r]	&\sqgeq T[r_\p](R[u])		\tag{r3} \label{bsp1:r3}	\\
R[s_p]	&\sqgeq R[s]				\tag{r4} \label{bsp1:r4}	\\
R[s_p]	&\sqgeq R[u]				\tag{r5} \label{bsp1:r5}	\\
R[r_p]	&\sqgeq f(R[s_p])			\tag{r6} \label{bsp1:r6}	\\
R[r_p]	&\sqgeq g(R[s_p])			\tag{r7} \label{bsp1:r7}.	
\end{align*}
Wir lösen zuerst das Ungleichungssystem zur Bestimmung der Transferfunktionen. 
Da genauer nur der Wert $T[r_\p]$ zum Lösen des zweiten Ungleichungssystems relevant ist, beschränken wir uns darauf, die hier möglichen Werte zu berechnen.
\begin{itemize}
\item Zunächst wird jede Variable mit der Abbildung initialisiert, die konstant den Wert $\bot$ liefert. Die Workset enthält alle Ungleichungen, also ist $W = \{\eqref{bsp1:t1},\dots,\eqref{bsp1:t6}\}$.
\item Unabhängig von der Ausführungsreihenfolge wird \eqref{bsp1:t4} nur einmal betrachtet: Da diese Ungleichung zu Beginn nicht erfüllt ist, muss sie mindestens einmal betrachtet werden. 
Dabei wird sie aus der Workset herausgenommen. Da die Ungleichung auf der rechten Seite keine Variable benutzt, wird sie zu keinem Zeitpunkt erneut der Workset hinzugefügt. 
Also wird sie genau einmal betrachtet und erhält den Wert $T[s_\p] = \id$.
\item Anschließend sind die Ungleichungen \eqref{bsp1:t5} und \eqref{bsp1:t6} nicht erfüllt. Diese werden ebenfalls genau einmal betrachtet. 
Im ersten Schritt ist $T[r_\p] = f$ oder $T[r_\p] = g$ und es bleibt, die jeweilige andere Ungleichung zu betrachten. 
\item Somit ist zuletzt $T[r_\p] \in \{f\widening^\ast g, g \widening^\ast f\}$.
\end{itemize}
Wir wenden uns nun dem Ungleichungssystem zur Bestimmung des Wertes von $x$ in den Programmpunkten zu. Sei $\underline{R}$ eine Lösung, die vom Workset-Algorithmus mit Widening berechnet wurde. 
Wir zeigen nun, dass $\underline{R}[r]$ und $\hat{A}[r]$ stets unvergleichbar sind:

Zu Beginn wird wieder jede Variable mit $\bot$ und die Workset mit der Menge aller Ungleichungen initialisiert. 
Nun ist \eqref{bsp1:r1} die einzige nicht erfüllte Ungleichung und $R[s]$ enthält den Wert $\bot \widening \init = l_1$. 
Mit derselben Argumentation wie zuvor folgt, dass diese Ungleichung nie wieder betrachtet wird. Anschließend ist die Ungleichung \eqref{bsp1:r2} unerfüllt. 
Also wird $R[u]= \underline{T}[r_\p](l_1)$ berechnet. 
Da diese Ungleichung nur die Variable $R[s]$ benutzt, diese aber nur in Ungleichung \eqref{bsp1:r1} auf der linken Seite vorkommt 
und letztere Ungleichung nach voriger Argumentation nie wieder betrachtet wird, wird auch \eqref{bsp1:r2} nie wieder betrachtet. 
Ebenso wird \eqref{bsp1:r3} genau einmal betrachtet. Damit ist dann $R[r] = \underline{T}[r_\p](\underline{T}[r_\p](l_1))$. Wir berechnen nun die möglichen Werte:
Falls $\underline{T}[r_\p]=f \widening^\ast g$, so ist 
\[
R[r]
= (f\widening^\ast g) (f(l_1) \widening g(l_1))
= (f\widening^\ast g) (l_1 \widening l_2)
= (f\widening^\ast g) (l_3)
= l_4 \widening l_5 = l_7.
\]
Andernfalls ist $\underline{T}[r_\p]=g \widening^\ast f$. Dann ist 
\[
R[r]
= (g\widening^\ast f) (g(l_1) \widening f(l_1))
= (g\widening^\ast f) (l_2 \widening l_1)
= (g\widening^\ast f) (l_2)
= l_3 \widening l_4 = l_7.
\]
In beiden Fällen ist also $\underline{R}[r]=l_7$. Dies ist unvergleichbar zu $l_8 = \hat{A}[r]$.
\end{bsp}
In dem Beispiel haben wir benutzt, dass der funktionale Ansatz einmalig den Effekt einer Prozedur festlegt und mit diesem dann Werte in Knoten der $\main$-Prozedur berechnet, 
während im Call-String-Ansatz Werte durch Propagation durch die Knoten der entsprechenden Prozedur bestimmt werden und dort die Ausführungsreihenfolgen voneinander verschieden gewählt werden konnten. 
Somit konnten wir unvergleichbare Werte im Endknoten der Prozedur $\main$ erzeugen. 

Das folgende Beispiel zeigt, dass es ein Programm, ein universell-distributives Framework und einen Widening-Operator gibt, 
so dass der funktionale Ansatz eine Lösung liefert, die zu allen Lösungen des Call-String-Ansatzes unvergleichbar ist.
Dazu finden wir eine Lösung des funktionalen Ansatzes, die mit keiner Lösung des Call-String-Ansatzes vergleichbar ist. 
Konkret werden wir dazu unvergleichbare Werte im Endknoten einer Prozedur $\p$ berechnen. 
Dazu müssen wir ausnutzen, dass die Informationen an den Knoten der aufgerufenen Prozedur in beiden Ansätzen auf verschiedene Weisen berechnet wird: 
Der Call-String-Ansatz betrachtet für jeden Aufruf eine eigene Kopie und initialisiert den Startknoten der Prozedur mit dem jeweiligen Wert des Knotens, von dem aus die Prozedur aufgerufen wird. 
Im funktionalen Ansatz dagegen gibt es keine verschiedenen Kopien und der Startknoten dieser Prozedur muss mit den Werten aller aufrufenden Knoten initalisiert werden. 
\begin{bsp}\label{bsp2}
Wir betrachten den vollständigen Verband $(L, \sqleq)$, der durch das Hasse-Diagramm in Abbildung \ref{ggbsp2-verband} definiert ist.
\begin{figure}[ht]\begin{center}\begin{tikzpicture}[node distance = 0.5cm]
\node	(top) 					{$\top$};
\node	(6) [below right=of top]{$l_6$};
\node	(5) [below left=of top]	{$l_5$};
\node	(4) [below right=of 5]	{$l_4$};
\node	(3) [below left=of 4]	{$l_3$};
\node	(2) [below =of 4]		{$l_2$};
\node	(1) [below right=of 4]	{$l_1$};
\node	(bot) [below =of 2]		{$\bot$};
\node	(inv) [left=of top] 	{};
\node 	(L)	[left=of inv]		{$L=$};
\path[-]
  (top)	edge node {} (6)
  (top)	edge node {} (5)
  (6)	edge node {} (4)
  (5)	edge node {} (4)
  (4)	edge node {} (3)
  (4)	edge node {} (2)
  (4)	edge node {} (1)
  (1)	edge node {} (bot)
  (2)	edge node {} (bot)
  (3)	edge node {} (bot)
;
\end{tikzpicture}
\caption{Hasse-Diagramm des zugrundeliegenden Verbandes in Beispiel \ref{bsp2}.}
\label{ggbsp2-verband} 
\end{center}\end{figure}
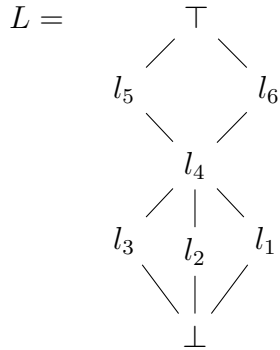

Definiere weiter $f_j: L \to L$ durch
\[ f_j(l) := \begin{cases} \bot & \text{ falls }l=\bot \\ l_j & \text{ sonst} \end{cases} \]
für $1 \le j \le 3$.
Diese Abbildungen sind offenbar universell-distributiv. 
Mit $\mathcal{F}$ bezeichnen wir das kleinste unter Komposition abschlossene System von Abbildungen, das $\id_L$, $h_1$, $h_2$ und $h_3$ enthält. 
Da die Komposition universell-distributiver Abbildungen wieder universell-distributiv ist, ist $(L,\sqleq,\mathcal{F})$ ein universell-distributives monotones Framework.

Ebenso ist die Abbildung $\widening: L \times L \to L$ mit
\begin{align*}
x \widening y = \begin{cases} 
l_5 & \text{ falls } (x,y) \in \{(l_1,l_3), (l_3,l_1), (l_2,l_3), (l_3,l_2)\} \\
l_6 & \text{ falls } (x,y) \in \{(l_1,l_2), (l_2,l_1)\} \\
x \sqcup y &\text{ sonst}
\end{cases}
\end{align*}
offenbar ein Extrapolationsoperator. Da $L$ endlich ist, ist $\widening$ sogar ein Widening-Operator. Weiter ist $\widening$ symmetrisch und erfüllt $x \widening y = x$, falls $y \sqleq x$. 
Also ist $\widening_{\text{idp}} = \widening$. Da $L$ endlich ist, ist nach Lemma \ref{lem-widening-lift} auch der Lift $\widening^\ast$ von $\widening$ auf $(L\to L)$ ein Widening-Operator.

Sei nun $\Var = \{x\}$. Wir betrachten das Flussgraphsystem aus Abbildung \ref{ggbsp2-fgs}.
\begin{figure}[ht]\begin{flowgraph}
\node				(procmain)						{$\main$:};
\node[startnode]	(s) 		[below=of procmain]	{$s$};
\node[stdnode]		(u)			[below left=of s] 	{$u$};
\node[stdnode]		(v)			[below right=of s] 	{$v$};
\node[endnode]		(r)			[below right=of u]	{$r$};

\node				(inv1)		[right=of procmain]	{};
\node				(inv2)		[right=of inv1]	{};
\node				(inv3)		[right=of inv2]	{};
\node				(procp) 	[right=of inv3]		{$\p$:};
\node[startnode]	(sp) 		[below=of procp]	{$s_\p$};
\node[endnode]		(rp)		[below=of sp] 		{$r_\p$};
\path[->]
  (s) 	edge [swap]				node {$\mathtt{x:=f_1(x)}$} (u)
  (s)	edge 					node {$\mathtt{x:=f_2(x)}$}  (v)
  (u)	edge [swap]				node {$\mathtt{p()}$}  (r)
  (v)	edge 					node {$\mathtt{p()}$}  (r)
  (sp)	edge [bend right,swap] 	node {$\mathtt{x:=f_3(x)}$} (rp)
  (sp)	edge [bend left] 		node {$\mathtt{x:=x}$} (rp)  
;
\end{flowgraph}
\caption{Flussgraphsystem in Beispiel \ref{bsp2}.}
\label{ggbsp2-fgs}
\end{figure}

Wir stellen nun die zugehörigen Ungleichungssysteme auf, die den Wert von $x$ in den Programmpunkten bestimmen. Dabei sei $\init = \top$. 
Schreibe im Folgenden kurz $e_1$ für $(u,\mathtt{p()},r)$ und $e_2$ für $(v,\mathtt{p()},r)$.

Die Ungleichungssysteme zum funktionalen Ansatz haben die folgende Gestalt:
\begin{align*}
T[s] &\sqgeqmap \id					& T[s_\p] &\sqgeqmap \id \\
T[u] &\sqgeqmap f_1 \circ T[s]		& T[r_\p] &\sqgeqmap f_3(T[s_\p]) \\
T[v] &\sqgeqmap f_2 \circ T[s]		& T[r_\p] &\sqgeqmap T[s_\p] \\
T[r] &\sqgeqmap T[r_\q] \circ T[u]	& \\
T[r] &\sqgeqmap T[r_\q] \circ T[v]	& 
\end{align*}
und als Lösung des letzteren Ungleichungssystems wählen wir
\begin{align*}
R[s] &\sqgeq \init			& R[s_\p] &\sqgeq R[u] \\
R[u] &\sqgeq f_1(R[s])		& R[s_\p] &\sqgeq R[v] \\
R[v] &\sqgeq f_2(R[s])		& R[r_\p] &\sqgeq f_3(R[s_\p])\\
R[r] &\sqgeq T[r_\q](R[u])	& R[r_\p] &\sqgeq R[s_\p]\\
R[r] &\sqgeq T[r_\q](R[v]).	& 
\end{align*}
Sei $\underline{T}$ eine vom Workset-Algorithmus mit Widening berechnete Lösung des ersteren Ungleichungssystems. Dabei sei 
\begin{align*}
\underline{T}[r_\p] &= f_3 \widening^\ast \id
\intertext{und}
\underline{R}[s] &= \init = \top \\
\underline{R}[u] &= l_1\\
\underline{R}[v] &= l_2\\
\underline{R}[s_\p] &= l_1 \widening l_2 = l_6 \\
\underline{R}[r_\p] &= f_3(l_6) \widening l_6 = l_3 \widening l_6 = l_6 \\
\underline{R}[r] &= ((f_3 \widening^\ast \id) (l_1)) \widening ((f_3 \widening^\ast \id) (l_2))) = (l_3 \widening l_1) \widening (l_3 \widening l_2)= l_5 \widening l_5 = l_5.
\end{align*}
Wir müssen nun die Lösungen des entsprechenden Workset-Algorithmus des Call-String-Ansatzes betrachten.
Das Ungleichungssystem zur Berechnung der Call-Strings einer Prozedur hat folgende Gestalt:
\begin{align*}
\CS[\main] 	&\supseteq \{\eps\}	\\
\CS[\p] 	&\supseteq \{e_1\} \\
\CS[\p] 	&\supseteq \{e_2\}. 
\end{align*}
Dies hat als Lösung 
\[\underline{\CS}[\main] = \{\eps\} \text{ und } \underline{\CS}[\p] = \{e_1,e_2\}.\]
Weiter ist 
\begin{align*}
A[s,\eps] 	&\sqgeq \init			& A[s_\p,e_1] &\sqgeq A[u,\eps] \\
A[u,\eps] 	&\sqgeq f_1(A[s,\eps])	& A[r_\p,e_1] &\sqgeq f_3(A[s_\p,e_1])\\
A[v,\eps] 	&\sqgeq f_2(A[s,\eps])	& A[r_\p,e_1] &\sqgeq A[s_\p,e_1]\\
A[r,\eps] 	&\sqgeq A[r_\p,e_1]		& A[s_\p,e_2] &\sqgeq A[v,\eps] \\
A[r,\eps] 	&\sqgeq A[r_\p,e_2]		& A[r_\p,e_2] &\sqgeq f_3(A[s_\p,e_1])\\
			&						& A[r_\p,e_2] &\sqgeq A[s_\p,e_2].
\end{align*}
Sei $\underline{A}$ eine Lösung, die der Workset-Algorithmus mit Widening berechnet hat. Wir zeigen nun $\underline{A}[r_\p](e_i) = l_5$ für $1 \le i \le 2$.
Jede Ausführung des Workset-Algorithmus liefert zunächst
\begin{align*}
\underline{A}[s,\eps] &= \top \\
\underline{A}[u,\eps] &= l_1 \\
\underline{A}[v,\eps] &= l_2 \\
\underline{A}[s_\p,e_1] &= l_1 \\
\underline{A}[s_\p,e_2] &= l_2.
\intertext{Da $\widening$ symmetrisch ist, ist irrelevant, in welcher Reihenfolge die Information an den Endknoten der Prozedur $\p$ beim jeweiligen Aufruf propagiert wird. Damit ist}
\underline{A}[r_\p,e_1] &= f_3(l_1) \widening l_1 = l_3 \widening l_1 = l_5 \intertext{und}
\underline{A}[r_\p,e_2] &= f_3(l_2) \widening l_2 = l_3 \widening l_2 = l_5.
\end{align*}
Es folgt $\hat{A}[r_\p] = \underline{A}[r_\p,e_1] \sqcup \underline{A}[r_\p,e_2] = l_5 \sqcup l_5 = l_5$, was unvergleichbar ist zu $l_6 = \underline{R}[r_\p]$. 
Also ist $\underline{R}$ unvergleichbar zu allen möglichen Lösungen $\hat{A}$.
\end{bsp}
Diese Beispiele sind einfach, da die betrachteten Programme und vollständigen Verbände sehr klein sind. 
Umso erstaunlicher ist es, dass die hier nicht zwingend notwendige Verwendung eines Widening-Operators schon zu Unvergleichbarkeit der Lösungenmengen des funktionalen bzw. des Call-String-Ansatzes führt. 
Die angegebenen Programme und universell-distributiven Frameworks können auch auf einfache Weise so abgeändert werden, dass der Workset-Algorithmus ohne Widening nicht terminiert. 

Dazu erweitern wir den Verband um Elemente $\lambda_n$ für $n \in \nn_0$. 
Dabei \quotes{ersetzen} wir das Element $\top$ durch $\lambda_0$ und erweitern die Ordnung um $\lambda_n \sqsupset \lambda_{n+1}$ für jedes $n \in \nn$. 
Wir fügen wiederum $\top$ als ein größtes Element hinzu. 
Wir erweitern die Variablenmenge um eine frische Variable $y$ und fügen eine Prozedur $\q$ hinzu, die aus den aufeinanderfolgenden Anweisungen $\mathtt{y:=k(y)}$ und $\mathtt{q()}$ besteht. 
Dabei sei $k(l) = \bot$ für Elemente aus dem \quotes{alten} Verband und sonst $k(\lambda_n) = \lambda_{n+1}$. 
Außerdem verändern wir die Prozedur $\main$, indem wir an den Knoten $r$ eine mit $\mathtt{y:=0}$ beschriftete Kante und daran eine mit dem Aufruf $\mathtt{q()}$ beschriftete Kante anhängen. 
Der Knoten, zu dem diese Kanten führen, ist dann der neue Endknoten der Prozedur $\main$. 
Den Widening-Operator setzen wir folgendermaßen fort: 
Für die Elemente aus dem \quotes{alten} Verband und $\lambda_0$ sei der Operator wie zuvor definiert, wobei $\lambda_0$ die Rolle von $\top$ übernimmt. 
Sei nun $x$ beliebig und $y = \lambda_n$ für ein $n > 0$ oder $y = \top$. Ist $x \sqgeq y$, so sei $x \widening y = x$. Andernfalls sei $x \widening y = \top$. 
Dies ist offenbar ein Widening-Operator.

Will man nun zu einem so erweiterten Programm Informationen mit dem Workset-Algorithmus berechnen, so terminiert die Berechnung für $y$ ohne Widening-Operatoren nicht. 
Um eine Lösung zu berechnen, ist also ein Widening-Operator nötig. Damit wird für $y$ der Wert $\top$ berechnet. Für $x$ wird derselbe Wert wir im ursprünglichen Programm berechnet. 
Damit erhalten wir eine entsprechende Unvergleichbarkeit also auch in Situationen, in denen Widening nötig ist. 

Insbesondere zeigen diese Beispiele also, dass die zu Beginn des Abschnittes formulierten Aussagen (A) bis (K) in dieser Allgemeinheit falsch sind.


\chapter{Fazit}\label{chap:fazit}
In dieser Arbeit wurden zwei Ansätze zur Analyse interprozeduraler Programme vorgestellt, die von Sharir und Pnueli in \cite{sharir-pnueli} entwickelt wurden: 
Dies ist zum einen der \emph{funktionale Ansatz}, der zunächst für jede Prozedur eine Transferfunktion als \emph{Summary-Information} berechnet und damit für jeden Programmknoten eine Information bestimmt. 
Zum anderen haben wir den \emph{Call-String-Ansatz} vorgestellt, der eine Analyse eines Programms mit nur einer Prozedur simuliert, 
indem Aufrufkanten durch die aufgerufenen Prozeduren \quotes{ersetzt} und diese verschiedenen \quotes{Kopien} der Prozeduren durch \emph{Call-Strings} unterschieden werden.

Dazu haben wir zunächst 
\emph{Ungleichungssysteme} definiert und gezeigt, 
dass deren Lösungen sich als Präfixpunkte einer geeigneten Abbildung schreiben lassen und die \quotes{beste} Lösung dem kleinsten Fixpunkt dieser Abbildung entspricht. 
Außerdem haben wir Flussgraphsysteme eingeführt, mit denen wir Programme darstellen, und mit dem Workset-Algorithmus ein Mittel bereit gestellt, mit dem Ungleichungssysteme gelöst werden können. 

Anschließend 
haben wir den \emph{funktionalen} und den \emph{Call-String-Ansatz} vorgestellt und gezeigt, 
dass es sich dabei um \emph{korrekte Analysen} handelt. Wir haben außerdem gezeigt, dass die Ansätze \emph{präzise Analysen} beschreiben, wenn die verwendeten Abbildungen universell-distributiv 
oder zumindest positiv-distributiv sind und in letzterem Fall zusätzlich jeder Programmpunkt erreichbar ist. In diesem Fall liefern beide Ansätze dasselbe Ergebnis und sind in diesem Sinne gleichwertig. 
Weiter haben wir mit der Betrachtung der Polyederanalyse aus \cite{CH78-POPL} begonnen, die den funktionalen Ansatz variiert, 
indem die Summary-Informationen nicht als Transferfunktionen, sondern in einer anderen Darstellung berechnet werden. 
Dabei wurden einerseits Matrizenmengen wie in \cite{seidl07} und andererseits Relationen wie in \cite{CH78-POPL} als Summary-Informationen benutzt. 
Diese wiederum haben wir als Transferfunktionen aufgefasst und damit jeweils ein Ungleichungssystem zur Berechnung die erreichbaren Variablenwerte in jedem Programmpunkt aufgestellt. 
Wir haben gezeigt, dass diese beiden Ungleichungssysteme das gleiche Ergebnis liefern und somit gleichwertig sind. 
Zuletzt haben wir eine Berechnung eines Ungleichungssystems mithilfe des Workset-Algorithmus betrachtet und festgestellt, dass diese nicht unbedingt terminieren muss.

Darum haben wir 
weiter mit der \emph{abstrakten Interpretation} ein von Cousot und Cousot in \cite{CousotCousot76-1} vorgestelltes Mittel betrachtet, 
mit dem neue korrekte oder präzise Analysen konstruiert werden können, bei denen die Berechnung der Lösung durch den Workset-Algorithmus terminieren kann. 
Dazu haben wir zunächst allgemein \emph{Galois-Verbindungen} und abstrakte Interpretationen von Ungleichungssystemen studiert. 
Diese haben wir anschließend auf den funktionalen und den Call-String-Ansatz übertragen, indem wir die in den beiden Ansätzen verwendeten Transferfunktionen für Basisanweisungen abstrakt interpretiert haben.
Wir haben außerdem gezeigt, dass eine korrekt, präzise oder kanonische abstrakte Interpretation dieser Transferfunktionen 
eine korrekte, präzise bzw.~kanonische abstrakte Interpretation der beiden Ansätze induziert. 
Dabei hängt die kanonische abstrakte Interpretation nur von der zugrundeliegenden Galois-Verbindung ab und ist die präziseste korrekte abstrakte Interpretation. Außerdem ist sie genau 
dann präzise, wenn die Galois-Verbindung präzise abstrakte Interpretationen zulässt. Wir haben gesehen, dass dies nicht immer der Fall ist. 
Somit konnten wir abschließend feststellen, dass die kanonische die beste abstrakte Interpretation ist. 
Ist sie präzise und gelten die Anforderungen an die Distributivität der verwendeten Abbildungen wie zuvor, 
so liefern die kanonischen abstrakten Interpretationen des funktionalen und des Call-String-Ansatzes dasselbe Ergebnis und sind wiederum gleichwertig.

Weiter haben wir wieder die Polyederanalyse untersucht. Dabei haben wir zunächst rückblickend festgestellt, 
dass das Auffassen einer Matrizenmenge oder einer Relation als Transferfunktion eine präzise abstrakte Interpretation ist. 
Außerdem haben wir die Polyederanalyse weitergeführt, indem wir zusätzlich \emph{konvexe Hüllen} gebildet haben, 
um anstelle der erreichbaren Variablenwerte konvexe Polyeder zu berechnen, in denen die erreichbaren Variablenwerte enthalten sind. 
Dabei haben wir gesehen, dass dies eine präzise abstrakte Interpretation ist, wenn Matrizenmengen als Summary-Informationen verwendet werden, 
wohingegen die Bildung konvexer Hüllen bei Relationen als Summary-Informationen zwar eine korrekte, aber nicht unbedingt präzise abstrakte Interpretation ist. 
Die Ansätze sind also nicht gleichwertig und der Ansatz mit Matrizenmengen ist dem mit Relationen vorzuziehen. 
Zuletzt haben wir gesehen, dass eine solche abstrakte Interpretation noch keine Terminierung garantieren muss.

Zuletzt 
haben wir darum ein weiteres Mittel untersucht, das Terminierung des Workset-Algorithmus erzwingen soll. 
Dieses Mittel ist die Verwengun der von Cousout und Cousot in \cite{CousotCousot77-1} eingeführten \emph{Widening-Operatoren} im Workset-Algorithmus. 
Dabei wird Terminierung garantiert, wenn das betrachtete Ungleichungssystem nur endlich viele Variablen enthält. 
Werden dagegen unendlich viele Ungleichungen betrachtet, berechnet der Algorithmus unter der zusätzlichen Voraussetzung von \emph{Fairness} für jede Variable in endlicher Zeit ein Ergebnis, 
terminiert aber nicht. 
Somit kann der Workset-Algorithmus für den Call-String-Ansatz, der im Allgemeinen unendlich viele Ungleichungen betrachtet, nicht zum Lösen, sondern nur als Referenzpunkt betrachtet werden. 
Außerdem haben wir gesehen, dass der Workset-Algorithmus mit Widening auch für den funktionalen Ansatz keine Terminierung garantiert, 
da im Allgemeinen aus einem Widening-Operator für den zugrundeliegenden Verband kein kanonischer Widening-Operator für 
die Berechnung der Transferfunktionen als Summary-Informationen gewonnen werden kann. 
Außerdem haben wir gezeigt, dass im Allgemeinen der Widening-Operator im Workset-Algorithmus in einer nicht-monotonen Weise verwendet ist und somit keine eindeutige kleinste Lösung existiert. 
Somit mussten anstelle von kleinsten Lösungen die gesamten Lösungsmengen verglichen werden. 
Wir haben aber gezeigt, dass die Lösungsmengen, die vom Workset-Algorithmus mit Widening für den funktionalen bzw.~den Call-String-Ansatz berechnet werden, im Allgemeinen nicht vergleichbar sind. 
Eine Koinzidenzaussage wie 
zuvor 
ist also nicht möglich. 
Ein Workset-Algorithmus mit Widening kann demnach in manchen Situationen zwar helfen, ein Ungleichungssystem zu lösen, aber es ist nicht möglich, Aussagen über die Präzision dieser Lösungen zu treffen.

Es sind mehrere Möglichkeiten denkbar, auf diese Arbeit aufzubauen. 
Zunächst einmal könnten weitere Anweisungen wie nicht-deterministische Zuweisungen der Gestalt $\mathtt{x := ?}$ und affine bedingte Verzweigungen wie in \cite{seidl07} betrachtet werden. 
Es könnte dann überprüft werden, ob mit diesen zusätzlichen Anweisungen der Ansatz mit Matrizenmengen wieder das präzisere Ergebnis liefert. 
Außerdem wird in \cite{seidl07} ein Widening-Operator vorgestellt, um die berechneten konvexen Matrizenmengen endlich darstellen zu können. Dieser könnte näher untersucht werden.

Aber auch allgemein wäre eine weitere Untersuchung von Widening-Operatoren interessant. 
Die Bedingungen, die in dieser Arbeit an Widening-Operatoren gestellt wurden, sind für eine Aussage über die Vergleichbarkeit des funktionalen und des Call-String-Ansatzes zu allgemein. 
Es könnte also untersucht werden, ob und welche Eigenschaften von Widening-Operatoren es gibt, die eine solche Vergleichbarkeit garantieren. 
Alternativ könnten andere Mittel als der Workset-Algorithmus zur Berechnung der Lösungen eines Ungleichungssystems untersucht 
und dabei wiederum die Ergebnisse des funktionalen und des Call-String-Ansatzes verglichen werden.



\begin{thebibliography}{NNH99}

\bibitem[CC76]{CousotCousot76-1}
Patrick Cousot and Radhia Cousot.
\newblock Static determination of dynamic properties of programs.
\newblock In {\em Proceedings of the Second International Symposium on
  Programming}, pages 106--130. Dunod, Paris, France, 1976.

\bibitem[CC77]{CousotCousot77-1}
Patrick Cousot and Radhia Cousot.
\newblock Abstract interpretation: a unified lattice model for static analysis
  of programs by construction or approximation of fixpoints.
\newblock In {\em Conference Record of the Fourth Annual ACM SIGPLAN-SIGACT
  Symposium on Principles of Programming Languages}, pages 238--252, Los
  Angeles, California, 1977. ACM Press, New York, NY.

\bibitem[CH78]{CH78-POPL}
Patrick Cousot and Nicolas Halbwachs.
\newblock Automatic discovery of linear restraints among variables of a
  program.
\newblock In {\em Conference Record of the Fifth Annual ACM SIGPLAN-SIGACT
  Symposium on Principles of Programming Languages}, pages 84--97, Tucson,
  Arizona, 1978. ACM Press, New York, NY.

\bibitem[LNS82]{folk-tale}
J.-L. Lassez, V.~L. Nguyen, and E.~A. Sonenberg.
\newblock Fixed point theorems and semantics: A folk tale.
\newblock {\em Information Processing Letters}, 14(3):112--116, 1982.

\bibitem[NNH99]{Niel}
Flemming Nielson, Hanne~Riis Nielson, and Chris Hankin.
\newblock {\em Principles of Program Analysis}.
\newblock Springer-Verlag New York, Inc., 1999.

\bibitem[SFP07]{seidl07}
Helmut Seidl, Andrea Flexeder, and Michael Petter.
\newblock Interprocedurally analysing linear inequality relations.
\newblock In {\em Programming Languages and Systems}, volume Volume 4421/2007
  of {\em Lecture Notes in Computer Science}, pages 284--299. Springer Berlin /
  Heidelberg, 2007.

\bibitem[SP81]{sharir-pnueli}
Micha Sharir and Amir Pnueli.
\newblock Two approaches of interprocedural data flow analysis.
\newblock In Steven~S. Muchnick and Neil~D. Jones, editors, {\em Program flow
  analysis: theory and applications}, chapter~7, pages 189--233. Prentice Hall,
  Englewood Cliffs, New Jersey, 1981.

\bibitem[WM92]{Wilhelm-Maurer}
Reinhard Wilhelm and Dieter Maurer.
\newblock {\em {\"U}bersetzerbau - Theorie, Konstruktion, Generierung}.
\newblock Springer, 1992.

\end{thebibliography}


\cleardoublepage
\thispagestyle{empty}
\quad 
\vfill

Hiermit erkläre ich, dass ich die vorliegende Diplomarbeit selbständig verfasst und keine anderen als die angegebenen Quellen und Hilfsmittel benutzt sowie alle Stellen, 
die wörtlich oder sinngemäß aus Veröffentlichungen entnommen worden sind, als solche kenntlich gemacht habe.
 
\vspace{0.75cm}

Münster, im Dezember 2010 

\vspace{0.25cm}

Dorothea Jansen

\quad 


\end{document}